\documentclass[12pt,twoside]{report}

\usepackage[a-1b]{pdfx} % Need this to create a PDF/A file

\usepackage{pdfpages}

\usepackage{amsmath}
\usepackage{amsfonts}
\usepackage{amssymb}
\usepackage{amsthm}
\usepackage{mathrsfs}
\usepackage{mathtools}
\usepackage{nccmath}
\usepackage{bbm}
\usepackage{bigdelim}
\usepackage{rotating}

\usepackage{graphicx}
\usepackage{xcolor}
\usepackage{stfloats}
\usepackage{subcaption}
\usepackage[font=small,labelfont=bf]{caption}

\usepackage{pifont}
\usepackage{microtype}

\usepackage{booktabs}
\usepackage{adjustbox}
\usepackage{makecell}
\usepackage{multirow}

\usepackage{flushend}
\usepackage{enumitem}

\usepackage[linesnumbered,ruled,vlined]{algorithm2e}
\usepackage{listings}

\SetCommentSty{mycommfont}

\newtheorem{theorem}{Theorem}[section]

\newtheorem{lemma}[theorem]{Lemma}

\theoremstyle{definition}
\newtheorem{definition}{Definition}[section]

\newenvironment{sketch}{%
  \proof}{\endproof}

\DeclarePairedDelimiter{\nint}\lfloor\rceil

\SetKwInput{KwInput}{Input}                % Set the Input
\SetKwInput{KwOutput}{Output}              % set the Output
\SetKwInput{KwInit}{Init.}              % set the Initialization

\newcommand{\cmark}{\ding{51}}%
\newcommand{\xmark}{\ding{55}}%

\fnbelowfloat 

\newdimen \jot \jot=5mm
\allowdisplaybreaks

\newenvironment{Tabular}[2][1]
  {\def\arraystretch{#1}\tabular{#2}}
  {\endtabular}

\usepackage[T1]{fontenc}
\usepackage{lmodern} % load a font with all the characters
\usepackage{tocbibind} % need this to contents adding for TOC
\usepackage{setspace}
\usepackage{RJournal_nogeom} % Changes the colors of links among other things
\usepackage[hypcap=true]{caption}
\hypersetup{linktocpage}
\usepackage{array}

\usepackage{color}
\makeatletter
\def\maxwidth{ %
  \ifdim\Gin@nat@width>\linewidth
    \linewidth
  \else
    \Gin@nat@width
  \fi
}
\makeatother

\definecolor{fgcolor}{rgb}{0.345, 0.345, 0.345}
\usepackage{framed}
\makeatletter
 {\par\unskip\endMakeFramed%
 \at@end@of@kframe}
\makeatother

\definecolor{shadecolor}{rgb}{.97, .97, .97}
\definecolor{messagecolor}{rgb}{0, 0, 0}
\definecolor{warningcolor}{rgb}{1, 0, 1}
\definecolor{errorcolor}{rgb}{1, 0, 0}
\makeatletter
\newcommand\gobblepars{%
    \@ifnextchar\par%
        {\expandafter\gobblepars\@gobble}%
        {}}
\makeatother

\definecolor{mygreen}{rgb}{0,0.6,0}
\definecolor{mygray}{rgb}{0.5,0.5,0.5}
\definecolor{mymauve}{rgb}{0.58,0,0.82}

\usepackage[
style = numeric, 
sorting = none, 
maxbibnames = 99,
backend = bibtex,
natbib = true
]{biblatex}

\AtEveryBibitem{
\clearfield{note}
\clearfield{month}
}

\graphicspath{{rnw_chapter/figure/}{rnw_chapter/}} % change it accordingly!

\begin{document}

\newcommand{\bm}[1]{ \mbox{\boldmath $ #1 $} }
\newcommand{\bin}[2]{\left(\begin{array}{@{}c@{}} #1 \\ #2
             \end{array}\right) }
\renewcommand{\contentsname}{Table of Contents}
\baselineskip=24pt
 
% Create cover page of dissertation !
\pagenumbering{roman}
\thispagestyle{empty}
\begin{center}
\vspace*{.25in}
{\bf\LARGE{ LISTENING TO MULTI-TALKER CONVERSATIONS: MODULAR AND END-TO-END PERSPECTIVES }}\\ % change it accordingly, title must be in ALL CAPS!
\vspace*{.75in}
{\bf by} \\*[18pt]
\vspace*{.2in}
{\bf Desh Raj}\\ % change it accordingly!
\vspace*{1in}
{\bf A dissertation submitted to Johns Hopkins University\\
in conformity with the requirements for the degree of\\
Doctor of Philosophy }\\
\vspace*{.75in}
{\bf Baltimore, Maryland} \\
{\bf February, 2024} \\     % change it accordingly! Put month of submission, not of thesis defense (if different)
\vspace*{.5in}
\begin{small}
{\bf \copyright{ }2024 Desh Raj} \\ % change the year if needed!
{\bf All rights reserved}
\end{small}
\end{center}
\newpage 

% Add acknowledgements
\pagestyle{plain}
\pagenumbering{roman}
\setcounter{page}{2}

\chapter*{Abstract}
\addcontentsline{toc}{chapter}{Abstract}%

Since the first speech recognition systems were built more than 30 years ago, improvement in voice technology has enabled applications such as smart assistants and automated customer support. 
Conversational intelligence of the future is expected to move beyond single-user applications of voice technologies to actively participate in human conversations, including in scenarios such as note-taking, fact-checking, or collaborative learning in peer groups. 
For such systems, recognizing free-flowing multi-party conversations is a crucial and challenging component that still remains unsolved.
In this dissertation, we focus on this problem of \textit{speaker-attributed multi-talker speech recognition} for the meeting transcription task, and propose two perspectives which result from its probabilistic formulation.

In the modular perspective, speaker-attributed transcription is performed through a pipeline of sub-tasks involving speaker diarization, target speaker extraction, and speech recognition.
Our first contribution is a novel method to perform overlap-aware speaker diarization by reformulating spectral clustering as a constrained optimization problem.
We also describe an algorithm to ensemble diarization outputs, and show that it can be used to either combine several overlap-aware systems, or to perform multi-channel diarization by late fusion.
Once speaker segments are identified, we robustly extract single-speaker utterances from the mixture using a GPU-accelerated implementation of guided source separation.
This eventually allows us to use an off-the-shelf ASR system to obtain speaker-attributed transcripts.

Since the modular approach suffers from error propagation, we propose an alternate ``end-to-end'' perspective on the problem.
For this, we describe the Streaming Unmixing and Recognition Transducer (SURT) which extends neural transducers for multi-talker ASR by incorporating an unmixing component.
We show how to train SURT models efficiently by carefully designing the network architecture, objective functions, and mixture simulation techniques.
Finally, we add an auxiliary speaker branch to enable joint prediction of speaker labels synchronized with the speech tokens, and propose a novel speaker prefixing approach for ensuring label consistency through the recording.
We demonstrate that training on synthetic mixtures and adapting with real data helps these models transfer well for streaming transcription of real meeting sessions.

\vspace*{\fill}

\noindent \textbf{Primary reader}: Sanjeev Khudanpur\\
% \textbf{Secondary reader}:

\subsection*{Thesis Committee}

\noindent
Sanjeev Khudanpur \\
Daniel Povey (Xiaomi Inc., Beijing) \\
Jinyu Li (Microsoft Corp., Redmond)
\chapter*{Acknowledgments}
\addcontentsline{toc}{chapter}{Acknowledgements}%

When I started on this Ph.D. journey in 2018, I had no inkling of how much of a life-altering process it would turn out to be.
These last five years have been absolutely exceptional, and as with all things magical, I split my acknowledgments for this dissertation into seven parts.

% 1. Parents, family, brother
First, I am indebted to my parents for making immense sacrifices such that I get the best education possible, which in turn enabled me to come further than I had imagined.
I thank my grand-parents for always pushing me to dream bigger, even if it came at the cost of being away for years on end.
I thank my brother, Abhishek Raj, who has been an inspiration my whole life --- if I have been able to sail in high winds, it's because I know I have a light to guide me home.
I thank my sister, Shanu Amit Srivastava, who was my first teacher and instilled in me a love and thirst for learning.

% 2. Sanjeev and Thesis committee
Second, I am infinitely grateful to my supervisor, Sanjeev Khudanpur, for the constant support and advice over the last several years.
Besides teaching me how to think critically about research questions, I have also learnt from you the importance of going out of my comfort zone, making strong connections, and enjoying the process of research.
I thank my second advisor, Dan Povey, from whom I learnt the importance of well-written, open-source code, and which has since become a guiding principle of my research.
I will also forever be grateful to Dan for offering me a PhD position in his group at a time when I had no background in speech research. 
I thank Jinyu Li for not just being a great mentor, but also for taking out time from managing a large industry group to be a part of my qualification and thesis committees.

% 3. Other advisors: Dan, Yenda, Shinji, Paola, Piotr + Internship mentors at Microsoft and Meta 
I have been fortunate to have had several unofficial advisors during my time at JHU, each of whom taught me several things.
Shinji Watanabe showed me the importance of asking big picture questions (in fact, my dissertation research branched off of working with Shinji on the CHiME-6 project) and of staying humble in success.
From Paola Garcia, I learnt that caring about the researcher is just as important as caring about their research.
Paola also taught me the fundamentals of speaker diarization, which forms several chapters of this dissertation.
I thank Jan ``Yenda'' Trmal for his expertise with Kaldi, and Piotr Zelasko for helping me become a better Python developer.
During this Ph.D., I spent two summers at Microsoft and Meta, respectively.
I thank Jinyu Li, Liang Lu, Zhuo Chen, and Naoyuki Kanda at Microsoft for their guidance in initiating me into the ``end-to-end'' multi-talker methods.
At Meta, I thank Ozlem Kalinli for hiring me into the amazing speech team, and Junteng Jia, Chunyang Wu, and Jay Mahadeokar for giving me the freedom to try and fail.

% 4. Labmates: old and current, and CLSP friends
My fourth acknowledgment is reserved for friends, the ones I had coming in and the ones I made along the way.
I thank my ``B3 family'' for always having my back --- many a Sunday afternoon was spent Zoom-ing with this band of brothers spread across 5 countries even before the pandemic made it cool.
I thank my labmates in the ``Kaldi group'' for the fruitful discussions and gossip we shared in Hackerman 322.
I had zero knowledge about speech processing when I joined the program, and much of what I have learnt I owe to David Snyder, Vimal Manohar, Matthew Wiesner, Matt Maciejewski, Yiming Wang, and Ke Li.
I am indebted to Hainan Xu and Xiaohui Zhang who were incredibly helpful this past year as I was looking for full-time positions.
I have also gained much from the friendships of Ashish Arora, Dongji Gao, Zili Huang, Ruizhe Huang, Yiwen Shao, Fei Wu, Amir Hussein, Jonathan Chang, Henry Li, and Cihan Xiao.
For Jinyi Ondel Yang, I reserve a special place in my heart --- I cherish our long walks that made the pandemic a little more sufferable, our expeditions for Indian, Chinese, and Vietnamese food, and our mutual love for the French language (and its speakers).
Outside the Kaldi group, I thank Elias Stengel-Eskin, David Mueller, Aaron Mueller, Kelly Marchisio, Suzanna Sia, Mitchell Gordon, Jacob Buckman, and Craig Guo for (often beverage-infused) conversations about big ideas in machine learning and beyond.
I thank Samik Sadhu for hosting a number of potluck parties, and Xuan Zhang for accompanying me to them.

% 5. CLSP and CS admin
My fifth vote of thanks is owed to the CS and CLSP administrative staff who ensured that all the technicalities were met so that I could focus on the research.
For this, I tip my hat to Ruth Scally, Lauren Meek, and Kim Franklin for their continued support with all things admin.
I also thank Joe McKnight at the HLTCOE for always being proactive about resource allocation on that cluster, without which these experiments could not have been performed.

% 6. All co-authors
This dissertation is the result of several publications which would not have seen the light of day without the efforts of all my amazing co-authors.
Beside the individuals mentioned earlier, I thank Aswin Subramanian, Jesus Villalba, Pavel Denisov, Hakan Erdogan, Mao-kui He, Takuya Yoshioka, Andreas Stolcke, Katerina Zmolikova, Marc Delcroix, Yashesh Gaur, Samuele Cornell, Xuankai Chang, and Niko Moritz.
If I have seen further, it is by standing on the shoulders of these giants.

% 7. Marie
Finally, I am grateful to Marie-Philippe Gill, whose love and support has stayed strong through the hills and valleys of this journey.
I will always treasure these years at Johns Hopkins, not just because they made me a better person and a better researcher, but also because it was here that I met the love of my life.
When I count my blessings, I count you twice.

\renewcommand{\contentsname}{Table of Contents}
\baselineskip=24pt

\pagestyle{plain}
\baselineskip=24pt
\tableofcontents

\listoftables
\listoffigures

\cleardoublepage % Needed because our intro chapter doesn't really have anything
\pagenumbering{arabic}

% \pagestyle{fancy}
% \fancyhf{}
% \fancyhead[EL]{\nouppercase\leftmark}
% \fancyhead[OR]{\nouppercase\rightmark}
% \fancyhead[ER,OL]{\thepage}

\cleardoublepage

\setstretch{2.0}

%% The above was the recommended setup by https://github.com/weitzner/jhu-thesis-template but it's no longer needed
%% after Muschelli's changes which stores different chapters in their
%% respective directories. You will still need to add your chapters as
%% TeX files or Rnw files (see rnw_chapter as an example) and please
%% remember to update the makefile accordingly.

%% PLEASE UNCOMMENT THE CHAPTERS AND APPENDICES TO COMPILE THE FULL DISSERTATION

\chapter{The ``Who Spoke What`` Problem}
\label{chap:intro}

\section{Motivation}
\label{sec:intro_motivation}

Advances in artificial intelligence (AI) in the last several decades have been limited to task-specific improvements that were strictly categorized into different modalities. 
For instance, methods have been developed to detect and segment objects in images and video (vision), to analyze sentiments present in textual extracts (language), and to transcribe an audio recording in noisy conditions (speech). 
With the rise of deep neural networks and the convergence of modeling strategies used to address these diverse tasks and modalities, the next version of AI is expected to comprise systems that can learn simultaneously from several sensors, similar to how humans learn. 
Neuro-symbolic learning methods are expected to be complemented by extensive knowledge graphs to enable common sense reasoning in complex scenarios, such as participating in human conversations.

Since the first speech recognition systems were built more than 50 years ago, improvement in voice technology has enabled applications such as voice assistants on smartphones, semi-automated customer support, and embedded systems. 
Through years of research on speech enhancement and robust speech processing, these systems are now deployed in diverse settings such as on smart home speakers and vehicle controls. 
Nevertheless, present systems are passive listeners which transcribe single-speaker utterances and feed into downstream language understanding components. 
Conversational intelligence of the future is expected to comprise systems that can actively participate in human conversations, including scenarios such as note-taking or fact-checking in meetings, collaborative learning in education, or simply recommending grocery items in households. 
While such systems would require intelligence in diverse modalities --- dialog systems for context handling, emotion recognition from speech and video, common sense reasoning, to name a few --- their ability to recognize free-flowing multi-party conversations is a crucial and complex task that needs to be solved.

\section{Background}
\label{sec:intro_bg}

Multi-talker speech recognition of free-flowing conversations is a well-known problem. 
In the offline setting, a long-form audio recording (ranging between several minutes up to a few hours) is provided, and the expected output is a speaker-attributed transcription with time marks. 
When deployed online, streaming audio is provided as input with the same transcription requirements. 
In addition to their use in conversational agents, these systems have several other applications --- such as real-time meeting transcription for hearing-impaired participants, and generating automatic subtitles for movies or video streams, to name a few. 
When used in conjunction with language understanding or dialog systems, they also enable real conversational AI. 
However, systems for solving this problem are still far from human parity, often achieving between 30\% and 50\% error rates on the task. 
As such, the problem is rewarding both in its technical difficulty as well as its ramifications on real-world applications. 

In the 2000s, several advances were made as a result of NIST evaluations~\cite{Fiscus2007TheRT} and the AMI project~\cite{Carletta2005TheAM} that were aimed at tackling the multi-talker recognition problem, primarily in the offline setting.
More recently, challenges such as DIHARD~\cite{Ryant2019TheSD} and CHiME~\cite{Watanabe2020CHiME6CT} have focused on diarization and speech recognition tasks in very challenging scenarios. 
As a result, we are closer today to solving the multi-talker conversation transcription problem than we have ever been.

Nevertheless, there are several challenges yet to be addressed.
While recognition of clean, read speech is claimed to have surpassed human parity\footnote{Such claims must always be taken with a grain of salt, for they are made on specific data sets, and their generality is rarely (if ever) tested~\cite{Amodei2016DS2,Xiong2017TowardHP}}, the same cannot be said of conversations in the wild.
A careful selection of deep learning advances in acoustic modeling, language modeling, and system combination was shown to reach professional transcription levels in conversational speech~\cite{Xiong2016AchievingHP}, but this evaluation was limited to telephonic conversations between two speakers.
Although this is an important development, real multi-talker conversations raise many additional challenges, most notably the case of overlapping speech. 
Studies have shown that meetings can contain up to 20\% overlapping speech~\cite{Carletta2005TheAM}, which has implications for both diarization and ASR --- diarization systems which make single-speaker assumptions miss the interfering speaker completely, and ASR systems trained on clean utterances are more error-prone on these overlapped regions. 
Combined with the effect of non-stationary noise and reverberation in real recordings, the error rates in these settings may be degraded by up to 86\%~\cite{Yoshioka2019MeetingTU} in meetings, and 52\% in dinner-party settings~\cite{Watanabe2020CHiME6CT}. 
Real-time recognition adds an extra layer of difficulty to this problem.

In popular literature, the challenge posed by overlapped speech is often referred to as the ``cocktail party problem'', and is challenging enough that it encompasses several modules that are entire fields in speech processing research. 
The conventional approach for multi-talker ASR is through a cascade of front-end and back-end components, where a separation module feeds into a single-talker ASR. 
While this is an appealing solution, and several advances have indeed been made in speech enhancement (including separation), most of these techniques are designed for (and evaluated on) short, fully overlapping (and often simulated) mixtures.
In several cases, they also make unrealistic assumptions, such as prior knowledge of the number of speakers in the mixture. 
It may further be preferable from an application perspective to have fewer independent components in the pipeline, since cascaded modules tend to compound the overall latency in deployed systems. 
Recently, there have been efforts towards ``continuous speech separation,'' which seeks to situate separation techniques in more realistic settings of long-form conversations containing partially overlapped speech~\cite{Chen2020ContinuousSS}, such as that found in multi-talker conversations. 

The diarization and ASR communities have independently sought to develop methods that handle multi-talker overlapping speech.
For diarization, existing approaches to solve the overlap problem involve using externally trained overlap detectors to identify frames in the recording which contain overlapping speech. 
Once overlaps are detected, an ``overlap assignment'' stage assigns additional speaker labels to the overlapping frames~\cite{Bullock2019OverlapawareDR}. 
A second class of methods uses end-to-end neural systems to perform overlapping diarization in a supervised setting~\cite{Fujita2020EndtoEndND}. 
For multi-talker ASR, permutation-invariant training (PIT), which was first proposed for speech separation, has successfully been employed~\cite{Yu2017RecognizingMS}; however, the transcriptions produced by such a method are unordered \textit{across} utterances, which makes them dependent on an external speaker tracking or diarization module. 
A new framework called serialized output training~\cite{Kanda2020SerializedOT} aims to mitigate some of the issues with PIT, and has been used to jointly perform ASR, speaker identification, and counting~\cite{Kanda2020JointSC,Kanda2020InvestigationOE}.
Similarly, noise robust ASR and diarization continues to be an important direction of research~\cite{Zhang2018DeepLF}. 
In addition to improvements in speech enhancement~\cite{Kinoshita2020ImprovingNR}, existing research has explored noise-aware and multi-condition training for ASR~\cite{Seltzer2013AnIO}, and feature mapping for diarization~\cite{Zhu2017FeatureMF}. 

With these advancements in diarization and ASR systems, there is increasing interest in developing systems that perform multi-talker speech recognition for unsegmented recordings. 
The latest edition of the CHiME challenge~\cite{Watanabe2020CHiME6CT} included a track for evaluating systems for dinner-party conversations. 
An iterative training strategy combining clustering-based diarization with target-speaker ASR has been proposed recently~\cite{Kanda2019SimultaneousSR}. 
The need for controlled but realistic data to investigate such systems has resulted in several new datasets such as LibriCSS~\cite{Chen2020ContinuousSS} and LibriMix~\cite{Cosentino2020LibriMixAO}. 
The research agenda at the JSALT 2020 workshop\footnote{\url{https://www.clsp.jhu.edu/workshops/20-workshop/}} included a project aimed at ``building fully contained multi-talker audio transcription systems based on speech separation and extraction''. 

Due to the amount of research conducted on the cocktail-party problem and its various subproblems, it is near impossible to provide a comprehensive review of past work.
Instead, we will discuss the relevant background work for each chapter as we proceed through the dissertation. 
In the remainder of this chapter, we will define the core problem that is addressed in this dissertation, along with the relevant datasets and evaluation metrics.
We will also provide a brief outline of the dissertation, including a summary of its two parts.

\section{Problem definition}
\label{sec:intro_problem}

Our objective in this work is to solve the problem of speaker-attributed multi-talker speech recognition.
Given a single or multi-channel recording $\mathcal{R}$ (such as that from a meeting), the goal is to transcribe all the speech in the recording and attribute the words to the corresponding speakers.
Formally, if the recording consists of $K$ speakers, and $\mathbf{w}_k$ denotes the sequence of words uttered by speaker $k$, we aim to find a mapping $W$ such that
\begin{equation}
    W(R) = \{\mathbf{w}_1,\ldots,\mathbf{w}_K\}.
\end{equation}

As we will see throughout this dissertation, it is possible to estimate $W$ indirectly by combining several components (which we call the ``modular'' approach), or to estimate it directly in an ``end-to-end'' manner.
The former decomposition results in several problems which are related to the task of speaker-attributed transcription, as summarized in Table~\ref{tab:intro_tasks}.
In the following chapters, we will focus on the problems of speaker diarization and target-speaker extraction as components for one such modular pipeline.
Other combinations of components can be used for different pipelines, as described in \citet{Raj2020IntegrationOS}.
In Chapter~\ref{chap:modular}, we will show how our modular pipeline falls out from a probabilistic formulation of the above problem.

\begin{table}[tp]
    \centering
    \adjustbox{max width=\linewidth}{
    \begin{tabular}{@{}lll@{}}
    \toprule
    \textbf{Task} & \textbf{Input} & \textbf{Output} \\
    \midrule
    Speech enhancement & Mixed recording $\mathcal{R}$ & Enhanced recording $\mathcal{R}^{\ast}$ \\
    Speech separation & Mixed recording $\mathcal{R}$ & Separated audio signals $\mathcal{R}_1,\ldots,\mathcal{R}_K$ \\
    Target-speaker extraction & 
    \begin{Tabular}[1]{@{}l@{}}
    Mixed recording $\mathcal{R}$ \\
    Speaker identity $k$
    \end{Tabular} & Speaker-specific audio $\mathcal{R}_k$ \\
    
    \midrule
    Speaker diarization & Mixed recording $\mathcal{R}$ & 
    \begin{Tabular}[1]{@{}l@{}}
    Homogeneous speaker segments \\ 
    $\{(\Delta_j, u_j): 1\leq j \leq N\}$
    \end{Tabular}
    \\
    Speech recognition & Segmented utterance $\mathbf{X}_j$ & Segment transcript $\mathbf{y}_j$ \\
    Target-speaker ASR & \begin{Tabular}[1]{@{}l@{}}
    Mixed recording $\mathcal{R}$ \\
    Speaker identity $k$
    \end{Tabular} & Speaker-specific transcript $\mathbf{w}_k$ \\
    \bottomrule
    \end{tabular}}
    \caption{A summary of tasks related to the problem of speaker-attributed multi-talker speech recognition. The two sections denote a categorization of tasks into ``front-end'' and ``back-end''.}
    \label{tab:intro_tasks}
\end{table}

\section{Data}
\label{sec:intro_data}

We focus on meeting transcription as an application of multi-talker speaker-attributed ASR.
Throughout this dissertation, we will demonstrate the performance of our methods on several publicly available meeting benchmarks: LibriCSS, AMI, ICSI, and AliMeeting.
These benchmarks are detailed in Appendix~\ref{chap:appendix_data}, and their summary statistics are shown in Table~\ref{tab:intro_stats}.
We have selected these corpora because they provide distant microphone recordings as well as close-talk recordings for the meeting sessions, which makes it convenient to perform controlled evaluations, as well as for training models.
In this dissertation, we will often refer to different microphone settings for these corpora:
\begin{enumerate}
    \item \textbf{IHM}: These are individual close-talk microphone recordings, one per participant, recorded on either headset or lapel microphones.
    \item \textbf{IHM-Mix}: These are digitally summed versions of the IHM recordings, which provides overlapped speech conditions without far-field and background noise artifacts.
    \item \textbf{SDM}: This refers to single distant microphone, i.e., a recording taken from one of the far-field microphones.
    \item \textbf{MDM}: This refers to multiple distant microphones, and the recordings may or may not be beamformed.
\end{enumerate}

\begin{table}[t]
\centering
\caption{Statistics of datasets used for evaluations. The $k$-speaker durations are in terms of fraction of total speaking time.}
\label{tab:intro_stats}
\adjustbox{max width=\linewidth}{
\begin{tabular}{@{}lrrrrrrrrrrr@{}}
\toprule
\multirow{2}{*}{} & \multicolumn{2}{c}{\textbf{LibriCSS}} & \multicolumn{3}{c}{\textbf{AMI}} & \multicolumn{3}{c}{\textbf{ICSI}} & \multicolumn{3}{c}{\textbf{AliMeeting}} \\
\cmidrule(r{5pt}){2-3} \cmidrule(l{4pt}){4-6} \cmidrule(l{4pt}){7-9} \cmidrule(l{4pt}){10-12}
 & \textbf{Dev} & \textbf{Test} & \textbf{Train} & \textbf{Dev} & \textbf{Test} & \textbf{Train} & \textbf{Dev} & \textbf{Test} & \textbf{Train} & \textbf{Dev} & \textbf{Test} \\ 
\midrule
\textbf{Duration (h:m)} & 1:00 & 9:05 & 79:23 & 9:40 & 9:03 & 66:38 & 2:16 & 2:45 & 111:21 & 4:12 & 10:46 \\
\textbf{Num. sessions} & 6 & 54 & 133 & 18 & 16 & 70 & 2 & 3 & 209 & 8 & 20 \\
\textbf{Silence (\%)} & 6.2 & 6.7 & 18.1 & 21.5 & 19.6 & 55.2 & 25.9 & 25.9 & 7.11 & 7.7 & 8.0 \\
\textbf{1-speaker (\%)} & 81.3 & 81.2 & 75.5 & 74.3 & 73.0 & 82.1 & 90.3 & 84.9 & 52.5 & 62.1 & 63.4  \\
\textbf{2-speaker (\%)} & 18.6 & 18.5 & 21.1 & 22.2 & 21.0 & 15.7 & 9.0 & 13.6 & 32.8 & 27.6 & 24.9 \\
\textbf{>2-speaker (\%)} & 0.1 & 0.4 & 3.4 & 3.5 & 6.0 & 2.2 & 0.7 & 1.4 & 14.7 & 10.2 & 11.7 \\
\bottomrule
\end{tabular}}
\end{table}

\section{Evaluation metrics}
\label{sec:intro_metrics}

Our objective in this dissertation is to demonstrate the performance of modular and end-to-end systems for speaker-attributed multi-talker ASR.
Modular systems comprise a pipeline of speaker diarization, target speaker extraction, and ASR components, while end-to-end systems seek to directly estimate speaker-labeled transcripts.
In this section, we describe the metrics that we will use to evaluate the individual components or the final system in the subsequent chapters.

\subsection{Diarization metrics}

We use \textbf{diarization error rate (DER)}, as described in \citet{Mir2006RobustSD}, to evaluate the speaker diarization systems proposed in this dissertation.
DER comprises the sum of three different kinds of error rates --- missed speech (MS), false alarms (FA), and speaker confusion (SC) --- and is given as
\begin{equation}
    DER = MS + FA + SC.
\end{equation}
Since speaker diarization systems produce relative speaker labels (and not absolute labels), we first find the optimal mapping of reference and hypothesis speakers which would minimize the DER.
This is usually done by formulating the problem as a linear sum assignment problem, and solved using the Hungarian algorithm~\cite{Munkres1957AlgorithmsFT}.

\subsection{ASR metrics}

We use the \textbf{word error rate (WER)} metric to evaluate single-speaker ASR systems.
WER is computed as the Levenshtein distance~\cite{Levenshtein1965BinaryCC} between the reference $\mathbf{w}^{\ast}$ and hypothesis $\hat{\mathbf{w}}$, and is the most popular metric for reporting ASR performance.
Informally, the Levenshtein distance between two strings is the minimum number of edits, (insertions, deletions, and substitutions), needed to transform one string into the other.
The process of alignment can be described as the application of these edits to $\hat{\mathbf{w}}$, and the best alignment is the one that uses the fewest number of edits. 
Note that the Levenshtein distance considers only monotonic, i.e. left-to-right, alignment of the hypothesis transcript $\hat{\mathbf{w}}$, with reference transcript, $\mathbf{w}^{\ast}$, and is efficiently computed via dynamic programming.

\subsection{Signal-level metrics}

In Chapter~\ref{chap:gss}, we will describe target speaker extraction (TSE) using guided source separation~\cite{Boeddeker2018FrontendPF}.
Since the objective of TSE is to estimate clean, single-speaker signal, given a noisy multi-talker signal as input, we will additionally use signal-level metrics as an intrinsic measure of system performance.
In particular, we will use the following metrics popular in speech enhancement literature:

\begin{itemize}
    \item \textbf{PESQ~\cite{Rix2001PerceptualEO}:} PESQ (Perceptual Evaluation of Speech Quality), originally developed to assess the quality of speech codecs, measures the similarity between a reference (original) and a processed (separated) audio signal in terms of perceived quality. It takes into account various factors related to human auditory perception, such as loudness, sharpness, and distortion. The output of PESQ is a single value, typically ranging from -0.5 to 4.5, with higher values indicating better perceived quality.
    \item \textbf{SI-SDR~\cite{LeRoux2018SDRH}:} SI-SDR (Scale Invariant Signal-to-Distortion Ratio) measures the quality of a separated audio signal by quantifying the ratio between the target (clean) signal and the unwanted interference or distortion caused by the separation process. The ``scale-invariant'' aspect of the metric means that it is not sensitive to the scaling of the signals, making it more robust in scenarios where the amplitude of the signals may vary. Mathematically, it is defined as
    \begin{equation}
        \text{SI-SDR} = 10 \cdot \log_{10}\left(\frac{{\|s_{\text{target}}\|^2}}{{\|s_{\text{distortion}}\|^2}}\right),
    \end{equation}
    where $s_{\text{target}}$ is the clean signal and $s_{\text{distortion}}$ is the distortion caused by the enhancement process.
    \item \textbf{STOI~\cite{Taal2010ASO}:} STOI (Short-Time Objective Intelligibility) specifically evaluates the intelligibility of speech by measuring the similarity between a reference  and a processed audio signal in terms of how well the speech can be understood. STOI takes into account factors related to human auditory perception, such as the modulation spectrum and the presence of noise. It operates on short-time frames, which allows it to capture variations in intelligibility over time. The output of STOI is a value between 0 and 1, where 1 indicates perfect intelligibility (i.e., the separated speech is identical to the original), and 0 indicates no intelligibility (i.e., the separated speech is completely unintelligible).
\end{itemize}

\subsection{Multi-talker ASR metrics}

Our primary metric for evaluating the final speaker-attributed transcription systems is the \textbf{concatenated minimum-permutation word error rate (cpWER)}~\cite{Watanabe2020CHiME6CT}.
The key idea in cpWER is to find the best permutation between reference and hypothesis speakers that minimizes the average WER between the speaker's concatenated transcripts.
Given reference transcripts for a session and the corresponding hypothesis, the computation follows three steps:
\begin{itemize}
    \item Concatenate all utterances of each speaker for both reference and hypothesis files.
    \item Compute the WER between the reference and all possible speaker permutations of the hypothesis.
    \item Pick the lowest WER among them.
\end{itemize}

In some cases, we may also want to evaluate multi-talker ASR systems without speaker labels.
For example, the SURT model that we describe in Chapter~\ref{chap:surt} transcribes all the speech in a mixture without attributing the words to speakers.
In this case, we will use the \textbf{optimal reference combination word error rate (ORC-WER)} metric, proposed concurrently in \citet{Sklyar2021MultiTurnRF} and \citet{Raj2021ContinuousSM}.
ORC-WER is used for evaluating systems where the reference contains a list of (possibly overlapping) segments ordered by start time, and the hypothesis contains multiple streams of output transcript.
Since there is no obvious mapping from reference segments to the output streams, we find the best permutation of reference to stream that minimizes the overall word error rate.
\citet{vonNeumann2022OnWE} proposed a polynomial-time implementation of ORC-WER using multi-dimensional Levenshtein distance, and released it through the \texttt{meeteval}\footnote{\url{https://github.com/fgnt/meeteval}} toolkit.
They showed that ORC-WER is a lower bound on cpWER and becomes equal to the cpWER when no speaker errors are present.
We used their open-source implementation for evaluating our SURT models in this dissertation.
The overall algorithm for ORC-WER computation is given in Algorithm~\ref{alg:orc-wer}.
We assume that the references $\mathcal{R}_n$ are ordered according to the start times of utterances.

\begin{algorithm}[tp]
\DontPrintSemicolon
  
  \KwInput{Ref.: $\mathcal{R}_1,\ldots,\mathcal{R}_N$; hyp.: $\mathcal{H}_1,\ldots,\mathcal{H}_C$}
  \KwOutput{$W_{\text{orc}}$}

  $\mathcal{R} \leftarrow \mathcal{R}_1 \diamond \ldots \diamond \mathcal{R}_N$ \tcp*[r]{Insert channel change token}

  $\xi \leftarrow \{\}$ \hfill \tcp*[r]{Memoize costs to avoid recomputing}

  \SetKwFunction{FMain}{levenshtein}
  \SetKwProg{Fn}{Function}{:}{}
  \Fn{\FMain{$R$, $H$, $i=\phi$}}{
        $K = (R,H,i)$\; 
        \If{$K$ in $\xi$}{
            \KwRet $\xi[K]$\;
        }
        \If{$i\neq \phi$ and $\mathcal{R}[R-1] == \diamond$}{
            \tcp*[l]{Reached a channel change token}
            $\xi[K] \leftarrow$ \FMain{$R-1$,$H$,$\phi$} \;
            \KwRet $\xi[K]$\;
        }
        \eIf{$i=\phi$}{
            $C \leftarrow \emptyset$ \tcp*[r]{Find minimum over all hypotheses}
            \For{$i$ in range($|\mathcal{H}|$)}{
                $C[i] \leftarrow$ \FMain{$R$,$H$,$i$}\;
            }
            $\xi[K] \leftarrow \min(C)$\;
            \KwRet $\xi[K]$\;
        }{
            \tcp*[l]{Regular Levenshtein update}
            \uIf{$R = 0$}{
                % \tcp*[l]{Empty reference}
                $c \leftarrow \sum_{h\in H} h$\;
            }
            \uElseIf{$H[i]=0$}{
                % \tcp*[l]{Empty hypothesis}
                $c \leftarrow 1$ + \FMain{$R-1$,$H$,$i$}\;
            }
            \uElseIf{$\mathcal{R}[R-1]=\mathcal{H}[i][H[i]-1]$}{
                % \tcp*[l]{Correct match}
                $H[i] \leftarrow H[i]-1$\;
                $c \leftarrow$ \FMain{$R-1$,$H$,$i$}\;
                $H[i] \leftarrow H[i]+1$\;
            }
            \uElse{
                % \tcp*[l]{No correct match}
                $ins \leftarrow$ \FMain{$R-1$,$H$,$i$}\;
                $H[i] \leftarrow H[i]-1$\;
                $del \leftarrow$ \FMain{$R$,$H$,$i$}\;
                $sub \leftarrow$ \FMain{$R-1$,$H$,$i$}\;
                $H[i] \leftarrow H[i]+1$\;

                $c \leftarrow 1 + \min(ins,del,sub)$\;
            }
            $\xi[K] \leftarrow c$\;
            \KwRet $\xi[K]$\;
        }
  }

  \KwRet \FMain{$|\mathcal{R}|$,$\{|h|:h\in \mathcal{H}\}$}\;
  
\caption{ORC-WER computation}
\label{alg:orc-wer}
\end{algorithm} 

Fig.~\ref{fig:sawer_types} shows the difference between cpWER and ORC-WER using some toy examples.
For the example shown in Fig.~\ref{fig:intro_wer1}, the hypothesis assigns all words to the same channel (or speaker).
Since ORC-WER is speaker-agnostic, the resulting error is 0, whereas for cpWER, there are 2 insertions and 2 deletions, resulting in 4 errors.
In the example shown in Fig.~\ref{fig:intro_wer2}, both the metrics result in the same absolute error.

\begin{figure}[t]
\begin{subfigure}{0.49\linewidth}
\centering
\includegraphics[width=\linewidth]{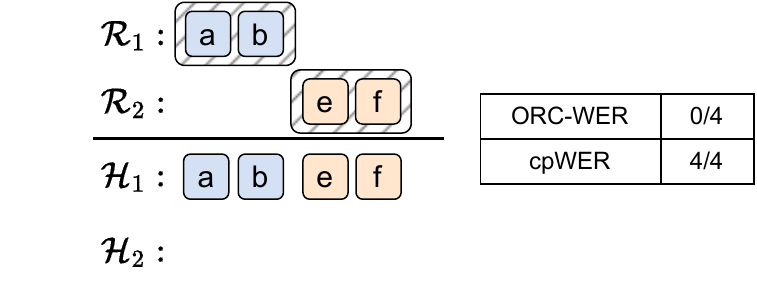}
\caption{}
\label{fig:intro_wer1}
\end{subfigure}
\begin{subfigure}{0.49\linewidth}
\includegraphics[width=\linewidth]{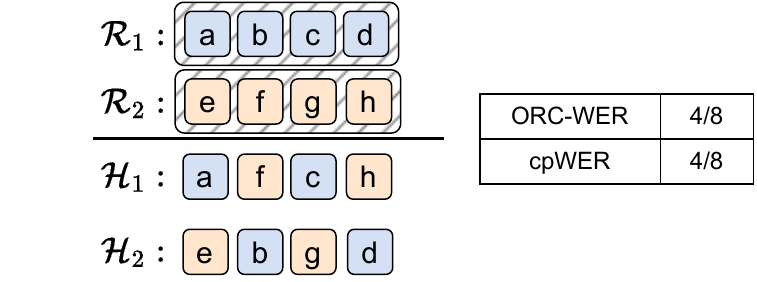}
\caption{}
\label{fig:intro_wer2}
\end{subfigure}
\caption[Toy examples to demonstrate differences of multi-talker WER definitions.]{Toy examples to demonstrate differences of multi-talker WER definitions, based on Figure 1 in \citet{vonNeumann2022OnWE}. Each solid box is a word, and grey hatched box is an utterance. Error counts for ORC-WER and cpWER are shown in the tables.}
\label{fig:sawer_types}
\end{figure}

\section{Outline of the dissertation}
\label{sec:intro_outline}

In the previous sections, we have described the problem of speaker-attributed multi-talker ASR, and the corresponding datasets and evaluation criteria.
We tackle this problem from two perspectives, and this dissertation is also organized along the same lines.

The first part (i.e., the modular approach) comprises chapters \ref{chap:oasc}, \ref{chap:doverlap}, \ref{chap:gss}, and \ref{chap:modular}.
Since overlapping speech is a major challenge for multi-talker ASR, we begin in Chapter~\ref{chap:oasc} by proposing a new method for overlap-aware speaker diarization using spectral clustering.
We show that our method is easy to integrate into existing clustering-based pipelines and provides significant DER improvements compared to methods that ignore overlapping speech.
In the last few years, there have been tremendous advances in overlap-aware diarization, with different approaches having their own advantages and limitations.
In Chapter~\ref{chap:doverlap}, we describe the DOVER-Lap algorithm to ensemble the outputs of these systems, resulting in significant gains over single-best systems.
In the multi-channel setting, DOVER-Lap can also be used to combine the outputs from different channels as a late fusion strategy.
Once the homogeneous speaker segments have been identified, we use the guided source separation (GSS) method to extract speaker-specific signals from the mixture.
In Chapter~\ref{chap:gss}, we describe a GPU-accelerated GSS implementation that provides close to 300x speed-up for inference, thus providing large RTF improvements for the pipeline.
Finally, we present the full pipeline, along with its probabilistic formulation, in Chapter~\ref{chap:modular}.
Qualitative and quantitative analyses of the pipeline throws light on the limitations of the modular approach.

The second part of the dissertation describes an alternate end-to-end perspective in Chapters \ref{chap:surt} and \ref{chap:surt2}.
In Chapter~\ref{chap:surt}, we first introduce the Streaming Unmixing and Recognition Transducer (SURT), an extension of the popular neural transducers for handling multi-talker overlapping speech.
SURT combines ``unmixing'' and ``recognition'' components trained end-to-end with an ASR loss.
By using a fixed ordering of references instead of permutation invariant training, and through various choices in model design, architecture, training mixture simulation, and objective functions, we show how to train the SURT models efficiently for speaker-agnostic transcription of real meetings.
Once we have an efficient training pipeline, we extend SURT for speaker-attributed transcription by adding an auxiliary speaker branch to the recognition component.
For this, we propose blank factorization to synchronize the emission from the branches, and speaker prefixing to ensure speaker label consistency through the recording.
Our final SURT model shows promising results on AMI, although there still exists a gap compared to offline, modular systems.

\section{Publications and Code}
\label{sec:intro_papers}

The work described in this dissertation has been published at the following venues.

\begin{itemize}[noitemsep]
\item Chapter~\ref{chap:oasc}: \citet{Raj2020MulticlassSC}
\item Chapter~\ref{chap:doverlap}: \citet{Raj2020DOVERLapAM,Raj2021ReformulatingDL}
\item Chapter~\ref{chap:gss}: \citet{Raj2022GPUacceleratedGS}
\item Chapter~\ref{chap:modular}: \citet{Arora2020TheJM}
\item Chapter~\ref{chap:surt}: \citet{Raj2021ContinuousSM,Raj2023Surt20}
\item Chapter~\ref{chap:surt2}: \citet{Raj2024OnSAS}
\end{itemize}

Additionally, this dissertation has also resulted in open-source software contributions, as summarized in Table~\ref{tab:intro_software}.
In several cases such as Chapters \ref{chap:doverlap} and \ref{chap:gss}, we implemented stand-alone Python packages which can be integrated into existing frameworks or recipes.
For example, our GSS implementation has recently been used as part of the community baseline in the CHiME-7 DASR challenge~\cite{Cornell2023TheCD}.
For others, such as the SURT model described in Chapters \ref{chap:surt} and \ref{chap:surt2}, the open-source implementation was done as part of a popular framework such as \texttt{icefall}.
In such cases, the recipe was also accompanied by development and contributions in the \texttt{k2} and \texttt{Lhotse} toolkits.\footnote{This dissertation was partially supported by an NSF CIRC grant for the development of next-generation speech tools (Lhotse, k2, and icefall.)}

\begin{table}[t]
\centering
\adjustbox{max width=\linewidth}{
\begin{tabular}{lll}
\toprule
\textbf{Chapter} & \textbf{GitHub repositories} & \textbf{Remarks} \\
\midrule

\ref{chap:oasc} & 
\begin{tabular}{l}
\texttt{desh2608/diarizer} \\
\texttt{desh2608/spyder}
\end{tabular} & 
\begin{tabular}{@{}l@{}}
Recipes for clustering-based diarization \\
Python package for fast DER computation
\end{tabular} \\

\ref{chap:doverlap} & \texttt{desh2608/dover-lap} & Python implementation of DOVER-Lap \\

\ref{chap:gss} & \texttt{desh2608/gss} & CuPy-based fast GSS with recipes \\

\ref{chap:modular} & \texttt{desh2608/icefall/tree/multi\_talker} & Icefall recipes for multi-talker ASR \\

\ref{chap:surt} \& \ref{chap:surt2} & \texttt{k2-fsa/icefall} & LibriCSS and AMI SURT recipes \\
\bottomrule
\end{tabular}}
\caption{Summary of software contributions as a result of this dissertation.}
\label{tab:intro_software}
\end{table}

% \cleardoublepage

\chapter{Overlap-aware Speaker Diarization using Spectral Clustering}
\label{chap:oasc}

One of the challenges of multi-talker speaker-attributed ASR is to identify and demarcate unique speakers in the recording.
In this chapter, we will begin by formalizing this ``speaker diarization'' problem, and then describe a popular clustering-based paradigm to solve it.
We will identify how the naive clustering solution fails to address overlapping speech, resulting in high error rates for meeting-like scenarios.
Finally, we will show that an alternative formulation of multi-class spectral clustering allows us to incorporate an external overlap detection module into the process, effectively reducing overlap related errors in the diarization output.

\section{A background in speaker diarization}
\label{sec:oasc_background}

% Speech is one of the most popular modalities for communication and dissemination of ideas.
% %
% In everyday life, we come across a wide variety of scenarios where two or more people interact using the spoken word. 
% %
% These interactions may occur in the form of telephone conversations, multi-speaker meetings, or simply a conversation over a dinner party. 
% %
% In the digital age, video/audio conferencing has also become a popular mode of communication. 
% %
% Given its widespread use, a natural question is whether we can identify the segments of speech spoken by different speakers. 
% %
% Such segmentation is often useful for applications such as meeting transcription, and also enables downstream psycho-linguistic analysis such as the study of language acquisition in infants, pyschotherapy and human interaction, and observing collaborative learning in peer groups.

Speaker diarization (or ``who spoke when'') is defined as the task of segmenting speech into speaker-homogeneous regions~\citep{Mir2012SpeakerDA,Tranter2006AnOO}. 
The input in this task is an audio recording (single or multi-channel) which may extend from a few minutes up to several hours (or days, such as in scenarios like child language acquisition). 
The desired output is a set of (possibly overlapping) segments with associated speaker labels. 
Formally, given an audio recording $R$, the diarization system is a function $f$ that estimates a set of speaker-labeled time segments, i.e.,
\begin{equation}
\label{eq:oasc_1}
f(R) = \{(\Delta_j, u_j): 1\leq j \leq N\},
\end{equation}
where $\Delta_j = (t_j^{\mathrm{st}}, t_j^{\mathrm{en}})$ denotes the segment boundaries and $u_j \in [K]$ denotes the speaker label. 
$K$ and $N$ are the estimated number of speakers and speaker-homogeneous segments in $R$, respectively. 
Depending on the use-case, the true value of $K$ may or may not be known beforehand. 

\begin{figure}[t]
    \centering
    \includegraphics[width=0.8\linewidth]{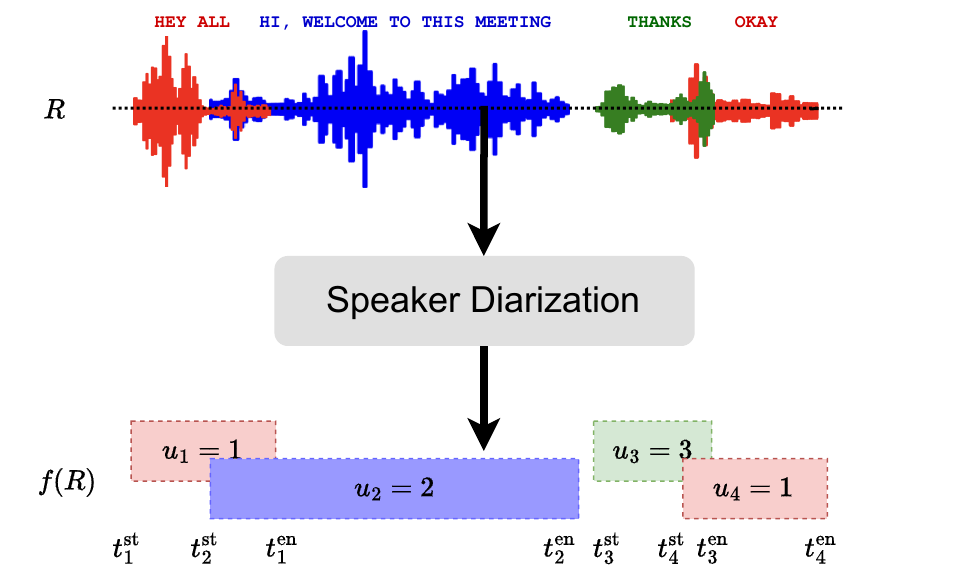}
    \caption{An illustration of the speaker diarization task. The system estimates $N=3$ speakers in the recording.}
    \label{fig:diarization}
\end{figure}

The speaker labels $u_j$ are \textit{relative} labels, i.e., they are only consistent within the recording. 
Fig.~\ref{fig:diarization} illustrates an example of the diarization task where the system predicts $N=4$ segments containing $K=3$ speakers. 
Note that some parts of the audio may be ``non-speech,'' i.e., they do not have any assigned speaker (e.g., the segment between $t_2^{\mathrm{en}}$ and $t_3^{\mathrm{st}}$, whereas some other parts may have ``overlapping speech'' and are assigned to multiple speakers (e.g., the segment between $t_2^{\mathrm{st}}$ and $t_1^{\mathrm{en}}$).

\section{The clustering paradigm}
\label{sec:oasc_clustering}

The conventional solution to speaker diarization is based on the ``clustering'' principle.
This is a pipeline consisting of four major components, as shown in Fig.~\ref{fig:clustering}.
These components are described in the following sections.

\begin{sidewaysfigure}
    \centering
    \includegraphics[width=0.95\linewidth]{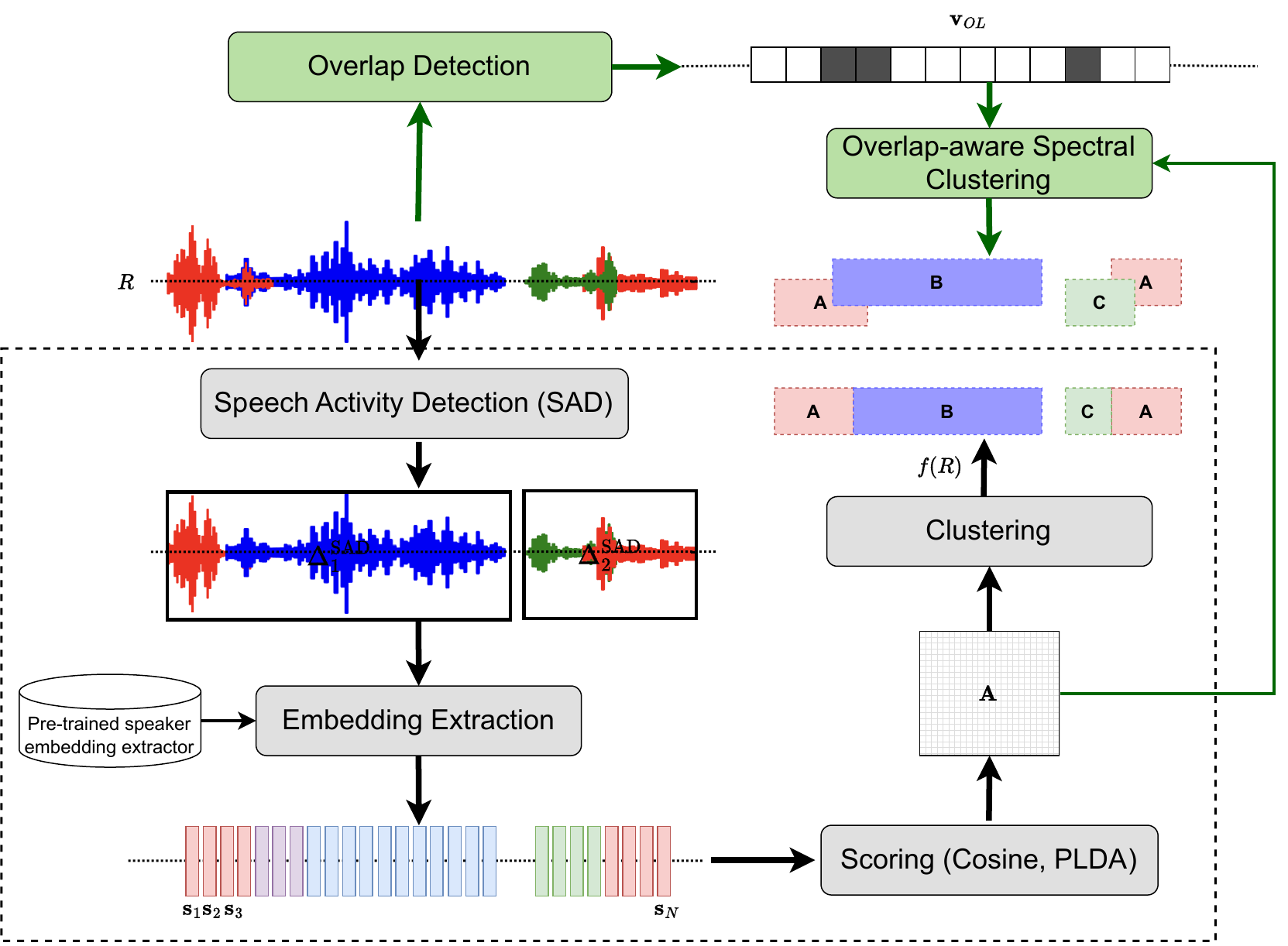}
    \caption{Components of the clustering-based diarization system (in dotted box). Green arrows and boxes show the new components for overlap-aware diarization (Section~\ref{sec:oasc_method}).}
    \label{fig:clustering}
\end{sidewaysfigure}

\subsection{Speech activity detection}

First, a speech activity detection (SAD) module is used to filter out the non-speech regions of the recording.
Traditional methods in SAD often rely on signal processing techniques, such as energy-based thresholding~\citep{Davis2006StatisticalVA}, analysis of the frequency modulation including harmonics~\citep{Hsu2015RobustVA} and formants~\citep{Yoo2015FormantBasedRV}, and statistical modeling~\citep{Cho2001AnalysisAI,Tan2010VoiceAD}. 
These methods analyze the characteristics of the audio signal to distinguish between speech and non-speech segments. 
While effective in controlled environments, they may face challenges in noisy or dynamic settings.
Machine learning and deep learning techniques have significantly advanced the field of SAD.
Supervised learning approaches, using annotated datasets, allow models to learn discriminative features for differentiating speech from non-speech segments. 
Popular algorithms include support vector machines (SVMs)~\citep{Enqing2002ApplyingSV}, Gaussian mixture models (GMMs)~\citep{Fukuda2010LongTermSA}, hidden Markov models (HMMs)~\citep{Sohn1999ASM}, and more recently, deep learning architectures like convolutional~\citep{Lin2019OptimizingVA,Vafeiadis2019TwoDimensionalCR} and recurrent neural networks~\citep{Hughes2013RecurrentNN,Eyben2013ReallifeVA}.
Unsupervised and semi-supervised techniques have also emerged to address scenarios where labeled data is scarce~\citep{Ying2011VoiceAD,Sholokhov2018SemisupervisedSA}. 
Additionally, context-aware approaches have gained prominence, leveraging additional information like contextual cues, speaker characteristics, or linguistic features to enhance SAD performance~\citep{Segbroeck2013ARF,Kim2018VoiceAD}.

In this chapter, since our objective is to improve the clustering stage using external overlap information, we will use an oracle SAD to detect speech segments.
Later, in Chapter~\ref{chap:modular}, we will replace this oracle with a neural network based SAD to evaluate the full pipeline of speaker-attributed ASR.
Regardless of the SAD model used, we can mathematically formulate the problem as
\begin{equation}
    \mathrm{SAD}(R) = \{\Delta_n^{\mathrm{SAD}}:1\leq n \leq N_{\mathrm{SAD}}\},
\end{equation}
where $\Delta_n^{\mathrm{SAD}}$, as defined in \eqref{eq:oasc_1}, denotes a \textit{speech} segment in $R$, and $N_{\mathrm{SAD}}$ is the number of speech segments identified.
Note that each of these speech segments may contain multiple speakers at this point.

\subsection{Embedding extraction}

Once the non-speech regions have been discarded, the speech regions of the recording, i.e., $\{\Delta_1^{\mathrm{SAD}},\ldots,\Delta_{N_{\mathrm{SAD}}}^{\mathrm{SAD}}\}$, are divided into small, overlapping subsegments (or window), under the assumption that no speaker changes would happen within these subsegments.
Conventionally, a sliding window of size 1.5s and shift 0.75s is used to obtain these subsegments.
% , which we denote as $\{\Delta_1^{\mathrm{w}},\ldots,\Delta_T^{\mathrm{w}}\}$.\footnote{Knowing the size and shift of the windows, we can exactly determine $T$ from $S$, but we skip these details here for brevity.}
%
A pre-trained speaker embedding extractor is then used to obtain fixed-dimensional embeddings, which we denote as $\mathbf{s}_1,\ldots,\mathbf{s}_N$, where $\mathbf{s}_n \in \mathbb{R}^d$ is the embedding for one such window.
The diarization process assigns one or more speaker labels to each of these $N$ windows.

Originally, unsupervised GMM-based methods, such as i-vectors~\citep{Dehak2011FrontEndFA} were used to obtain $\{\mathbf{s}_n\}$, the fixed-dimensional embeddings.
However, deep neural network (DNN) based techniques, such as d-vectors~\citep{Variani2014DeepNN} and x-vectors~\citep{Snyder2018XVectorsRD} have become more popular in the last few years due to their better speaker discriminative capabilities.
Researchers have also explored using attention-based methods~\citep{Okabe2018AttentiveSP,Wang2020MultiResolutionMA}, triplet loss~\citep{Novoselov2018TripletLB}, and domain adaptation~\citep{Wang2018UnsupervisedDA} to improve the speaker embeddings, which in turn improves their applicability in diarization.
In this chapter, we use the x-vector based method to obtain embeddings.

The x-vector model~\citep{Snyder2018XVectorsRD} is based on time-delay neural networks (TDNNs)~\citep{Peddinti2015ATD}, which model temporal dependencies in speech, usually better than RNNs or multi-layer perceptrons (MLP) with context windows. 
The architecture consists of TDNN layers operating at the frame level, followed by statistics pooling layers operating at the segment level.
These pooling layers mitigate the dependency on the utterance length, and are especially advantageous for speaker diarization since these systems process segments that are shorter than the typical window length for speech recognition.
The x-vector neural network is trained in a supervised manner with a speaker classification objective.
In Chapter~\ref{chap:modular}, we will replace the TDNN-based x-vector with a wide ResNet architecture that has been shown to outperform the original model on speaker recognition tasks~\citep{Landini2020BayesianHC}.

\subsection{Scoring}

Once we have the embeddings $\mathbf{s}_n$, we score them pair-wise using some distance metric. 
The simplest of such metrics is the cosine similarity between two vectors, which is defined as
\begin{equation}
    \cos(\theta) = \frac{\mathbf{s}_1 \cdot \mathbf{s}_2}{\lVert\mathbf{s}_1\rVert \lVert\mathbf{s}_2\rVert}. 
\end{equation}

Cosine similarity is a popular metric in speaker diarization systems which use spectral clustering or mean shift clustering algorithms. 
It does not require any training or parameter tuning; however, it cannot project or weight the embedding vectors to enhance the similarity measurement.

Another distance metric often used in speaker recognition literature is probabilistic linear
discriminant analysis (PLDA)~\cite{Ioffe2006ProbabilisticLD}. 
PLDA produces a comparison score, which is the log of the ratio of the probability that the embeddings were produced by the same speaker, versus the probability that they were produced by different speakers. 
If the comparison score is above a decision threshold $\theta$, we posit that the embeddings belong to the same speaker, otherwise, we declare that they belong to different speakers.
Since these PLDA scores may be negative and are asymmetric by definition, they cannot directly be used for spectral clustering, since the method requires a positive, symmetric affinity matrix.
Although these limitations can be bypassed using shifting and symmetrization steps, we just use cosine similarity in this chapter as our distance metric.
This results in an affinity matrix, $\mathbf{A} \in \mathbb{R}^{N\times N}$, given as
\begin{equation}
    \mathbf{A}_{ij} = \frac{\mathbf{s}_i \cdot \mathbf{s}_j}{\lVert\mathbf{s}_i\rVert \lVert\mathbf{s}_j\rVert}.
\end{equation}

\subsection{Clustering}

The final step in the diarization process, and the one we are most interested in, is clustering performed on the resulting affinity matrix.
Among the most widely used clustering methods in speaker diarization are agglomerative hierarchical clustering (AHC) and spectral clustering~\citep{Luque2012OnTU}.

\subsubsection{Agglomerative hierarchical clustering}

AHC is based on an iterative process of merging the existing clusters of speech segments until the distance of the closest cluster pair meets a predetermined stopping criterion~\citep{GarciaRomero2017SpeakerDU}. 
Given the affinity matrix $\mathbf{A}$, which provides pair-wise similarity scores between the $N$ sub-segment embeddings, an iterative process is followed until some stopping criterion is reached.
We start by assigning each of the $N$ embeddings to a unique cluster (or speaker).
At each step, the two closest clusters are merged based on a linkage criterion, such as the distance between their closest (single linkage), farthest (complete linkage), or average (average linkage) data points.
After merging the clusters, $\mathbf{A}$ is updated to reflect the distances between the new clusters and the remaining clusters.
This iterative merging process is often visualized as a dendrogram.

For speaker diarization tasks, the merging process can be stopped using either a threshold for similarity or a target number of speakers.
Usually, the stopping criterion is adjusted based on a development set to get an accurate number of clusters. 
If the number of speakers is known or estimated in advance, the AHC process can be stopped when the number of merged clusters reaches the predetermined number of speakers.

\subsubsection{Spectral clustering}
\label{sec:oasc_njw}

Spectral clustering was first applied to speaker diarization in \citet{Ning2006ASC} using the Ng-Jordan-Weiss (NJW) algorithm~\citep{Ng2001OnSC}. 
\citet{Bassiou2010SpeakerDE} extended this to the case of an unknown number of speakers by using the eigengap criterion. 
Agglomerative and spectral clustering methods for meeting diarization were compared in \citet{Luque2012OnTU}. 
After i-vectors were proposed for speaker recognition~\citep{Dehak2011FrontEndFA}, they were combined with cosine scoring and spectral clustering to perform diarization in \citet{Shum2012OnTU}.
It was further observed that spectral clustering was more robust to non-stationary environmental noise compared to other clustering methods~\citep{Tawara2015ACS}.
More recently, with the ubiquitousness of deep neural networks, several researchers have proposed methods to incorporate DNNs with spectral clustering. 
\citet{Lin2019LSTMBS} proposed a supervised method to measure the similarity matrix between all segments of an audio recording with BLSTMs, and applied spectral clustering on top of the similarity matrix. 
Other approaches use DNN-based speaker embeddings, such as x-vectors~\citep{Snyder2018XVectorsRD}, to compute the similarty matrix between segment pairs~\citep{Park2020AutoTuningSC,Medennikov2020TargetSpeakerVA}. 
Additionally, \citet{Park2020AutoTuningSC} introduced $p$-binarization and normalized maximum eigengap (NME) techniques to automatically estimate the number of speakers in the recording.

The traditional approach to spectral clustering is using the algorithm presented in \citet{Ng2001OnSC}.
This algorithm consists of the following steps:
\begin{enumerate}
    \item The Laplacian, $\mathbf{L}$, of the affinity matrix $\mathbf{A}$ is computed as $\mathbf{L} = \mathbf{D} - \mathbf{A}$ (for the case of unnormalized Laplacian), where $\mathbf{D}$ is the degree matrix.
    \item Eigen-decomposition of $\mathbf{L}$ is used to obtain the eigenvectors and eigenvalues as $\mathbf{L} = \mathbf{X}\Lambda \mathbf{X}^T$.
    \item The optimal number of clusters, $K$, is estimated by finding the maximum eigengap in the eigengap vector, i.e., 
    $$
    \mathbf{e}_{\mathrm{gap}} = \{\lambda_2-\lambda_1,\lambda_3-\lambda_2,\ldots,\lambda_n-\lambda_{n-1}\}.
    $$
    \item The eigenvectors $\mathbf{v}_1,\ldots,\mathbf{v}_K$ corresponding to the $K$ smallest eigenvalues are collected into a $\mathbf{U}\in \mathbb{R}^{N\times K}$ matrix, and the rows $\mathbf{u}_n$ of $\mathbf{U}$ are clustered into $K$ clusters using K-means clustering. 
\end{enumerate}

In Section~\ref{sec:oasc_method}, we will describe an alternative formulation of spectral clustering based on solving a constrained optimization problem, which is more amenable for overlap-aware speaker diarization.

\section{Overlap-aware diarization}
\label{sec:oasc_diar}

Although the clustering paradigm described above has proved to be effective through the use of DNN-based speaker embeddings, it cannot handle overlapping speech by design, since the clustering process assigns each segment to exactly one speaker. 
Existing approaches to solve the overlap problem fall into two categories. 
In the first category, an externally trained overlap detection module identifies frames in the recording which contain overlapping speech. 
This ``overlap detection'' may be performed using HMMs~\citep{boakye2008overlapped, Huijbregts2009SpeechOD, Yella2012SpeakerDO} or neural networks~\citep{Geiger2013DetectingOS, Andrei2017DetectingOS, Hagerer2017EnhancingLR, Kunesov2019DetectionOO}. 
Once overlaps are detected, an ``overlap assignment'' stage assigns additional speaker labels to the overlapping frames. 
\citet{Bullock2019OverlapawareDR} proposed overlap-aware resegmentation, which leverages the variational Bayes (VB)-HMM method used originally for diarization in \citet{Dez2018SpeakerDB}, and applied to resegmentation in \citet{Sell2015DiarizationRI}. 
In the second framework, end-to-end systems~\citep{Fujita2020EndtoEndND,Huang2020SpeakerDW} are used to perform overlapping diarization in a supervised setting.
While speaker diarization in overlapping settings has been studied extensively, there is no prior work, to the best of our knowledge, on incorporating overlap awareness into spectral clustering based diarization. 
%
% \citet{Medennikov2020TargetSpeakerVA} proposed target speaker voice activity detection (TS-VAD), which uses x-vector based spectral clustering for initial estimate of speaker i-vectors, and thereafter performs frame-level multi-label classification to predict speaker activities in a speech frame. 

In this chapter, we extend spectral clustering for overlap-aware speaker diarization.
Specifically, we train an external overlap detector, and use its classification decision during clustering of the segment-level embeddings. 
Our method relies on the two-step clustering formulation proposed in~\citet{Yu2003MulticlassSC}. 
The remainder of this chapter is organized as follows. 
We start by giving a detailed description of our method in Section~\ref{sec:oasc_method}, where we discuss $p$-binarization and NME for estimating the number of speakers, followed by the mathematical formulation of multi-class spectral clustering. 
We then introduce our modification of the method to perform overlap-aware diarization. 
This is followed by a description of our experimental setup and results in Sections~\ref{sec:oasc_experiments} and \ref{sec:oasc_results}, respectively. 
We present results on the AMI meeting corpus and the LibriCSS dataset, with detailed analysis of the performance of the method on different overlap conditions.

\section{Methodology}
\label{sec:oasc_method}

Our diarization system follows the clustering paradigm outlined in Section~\ref{sec:oasc_clustering}. 
We focus on the final clustering stage, and specifically on how to make the clustering process overlap-aware. 
In particular, given a sequence of windowed embeddings $S = (\mathbf{s}_1,\ldots,\mathbf{s}_N)$, the objective is to compute a label sequence $L = (\boldsymbol{\ell}_1,\ldots,\boldsymbol{\ell}_N)$, where $\boldsymbol{\ell}_n$ may be a single label or, in case of overlapping segment, multiple labels. 
For convenience, we assume that overlaps can occur between at most two speakers, so $\boldsymbol{\ell}_n$ will be a 2-tuple for an overlapping segment. 

In this section, we will assume that we have an overlap detector which decides, for each window, whether or not it contains overlapping speech. 
We denote the as
\begin{equation}
\label{eqn:overlap}
\mathrm{OVL}(R) = \mathbf{v}_{OL},
\end{equation}
where $\mathbf{v}_{OL} \in \{0,1\}^N$, and $\mathbf{v}_{OL}^t = 1$ indicates that ${\ell}_n$ must be a 2-tuple. 
Later, in Section~\ref{sec:oasc_overlap}, we will describe the formulation of one such method for computing $\mathrm{OVL}(R)$, using a hybrid HMM-DNN approach.

Given $S$ and $\mathbf{v}_{OL}$, overlap-aware diarization seeks to compute an optimal label sequence $L$ which minimizes the diarization error. 
Since we do not additionally have information about the number of speakers $K$ in the recording, we first estimate it using the heuristic described in~\citet{Park2020AutoTuningSC}. 
Subsequently, we perform multi-class spectral clustering to group the $N$ windowed subsegments into the estimated $\widehat{K}$ clusters using the optimal discretization procedure proposed in \citet{Yu2003MulticlassSC}, where we make a key modification to constrain the optimization process on the output $\mathbf{v}_{OL}$ of our overlap detector.

\subsection{Estimating number of speakers}

Given $S$, we compute the affinity matrix $\mathbf{A} \in [-1,1]^{N\times N}$ of raw cosine similarity values. 
Then, $p$-binarization is performed on this matrix by replacing the $p$ highest similarity values in each row with 1, and the rest with 0, followed by a symmetrization operation, 
\begin{equation}
\bar{\mathbf{A}}_p = \frac{1}{2}(\mathbf{A}_p + \mathbf{A}_p^T).
\end{equation}

We compute the unnormalized Laplacian for this matrix,
\begin{equation}
\mathbf{L}_p = \mathbf{D}_p - \bar{\mathbf{A}}_p,
\end{equation}
where $\mathbf{D}_p = \mathrm{diag}\{d_1,\ldots,d_N\}$, $d_i = \sum_{n=1}^N a_{i,n}$, also known as the ``degree'' of node $i$.

The properties of the unnormalized Laplacian of the affinity matrix have been studied extensively~\citep{Luxburg2007ATO}, and it is known that $\mathbf{L}_p$ has $N$ non-negative, real eigenvalues $0=\lambda_1 \leq \lambda_2 \leq \ldots \leq \lambda_N$.
Furthermore, an implication of the Davis-Kahan perturbation theory~\citep{Stewart1990MatrixPT} proposes an eigengap heuristic for the optimal number of clusters. 
Specifically, let $\mathbf{e}_p$ denote the vector of differences in consecutive eigenvalues (in increasing order). We compute the quantities
\begin{equation}
g_p = \frac{\max(\mathbf{e}_p)}{\lambda_{p,N}+\epsilon},\quad \text{and} \quad r(p) = \frac{p}{g_p}.
\end{equation}

Then, the optimal number of clusters is given as
\begin{equation}
    \widehat{K} = \arg\max(\mathbf{e}_{\hat{p}}), \quad \mathrm{where}~~\hat{p} = \mathrm{arg}\min_p r(p).
\end{equation}

In the multi-class spectral clustering procedure below, we will use this estimate $\hat{K}$ for the number of clusters, and drop the subscript $p$ from the matrices like $\mathbf{L}_p$, $\mathbf{D}_p$ and $\bar{\mathbf{A}}_p$ for brevity.

\subsection{Multi-class spectral clustering}

Bipartite graph partitioning using the affinity matrix Laplacian $\mathbf{L}$ is solved by node assignment based on the underlying Fiedler vector (eigenvector corresponding to the second smallest eigenvalue)~\citep{Fiedler1973AlgebraicCO}. 
The Ng-Jordan-Weiss algorithm~\citep{Ng2001OnSC} is a popular extension of this principle for multi-way partitioning of the graph. 
It applies K-means clustering on the first $K$ eigenvectors of $L$, i.e., in the $K$-eigenspace of the Laplacian, as described in Section~\ref{sec:oasc_njw}. 
It is known that if the original samples are separable into $K$ groups using some transformation, then their projection on the $K$-eigenspace can be easily grouped using K-means clustering. 
Recent work on speaker diarization through spectral clustering of x-vectors, as in \citet{Park2020AutoTuningSC} and \citet{Medennikov2020TargetSpeakerVA}, has employed this algorithm. 
However, there are two major limitations of this approach. 
First, the K-means clustering process may get stuck in bad local optima, particularly when the affinity matrix is noisy. 
Second, and particularly relevant for our case, it is difficult to extend this method to handle overlaps. 
To remedy these issues, we use an alternative formulation of spectral clustering, proposed in \citet{Yu2003MulticlassSC}. 

Given $\mathbf{A}$ and $\mathbf{D}$ as defined earlier, the clustering problem requires estimating the assignment matrix $X$. 
In graph partitioning terms, this can be represented as
\begin{equation}
\label{eqn:pncx}
\begin{aligned}
\mathrm{maximize} \quad \epsilon(X) &= \frac{1}{K}\sum_{k=1}^K \frac{X_k^T \mathbf{A} X_k}{X_k^T \mathbf{D} X_k} \\
\mathrm{subject~to}\,\,\,\,\,\, \quad X &\in \{0,1\}^{N\times K}, \\
    X \boldsymbol{1}_K &= \boldsymbol{1}_N.
\end{aligned}
\end{equation}

Intuitively, the objective function $\epsilon(X)$ seeks to maximize the average ``link-ratio'', i.e., the fraction of all link weights in a group that stay within the group. 
The constraint $X \boldsymbol{1}_K = \boldsymbol{1}_N$ enforces the condition that each sample can belong to exactly 1 cluster. 
We will see later (cf. Section~\ref{sec:oasc_new}) how this constraint can be modified for our overlap-aware scenario.

The optimization problem in~\eqref{eqn:pncx} is NP-complete due to the discrete constraints on $X$. 
Instead of solving this original problem, we solve a relaxed version of this problem which ignores the constraints. 
Let 
\begin{equation}
\label{eqn:z_def}
Z = f(X) = X(X^T\mathbf{D}X)^{-\frac{1}{2}}. 
\end{equation}

It is easy to verify that $Z^T\mathbf{D}Z = I_K$. 
We can rewrite the above problem~\eqref{eqn:pncx}, by ignoring the constraints, as

\begin{equation}
\label{eqn:pncz}
\begin{aligned}
\mathrm{maximize}\,\,\,\,\,\,\, \quad \epsilon(Z) &= \text{tr}(Z^T\mathbf{A}Z) \\
\mathrm{subject~to} \quad Z^T\mathbf{D}Z &= I_K.
\end{aligned}
\end{equation}

Since $Z$ has been relaxed into the continuous domain, the new optimization problem becomes tractable. 
Let 
\begin{equation}
P = \mathbf{D}^{-1}\mathbf{A},
\end{equation}
and suppose the eigen-decomposition of $P$ is given as $PV = VS$. 
Let $\Lambda^{\ast} = \mathrm{diag}(s_1,\ldots,s_K)$ and $Z^{\ast}$ contain the first $K$ columns of $V$. 
Then, the global optimum of the problem described in \eqref{eqn:pncz} occurs at
\begin{equation}
\label{eqn:z_sol}
    \{Z^{\ast}R: R^TR = I_K, PZ^{\ast} = Z^{\ast}\Lambda^{\ast}\}.
\end{equation}

This implies that the global optimum is not unique; rather, it is a subspace spanned by the first $K$ eigenvectors of $P$ through orthonormal matrices. 
The matrix $Z$ is a continuous solution to our clustering problem. 
To obtain a discrete solution, we solve for a discrete approximation for $Z$. 
First, we note from \eqref{eqn:z_def} that
\begin{equation}
X = f^{-1}(Z) = \operatorname{Diag}\left({\mathrm{diag}}^{-\frac{1}{2}}\left(Z Z^{T}\right)\right) Z.
\end{equation}

Using this transformation, we can characterize the solution obtained in \eqref{eqn:z_sol} as
\begin{equation}
     \{\Tilde{X}^{\ast}R: R^TR = I_K, \Tilde{X}^{\ast} = f^{-1}(Z^{\ast})\}.
\end{equation}

Now, our discretization problem is to find an $X$ which approximates $\Tilde{X}^{\ast}R$ for some orthonormal $R$, such that $X$ obeys the discrete constraints from problem~\eqref{eqn:pncx}. 
Mathematically, this is formulated as
\begin{equation}
\label{eqn:pod}
\begin{aligned}
\mathrm{minimize} \quad \phi(X,R) &= \left\lVert X - \Tilde{X}^{\ast}R \right\rVert^2 \\
\mathrm{subject~to}\, \quad\quad\quad X &\in \{0,1\}^{N\times K}, \\
    X \boldsymbol{1}_K &= \boldsymbol{1}_N, \\
    R^TR &= I_K.
\end{aligned}
\end{equation}

It is difficult to minimize $\phi(X,R)$ jointly in $X$ and $R$, so we optimize it alternately in $X$ and $R$. Suppose we are given some $R^{\ast}$, then the problem~\eqref{eqn:pod} reduces to
\begin{equation}
\label{eqn:podx}
\begin{aligned}
\mathrm{minimize} \quad \phi(X) &= \left\lVert X - \Tilde{X}^{\ast}R^{\ast} \right\rVert^2 \\
\mathrm{subject~to} \quad\quad X &\in \{0,1\}^{N\times K}, \\
    X \boldsymbol{1}_K &= \boldsymbol{1}_N.
\end{aligned}
\end{equation}

The solution to this problem is given by non-maximal suppression, i.e.,
\begin{equation}
\label{eqn:podx_sol}
    X^{{\ast}}(i, l)=\left\langle l=\arg \max _{k \in[K]} \tilde{X}(i, k)\right\rangle, \quad i \in \{1,\ldots,N\}.
\end{equation}
Intuitively, we set the largest entry in each row as one and zero out all the others. This ensures that each sample belongs to exactly 1 cluster. 
In the next section, we will see how to reformulate problem~\eqref{eqn:podx} for the case when some samples can belong to more than one clusters.
Next, we fix $X^{\ast}$ and solve the following problem for $R^{\ast}$:
\begin{equation}
\label{eqn:podr}
\begin{aligned}
\mathrm{minimize} \quad \phi(R) &= \left\lVert X^{\ast} - \Tilde{X}^{\ast}R \right\rVert^2 \\
\mathrm{subject~to} \quad R^TR &= I_K.
\end{aligned}
\end{equation}

The solution to this problem is given by
\begin{equation}
    R^{\ast} = \Tilde{U}U^T,
\end{equation}
where $(U,\Omega,\tilde{U})$ is a singular value decomposition of $X^{{\ast}T}\tilde{X}^{\ast}$.

We solve the two problems~\eqref{eqn:podx} and \eqref{eqn:podr} iteratively until convergence, and finally return $X^{\ast}$ as the output of the clustering procedure.

\subsection{Extension to overlap-aware clustering}
\label{sec:oasc_new}

From \eqref{eqn:overlap}, suppose the output of our overlap detector is given by $\mathbf{v}_{OL} \in \{0,1\}^N$, where $\mathbf{v}_{OL}^t = 1$ indicates that ${\ell}_n$ must be a 2-tuple. 
Then, we can reformulate problem~\eqref{eqn:podx} to include the overlap constraint as
\begin{equation}
\label{eqn:podx_ol}
\begin{aligned}
\mathrm{minimize} \quad  \phi(X) &= \left\lVert X - \Tilde{X}^{\ast}R^{\ast} \right\rVert^2 \\
\mathrm{subject~to} \quad \quad X &\in \{0,1\}^{N\times K}, \\
    X \boldsymbol{1}_K &= \boldsymbol{1}_N + \mathbf{v}_{OL}.
\end{aligned}
\end{equation}

Intuitively, this solves the same optimum discretization problem, but the cluster exclusivity constraint has been modified to represent the condition that some samples may belong to more than one cluster. 
Similar to how we used non-maximal suppression in \eqref{eqn:podx_sol} to solve the problem previously, the solution to the modified problem is again given by non-maximal suppression, with the exception that for samples belonging to more than one group, we set the largest two entries to 1, while zeroing out the others. 
Mathematically,
\begin{align}
\label{eqn:podx_newsol}
    \hat{X}^{{\ast}}(i, l)=X^{{\ast}}(i, l) + \mathbf{v}_{OL}^{(i)} \times \left\langle l=k^{\prime}_i\right\rangle, 
\end{align}
$~\forall i \in \{1,\ldots,N\}$, where $X^{{\ast}}(i, l)$ is defined in \eqref{eqn:podx_sol}, and $k^{\prime}_i$ is the index of the second largest element of $\Tilde{X}_i$. 
Although we only consider the case of 2-speaker overlaps here, it is easy to extend this method to the case of an arbitrary number of overlapping speakers.

\section{Overlap detection}
\label{sec:oasc_overlap}

Our proposed overlap-aware diarization method relies heavily on the performance of $\mathrm{OVL}(R)$, the overlap detector. 
In this section, we detail an HMM-DNN based overlap detector. 
Our model is similar to the speech activity detector previously used in the CHiME-6 baseline system~\citep{Watanabe2020CHiME6CT}.

We first trained a neural network classifier to assign each frame in an utterance a label from $\cal C$ = \{\textit{silence}, \textit{single}, \textit{overlap}\}, denoting silence, single speaker, or overlapping regions, respectively. 
We used the architecture shown in Figure~\ref{fig:od_nnet}, consisting of time-delay neural network (TDNN) layers to capture long temporal contexts \citep{Peddinti2015ATD}, interleaved with bidirectional long short term memory (BLSTM) layers with projection, to incorporate utterance-level statistics.

The posteriors obtained from the classifier were scaled with an external bias parameter tuned on the development data to reduce the false alarm rate. 
We then post-processed the per-frame classifier outputs to enforce minimum and maximum silence/single/overlap durations, by constructing a simple HMM whose state transition diagram encodes these constraints. 
Fig.~\ref{fig:oasc_fst} shows an example of a simplified HMM where these durations are set to small quantities.
Treating the per-frame posteriors like emission probabilities, we performed Viterbi decoding to obtain the most likely label-sequence. 
Furthermore, state transitions between the silence and overlap states were prohibited, mimicking real-world observations where it is highly unlikely for two speakers to start or stop speaking simultaneously. 

\begin{figure}[t]
\centering
\includegraphics[width=0.6\linewidth]{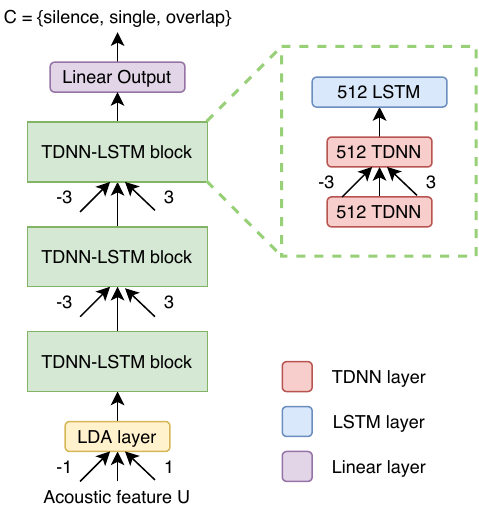}
\caption{Architecture of the neural network used for frame-level classification for overlap detection.}
\label{fig:od_nnet}
\end{figure}

\begin{sidewaysfigure}
    \centering
    \includegraphics[angle=270,width=\linewidth]{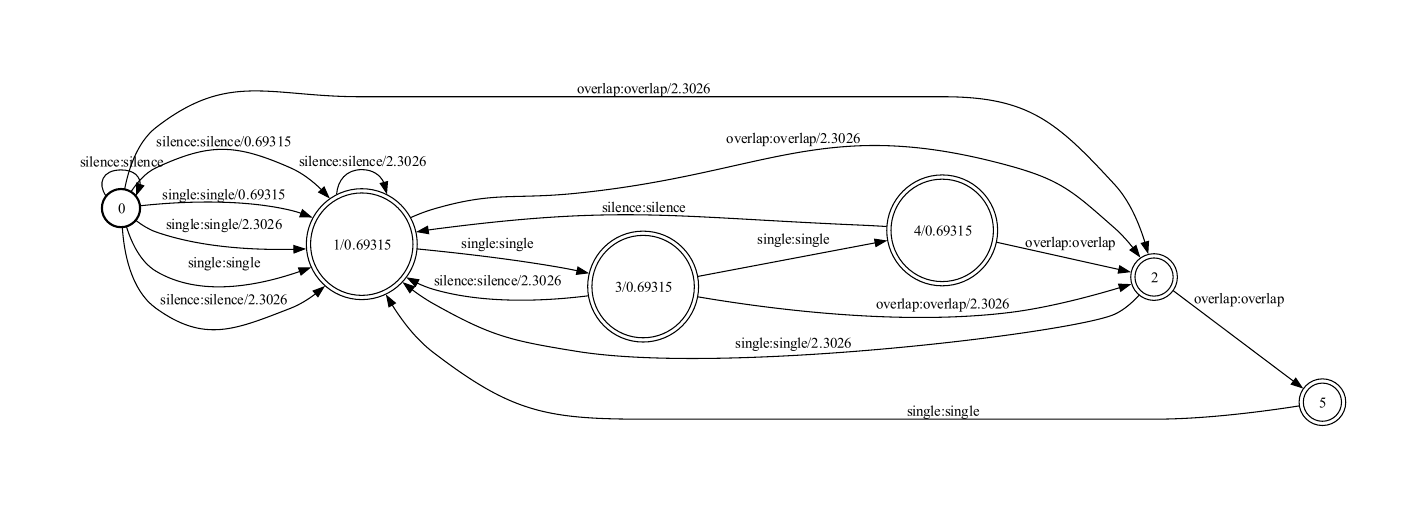}
    \caption{Simple HMM used for decoding graph construction. Here, we set the minimum silence/speech/overlap duration as 0.0s/0.0s/0.1s, and maximum speech/overlap duration as 0.03s/0.02s, respectively. Note that each transition correspond to one frame in the input, which has a frame shift of 0.01s in this example.}
    \label{fig:oasc_fst}
\end{sidewaysfigure}

\section{Experimental setup}
\label{sec:oasc_experiments}

\subsection{Datasets}

We performed experiments on the AMI meeting corpus and the LibriCSS data.
Recall from Section~\ref{sec:intro_data} that AMI contains 4-speaker meetings, while LibriCSS is a simulated corpus containing 8-speaker recordings.
For our AMI experiments, we used the mix-headset recordings, which are obtained by summing the individual headset signals from the participants in the meeting. 
For our LibriCSS experiments, we selected the recordings from the first channel of the array. 
We used this dataset to conduct a performance analysis of our proposed method on different overlap conditions. 

\subsection{Baselines}

We first have single-speaker baselines: (i) agglomerative hierarchical clustering (AHC) of x-vectors with probabilistic linear discriminant analysis (PLDA) scoring~\citep{Sell2018DiarizationIH}, (ii) spectral clustering of x-vectors with cosine scoring (using the Ng-Jordan-Weiss method)~\citep{Park2020AutoTuningSC}, and (iii) Bayesian HMM based x-vector clustering (VBx)~\citep{Diez2019BayesianHB,Dez2020OptimizingBH}. 
We used the same x-vector extractor for all the baselines (described in Section~\ref{sec:oasc_implementation}), so that the difference in their performance was only due to the clustering process. 
Furthermore, we used the same PLDA model (trained on a subset of the AMI training data) for the AHC and VBx baselines\footnote{The VBx diarization system has been shown to obtain significant gains with a PLDA interpolated between general data and in-domain data, but we did not use this method in this chapter.}. 
For both these baselines, hyperparameters were tuned on the development set. 
No hyperparameter selection is required for spectral clustering since it is auto-tuned. 
We used a ground-truth VAD for these baselines as well as for our proposed method.

We also compare our approach with diarization methods that do not ignore overlaps. 
These include: (i) overlap-aware VB resegmentation~\citep{Bullock2019OverlapawareDR} and (ii) region proposal networks (RPNs)~\citep{Huang2020SpeakerDW}. 
For the former, we report the results from the paper, which uses a neural VAD. For RPNs, we filtered out non-speech regions using the ground truth VAD.
All results are reported in terms of diarization error rate (DER), and we further break it down into missed speech (MS), false alarms (FA), and speaker confusion (Conf.) errors.

\subsection{Implementation details}
\label{sec:oasc_implementation}

\noindent
\textbf{X-vector extractor}. We used an x-vector extractor similar to the ones described in earlier studies~\citep{GarciaRomero2017SpeakerDU,Park2020AutoTuningSC}. 
The model consists of TDNN layers with statistics pooling, and we extracted 128-dim embeddings from the pre-final layer. 
It was trained on VoxCeleb data~\citep{nagrani2017voxceleb} with simulated room impulse responses~\citep{ko2017study} using the Kaldi toolkit~\citep{povey2011kaldi}, and released as part of the CHiME-6 baseline~\citep{Watanabe2020CHiME6CT}.

\noindent
\textbf{Overlap detector}. We trained an HMM-DNN overlap detector (described in Section~\ref{sec:oasc_overlap}) using Kaldi. 
We used 40-dim MFCCs features as input, and trained the classifier on in-domain training data. 
For AMI, we used targets obtained from annotations of the official training set. 
Since LibriCSS does not have corresponding training data, we generated simulated mixtures with reverberation using Librispeech training utterances~\citep{Panayotov2015LibrispeechAA} and used force-aligned targets for training our overlap detector. 
The decoding graph was created with additional constraints on the minimum and maximum durations: 0.03s and 10.0s for single speakers and 0.1s and 5.0s for overlapping regions, respectively.
Note that our clustering method itself is independent of the overlap detector used. 

\noindent
\textbf{Overlap-aware spectral clustering}. We extended the spectral clustering algorithm in scikit-learn~\citep{scikit-learn} for our implementation. 
Since the overlap detector provides frame-level classification decisions whereas x-vectors were extracted for 1.5s segments, we assumed that a segment is ``overlapping'' if at least half of it lies in overlapping regions. 
For estimating the number of speakers $\widehat{K}$, we swept the binarization factor $p$ in the range from 2 to 20, similar to what was done in \citet{Park2020AutoTuningSC}.   

\section{Results and discussion}
\label{sec:oasc_results}

\subsection{Overlap detection on AMI}

We present the results obtained by our overlap detector using 40-dim MFCC features on AMI mix-headset data in Table~\ref{tab:overlap_results}. 
We can see that the performance is comparable to previous studies on this dataset, without using waveform-level learned features. 
Furthermore, the results are similar regardless of the input features used for the overlap detection.
In general, we prefer an operating point with a relatively high precision, at the cost of low recall.
This is because if the overlap detector falsely identifies a sub-segment as an overlapping one, the clustering process would assign 2 speakers to it, resulting in increased false alarms.
On the other hand, if we fail to detect an overlapped sub-segment, it would not result in any extra errors compared with a regular (non-overlapping) diarization system.
For the experiments reported in the subsequent sections, we used the output from the above overlap detector using MFCC features.

\begin{table}[t]
\centering
\caption{Overlap detection results on AMI mix-headset data, in terms of Precision (\%) and Recall (\%).}
\label{tab:overlap_results}
\begin{adjustbox}{max width=\linewidth}
\begin{tabular}{@{}lcccc@{}}
\toprule
\multicolumn{1}{@{}l}{\multirow{2}{*}{\textbf{Model (feature type)}}} & \multicolumn{2}{c}{\textbf{Dev}} & \multicolumn{2}{c}{\textbf{Eval}} \\
\cmidrule(r{4pt}){2-3} \cmidrule(l){4-5}
\multicolumn{1}{c}{} & \textbf{Precision} & \textbf{Recall} & \textbf{Precision} & \textbf{Recall} \\
\midrule
ConvNet (Spectogram)~\citep{Kunesov2019DetectionOO} & 80.5 & 50.2 & 75.8 & 44.6 \\
E2E BLSTM (MFCC)~\citep{Bullock2019OverlapawareDR} & 90.0 & 52.5 & 91.9 & 48.4 \\
E2E BLSTM (SincNet)~\citep{Bullock2019OverlapawareDR} & 90.0 & 63.8 & 86.8 & 65.8 \\
\midrule
Our method (FBank) & 81.2 & 68.1 & 85.1 & 63.6 \\
Our method (MFCC) & 83.9 & 68.5 & 86.4 & 65.2 \\
Our method (MFCC+pitch) & 79.1 & 69.3 & 82.4 & 66.2 \\
\bottomrule
\end{tabular}
\end{adjustbox}
\end{table}

\subsection{Diarization results for AMI}
\label{sec:ami_result}

Table~\ref{tab:diar_results} shows our proposed overlap-aware spectral clustering method compared with baselines, evaluated on the AMI mix-headset eval data. 
First, the simple clustering-based baselines (AHC, spectral, and VBx) obtain a high MS error rate, since they cannot assign more than one speaker to any sub-segment.
Since we used an oracle VAD, the corresponding FA rate is 0\%, but this is likely to increase when using automatic VADs.
The speaker confusion varies depending on the clustering method, and VBx is found to be the most accurate in this setting.

Next, we see that the baselines with overlap assignment (such as VB resegmentation~\citep{Bullock2019OverlapawareDR} and RPN~\citep{Huang2020SpeakerDW}) improve the overall DER, mainly due to improvement in MS.
Nevertheless, RPN was found to significantly increase FAs, perhaps since it is trained on overlapping mixtures, and is prone to over-assignment of speakers.

Using our overlap detector trained on the AMI train set, our overlap-aware spectral clustering method was able to improve the DER from 28.3\% for the AHC/PLDA baseline, to 24.0\%, which is a relative improvement of 15.2\%.
This compares favorably with the performance of other diarization methods like overlap-aware VB resegmentation and RPNs. 
Furthermore, it is possible to reduce the DER to 21.5\% using an oracle overlap detector.
For all the recordings in the test set, the NME-based speaker counting approach estimated between 3 and 6 speakers, which is close to the oracle count of 4 speakers.

\begin{table}[t]
\centering
\caption{Diarization results on AMI mixed-headset eval set.}
\label{tab:diar_results}
\begin{adjustbox}{max width=\linewidth}
\begin{tabular}{@{}lcccc@{}}
\toprule
\multicolumn{1}{@{}l}{\textbf{Method}} & \textbf{MS} & \textbf{FA} & \textbf{Conf.} & \textbf{DER} \\ \midrule
AHC/PLDA & 19.9 & 0.0 & 8.4 & 28.3 \\
Spectral/cosine & 19.9 & 0.0 & 7.0 & 26.9 \\
VBx~\citep{Diez2019BayesianHB} & 19.9 & 0.0 & 6.3 & 26.2 \\ \midrule
VB resegmentation~\citep{Bullock2019OverlapawareDR} & 13.0 & 3.6 & 7.2 & 23.8 \\
RPN~\citep{Huang2020SpeakerDW} & 9.5 & 7.7 & 8.3 & 25.5 \\ \midrule
Our method & 11.3 & 2.2 & 10.5 & 24.0 \\
Our method + oracle OD & 7.4 & 1.8 & 12.3 & \underline{21.5} \\ 
Our method + noise aug. & 11.3 & 2.2 & 10.1 & \textbf{23.6} \\ \bottomrule
\end{tabular}
\end{adjustbox}
\end{table}

A detailed analysis of the results reveals that although the missed speech reduces substantially (from 19.9\% to 11.3\%) as a result of overlap detection, there is also a significant increase in speaker confusion errors (from 8.4\% to 10.5\%). 
We conjecture that since the x-vector extractor was trained only on single-speaker utterances, a mismatch in the overlap regions of the recording results in noisy samples. 
The speaker confusion improved by 0.4\% when we used an x-vector extractor trained with noise augmentation, using noises from the MUSAN corpus~\citep{Snyder2015MUSANAM}. 
To verify our hypothesis further, we show the T-SNE plots for the non-overlapping and overlapping segments in Fig.~\ref{fig:non_ovl} and \ref{fig:ovl}, respectively. 
We can see that while the embeddings for the non-overlapping segments are well separated, those for overlapping segments may often be noisy, leading to clustering errors.

\begin{figure}[t]
\begin{subfigure}{0.49\linewidth}
\centering
\includegraphics[width=\linewidth]{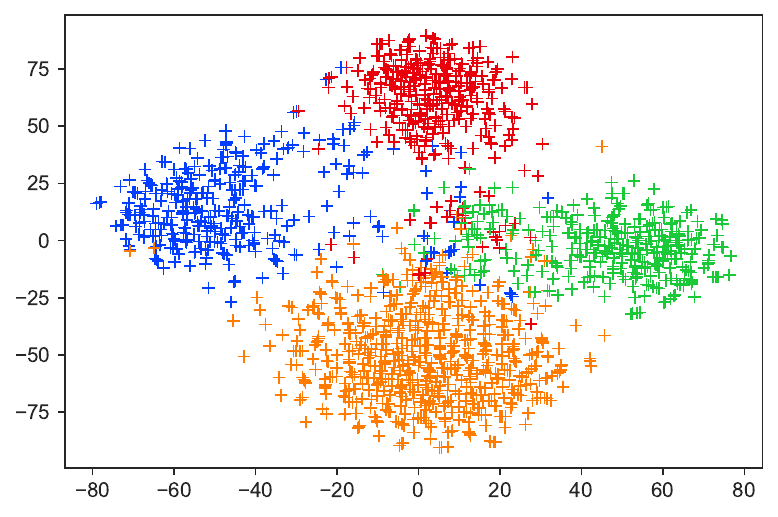}
\caption{}
\label{fig:non_ovl}
\end{subfigure}
\begin{subfigure}{0.49\linewidth}
\centering
\includegraphics[width=\linewidth]{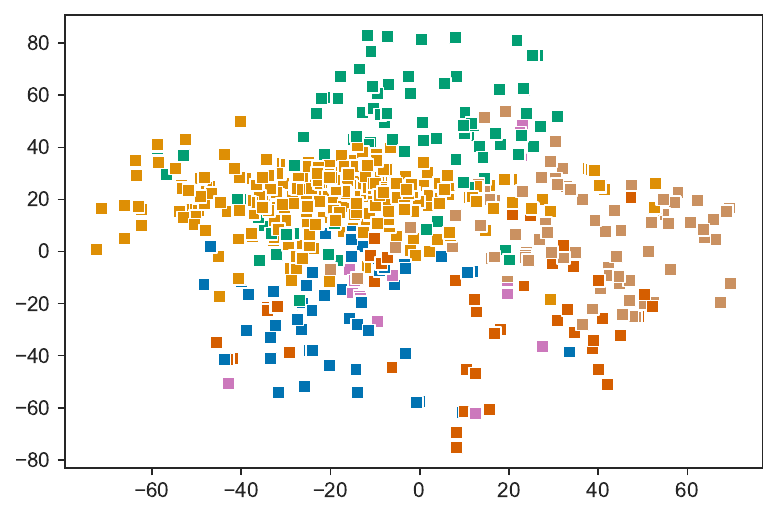}
\caption{}
\label{fig:ovl}
\end{subfigure}\hfill
\caption{T-SNE plots of x-vector embeddings for (a) non-overlapping, and (b) overlapping segments for the recording \texttt{EN2002a} in the AMI eval set (containing 4 speakers). Colors denote the speaker assigned to the segment. For (b), each color represents a distinct pair of speaker labels, resulting in 6 differently colored clusters.}
\label{fig:xvec_plot}
\end{figure}

\begin{table}[t]
\centering
\caption{Diarization performance on LibriCSS evaluation set (sessions 2-10), evaluated condition-wise, in terms of \% DER. 0S and 0L refer to 0\% overlap with short and long inter-utterance silences, respectively. Our overlap detector obtained 96.3\% precision and 83.8\% recall on this data. RPN does not use ground truth VAD.}
\label{tab:libricss}
\begin{adjustbox}{max width=\linewidth}
\begin{tabular}{@{}lccccccc@{}}
\toprule
\multicolumn{1}{c}{\multirow{2}{*}{\textbf{Method}}} & \multicolumn{6}{c}{\textbf{Overlap ratio in \%}} & \multicolumn{1}{c}{\multirow{2}{*}{\textbf{Average}}} \\
\cmidrule(r{4pt}){2-7}
\multicolumn{1}{c}{} & \multicolumn{1}{c}{\textbf{0L}} & \multicolumn{1}{c}{\textbf{0S}} & \multicolumn{1}{c}{\textbf{10}} & \multicolumn{1}{c}{\textbf{20}} & \multicolumn{1}{c}{\textbf{30}} & \multicolumn{1}{c}{\textbf{40}} & \multicolumn{1}{c}{} \\
\midrule
AHC/PLDA & 5.1 & 2.7 & 11.5 & 19.3 & 22.7 & 28.8 & 16.3 \\
Spectral/cosine & 2.2 & 1.9 & 8.5 & 13.5 & 19.0 & 23.4 & 12.6 \\
RPN & 4.5 & 9.1 & 8.3 & 6.7 & 11.6 & 14.2 & 9.5 \\
\midrule
Our method & 2.6 & 3.4 & 6.8 & 10.0 & 13.9 & 15.2 & 9.3 \\
 + oracle OD & 2.2 & 3.3 & 6.7 & 9.6 & 12.9 & 14.4 & 8.8 \\
\bottomrule
\end{tabular}
\end{adjustbox}
\end{table}

\subsection{Analysis on LibriCSS}

Table~\ref{tab:libricss} shows a breakdown of diarization errors obtained by our system, 
compared with some of the baselines. 
It is evident that as the overlap ratio increases from 0 to 40\%, the difference in performance becomes more significant. 
On average, our method provided a 42.9\% relative DER improvement compared to a baseline AHC system, and this increased to 46.0\% relative on using an oracle overlap detector. 
We note here that since LibriCSS does not have a corresponding training set, the PLDA was trained on Librispeech utterances. 
The mismatch between clean training data versus overlapping mixtures at test time may be particularly detrimental to the performance of the AHC system. 
As the overlap ratio increases, RPN performs better than our method. 
We again attribute this to the fact that our RPN model is trained on closely matched overlapping speech, whereas the x-vector extractor was trained on single-speaker utterances, which results in a higher speaker confusion.
Nevertheless, in real meetings, it is relatively rare for overlap statistics to be more than 30\% (for example, AMI has 19.9\% overlapping speech), and in such conditions, our method obtains similar performance to the discriminatively-trained RPN.

\section{Conclusion}
\label{sec:oasc_conclusion}

Speaker diarization is an important component of the multi-talker ASR problem, since it performs speaker attribution for the transcriptions.
In this chapter, we showed that the traditional clustering-based diarization system is inherently incapable of handling overlapping speech, which forms a significant fraction of multi-talker recordings.
We proposed a new method for overlap-aware speaker diarization using spectral clustering by leveraging an external overlap detector to identify the overlapping subsegments, and then assigning these segments to multiple speakers during clustering. 
Our method provided significant improvements over conventional single-speaker clustering models, and was competitive with other overlap-aware diarization methods.

As described in Section~\ref{sec:oasc_diar}, there are several newer approaches for overlap-aware diarization, such as those which reformulate the problem as multi-label classification.
These methods, which are based on supervised training on multi-speaker mixtures, are often more robust at handling overlapping speech.
Nevertheless, clustering-based approaches are still prominent because of their robustness across acoustic domains and number of speakers.
Given such complementary systems, we ask the question: \textit{can we benefit by ensembling them to benefit from their respective strengths?}
This question is addressed in the next chapter.

% \cleardoublepage

\chapter{Ensembles of Diarization Systems}
\label{chap:doverlap}

In the previous chapter, we described a method for overlap-aware speaker diarization using spectral clustering.
Considering the importance of diarization as a pre-processing step, it is not surprising that there is a large body of research tackling this problem with very different solutions.
These solutions have complementary strengths, such as performance on overlapping speech, accurate estimation of number of speakers, or robustness on out-of-domain acoustic environments.
In this chapter, we propose \textbf{DOVER-Lap}, an algorithm that can combine the outputs of these diverse systems, thereby improving error rates across a range of conditions.
In the context of our over-arching multi-talker ASR pipeline, we expect such ensembling methods to be fast and accurate for any number of diarization systems.
We will show how to adapt ideas from approximation algorithms to design a suite of such ensembling algorithms with trade-offs between computational cost and approximation ratios.

This chapter is organized as follows.
Section~\ref{sec:dl_intro} describes commonly used approaches for speaker diarization in literature, along with their advantages and disadvantages.
In Section~\ref{sec:dl_problem}, we will formalize the system combination problem and introduce DOVER-Lap, which decomposes the above problem into the sub-problems of ``label mapping'' and ``label voting.''
Section~\ref{sec:dl_mapping} reformulates label mapping as a graph partitioning problem, which allows us to describe a fast and accurate approximation algorithm for the task.
In Section~\ref{sec:dl_voting}, we present the relatively easier problem of voting, and describe a simple maximal voting strategy which solves it.
Finally, we present results for system combination in Section~\ref{sec:dl_results}, where we perform DOVER-Lap based ensembling for AMI and LibriCSS using a variety of speaker diarization systems.
We also experiment with multi-channel diarization through a late fusion of channel-level diarization outputs in Section~\ref{sec:dl_multi}, and show that this simple strategy is competitive with beamforming-based early fusion.

\section{Introduction}
\label{sec:dl_intro}

In the 2000s, NIST and DARPA organized several challenges to advance the state of speaker diarization research. 
Since statistical learning-based approaches for diarization were still fairly new at the time, the evaluations were suitably constrained --- participants were often provided oracle speech segmentation (i.e., non-speech regions were demarcated), the number of speakers was known beforehand, and the audio recordings were assumed to have zero overlapping speech. 
More recently, as the diarization task has matured, we have moved towards unconstrained evaluations. 
The CHiME-6~\cite{Watanabe2020CHiME6CT} and DIHARD~\cite{Ryant2019TheSD, Ryant2020TheTD} challenges, for example, contain challenging recordings with high noise and overlapping speech and a diverse number of speakers.

This shift in evaluation happened as a result of several advances in systems that perform speaker diarization. 
The traditional approach for speaker diarization, as described in Chapter~\ref{chap:oasc}, involved a clustering of segment-level speaker embeddings, optionally followed by resegmentation~\cite{GarciaRomero2017SpeakerDU, Park2020AutoTuningSC, Landini2020BayesianHC}. 
This approach requires separately optimized speech activity detection components, and also assumes that the recording does not contain overlapping speech. 
To alleviate the latter problem, several methods have been proposed which seek to employ overlap detection modules and use some heuristics of the clustering process to assign extraneous speakers to the overlapping segments~\cite{Raj2021MulticlassSC, Bullock2020OverlapawareDR}. 

More recently, supervised diarization methods such as region proposal networks (RPN), end-to-end neural diarization (EEND), and target-speaker voice activity detection (TS-VAD) have been proposed which inherently perform overlapping speaker assignment~\cite{Huang2020SpeakerDW,Fujita2020EndtoEndND,Medennikov2020TargetSpeakerVA}. 
An alternate paradigm for overlap-aware diarization involves using a continuous speech separation (CSS) module to first split the recording into 2 or 3 streams, followed by clustering of speakers across the streams~\cite{Raj2021IntegrationOS, Xiao2020MicrosoftSD}. 
%
% Fig.~\ref{fig:methods} provides an overview of some of these diarization systems.

% \begin{figure}[t]
%     \centering
%     \includegraphics[width=\linewidth]{figs/methods.pdf}
%     \caption{An overview of some of the diarization methods proposed in literature. The methods in grey boxes cannot handle overlapping speech.}
%     \label{fig:methods}
% \end{figure}

Machine learning tasks usually benefit from an ensemble of different systems~\cite{Rokach2009EnsemblebasedC}.
ROVER~\cite{Fiscus1997APS} is a popular post-processing method for combining the outputs of speech recognition systems through weighted majority voting. 
This method was generalized to be used with n-best lists or lattices, and called confusion network combination (CNC)~\cite{evermann2000posterior,stolcke2000sri}.
Lattice combination has also proved to be effective in systems developed for community challenges such as the CHiME-6 challenge~\cite{arora2020jhu}. 
However, it has traditionally been difficult to combine diarization systems at the output level due to the varying system-dependent time-segmentations of the hypotheses. 

DOVER~\cite{Stolcke2019DoverAM} was the first algorithm proposed to perform such combination through weighted majority voting (similar to ROVER for ASR) on homogeneous single-speaker regions across a recording. 
Evaluation results showed that it improved over random (or even oracle) channel selection, albeit with several caveats.
While DOVER provided a convenient method to combine diarization outputs, it was limited since it could not handle outputs containing overlapping segments. 
This limitation is of particular significance if we consider the increasing application of diarization to real multi-talker conversations containing high overlaps (such as the AMI~\cite{Carletta2005TheAM} meeting corpus and the CHiME challenges~\cite{Barker2018TheF,Watanabe2020CHiME6CT}), and the availability of methods that can perform diarization in such settings~\cite{Bullock2019OverlapawareDR,Huang2020SpeakerDW,Fujita2020EndtoEndND}.

Our objective in this chapter is to devise an algorithm which can perform hypotheses combination in overlap-aware settings. 
Towards this objective, we propose \textbf{DOVER-Lap} (D\underline{OVER} + \underline{Over}lap), which generalizes and builds upon the label mapping and label voting paradigms suggested in DOVER.
We reformulate the label mapping task as graph partitioning, and propose fast approximation algorithms to solve it.
We also modify the voting mechanism in DOVER such that multiple speakers may be assigned to a region based on the voting decision. 

We demonstrate through our experiments conducted on AMI and LibriCSS that DOVER-Lap can efficiently combine hypotheses from very different systems, ranging from conventional clustering-based to more recent end-to-end neural methods, while improving the DER performance over the single best system. 
To demonstrate the wide applicability of our technique, we also apply it to multichannel diarization through late fusion of outputs from far-field array microphones. 
We show that late fusion using DOVER-Lap is competitive with early fusion using dereverberation and beamforming, without any need for enhancement. 

% The key contributions in this chapter are three-fold: 
% \begin{enumerate}
%     \item We propose a fast and robust method for combining diarization hypothesis across multiple systems with overlap handling. An important part of this contribution is reformulating the label mapping using an approximation algorithm for the weighted maximal $k$-partite matching problem.
%     \item We demonstrate the effectiveness of our method on diverse systems, such as overlap-aware spectral clustering, VB resegmentation, region proposal network (RPN), and target-speaker voice activity detection (TS-VAD). To the best of our knowledge, a combination of hypotheses from such diverse systems has not been studied before.
%     \item Finally, we show that our method can additionally be used to perform multichannel diarization on far-field array microphones through late fusion.
% \end{enumerate}

\section{System combination with DOVER-Lap}
\label{sec:dl_problem}

\subsection{Problem formulation}

We assume that we are given the ``hypotheses'' $U_k$ from $K$ diarization systems, i.e., $U_1,\ldots,U_K$, where each system can perform overlapping speaker assignment and may contain different number of speakers, $c_1,\ldots,c_K$.
Mathematically, $U_k$ can be written as a set of tuples containing time intervals and the corresponding speaker assignment. 
If $U_k$ has $c_k$ speakers $\{u_k^1,\ldots,u_k^{c_k}\}$ and $N_k$ speaker-homogeneous segments, then 

\begin{equation}
    U_k = \{(\Delta_k^n, u_k^n): n \in \{1,\ldots,N_k\}\},
\end{equation}
where $\Delta_k^n$ is the time interval $n$ as defined in \eqref{eq:oasc_1}, and $u_k^n \in \{u_k^1,\ldots,u_k^{c_k}\}$ is the corresponding speaker.

We define the problem of combining the hypothesis as finding a joint diarization hypothesis $\hat{U} = \mathcal{D}(U_1,\ldots,U_K)$ such that $\hat{U}$ minimizes the chosen diarization error metric with respect to an unknown reference. 

A straightforward approach for solving this problem involves dividing the input hypotheses into small time durations, and performing majority voting individually with each such region. 
However, there are two issues with this solution. 
First, to perform any kind of voting (within a region), the system outputs need to be in the same label space. 
Second, overlap-aware systems may contain different number of speakers within any such region, and for overlap-aware combination, the voting mechanism needs to account for this possibility. 
In the next section, we describe how DOVER~\cite{Stolcke2019DoverAM} solved this problem by decomposing it into label mapping and label voting stages.

\begin{table}[t]
\centering
\caption{Notations used in this chapter.}
\label{tab:notation}
\begin{tabular}{@{}rl@{}}
\toprule
\textbf{Symbol} & \textbf{Definition} \\
\midrule
$U_1,\ldots,U_K$ & Diarization hypotheses \\
$c_k$ & Number of speakers in hypotheses $k$ \\
$C$ & Maximum number of speakers in any hypothesis \\
$u_k^i$ & Speaker $i$ in hypothesis $k$ \\
$\Delta_k^n$ & Interval $n$ in hypotheses $k$ \\
$V$ & Set of all hypotheses speakers \\
$E$ & Set of all edges $\{(u_k^{i},u_{\ell}^{j})\}$ \\
$w_e$ & Weight on edge $e = \{(u_k^{i},u_{\ell}^{j})\}$ \\
$(V_1,\ldots,V_C)$ & Clique set (output of label mapping) \\
$\mathcal{G}$ & Graph formed by vertex set $V$ and edge set $E$ with weights $w$ \\
\bottomrule
\end{tabular}
\end{table}

% \begin{figure}[t]
% \centering
% \begin{subfigure}{0.55\linewidth}
% \centering
% \includegraphics[width=\linewidth]{doverlap_chapter/figs/bipartite_graph.pdf}
% \label{fig:bipartite_graph}
% \end{subfigure}
% \begin{subfigure}{0.4\linewidth}
% \centering
% \raisebox{0.7\height}{\includegraphics[width=0.8\linewidth]{doverlap_chapter/figs/bipartite_matrix.pdf}}
% \label{fig:bipartite_matrix}
% \end{subfigure}\hfill
% \caption{DOVER uses the Hungarian method pair-wise for label mapping. It iteratively solves the weighted bipartite graph matching problem using linear sum assignment on the cost matrix.}
% \label{fig:hungarian}
% \end{figure}

\subsection{Preliminary: DOVER}
\label{sec:dl_dover}

DOVER (diarization output voting error reduction) was the first system combination method proposed for combining diarization outputs using a weighted majority voting technique~\cite{Stolcke2019DoverAM}.
It comprises two stages: label mapping, and label voting, where each stage addresses one of the problems outlined in the previous section.

\subsubsection{Label mapping}

In the \textbf{label mapping} stage, the objective is to find a set of $C$ speaker labels $\{\hat{u}^1,\ldots,\hat{u}^C\}$ for our combined output $\hat{U}$, and a mapping function $\mathcal{S}$ that takes as input one of the speaker labels from the given hypotheses and maps it to a label in the combined output, i.e., $\mathcal{S}(u_k^i) = \hat{u}^c$, where $c \in [C]$.

In DOVER, this mapping is done incrementally by considering the hypotheses pair-wise, with their order decided based on the average diarization error rate (DER) to all other hypotheses.
Suppose, for example, that the hypotheses have been ordered as $U_1,\ldots,U_K$. 
In this case, $U_1$ and $U_2$ are first mapped together using $U_1$ as reference to obtain $U_{1,2}$, which is then mapped together with $U_3$ to obtain $U_{1,2,3}$, and so on, until we have $U_{1,\ldots,K}$. 
At each iteration, the mapping is performed using the Hungarian method~\cite{Kuhn1955TheHM}, similar to how the reference and system outputs are mapped to a common space for evaluating DER. 
This is an instance of the weighted bipartite graph matching problem, and is poly-time solvable using linear sum assignment on the weight matrix. 
This incremental approach for label mapping treats it as an incremental assignment problem~\cite{Toroslu2007IncrementalAP}. 
Due to this treatment, it is unable to use the global pair-wise costs to map all the hypotheses to a common space simultaneously. 
Furthermore, the method is dependent on the order of hypotheses $H_k$, and choosing a bad initial pair may be detrimental to the final mapping. 
In Section~\ref{sec:dl_mapping}, we will reformulate this label mapping problem in terms of graph partitioning, and propose fast algorithms that seek to overcome these limitations with DOVER.

\subsubsection{Label voting}

\begin{figure}[t]
\centering
\includegraphics[width=0.7\linewidth]{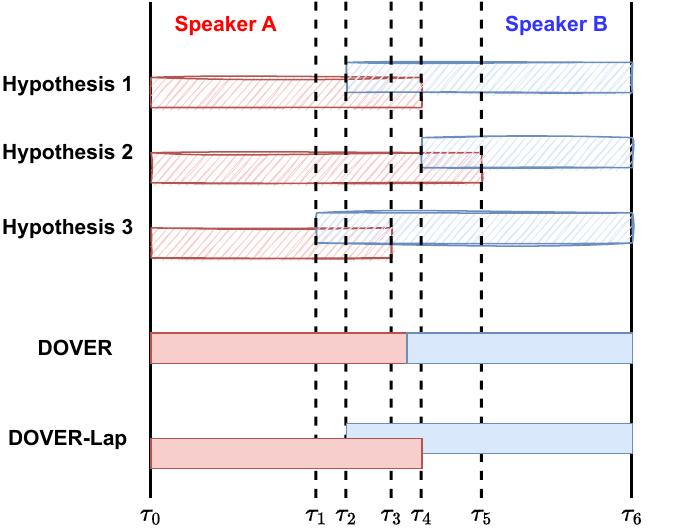}
\caption{An illustration of overlapping output produced by DOVER-Lap for overlapping hypotheses.}
\label{fig:dover_vs_dl}
\end{figure}

Once the hypotheses have been mapped to a common speaker label space, the \textbf{label voting} stage is performed. 
In DOVER, weighted majority voting is done on regions of input speech, where a region is defined as the maximal segment delimited by any of the original speaker boundaries from the input hypotheses.
Given the speaker-homogeneous intervals $\Delta_k^n$ for the $K$ hypothesis, we pool the start and end time-stamps for all the intervals and order them in monotonically increasing order.
Let this new sequence be denoted by $\{\tau_0,\ldots,\tau_{N-1}\}$, where $N = \sum_{k=1}^K 2N_k$; then, each interval $[\tau_n,\tau_{n+1}]$ is an example of a \textit{region} defined above.
For all illustration, see the intervals $[\tau_0,\tau_1],\ldots,[\tau_5,\tau_6]$ in Fig.~\ref{fig:dover_vs_dl}. 

Recall that overlap-aware diarization systems may contain multiple speakers in each such region. 
However, DOVER makes the single-speaker assumption in majority voting. 
Formally, if the audio is divided into $N$ regions, and $L_k = (\ell_k^1,\ldots,\ell_k^T)$ denotes the region-wise labels assigned by hypothesis $k \in K$, where $\ell_k^t \in \{\hat{u}^1, \ldots, \hat{u}^C\}$ (from label mapping), and $C$ is the number of speakers in combined label space, then the DOVER label voting stage computes, $\forall t \in T$,
\begin{equation}
\label{eqn:dover}
\ell_t = \arg \max_{\hat{u}^c \in \{\hat{u}^1, \ldots, \hat{u}^C\}} \left( \sum_{k \in K} w_k \mathbbm{1}(\ell_k^t = \hat{u}^c) \right),
\end{equation}
where $w_k \in [0,1]$ denotes a confidence weight assigned to hypothesis $k$. 
DOVER ranks the input hypotheses by their average DER to all other hypotheses, and applies a weight that decays slowly with rank: $w_k = \frac{1}{k^{0.1}}$. 
Since only a single speaker is assigned to every region, combination using DOVER may lead to high missed speech in the overlap case, as shown in Fig.~\ref{fig:dover_vs_dl}. 
We solve this problem through overlap-aware weighted majority voting, described in Section~\ref{sec:dl_voting}.

\section{Label mapping as a graph partitioning problem}
\label{sec:dl_mapping}

Again, consider $K$ diarization hypotheses $U_1,\ldots,U_K$, containing $c_1,\ldots,c_k$ speakers, respectively, such that $C = \max\{c_k, k\in [K]\}$. 
Let us denote each speaker as a node in a graph, i.e., $u_k^i$ is the node corresponding to the $i^{th}$ speaker in the $k^{th}$ hypothesis, and $V = \{u_k^i\}$ is the set of all speaker nodes. 
Let $E = \{(u_k^{i},u_{\ell}^{j}): \forall k,\ell \in [K], i \in U_k, j\in U_\ell, k\neq \ell\}$ denote the set of all edges. 
Informally, this means that there is an edge between any two nodes if the nodes belong to different hypotheses. 
Additionally, we have a weight function $w: e \rightarrow \mathbb{R}^+$, where $e$ denotes an edge. 
In practice, these edge weights are obtained by computing the relative overlap duration between the speakers in the recordings, i.e.,
\begin{equation}
\label{eq:edge_weight}
w(u_k^{i},u_{\ell}^{j}) = \frac{\Delta(u_k^{i})\cap \Delta(u_{\ell}^{j})}{\Delta(u_k^{i})\cup \Delta(u_{\ell}^{j})},
\end{equation}
where $\Delta(u)$ is the set of all segments where speaker $u$ is active in the recording. 
Clearly, $w: E \rightarrow [0,1]$, and a higher $w$ means that the corresponding speakers are more likely to occur in the same segments in the recording (i.e., they are more likely to represent the same speaker).

We define the graph as $\mathcal{G} = (V,E,w)$. 
It is easy to see that $\mathcal{G}$ is $K$-partite, and if $c_k = C, ~\forall k\in [K]$, then it is also complete. 
Fig.~\ref{fig:label_mapping_graph} illustrates this graphical formulation of the label mapping problem.
\begin{figure}[t]
    \centering
    \includegraphics[width=0.8\linewidth]{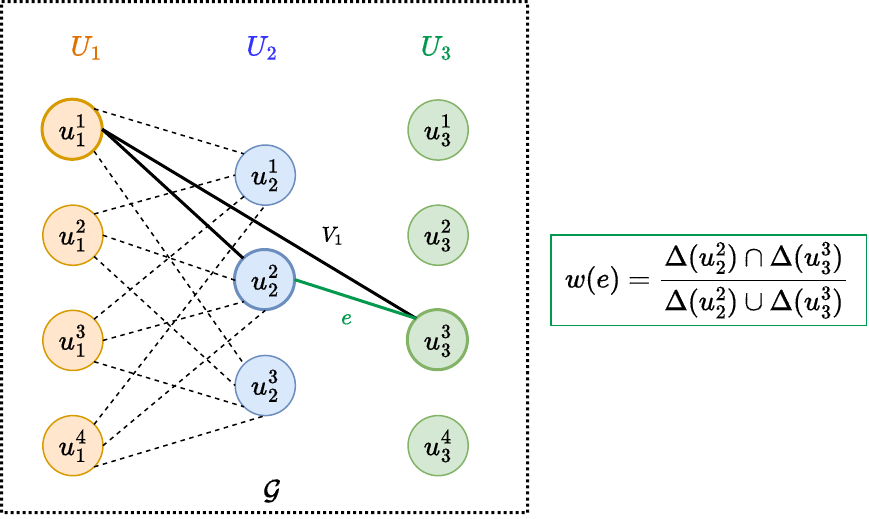}
    \caption{Illustration of the label mapping problem as a graph $\mathcal{G}$ for the case of $K=3$. $V_1$ denotes the clique formed by vertices $(u_1^1,u_2^2,u_3^3)$. $\Delta(u)$ represents the segments where speaker $u$ is active in the recording.}
    \label{fig:label_mapping_graph}
\end{figure}
Each $U_k$ in the graph is an \textit{independent set} (set of vertices with no edges between any pair), and the label mapping problem can be defined as: partition $V$ into $C$ vertex-disjoint cliques $\Phi = (V_1,\ldots,V_C)$, such that the partition maximizes $w(\Phi)$, i.e.,
\begin{equation}
\label{eq:objective}
    \widehat{\Phi} = \text{arg}\max_{\Phi} w(\Phi), ~~\text{where}~~ w(\Phi) = \sum_{c\in C} w(V_c) = \sum_{c\in C} \sum_{e\in E(V_c)} w(e),
\end{equation}
and $E(V_c)$ represents edges in the sub-graph induced by $V_c$. 
Intuitively, the objective maximizes the sum of all edge weights within the cliques.
The partition is \textit{orthogonal} since it may contain at most one vertex from every $U_k$.

It may not immediately be clear why maximizing the objective in \eqref{eq:objective} provides an optimal label mapping. 
Since the partition is orthogonal, each $V_c$ may represent a mapped speaker label.
By maximizing the total edge weights within cliques, we maximize the total relative overlap between speaker turns for speakers that are mapped to the same label. 
It is hard to demonstrate a theoretical correspondence between this objective and the DER metric that would eventually be used to evaluate the results of the system combination. 
However, we empirically demonstrate a correspondence using some recordings from the AMI dataset in Fig.~\ref{fig:weight_vs_der}. 
We note from the figure that as the objective (x-axis) improves, the DER decreases, suggesting that our proposed objective is a good proxy for minimizing the final DER.

\begin{figure}[t]
    \centering
    \includegraphics[width=0.8\linewidth]{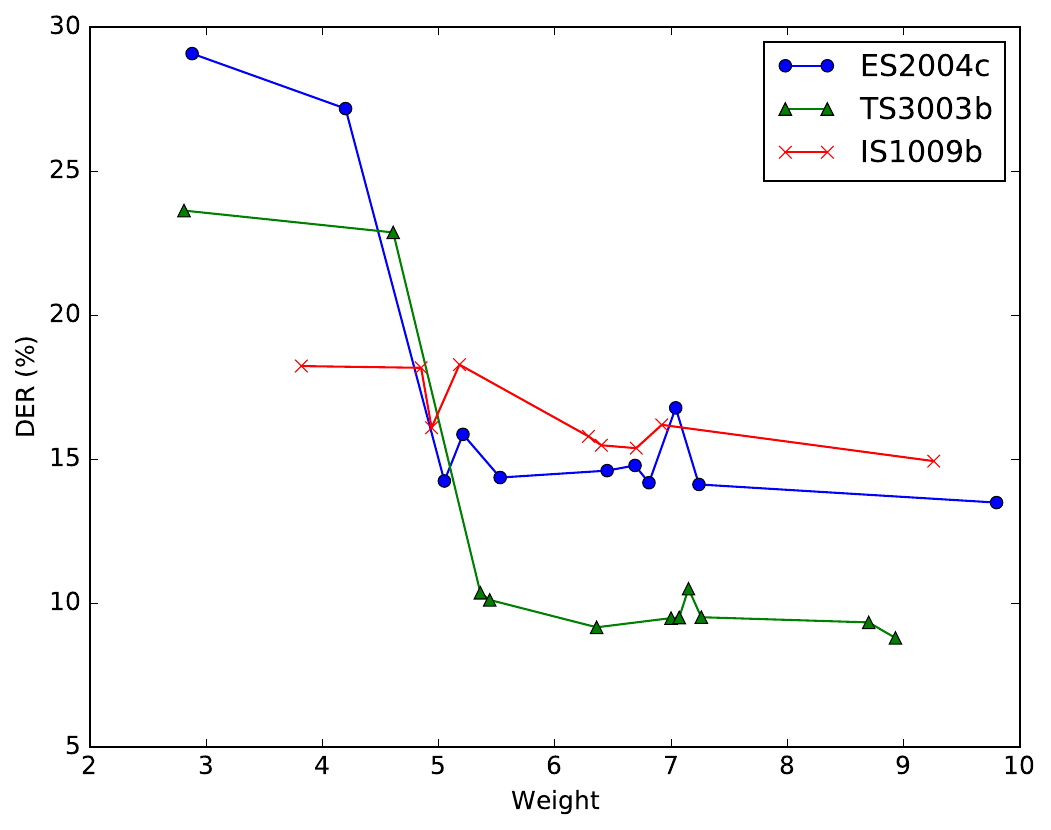}
    \caption{Partition weight $w(\Phi)$ versus diarization error rate (DER) for three arbitrarily chosen recordings from the AMI evaluation set, showing that DER tends to improve with weight.}
    \label{fig:weight_vs_der}
\end{figure}

\subsection{Some properties of graph \texorpdfstring{$\mathcal{G}$}{G}}

\begin{definition}[Tur\'an graph]
The Tur\'an graph $T(n,r)$ is a complete multi-partite graph formed by partitioning a set of $n$ vertices into $r$ subsets, with sizes as equal as possible, and connecting two vertices by an edge if and only if they belong to different subsets.
\end{definition}

\begin{lemma}
\label{lemma:turan}
$\mathcal{G}$ is equivalent to some $T(CK,K)$ Tur\'an graph, i.e., it is equivalent to a complete $K$-partite graph $K_{C,C,C,\ldots}$.
\end{lemma}

\begin{lemma}
\label{lemma:exponential}
$\mathcal{G}$ has an exponential number of maximal cliques.
\end{lemma}

\begin{proof}
The proof is through a simple combinatorial argument. 
Since $\mathcal{G}$ is a complete $K$-partite graph, any maximal clique of $\mathcal{G}$ contains exactly 1 vertex from all its $K$ independent sets. 
Since each independent set has $C$ vertices, there are $C^K$ possibilities for a maximal clique.
\end{proof}

\subsection{An exponential algorithm based on clique enumeration}

In the previous section, we proved that $\mathcal{G}$ has an exponential number of maximal cliques.
Nevertheless, if $K$ is small, we can still consider a brute-force solution for \eqref{eq:objective} by enumerating all the cliques in the graph.
Such an algorithm is presented in Algorithm~\ref{alg:dl_mapping}, and it roughly follows the following steps after the graph construction.

\begin{enumerate}[nolistsep]
    \item Enumerate all the maximal cliques in $\mathcal{G}$. Let this set be denoted by $S$.
    \item Find the clique $V_c$ with the maximum weight in $S$.
    \item Add $V_c$ to the partition $\Phi$. Remove the vertices in $V_c$ from $\mathcal{G}$ and the associated edges.
    \item Repeat from Step 1 until no vertices remain in $\mathcal{G}$.
\end{enumerate}

\begin{algorithm}[t]
\setstretch{1.35}
\DontPrintSemicolon
  
  \KwInput{Graph $\mathcal{G} = (V,E,w)$}
  \KwOutput{Partition $\Phi$ = $V_1,\ldots,V_C$}
  
  $\Phi = \{\}$ 
  
  \tcp{Loop until no vertices remaining}
  \While{$V \neq \phi$}{
    $S$ = set of all maximal cliques in $V$ \tcp*{Enumerate all maximal cliques}
    
    $V_c$ = max($S$, key=$\sum_{e\in S_i} w(e)$) \tcp*{Get maximum weighted clique}
    
    $\Phi = \Phi \cup \{V_c\}$ \tcp*{Add clique to partition}
    
    $V = V \setminus \{V_c\}$ \tcp*{Remove clique vertices from V}
  }

\caption{Exponential mapping based on clique enumeration}
\label{alg:dl_mapping}
\end{algorithm}

\begin{theorem}
Algorithm~\ref{alg:dl_mapping} has time complexity $\mathcal{O}(C^K)$, where $K$ is the number of input hypotheses and $C$ is the maximum number of speakers in any hypothesis.
\end{theorem}

\begin{proof}
From Lemma~\ref{lemma:exponential}, since $\mathcal{G}$ has $C^K$ maximal cliques, it takes at least $\mathcal{O}(C^K)$ time to simply enumerate all the cliques, as is required in the first step of the algorithm. 
Hence, the mapping algorithm has complexity exponential in the size of the input. 
\end{proof}

In practice, we use tensor broadcasting operations to efficiently compute the weights of all cliques in the graph.
We build a \textit{cost tensor} $\mathbf{C} \in \mathbb{R}^{N_1\times \ldots \times N_K}$, where each element of the tensor represents the weight of the clique corresponding to that speaker group.
\begin{figure}[t]
    \centering
    \includegraphics[width=0.8\linewidth]{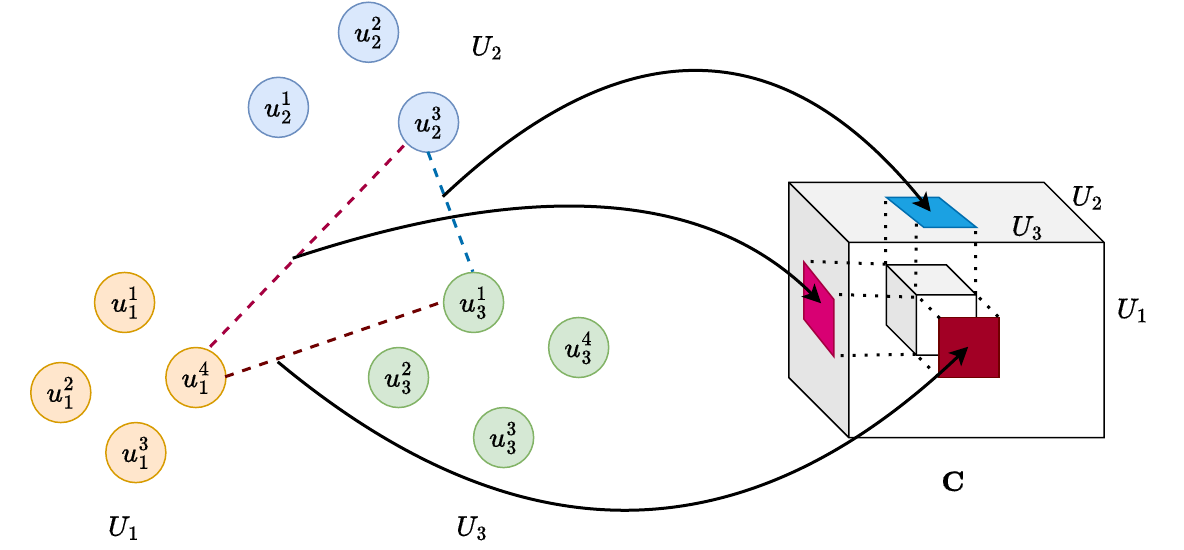}
    \caption{Computation of the \textit{cost tensor} $\mathbf{C}$ for the DOVER-Lap label mapping algorithm.}
    \label{fig:cost_tensor}
\end{figure}
This computation is illustrated in Fig.~\ref{fig:cost_tensor} for $K=3$. 
Since the hypotheses have 4, 3, and 4 speakers, respectively, $C \in \mathbb{R}^{4\times 3\times 4}$. 
In the figure, we show the computation of the tensor index corresponding to the tuple $(u_1^4,u_2^3,u_3^1)$. 
For this, we first compute all 3 pair-wise costs; the cost of mapping $u_1^4$ and $u_2^3$ together, for instance, is computed as the negative of the edge weight, which is the total overlapping duration between $u_1^4$ and $u_2^3$, divided by the sum of their speaking durations. 
Finally, $C(u_1^4,u_2^3,u_3^1)=-(M_{u_1^4,u_2^3}+M_{u_1^4,u_3^1}+M_{u_2^3,u_3^1})$.
This sum can be efficiently computed by first computing all the pairwise sums (i.e., the $K\choose{2}$ faces of the tensor, and then combining them using a broadcasting operation.

Nevertheless, due to the exponential dependency, the algorithm quickly becomes computationally intractable as $K$ increases. 
In Fig.~\ref{fig:mapping_time}, we computed the label mapping time for an increasing number of input hypotheses for the AMI and LibriCSS evaluation sets. 
For AMI (which contains 4 speakers; solid green line), the algorithm became infeasible beyond $K=10$. 
For LibriCSS (which contains 8 speakers; dotted green line), this limit was reached for an even smaller value of $K$, making combination impossible beyond 7 hypotheses.
As a comparison, we also show (in blue) the mapping time for the poly-time Hungarian algorithm we will describe in Section~\ref{sec:hungarian}.

\begin{figure}[t]
    \centering
    \includegraphics[width=0.8\linewidth]{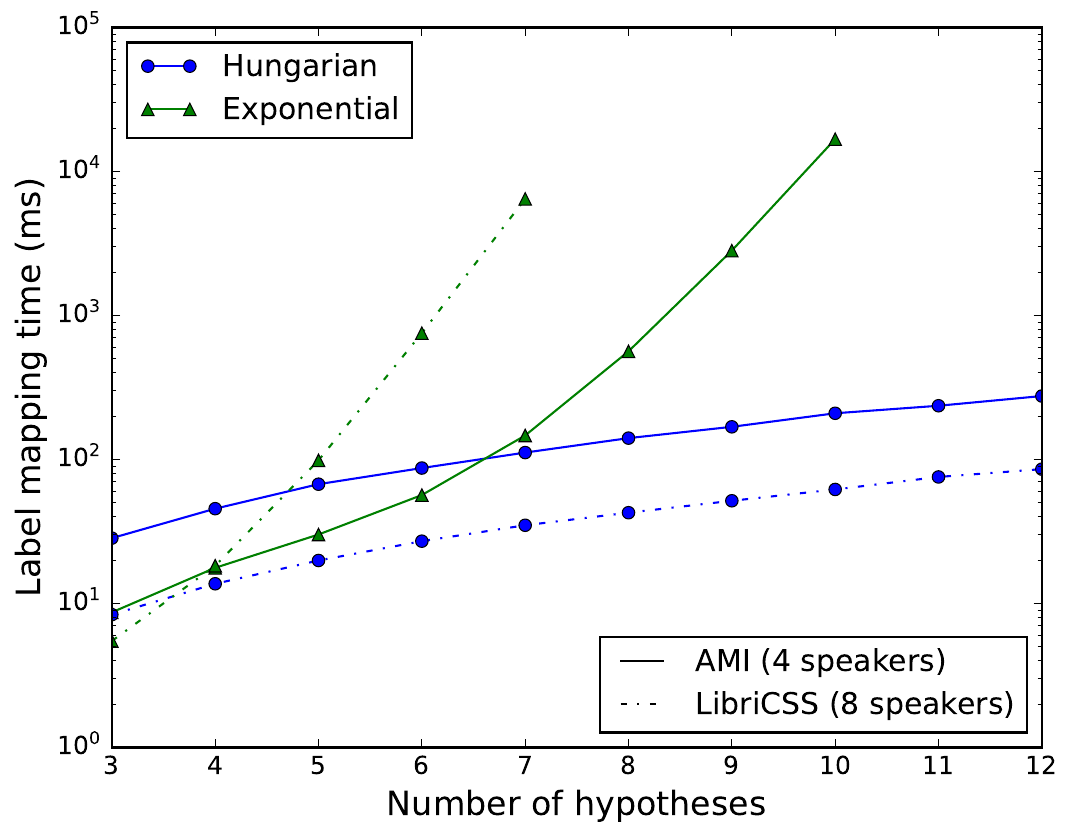}
    \caption{Label mapping time (in ms) for combining different number of hypotheses on the AMI and LibriCSS data. The y-axis is logarithmic.}
    \label{fig:mapping_time}
\end{figure}

\subsection{Polynomial-time Hungarian method}
\label{sec:hungarian}

In this section, we propose a poly-time algorithm based on greedy pair-wise combination of the hypotheses using the Hungarian method.
Let use first define some key concepts.

\subsubsection{Definitions and key concepts}

Consider the subgraph $\mathcal{G}_{ij}$ induced by considering all the vertices in sets $U_i$ and $U_j$, where $U_i$ and $U_j$ are two independent sets in $\mathcal{G}$. 
By construction, such a subgraph $\mathcal{G}_{ij}$ is a complete bipartite graph.
For our graphical formulation, we will first define a \textit{local} and \textit{global} mapping.

\begin{definition}
A \textit{matching} $M$ in $\mathcal{G}$ is a set of pairwise non-adjacent edges, none of which are loops; that is, no two edges share common vertices. A \textit{perfect matching} is a matching which covers all vertices.
\end{definition}

\begin{definition}
A \textit{local label mapping} $\psi$ is a matching on the bipartite graph $\mathcal{G}_{ij}$.
\end{definition}

\begin{definition}
A \textit{global label mapping} $\Psi$ is a function that assigns a vertex $u_k^i$ in the independent set $U_k$ to any one of the final cliques $V_c$.
\end{definition}
 
Our algorithm works by incrementally constructing the global label mapping using the local maps between pairs of independent sets. 
The local maps $\psi$ are computed using the Hungarian method as a subroutine, which is described later. 
We develop the following terminology to describe this method. 
For the sake of brevity, we drop the subscript $ij$ for the bipartite graph $\mathcal{G}_{ij}$, and just refer to it as $\mathcal{G}$.

\begin{definition}
A \textit{labeling} of $\mathcal{G}$ is a function $l: V \rightarrow \mathbb{R}$ such that 
$$\forall \{u,v\} \in E, \quad l(u) + l(v) \geq w(u,v).$$ 
\end{definition}

\begin{definition}
An \textit{equality subgraph} is a subgraph $\mathcal{G}_l = (V,E_l) \subseteq \mathcal{G} = (V,E)$, fixed on labeling $l$, such that
$$E_l = \{(u,v)\in E: l(u) + l(v) = w(u,v)\}.$$
\end{definition}

This means that $\mathcal{G}_l$ only includes those edges from the bipartite matching which allow the vertices to be perfectly feasible.

\begin{lemma}[The Kuhn-Munkres theorem]
\label{lemma:kuhn}
Given labeling $l$, if $\psi$ is a perfect matching on $\mathcal{G}_l$, then $\psi$ is a maximum-weight matching on $\mathcal{G}$.
\end{lemma}

\begin{proof}
Let $\psi^{\prime}$ be any perfect matching in $\mathcal{G}$. 
By definition of a labeling function, and since $\psi^{\prime}$ is perfect,
$$w(\psi^{\prime}) = \sum_{(u,v)\in \psi^{\prime}} w(u,v) \leq \sum_{(u,v)\in \psi^{\prime}} l(u) + l(v) = \sum_{v\in V}l(v).$$

This means that $\sum_{v\in V}l(v)$ is an upper bound on the weight $w(\psi^{\prime})$ of any perfect matching $\psi^{\prime}$ of $\mathcal{G}$. 
Now, let us consider $\psi$.
$$ w(\psi) = \sum_{(u,v)\in \psi} w(u,v) = \sum_{(u,v)\in \psi} l(u) + l(v) = \sum_{v \in V} l(v) \geq w(\psi^{\prime}). $$

Thus, $\psi$ is a maximum-weight matching in $\mathcal{G}$.
\end{proof}

\begin{definition}
Given a bipartite graph $\mathcal{G}$ and a matching $\psi$, the graph is said to contain an \textit{alternating path} if there exists a path in the graph which has alternating edges in $\psi$.
\end{definition}

\begin{definition}
An \textit{augmenting path} is a path in such a graph which has its endpoints (start and end vertices) unmatched (i.e., not in the matching).
\end{definition}

\subsubsection{The Hungarian method}

An important component of our algorithm is the Hungarian method for computing the maximum-weighted matching in bipartite graphs, due to Kuhn and Munkres~\cite{Kuhn1955TheHM, Munkres1957AlgorithmsFT}. 
In our algorithm, we use this as a subroutine to compute the local label mapping $\psi$.
In this section, we describe the algorithm and prove that it returns a maximum bipartite matching in $\mathcal{O}(n^3)$ time. 

The idea behind the algorithm is to find a perfect matching on some labeling $l$ on an equality subgraph, and use Lemma~\ref{lemma:kuhn} to claim that it is a maximum-weighted matching on $\mathcal{G}$. 
To achieve this, we start with an empty matching $\psi = \phi$ and a valid $l$ given as
$$ l ::= \forall x\in X, y\in Y : l(y) = 0, l(x) = \max_{y' \in Y}w(x,y')$$

We then repeat the steps of \textit{augmenting the matching} and \textit{improving the labeling}, until we obtain a perfect matching. 
These two steps are outlined below.

\noindent
\textbf{Step 1: Augmenting the matching}

Given $\mathcal{G}_l$ and some matching $\psi$, we find unmatched vertices $u,v\in V$ such that there is an augmenting path $\alpha$ from $u$ to $v$. 
If such a pair of vertices exist, we create this augmenting path and flip the edges in the matching, i.e., we replace the edges in $\psi$ with edges in the augmenting path that are in $E_l \setminus \psi$. 
This process increases the size of the matching, since we added previously unmatched vertices.

\noindent
\textbf{Step 2: Improving the labeling}

Let $S \subseteq X$ and $T \subseteq Y$ represent the set of vertices on either side of an ``almost'' augmenting path in $\psi$. 
Let $N_l(S) = \{v: \forall u\in S, (u,v)\in E_l\}$. 
If $N_l(S) = T$, then we cannot increase the alternating path and augment, so we must improve the labeling. 

Let $\delta_l = \min_{u\in S, v \notin T}(l(u) + l(v) - w(u,v))$. 
We improve $l$ to $l'$ as
\begin{equation}
l'(r) = 
\begin{cases}
l(r) - \delta_l & \text{if}~~r\in S, \\
l(r) + \delta_l & \text{if}~~r\in T, \\
l(r) & \text{otherwise}.
\end{cases}
\end{equation}
It is easy to show that $l'$ is a valid labeling by examining all modified edges.

Using the two subroutines of augmenting and improving described above, the Hungarian method iterates until $\psi$ is a perfect matching for $\mathcal{G}_l$.

\begin{lemma}
\label{lemma:hungarian}
The Hungarian method runs in $\mathcal{O}(n^3)$ time, where $n$ is the number of vertices in $\mathcal{G}$.
\end{lemma}

\begin{proof}
Each of the subroutines of augmenting the matching and improving the labeling increases the size of the matching by 1 edge. 
Since there can be at most $\frac{n}{2}$ edges in a matching, it takes $\mathcal{O}(n)$ rounds. 
Augmenting the matching requires $\mathcal{O}(n)$ time to find the right vertex, if one exists, and another $\mathcal{O}(n)$ to flip the matching. 
Improving the labeling also requires $\mathcal{O}(n)$ to find $\delta_l$. 
However, it can occur $\mathcal{O}(n)$ times if no augmenting path is found, therefore it requires $\mathcal{O}(n^2)$ steps in a single round. 
Hence, the total running time is $\mathcal{O}(n^3)$.
\end{proof}

\subsubsection{Algorithm}

We now describe the whole algorithm for label mapping in Algorithm~\ref{alg:dover_mapping}. 
The algorithm starts with a pair of hypotheses (independent sets), and computes a matching (local map) for them using the Hungarian method described in the previous section. 
It then \textit{merges} the pair w.r.t. the map $\psi$. 
This \textit{merge} operation is described next.

\begin{algorithm}[t]
\setstretch{1.35}
\DontPrintSemicolon
  
  \KwInput{Graph $\mathcal{G} = (V,E,w)$, $U = \{U_1,\ldots,U_K\}$}
  \KwOutput{Partition $\Phi$ = $V_1,\ldots,V_C$}
  
  $\Psi = \{\}$
 
  $\upsilon = U_1$
 
  \For{$k$ in $[2,K]$}{
    
    $\psi$ = Hungarian($\upsilon$, $U_k$) \tcp*{Compute local map}
    
    $\upsilon$ = Merge($\upsilon, U_k, \psi$) \tcp*{Merge pair w.r.t. local map}
    
    $\Psi$ = Update($\Psi,\psi$) \tcp*{Update global map}
  }
  
  $\Phi$ = Partition($U$,$\Psi$) \tcp*{Compute partition using global map}
  
\caption{Hungarian label mapping}
\label{alg:dover_mapping}
\end{algorithm}

Let $\mathcal{G}_{ij}$ be a bipartite graph induced from the graph $\mathcal{G}$ by considering the independent sets $U_i$ and $U_j$, as described earlier. 
Let $\psi$ be a matching on $\mathcal{G}_{ij}$, s.t. $\psi = \{u_1v_1, u_2,v_2, \ldots, u_C,v_C\}$, where $u_c \in U_i$ and $v_c \in U_j$. 
We merge $U_j$ with $U_i$ by identifying the vertex pairs $v_c$ with $u_c$. 
By Lemma~\ref{lemma:turan}, we can assume that we have performed graph completion on $\mathcal{G}$ and so there are no unmatched vertices. 
Let $U_{i(j)}$ denote the new vertex set. 
We remove the loops and edges within $U_{i(j)}$, and replace multiple edges by a single edge with weight equal to the sum of weights of the edges it is replacing.
The new weighted graph ($\mathcal{G}_{i(j)},w'$) is called the \textit{merge} of $U_j$ to $U_i$ from $\mathcal{G}$ along $\psi$. 
Clearly, $\mathcal{G}_{i(j)}$ is a ($k$-1)-partite graph.

In each iteration of the algorithm, we reduce the number of independent sets by one because of this merge operation. 
Furthermore, after every merge, we update the global mapping $\Psi$ using the local map $\psi$. 
This process simply involves creating a transitive map of the form $\psi_K(\ldots(\psi_2(\psi_1(\cdot))$ by composing the local label mappings at each iteration. 

\begin{theorem}
Algorithm~\ref{alg:dover_mapping} runs in polynomial time.
\end{theorem}

\begin{proof}
From Lemma~\ref{lemma:hungarian}, each iteration of the loop runs in $\mathcal{O}(C^3)$ time, and there are a total of $K-1$ iterations. 
Hence, the algorithm finishes in $\mathcal{O}(C^3K)$ time.
\end{proof}

\begin{theorem}
\label{thm:greedy}
Algorithm~\ref{alg:dover_mapping} is a $\frac{1}{C}$-approximation for the label mapping problem.
\end{theorem}

\begin{sketch}
The proof is obtained by induction on the number of hypotheses, $K$. 
The base case holds trivially since the Hungarian algorithm returns a maximal matching.
The inductive case also holds by algebraic manipulation of the sum of weights over the first matching and the remaining matching.
For a detailed proof, please refer to Appendix~\ref{sec:appendix_greedy}.
\end{sketch}

It can be shown through a reduction from the $c$-way $k$-coloring problem that there is no efficient ``deterministic'' algorithm with a better approximation ratio~\cite{He2000ApproximationAF}. 
In Appendix~\ref{sec:appendix_randomized}, we describe a label mapping algorithm based on ``randomized'' local search, which obtains a close to optimal approximation ratio in expectation.

\section{Label voting in DOVER-Lap}
\label{sec:dl_voting}

Similar to DOVER, we perform weighted majority voting on ``regions'' of the input. 
However, unlike the former, DOVER-Lap can assign multiple speakers to a region.
Consider a region $T$ (cf. intervals in Fig.~\ref{fig:dover_vs_dl}), and suppose the hypotheses $U_1,\ldots,U_K$ contain $n_1^T,\ldots,n_K^T$ speakers, respectively, in this region. 
Then, we compute the weighted mean rounded to the nearest integer as
\begin{equation}
\label{eq:dl_weight}
    \hat{n}_T = \nint{\sum_{k=1}^K w_k n_k^T},
\end{equation}
where $w_k$ are DOVER-like rank-based weights obtained by ranking the hypotheses in increasing order of their total relative overlap duration with all other hypotheses.
The highest weighted $\hat{n}_T$ speakers are then assigned to the region $T$. 
In case of ties, we assign all the tied speakers to the regions. 
Since we use rank-based weighting of the hypotheses, such ties only occur in very few regions. 
This assignment strategy allows multiple overlapping speakers to be present in the combined diarization output.

\section{System combination experiments}
\label{sec:dl_results}

\subsection{Setup}

We performed experiments on two datasets: the AMI meeting corpus~\cite{Carletta2005TheAM}, and the LibriCSS data~\cite{Chen2020ContinuousSS} (Section~\ref{sec:intro_data}). 
For AMI experiments, we used the mixed-headset recordings, while for LibriCSS, we selected the the first channel of the array in each recording.

We combine diarization results from the following systems.

\begin{enumerate}[leftmargin=*]

\item \textbf{Overlap-aware spectral clustering (SC)}~\cite{Raj2021MulticlassSC}: This is our proposed method from Chapter~\ref{chap:oasc}. 
It performs overlap-aware diarization by reformulating spectral clustering as a constrained optimization problem, and then discretizing it under the overlap constraints. 
An oracle SAD is used to remove non-speech segments.

\item \textbf{VB-based overlap assignment (VB)}~\cite{Bullock2019OverlapawareDR}: This method leverages Variational Bayes (VB)-HMM used originally for diarization in \cite{Dez2018SpeakerDB}. 
Using the output of an externally trained overlap detector, overlapping frames are assigned the top two speakers from the posterior matrix computed using VB inference. 
The same SAD and overlap detector as in the system above were used.

\item \textbf{Region proposal networks (RPN)}~\cite{Huang2020SpeakerDW}: It combines segmentation and embedding extraction into a single neural network, and jointly optimizes them using an objective function that consists of boundary prediction and speaker classification components. 
For AMI, we trained the RPN on force-aligned data from the AMI training set; for LibriCSS, it was trained on simulated meeting-style recordings with partial overlaps generated using utterances from the LibriSpeech~\cite{Panayotov2015LibrispeechAA} training set. 
% Since we used K-means clustering, we assumed that the oracle number of speakers for each recording is known. 
A post-processing step was applied using oracle SAD segments to filter non-speech.

\item \textbf{Target-speaker voice activity detection (TS-VAD)}~\cite{Medennikov2020TargetSpeakerVA}: It takes conventional speech features (e.g., MFCC) along with i-vectors for each speaker as inputs and produces frame-level activities for each speaker using a neural network with a set of binary classification output layers. 
The initial estimates for the speaker i-vectors were obtained using a spectral clustering system. 
For training the model for LibriCSS, we created simulated meeting-style data similar to that used for training the RPN model.
\end{enumerate}

\subsection{Diarization results on AMI}

Table~\ref{tab:doverlap_ami_results} shows the results on the AMI mix-headset data. 
We obtained diarization outputs using the VB, SC, and RPN models, and then combined them with our proposed DOVER-Lap method, using different label mapping algorithms.
The results are presented in terms of missed speech (MS), false alarm (FA), speaker confusion error (SE), and total diarization error rate (DER).
We used \texttt{spyder}\footnote{\url{https://github.com/desh2608/spyder}} for DER-based sorting in DOVER-Lap, and also for evaluating the final performances.

All three methods resulted in combined diarization outputs that outperformed the single best diarization system.
Since we were only combining 3 systems, we were able to use the exponential mapping algorithm, which resulted in the best DER of 19.86\% since it uses global clique weights to perform mapping.
The polynomial-time Hungarian mapping algorithm resulted in 20.46\% DER, which is still a 1\% absolute improvement over the best system.
Nevertheless, since it performs pair-wise mapping, the resulting SE was 0.6\% worse than that obtained by the exponential method.

\begin{table}[t]
\centering
\caption{System combination experiments using DOVER-Lap on the AMI evaluation set, reported in terms of MS, FA, SE, and DER. We combined 3 overlap-aware hypotheses: overlap-aware SC, VB-based overlap assignment, and regional proposal networks.}
\label{tab:doverlap_ami_results}
\begin{tabular}{@{}lrrrr@{}}
\toprule
\textbf{Method} & \textbf{MS} & \textbf{FA} & \textbf{SE} & \textbf{DER} \\
\midrule
Overlap-aware SC~\cite{Raj2021MulticlassSC} & 11.48 & 2.27 & 9.81 & 23.56 \\
VB-based overlap assignment~\cite{Bullock2020OverlapawareDR} & 9.84 & 2.06 & 9.60 & 21.50 \\
Region Proposal Networks~\cite{Huang2020SpeakerDW} & 9.49 & 7.68 & 8.25 & 25.42 \\
\midrule
\textbf{DOVER-Lap} & & & & \\
~~ Exponential mapping & 9.96 & 2.16 & 7.75 & 19.86 \\
~~ Hungarian mapping & 9.98 & 2.13 & 8.35 & 20.46 \\
~~ Randomized local search & 9.97 & 2.15 & 7.92 & 20.05 \\
\bottomrule
\end{tabular}
\end{table}

For randomized local search (RLS), we set $N$ and $M$ in Algorithm~\ref{alg:randomized} (Appendix~\ref{sec:appendix_randomized}) as 1000 and $2K+1$, respectively, where $K$ is the number of systems to be combined. 
As expected, RLS outperformed the Hungarian mapping algorithm, and the gains come primarily from lower speaker error (7.92\% compared with 8.35\%). 
However, this difference in performance is fairly small, especially when we consider that the RLS method requires a much longer processing time. 
This may be because the theoretical bounds are designed to hold in the setting when the size of inputs is fairly large. 
In our setting of combining diarization hypothesis, these ``large number'' assumptions are violated. 
Furthermore, while improving the objective in \eqref{eq:objective} typically leads to an improvement in DER, the relationship is not strictly monotonic as seen in Fig.~\ref{fig:weight_vs_der}.

Still, these results have important implications. 
Our experiments with the RLS method indicates that even with a theoretically stronger algorithm, it may not be possible to do much better than the fast Hungarian-based mapping algorithm due to the constraints of our setting. 
As such, it is unlikely that any further advances in combination performances under this framework would be obtained from better label mapping methods. 

\subsection{Diarization results on LibriCSS}

In Table~\ref{tab:dl_libricss}, we show the DER results on LibriCSS for four baseline diarization systems and DOVER-Lap with different label mapping methods, with a break down by overlap condition.
We also report a further break down by missed speech, false alarm, and speaker confusion, in Table~\ref{tab:libri_breakdown}.

\begin{table}[t]
\centering
\caption{Diarization performance on LibriCSS evaluation set (sessions 2-10), evaluated condition-wise, in terms of \% DER. 0S and 0L refer to 0\% overlap with short and long inter-utterance silences, respectively. The DL results are using rank-based weighting.}
\label{tab:dl_libricss}
\begin{adjustbox}{max width=\linewidth}
\begin{tabular}{@{}lccccccc@{}}
\toprule
\multicolumn{1}{c}{\multirow{2}{*}{\textbf{Method}}} & \multicolumn{6}{c}{\textbf{Overlap ratio in \%}} & \multicolumn{1}{c}{\multirow{2}{*}{\textbf{Average}}} \\
\cmidrule(r{4pt}){2-7}
\multicolumn{1}{c}{} & \multicolumn{1}{c}{\textbf{0L}} & \multicolumn{1}{c}{\textbf{0S}} & \multicolumn{1}{c}{\textbf{10}} & \multicolumn{1}{c}{\textbf{20}} & \multicolumn{1}{c}{\textbf{30}} & \multicolumn{1}{c}{\textbf{40}} & \multicolumn{1}{c}{} \\
\midrule
VB & 3.85 & 3.84 & 6.46 & 8.20 & 12.60 & 13.39 & 8.59 \\
SC & 2.57 & 3.41 & 6.80 & 10.03 & 13.86 & 15.17 & 9.34 \\
RPN & 4.45 & 9.11 & 8.33 & 6.68 & 11.59 & 14.21 & 9.50 \\
TS-VAD & 5.99 & 4.63 & 6.62 & 7.28 & 10.31 & 9.54 & 7.62 \\
\midrule
\textbf{DOVER-Lap} & & & & & & & \\
~~ Exponential & 1.77 & 1.74 & 3.47 & 4.09 & 6.66 & 6.83 & 4.38 \\
~~ Hungarian & 1.98 & 1.66 & 3.38 & 4.39 & 6.72 & 6.29 & 4.33 \\
~~ RLS  & 1.77 & 1.74 & 3.39 & 4.09 & 6.66 & 6.83 & 4.36 \\
\bottomrule
\end{tabular}
\end{adjustbox}
\end{table}

Similar to the results on AMI, we found that DOVER-Lap improved the average DER over the single best system (TS-VAD, in this case) significantly.
Specifically, we obtained a 43.2\% relative DER improvement, from 7.62\% to 4.33\%. 
This improvement was consistent across the different overlap conditions, even though the single best system themselves may differ depending on the condition. 
For instance, the clustering-based methods were better on low overlaps, while the supervised methods performed better on high overlap conditions.
DOVER-Lap was able to get the best of both techniques and obtain the best DERs across the board.
If we look at the performances on the overlapping regions (shown in the smaller font), we see that the Hungarian mapping outperforms other methods significantly, which results in better overall DER performance.

\begin{table}[t]
\centering
\caption{Diarization result break-down on LibriCSS evaluation set, in terms of \% missed speech (MS), false alarm (FA), and speaker confusion (Conf.). The numbers in smaller font are the corresponding error rates computed on only the overlapping regions.}
\label{tab:libri_breakdown}
\begin{tabular}{@{}lcccc@{}}
\toprule
\textbf{Method} & \textbf{MS} & \textbf{FA} & \textbf{Conf.} & \textbf{DER} \\ \midrule
VB & 1.68 {\scriptsize 8.77} & 0.48 {\scriptsize 0.00} & 6.42 {\scriptsize 19.86} & 8.59 {\scriptsize 28.63} \\
SC & 2.52 {\scriptsize 13.35} & 1.10 {\scriptsize 0.00} & 5.72 {\scriptsize 17.57} & 9.34 {\scriptsize 30.92} \\
RPN & 2.87 {\scriptsize 7.61} & 3.33 {\scriptsize 3.21} & 3.30 {\scriptsize 4.16} & 9.50 {\scriptsize 14.97} \\
TS-VAD & 3.24 {\scriptsize 11.72} & 1.52 {\scriptsize 0.78} & 2.86 {\scriptsize 4.88} & 7.62 {\scriptsize 17.37} \\
\midrule
\textbf{DOVER-Lap} & & & & \\
~~ Exponential & 1.73 {\scriptsize 8.52} & 0.77 {\scriptsize 0.02} & 1.88 {\scriptsize 4.81} & 4.38 {\scriptsize 13.35} \\ 
~~ Hungarian & \textbf{1.64} {\scriptsize 8.14} & 0.77 {\scriptsize 0.12} & 1.92 {\scriptsize 4.36} & \textbf{4.33} {\scriptsize 12.62} \\
~~ RLS & 1.73 {\scriptsize 8.52} & 0.77 {\scriptsize 0.02} & \textbf{1.86} {\scriptsize 4.81} & 4.36 {\scriptsize 13.36} \\
\bottomrule
\end{tabular}
\end{table}

\section{Late fusion for multi-microphone diarization}
\label{sec:dl_multi}

The DOVER-Lap algorithm combines diarization hypotheses, irrespective of the source of these hypotheses. 
In the experiments above, we have applied it for combining different diarization systems; however, it can also be applied to several other use cases. 
For instance, we may have a single-channel diarization system, but input signals from an array microphone. 
In such cases, the system can be independently run on each channel, and the outputs can be combined using DOVER-Lap --- this is a classic ``late fusion'' application (in contrast to early fusion techniques such as beamforming~\cite{HaebUmbach2019SpeechPF,Hain2007TheAS}).

In this section, we demonstrate the application of DOVER-Lap for late fusion on array microphones. 
We conducted our investigation on the LibriCSS dataset, which has 7 microphones arranged in a circular array. 
For our diarization system, we used the overlap-aware spectral clustering (SC)~\cite{Raj2021MulticlassSC} method described earlier. 
As shown in Table~\ref{tab:libri_breakdown}, the method obtained a DER of 9.34\% on LibriCSS using a single microphone.

Table~\ref{tab:fusion} shows the results for multichannel diarization using late fusion with DOVER-Lap (using the exponential mapping method). 
The single-channel system obtained a DER of 9.40\% on average (with a standard deviation of 0.23\%). 
For early fusion, we applied online weighted prediction error (WPE)~\cite{Nakatani2010SpeechDB} based dereverberation followed by delay-and-sum beamforming on the input channels. 
We used the Nara implementation~\cite{Drude2018NARAWPEAP} of WPE and the Beamformit tool~\cite{Mir2007AcousticBF} for beamforming. 
The corresponding DER was found to be 9.33\%, which is marginally better than the 7-channel average. 
Notably, simple beamforming without dereverberation degraded the DER to 9.71\%. 
Late fusion using DOVER-Lap improved over the average and best single system by achieving 9.02\% DER. 
Similar to our earlier results, we found that the improvement was mostly from reduced false alarms and speaker confusions.

\begin{table}[t]
\centering
\caption{Diarization results for multichannel LibriCSS evaluation set. Late fusion using DL achieved better performance compared to early fusion based on dereverberation and beamforming.}
\label{tab:fusion}
\begin{adjustbox}{max width=\linewidth}
\begin{tabular}{@{}lcccc@{}}
\toprule
\textbf{Method} & \textbf{MS} & \textbf{FA} & \textbf{Conf.} & \textbf{DER} \\ \midrule
7-channel avg. & \textbf{2.58} & 0.96 & 5.86 & 9.40 \\
7-channel best & 2.59 & 0.99 & 5.53 & 9.11 \\
\midrule
WPE + Beamforming & 2.91 & 0.96 & 5.86 & 9.33 \\
DL & 3.60 & \textbf{0.66} & \textbf{4.76} & \textbf{9.02} \\ \bottomrule
\end{tabular}
\end{adjustbox}
\end{table}

\section{Conclusion}

We proposed DOVER-Lap, a new method to combine the outputs from overlap-aware diarization systems. 
Our method was inspired by the label mapping and label voting approach in DOVER, but we modified the algorithms used in each of these stages. 
%
% We formulated the label mapping task as a graph partitioning problem and proposed a suite of approximation algorithms to solve it, along with derivations for their computational complexities and approximation ratios.
%
% A Hungarian mapping algorithm similar to DOVER was shown to be polynomial in time while provided strong empirical results.
%
% We also proposed a randomized local search method that theoretically outperformed the deterministic algorithms by obtaining a close to optimal approximation bound.
%
We demonstrated through experiments on AMI and LibriCSS that DOVER-Lap is effective at combining the outputs from different kinds of diarization systems, such as clustering-based, RPN, and TS-VAD. 
It provided consistent and significant improvements over the single best system for both datasets. 
We also showed its applicability to multi-channel diarization through late fusion, where it outperformed early fusion methods.
%
% Although the RLS method was found to be theoretically optimal, it only provided small empirical improvements in practice, leading us to conjecture that no further improvement in system combination may be obtained through better label mapping. 

Since we proposed DOVER-Lap in \citet{Raj2021DOVERLapAM} and \citet{Raj2021ReformulatingDL}, it has become the de-facto approach for combining diarization systems, and has been used by the top teams in several community challenges pertaining to speaker diarization and multi-talker ASR.
We have summarized these in Table~\ref{tab:dl_usage}, where we see that teams have used DOVER-Lap for combining a variety of systems, such as clustering-based, EEND, TS-VAD, separation-guided diarization, and multi-channel systems.
Across all such usage, DOVER-Lap has been found to give consistent and significant DER improvements.

\begin{table}[t]
    \centering
    \caption{Summary of DOVER-Lap usage in major speaker diarization and multi-talker ASR community challenges.}
    \label{tab:dl_usage}
    \begin{tabular}{llc}
    \toprule
    \textbf{Challenge/Team} & \textbf{Systems combined} & \textbf{Position} \\
    \midrule
    \multicolumn{3}{l}{\textit{DIHARD-3}} \\
    USTC-NELSLIP~\cite{Wang2021USTCNELSLIPSD} & Clustering, ITS-VAD & 1 \\
    Hitachi-JHU~\cite{Horiguchi2021TheHD} & VBx, EEND & 2 \\
    BUT~\cite{Landini2021ButSD} & VBx, SC, EEND & 5 \\
    \midrule
    \multicolumn{3}{l}{\textit{VoxSRC 2021}} \\
    DKU-Duke-Lenovo~\cite{Wang2021TheDS} & AHC, SC, TS-VAD & 1 \\
    ByteDance~\cite{Wang2021TheBS} & Multi-scale clustering & 2 \\
    Tencent~\cite{Zheng2021TencentSD} & VBx (multiple extractors) & 3 \\
    Huawei~\cite{Wang2020TheHS} & SC, VBx & - \\
    \midrule
    \multicolumn{3}{l}{\textit{VoxSRC 2022}} \\
    DKU-SMIIP~\cite{Wang2022TheDS} & AHC, SC, TS-VAD & 1 \\
    GIST-AiTeR~\cite{Park2022GISTAiTeRSF} & Multi-scale clustering & 3 \\
    BUCEA~\cite{Zhou2022TheBS} & VBx, SC & - \\ 
    \midrule
    \multicolumn{3}{l}{\textit{M2MeT at ICASSP 2022}} \\
    DUKE-DKU~\cite{Wang2022CrossChannelAT} & Multi-channel TS-VAD & 1 \\
    CUHK-Tencent~\cite{Zheng2022TheCS} & VBx, TS-VAD, FFM-TSVAD & 2 \\
    USTC-Ximalaya~\cite{He2022TheUS} & Multi-channel TS-VAD & 3 \\
    RoyalFlush~\cite{Tian2022RoyalflushSD} & Separation-guided diarization & 5 \\
    \midrule
    \multicolumn{3}{l}{\textit{CHiME-7 DASR}} \\
    USTC-NERCSLIP~\cite{Wu2023SemisupervisedMS} & Multi-channel NSD-MA-MSE & 1 \\
    NTT~\cite{Tawara2023NTTSD} & Multi-channel EEND-VC & 3 \\
    \bottomrule
    \end{tabular}
\end{table}

In the context of the broader multi-talker ASR problem, DOVER-Lap may be used to combine channel-wise diarization outputs from multi-channel arrays, as was done in \citet{Wu2023SemisupervisedMS} and \citet{Tawara2023NTTSD}. 
This is relatively efficient for end-to-end diarization systems which can process batched inputs, and may be particularly useful when the channels are ad-hoc microphone devices, such as smartphones placed near the speakers.
%

% \cleardoublepage

\chapter{Target Speaker Extraction with Guided Source Separation}
\label{chap:gss}

So far, we have focused exclusively on segmentation of long recordings through overlap-aware diarization approaches.
While this step is important for \textit{identification} of speaker-homogeneous segments, we are still required to \textit{extract} the speaker's utterance from possibly mixed audio, such that a regular speech recognition system may be able to transcribe it.
In this chapter, we will describe this ``target speaker extraction'' (TSE) task, showing how the target speaker information can be obtained from various means.
For multi-channel recordings where we have already obtained time-segmented speaker boundaries (either through an oracle or a diarization system), guided source separation (GSS)~\cite{Boeddeker2018FrontendPF} allows for unsupervised modeling of TSE.
Towards our objective for building efficient and accurate multi-talker ASR systems, we will describe our GPU-accelerated implementation of GSS, which is inspired by modern deep learning pipelines.
This efficient implementation will enable us to analyze the factors affecting GSS performance in some detail.
We will show how a strong TSE module can result in far-field ASR performance similar to that obtained using close-talk microphones.
In subsequent chapters, we will combine this extraction step with an overlap-aware diarization system (described earlier) to build a complete meeting transcription pipeline.

\section{Target speaker extraction}

Target speaker extraction (TSE) is the problem of estimating the speech signal of a target speaker in a mixture of several speakers, given auxiliary cues to identify the target.
In literature, TSE has alternatively been referred to as informed source separation, personalized speech enhancement, or audio-visual speech separation, depending on the context and the modalities involved.
TSE is motivated from the human ability of auditory attention in cocktail party settings, which has been the subject of much research in the last several decades~\cite{Bronkhorst2015TheCP}.
As illustrated in Fig.~\ref{fig:gss_tse}, neural models for TSE often exploit cues such as pre-recorded enrolment recordings, spatial information that provides the direction of the target, or visual tracking of faces, in order to focus on the target speaker in a multi-source mixture~\cite{molkov2019SpeakerBeamSA,Wang2018VoiceFilterTV,Gu2019NeuralSF,Rivet2014AudiovisualSS,Michelsanti2020AnOO}.
For more details about neural TSE, we refer the reader to the excellent review by \citet{molkov2023NeuralTS}.

\begin{figure}
    \centering
    \includegraphics[width=0.9\linewidth]{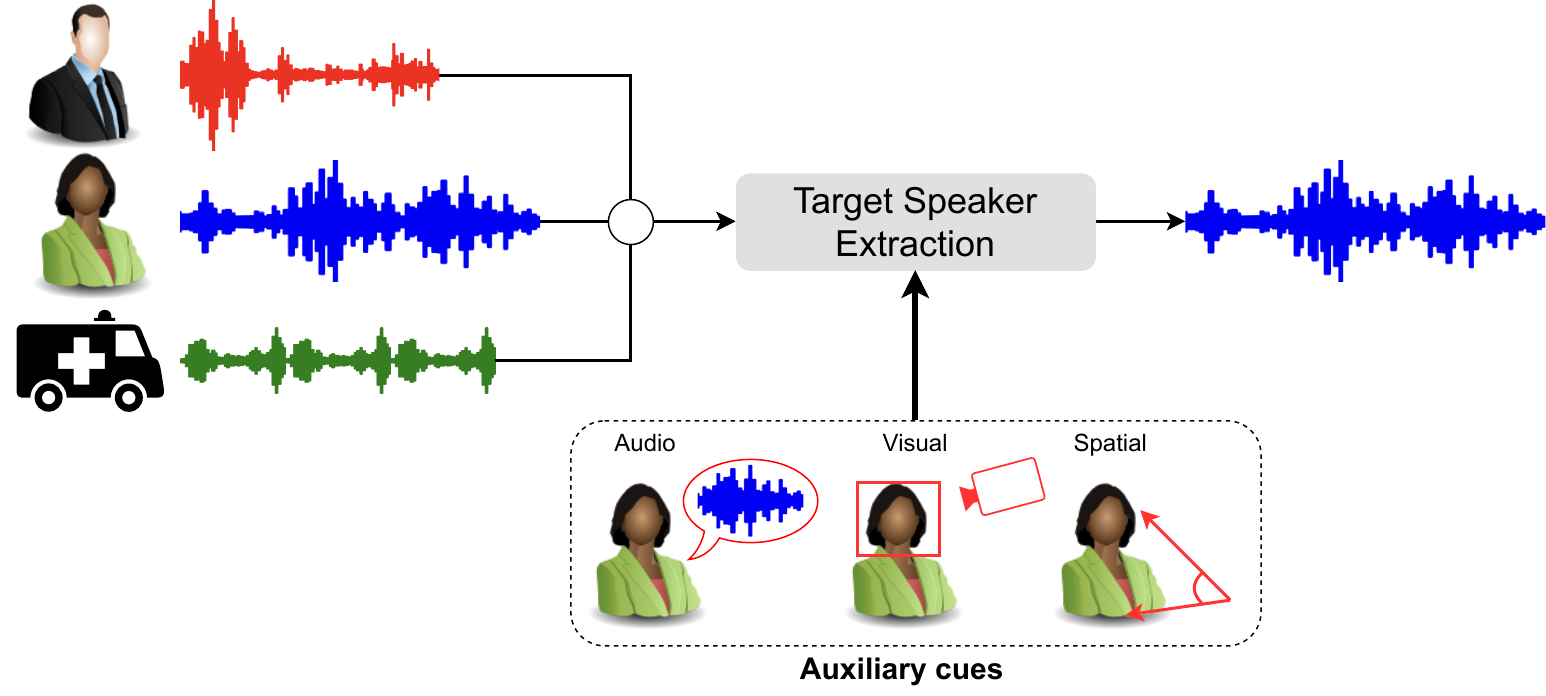}
    \caption{The target speaker extraction problem, and illustration of auxiliary cues. Adapted from \citet{molkov2023NeuralTS}.}
    \label{fig:gss_tse}
\end{figure}

Despite the large body of work on supervised neural TSE, current techniques have several limitations. 
First, the usage of auxiliary cues imposes additional data collection requirements.
For example, the target speaker may need to be registered with the device to enable enrollment recordings, or we may need to measure speaker locations to use spatial cues.
Often, such systems are trained on large amounts of synthetically mixed speech, or using multi-channel inputs assuming a specific array configuration.
This kind of training limits the model's performance when used in real-world settings, or makes it infeasible when the array configuration changes.

For these reasons, beamforming of multi-channel signals using unsupervised mask estimation remains a strong baseline for target-speaker ASR~\cite{Scheibler2020FastAS,Saijo2022SpatialLF,Ueda2021LowLO,Ikeshita2019AUF}. 
Among these, the recently proposed guided source separation (GSS) technique stands out as a particularly effective approach for handling noisy, overlapping speech using diarization information~\cite{Boeddeker2018FrontendPF,Kanda2019GuidedSS}. 
The method was first proposed for the CHiME-5 challenge, where it provided relative word error rate (WER) improvement of 21.1\% on the multi-array track using oracle diarization~\cite{Boeddeker2018FrontendPF}. 
It was later adopted as the challenge baseline for CHiME-6, and used by the winning systems on both oracle and unsegmented tracks~\cite{Watanabe2020CHiME6CT,Arora2020TheJM,Chen2020ImprovedGS,Medennikov2020TheSS}.

% Advantages of GSS
GSS relies on fundamental ideas from blind source separation (BSS), using spatial mixture models to model the sum of short-time Fourier transform (STFT) bins of multiple speakers~\cite{Comon2010HandbookOB}. 
It uses diarization information in two ways.
First, for the unsupervised mask estimation using BSS, speaker activities are used to
(i) estimate the number of mixture components, and (ii) avoid the speaker-frequency permutation problem when processing different frequency bins independently.
Once the speaker and noise masks are estimated, the speaker activities are again used to select the ``target'' speaker --- in particular, the speaker who is active throughout the segment is the target speaker.
We will describe the algorithm in detail in Section~\ref{sec:gss_gss} for the sake of completeness.

Despite its strong performance in the CHiME-5 and CHiME-6 challenges, GSS has seen limited adoption in other multi-talker benchmarks, most notably offline meeting transcription, primarily due to its significant computational cost. 
For instance, enhancing the CHiME-6 \texttt{dev} set using 80 CPU jobs requires approximately 20 hours with the original GSS implementation\footnote{\url{https://github.com/fgnt/pb_chime5}}. 
There have been some efforts to adapt the offline GSS algorithm for real-time enhancement by relying on limited right context~\cite{Horiguchi2021BlockOnlineGS}, but these are also CPU-bound.
These high inference times have also prevented detailed analysis of the method, particularly in terms of evaluation and impact of the various parameters.

% Contribution
In this chapter, we will describe our new, publicly-available GPU-accelerated implementation of GSS that aims to remove this computational bottleneck. 
We achieve this primarily by porting all the computations on to a GPU, and applying batching at several levels to maximize GPU memory utilization. 
Our implementation is inspired by modern deep learning pipelines where background CPU-based workers perform data loading of large tensors, while the data processing is performed by GPUs~\cite{Paszke2019PyTorchAI}. 
We describe our accelerated implementation in detail in Section~\ref{sec:gss_method}. 
The resulting 300x speedup allows us to perform ablation experiments using several benchmarks to analyze the importance of several factors that impact GSS performance such as WPE, noise class, context duration, number of BSS iterations, and number of channels, towards GSS performance.
While previous work has only evaluated GSS in terms of downstream ASR performance, we also measure signal-level enhancement metrics and the improvement in speaker information in the enhanced signal.

\section{Guided source separation}
\label{sec:gss_gss}

We first provide an overview of the GSS algorithm, as proposed in~\citet{Boeddeker2018FrontendPF}. 
Consider a multi-channel input recording provided in the form of STFT features $\mathbf{Y}_{t,f}\in \mathbb{C}^M$, where $t$ and $f$ are time and frequency bins, respectively, and $M$ is the number of channels. 
The GSS algorithm assumes the following model of the signal:

\begin{equation}
    \mathbf{Y}_{t,f} = \underbrace{\sum_{k\in K} \mathbf{X}_{t,f,k}^{\mathrm{early}}}_{\mathbf{X}_{t,f}^{\mathrm{early}}} + \underbrace{\sum_{k\in K} \mathbf{X}_{t,f,k}^{\mathrm{late}}}_{\mathbf{X}_{t,f}^{\mathrm{late}}} + \mathbf{N}_{t,f},
\end{equation}
where $K$ is the number of speakers in the recording, and ``early'' and ``late'' refer to components of the reverberation. 
For target-speaker extraction, the objective is to estimate the de-reverberated signal from a desired speaker $k$, i.e., $\widehat{\mathbf{X}}_{t,f,k}$. 
This estimation is performed in three steps, as described in Fig.~\ref{fig:gss}.

\begin{figure}[t]
    \centering
    \includegraphics[width=0.8\linewidth]{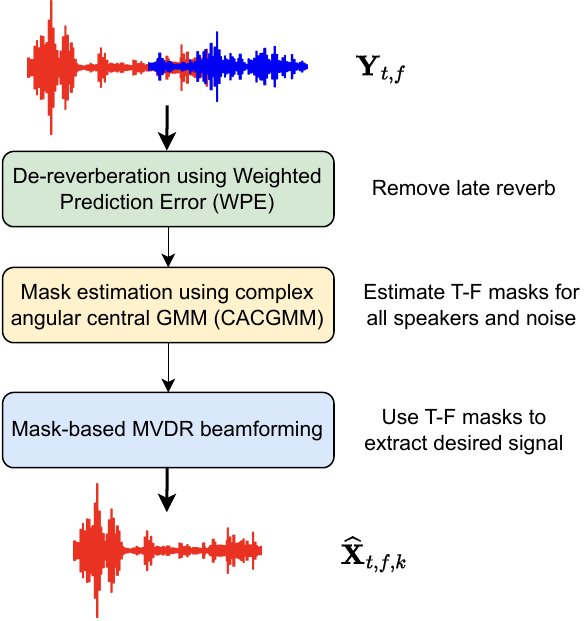}
    \caption{An overview of the guided source separation (GSS) method for target-speaker extraction.}
    \label{fig:gss}
\end{figure}

\subsection{De-reverberation using WPE}

First, we estimate $\mathbf{X}_{t,f}^{\mathrm{late}}$, i.e., the ``late'' part of the reverb, using the popular weighted prediction error (WPE) algorithm~\cite{Nakatani2008BlindSD,Nakatani2010SpeechDB}, which uses multi-channel linear prediction (MCLP).
%
% Let $\mathbf{h}_{t,f}^m$ denote the acoustic transfer function (ATF) between the speech source and the microphone $m$.
%
% Assuming that the duration of early reverberations is $D$, we can model the early reverberations as
% %
% \begin{equation}
%     \mathbf{X}_{t,f}^{\mathrm{early}} = \sum_{l=0}^{D-1} \mathbf{h}_{l,f}^M \ast \mathbf{S}_{t-l,f},
% \end{equation}
% %
% where $\mathbf{s}_{t-l,f}$ is the source mixture.

WPE models the late reverberation component of the signal at the first microphone ($m=1$) using an auto-regressive model on the received signal across all $M$ microphones as
\begin{equation}
    \mathbf{X}_{t,f}^{\mathrm{late}}[1] = \sum_{m=1}^M (\mathbf{g}_f^m)^H \mathbf{Y}_{t-D,f}^m. 
\end{equation}
Here, $\mathbf{g}_f$ is a regression vector for frequency bin $f$, $(\cdot)^H$ denotes the Hermitian transpose, and $D$ is the signal delay beyond which speech is not expected to have any residual auto-correlation.
Each STFT coefficient of the desired signal is modeled as a zero-mean Gaussian random variable with variance $\lambda_{t,f}$.
Then, we can solve for $\mathbf{g}_f^M$ and $\lambda_{t,f}$ using maximum likelihood estimation by minimizing the following cost function:
\begin{equation}
    \mathcal{L}(\Theta_f) = \sum_{t=1}^T \left( \log \lambda_{t,f} + \frac{\lVert \mathbf{Y}_{t,f}^1 - (\mathbf{g}_f)^H \mathbf{Y}_{t-D,f} \rVert^2}{\lambda_{t,f}} \right)
\end{equation}
The function cannot be solved analytically, so WPE solves it in a two-step iterative process, by alternatively keeping $\mathbf{g}_f$ and $\lambda_{t,f}$ fixed and minimizing w.r.t the other.

Once we have obtained the regression vectors, we estimate the late reverberation components, i.e., $\mathbf{X}_{t,f}^{\mathrm{late}}$ and remove it from the signal, followed by normalization to get unit STFT vectors, i.e.,
\begin{equation}
    \Tilde{\mathbf{Y}}_{t,f} = \frac{\mathbf{Y}_{t,f}-\widehat{\mathbf{X}}_{t,f}^{\mathrm{late}}}{\Vert\mathbf{Y}_{t,f}-\widehat{\mathbf{X}}_{t,f}^{\mathrm{late}}\Vert}.
\end{equation}

\subsection{Mask estimation using CACGMMs}

In the second stage, STFT masks are estimated for each speaker (and noise). 
The mask estimation technique is based on the sparsity assumption, which posits that at most one speaker is active in each time-frequency bin. 
Using this assumption, the vector in each T-F bin can be assumed to have been generated from a mixture model where each component of the mixture belongs to a different speaker (or noise class). 
In the case of GSS, each mixture component is a complex angular central Gaussian (CACG), and hence the mixture model is a CACGMM~\cite{Ito2016ComplexAC}. 
A CACGMM models sums of unit-normalized complex-valued random variables, and the probability density function at a frequency index $f$ is given by
\begin{equation}
\label{eq:gss_mixture}
   p(\tilde{\mathbf{Y}}_{t,f}) = \sum_{k\in K} \pi_{f,k} \mathcal{A}(\tilde{\mathbf{Y}}_{t,f};\mathbf{B}_{f,k}),
\end{equation}
where $\pi_{f,k}$ is the mixture weight of source $k$ at frequency index $f$, and $\mathcal{A}(\mathbf{y};\mathbf{B})$ is a CACG distribution parameterized by a covariance matrix $\mathbf{B}\in \mathbb{C}^{M\times M}$ as
\begin{equation}
    \mathcal{A}(\mathbf{y};\mathbf{B}) = \left(\frac{1}{2\pi}\right)^M \frac{(M-1)!}{\lvert\mathbf{B}\rvert} (\mathbf{y}^H \mathbf{B}^{-1} \mathbf{y})^{-M}.
\end{equation}

Mixture model parameters are usually estimated using the EM algorithm that alternates between estimating the state posteriors (in the E-step) and the parameters ($\mathbf{B}$) of the component model (in the M-step). 
However, there are two problems in applying EM independently for each frequency bin: (i) the number of sources $K$ in equation~\eqref{eq:gss_mixture} is unknown; and (ii) the same mixture component may correspond to different sources in different frequency bins. 
GSS solves both of these problems by assuming that speaker activities are known for the recording, either through an oracle or a diarization system. 
Given the speaker activities $a_{t,k} \in \{0,1\}$, we convert the time-invariant mixture weights to time-varying weights as
\begin{equation}
    \pi_{t,f,k} = \frac{\pi_{f,k}a_{t,k}}{\sum_{k' \in K}\pi_{f,k'}a_{t,k'}}.
\end{equation}

There may still be a permutation problem between the mixture components for the target speaker and the noise signal, since noise is present throughout the recording. 
To solve this problem, the GSS algorithm adds a ``context window'' to each utterance in which $a_{t,k}$ is zero for only the target speaker $k$. 
We run the EM algorithm on the CACGMM until convergence to obtain the final state posteriors $\gamma_{t,f,k}$ as the estimated speaker masks.

\subsection{Mask-based MVDR beamforming}

Finally, we compute the spatial covariance matrices for the target signal and background signal as
\begin{align}
\Phi_k(f) &= \frac{1}{T}\sum_t \gamma_{t,f,k} \tilde{\mathbf{Y}}_{t,f} \tilde{\mathbf{Y}}_{t,f}^H,\\
\Phi_{\mathrm{bg}}(f) &= \frac{1}{T}\sum_t \left(\sum_{k' \neq k}\gamma_{t,f,k'}\right) \tilde{\mathbf{Y}}_{t,f} \tilde{\mathbf{Y}}_{t,f}^H,
\end{align}
which are then used to compute the minimum-variance distortionless response (MVDR) filter~\cite{Souden2010OnOF,Erdogan2016ImprovedMB} as
\begin{equation}
    \mathbf{h}(f) = \frac{\Phi_\mathrm{bg}^{-1}(f)\Phi_k(f)\mathbf{e}_{\mathrm{ref}}}{\mathrm{tr}\left(\Phi_\mathrm{bg}^{-1}(f)\Phi_k(f)\right)},
\end{equation}
where $\mathbf{e}_{\mathrm{ref}}\in\{0,1\}^M$ is a one-hot vector indicating the reference channel, selected to maximize the signal-to-noise ratio. Finally, the enhanced STFT signal is computed as
\begin{equation}
    \widehat{\mathbf{X}}_{t,f,k} = \mathbf{h}(f)^H \tilde{\mathbf{Y}}_{t,f}.
\end{equation}

\section{GPU-accelerated inference}
\label{sec:gss_method}

% How we speed up original implementation:
% Batch together all frequency bins instead of iterating one at a time
% Perform einsum matrix operations on GPU
% Fix einsum path for CACG pdf estimation
% Smart batching to utilize GPU memory

The original GSS implementation is slowed down by four key factors: 

\begin{enumerate}[label=(\roman*),wide,labelwidth=!,labelindent=0pt]
    \item All the segments are processed sequentially, so processing time for a recording increases linearly with number of identified segments. 
    \item A context window (usually 15s) is used for all segments regardless of the segment duration, resulting in a lot of wasted computation for short segments. 
    \item For each segment, the CACGMM-based mask estimation is performed by iterating over all frequency bins (usually 513) sequentially.
    \item All computations (i.e., feature extraction, WPE, mask estimation, beamforming, and iSTFT) are implemented on the CPU in NumPy~\cite{Harris2020ArrayPW}.
\end{enumerate}

% A workaround for limitation (i) was provided by using MPI-based multi-processing (or Kaldi-style parallelization~\footnote{\url{https://kaldi-asr.org/doc/queue.html}}) to enhance segments concurrently on a multi-node CPU cluster. 
% %
% Nevertheless, enhancing the CHiME-6 \texttt{dev} set, for instance, may require close to 20 hours (wall clock time) even using 80 CPU jobs (\S~\ref{sec:gss_speed}).  

We propose to accelerate GSS-based inference by leveraging the power of modern GPU hardware and pipelines inspired by neural network training. 
First, to address limitation (iv), we use CuPy arrays which speed up array operations significantly using CUDA kernels, compared with regular NumPy-based array operations~\cite{cupy_learningsys2017}. 
Since the most computationally intensive operations in the pipeline (such as CACG probability estimation) involve matrix multiplications (through \texttt{einsum}), GPU-based CUDA kernels are more efficient. 
However, simply transferring all arrays to CuPy is not sufficient --- for example, limitations (i)--(iii) still require sequential processing, which limits GPU utilization. 
To maximize GPU utilization and improve real-time factor (RTF), we perform the following additional optimizations (shown in Fig.~\ref{fig:gss_batch}).

\begin{figure}[t]
    \centering
    \includegraphics[width=\linewidth]{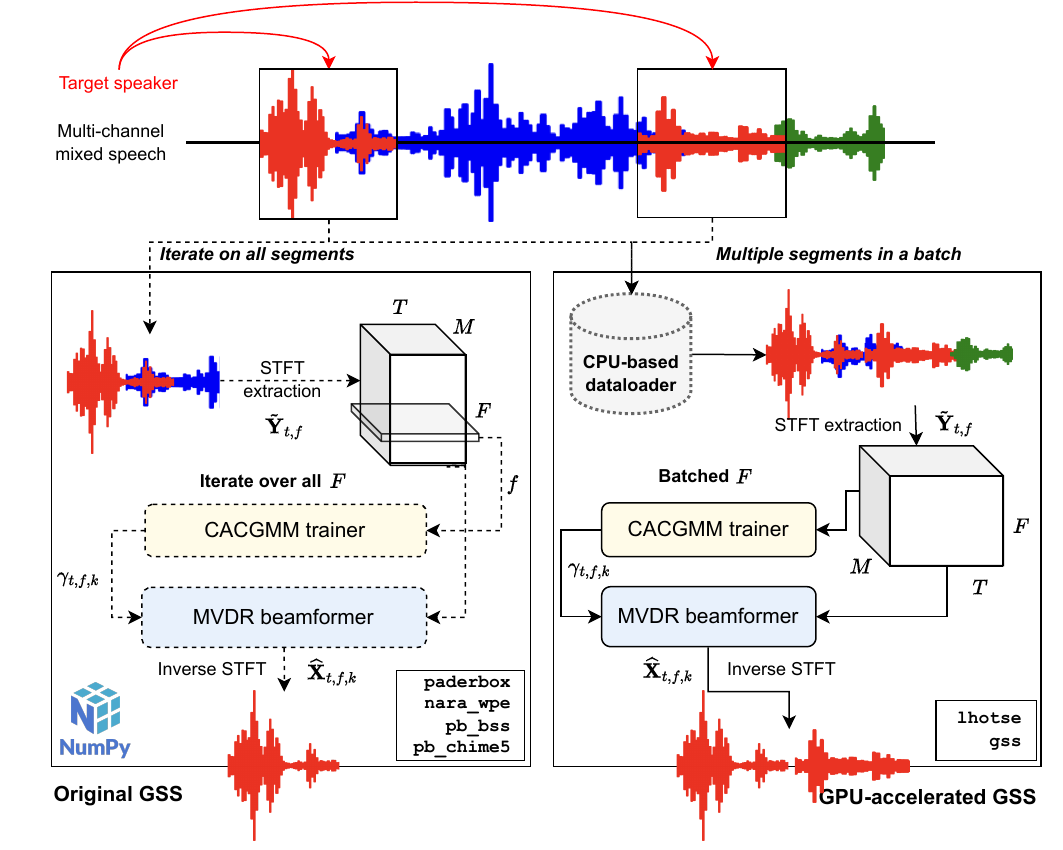}
    \caption{Overview of batch processing for GPU-accelerated GSS. Solid and dotted lines denote GPU-bound and CPU-bound operations, respectively. The WPE module is not shown.}
    \label{fig:gss_batch}
\end{figure}

\subsection{Segment batching}

Instead of processing each segment independently, we batch together multiple segments for inference. 
However, unlike neural network based training pipelines where batching is performed by stacking sequences in parallel, our batches are formed by concatenating segments sequentially along the time ($T$) axis to create ``super-segments.'' 
We choose this form of batching because (i) the \texttt{einsum}-based operations are designed to work with 3-D tensors, and (ii) parallel batching of segments with padding would result in wasted memory. 
Since multiple components of the inference (such as mask estimation and beamforming) compute statistics over the entire segment, we always create super-segments from segments of a recording that contain the same target speaker. 
Furthermore, we only use a single context window for the entire batch (instead of segment-wise context), which further amortizes the context window computations over multiple speaker segments. 
This batching technique should work well for the case when optimal reference channels do not vary over the duration of the recording (i.e., when speakers are stationary, which is common for meeting scenarios)\footnote{We also provide the option for using at most one segment per batch, for the case when speakers are not stationary (\S~\ref{sec:gss_speed}).}.

\subsection{CPU-based data-loaders}

We ensure that GPU idle time is minimized by off-loading the batch creation process to CPU-based data-loaders (possibly containing multiple workers), similar to deep learning pipelines.
These workers collect same-speaker segments in the background and load the audio from disk while the GPU is busy processing the previous batch.
Section~\ref{sec:gss_details} provides further details about our Lhotse-based data pipeline\cite{zelasko2021LhotseAS}.

\subsection{Frequency batching} 

To address (iii), we modified the CACGMM-based mask estimation to process 3-D tensors $(F,T,M)$ instead of 2-D arrays $(T,M)$. 
This simple change allows us to process all the frequency bins concurrently in a batch, significantly increasing GPU memory utilization.
Such a batching makes sense because all frequency bins are treated independently in the mask estimation process.

\subsection{Einsum path optimization} 

As mentioned above, several components in the GSS pipeline are implemented using \texttt{einsum}, which uses an optimal path contraction technique to find the path of minimum floating-point operations through the sequence (often resulting in up to 15x speed-up over a naive computation)~\cite{Smith2018opteinsum}. 
However, the optimal path finding itself is computationally demanding, with a complexity of $\mathcal{O}(N!)$ for $N$ arrays, and since it is performed several times during inference (for example, in each iteration of the CACGMM inference), it overshadows any speed-ups from the actual contracted sum. 
To remedy this, we cache the optimal computed path in the first iteration and re-use it in subsequent iterations.
In practice, since our tensor dimensions often have the same relative order across all batches (i.e., $M$<$F$<$T$), we can simply fix the optimal path for all \texttt{einsum} operations. This is because segment batching avoids very short segments that would otherwise result in $T$<$F$.

\begin{figure}[t]
    \centering
    \includegraphics[width=0.8\linewidth]{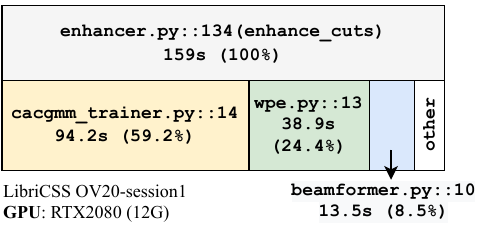}
    \caption{Representative output of profiler during enhancement of a single recording. Full stats are available at \href{https://github.com/desh2608/gss/blob/master/test.pstats}{this https url}.}
    \label{fig:profiling}
\end{figure}

Finally, once the enhanced waveform is obtained for the super-segment, we use background worker threads to chunk it into the original segments and save them to disk. 
With all these speed-ups, we were able to enhance a 10-minute LibriCSS recording in 159s (as shown in Fig.~\ref{fig:profiling}), of which mask estimation, WPE, and beamformer constituted 59.2\%, 24.4\%, and 8.5\% processing time, respectively. 
This is equivalent to a real-time factor (RTF) of approximately 0.3. We anticipate that further speed-ups could be obtained using GPUs with larger memory, by using bigger batches.

\section{Implementation details}
\label{sec:gss_details}

From an implementation perspective, we can divide the pipeline into two parts. 
The \textit{data processing} part is tasked with efficiently creating segments and corresponding speaker activities, while the \textit{inference} part performs the actual computations on GPU. 
We use Lhotse~\cite{zelasko2021LhotseAS} for all data processing, i.e., to store and read recording metadata, to represent speaker activities, and to perform segment batching to create super-segments. 
For batching we create buckets out of each speaker's segments on-the-fly, and sample speakers in a round-robin manner (using Lhotse's \texttt{DynamicBucketingSampler} and \texttt{RoundRobinSampler}), so that metadata from all segments do not need to be stored in memory. 
The super-segment obtained from the data-loader is converted to a CuPy array in-place, and all subsequent inference is performed on the GPU. 
Since we use Lhotse's supervision manifests to store speaker activities, it allows us to use either oracle segments, or read diarization output in the form of RTTM files with the same data processing pipeline. 
A typical recipe for enhancing a corpus using GSS is below:

\begin{lstlisting}[language=bash]
#!/bin/bash
# 1. Create Lhotse manifests for corpus
lhotse prepare libricss --type mdm $corpus_dir $data_dir
# 2. Prepare recording-level cuts
lhotse cut simple -r $data_dir/recordings.jsonl.gz -s $data_dir/supervisions.jsonl.gz $exp_dir/cuts.jsonl.gz
# 3. Prepare segment-level cuts
lhotse cut trim-to-supervisions --discard-overlapping $exp_dir/cuts.jsonl.gz $exp_dir/segments.jsonl.gz
# 4. Perform enhancement
gss enhance cuts --max-batch-duration 50.0 $exp_dir/cuts.jsonl.gz $exp_dir/segments.jsonl.gz $exp_dir/enhanced
\end{lstlisting}

\section{Experiments}

\subsection{Setup}
\label{sec:gss_setup}

We performed evaluations on three publicly-available meeting datasets: LibriCSS, AMI, and AliMeeting. 
Detailed statistics for all datasets are in Section~\ref{sec:intro_data}.
We used three different microphone settings for our experiments: IHM (individual headset microphone), SDM (single distant microphone), and GSS (GSS-enhanced multi-mic). 
Since LibriCSS does not provide headset recordings, we used the corresponding digitally-mixed LibriSpeech utterances to simulate IHM. 
For all datasets, the first channel of the first array was used for the SDM setting. 
For LibriCSS and AliMeeting, we used all available channels for GSS, whereas for AMI, we used the first of the two arrays. 

We evaluated the target speaker extraction quality using both intrinsic and extrinsic measures.
We consider the IHM recording as the clean signal, the SDM recording as the mixed signal, and the output of the GSS system as the enhanced signal.
We normalized the loudness of the signals to be equal to that of the clean signal, and also aligned the samples, before computing these metrics.
We used the \texttt{pyloudnorm}\footnote{\url{https://github.com/csteinmetz1/pyloudnorm}} and \texttt{fast-align-audio}\footnote{\url{https://github.com/nomonosound/fast-align-audio}} libraries for loudness normalization and sample alignment, respectively~\cite{steinmetz2021pyloudnorm}.
We report the signal-level enhancement performance in terms of PESQ, SI-SDR, and STOI, as described in Section~\ref{sec:intro_metrics}.
For extrinsic evaluation, we want to measure two aspects of the enhanced signal:
\begin{enumerate}
    \item Does the enhanced signal retain target speaker characteristics more than the mixed signal?
    \item Is the enhanced signal easier to transcribe than the mixed signal for an ASR system trained on single-speaker utterances?
\end{enumerate}

For the first question, we used embedding extractors to compute speaker embeddings for the mixed, clean, and enhanced signals --- let us denote the embeddings as $\mathbf{s}_{m}$, $\mathbf{s}_{c}$, and $\mathbf{s}_{e}$, respectively.
Next, we compute the following \textit{increase} in cosine similarity
\begin{equation}
\label{eq:gss_spk}
    \Delta_{\text{spk}} = cos(\mathbf{s}_{c}, \mathbf{s}_{e}) - cos(\mathbf{s}_{c}, \mathbf{s}_{m})
\end{equation}
as a measure of the improvement in speaker information when going from the mixed signal to the enhanced signal, where $cos(\mathbf{x},\mathbf{y})$ denotes the cosine similarity between $\mathbf{x}$ and $\mathbf{y}$.
Since the range of cosine similarity is [-1, 1], the range of $\Delta_{\text{spk}}$ is [-2, 2], with a higher value denoting a larger improvement in speaker information.
We used two off-the-shelf speaker embedding extractors from the SpeechBrain toolkit~\cite{Ravanelli2021SpeechBrainAG} to obtain the embeddings: \texttt{xvect} and \texttt{ecapa}.
The former is a conventional x-vector model using TDNN and stats pooling layers, and is trained with the regular softmax and cross-entropy loss~\cite{Snyder2018XVectorsRD}.
The latter is based on a newer ECAPA-TDNN model architecture and is trained with additive margin softmax (AM-softmax)~\cite{Desplanques2020ECAPATDNNEC}.
The models obtain equal error rates (EER) of 3.2\% and 0.8\% on the VoxCeleb-1 test set, respectively, meaning that \texttt{ecapa} has a better discriminative capability than \texttt{xvect}.

For the second question, i.e., we measured impact on ASR performance by training separate transducer-based ASR models for each benchmark.
We will define neural transducers briefly in Chapter~\ref{chap:modular}, when we describe the complete transcription pipeline.
Here, we provide the system configuration details for these models.
For LibriCSS, we used a pretrained Conformer-transducer~\cite{Gulati2020ConformerCT} trained on LibriSpeech.
In this model, the encoder is a 12-layer Conformer containing 8 attention heads, attention dimension of 256, and feed-forward dimension of 2048.
The prediction network is a 512-dimensional 1D convolution with a bi-gram context, and the joiner is a 512-dimensional feed-forward layer.
For AMI and AliMeeting, we trained a similar transducer on a combination of IHM, IHM with simulated reverb, SDM, and GSS-enhanced far-field recordings of the corresponding train set, and the resulting model was used to evaluate all microphone settings.
This model's configuration is similar to the one described earlier, with the exception that we used a Zipformer for the encoder~\cite{zipformer}.
The Zipformer encoder consists of 5 blocks, each containing multiple self-attention and feed-forward layers.
Each block is sub-sampled at different frame rates, with the intermediate ones more strongly down-sampled (up to 8x, for instance).
Within each block, self-attention weights are computed once and shared for all feed-forward layers in that block.
The result of these optimizations is that the encoder converges faster and better than the Conformer for the same model size.
In all cases, we applied three-fold speed perturbation and noise augmentation using MUSAN~\cite{Snyder2015MUSANAM} noises. 
We used a ``stateless'' decoder consisting of a convolutional layer with a bi-gram context.
The model was trained using a pruned RNN-T loss~\cite{Kuang2022PrunedRF} implemented in k2\footnote{\url{https://github.com/k2-fsa/{k2,icefall}}}. 
For decoding, we used a WFST-based parallel beam search method with beam size 4~\cite{Kang2022FastAP}.
The ASR performance is reported in terms of word error rate (WER).

\subsection{Results}

\subsubsection{Signal-level enhancement}

First, we measure the improvement in signal-level metrics: PESQ, SI-SDR, and STOI, as shown in Table~\ref{tab:gss_signal}.
For LibriCSS, we show a breakdown of these metrics by overlap ratio, in order to better understand the effect of GSS on target speaker extraction.
We see that GSS-based TSE provides a consistent improvement across all metrics, as seen in the $\Delta$ for the metrics.
The perceptual quality, as measured by PESQ, decreases with increase in overlap ratio, and the relative improvement provided by GSS also reduces.
The SI-SDR was found to be negative for both the mixture and enhanced signal, which means that the signal is weaker than the distortion.
However, GSS resulted in a relative increase in SI-SDR by 19.73 dB, 11.53 dB, and 17.92 dB, respectively, for LibriCSS, AMI, and AliMeeting.
Intelligibility for LibriCSS, as measured by STOI, decreased for input mixtures with increase in overlap ratio, but the corresponding output remained consistent.
For AMI and AliMeeting, the STOI improvement was found to be lower, which may be because real meetings may contain very short, often unintelligible utterances.

\begin{table}[t]
    \centering
    \caption{Signal-level enhacement performance of GSS in terms of PESQ, SI-SDR, and STOI metrics. For each metric, we report the measure on the input mixture (\texttt{in}), the enhanced signal (\texttt{out}), and the difference ($\Delta$).}
    \label{tab:gss_signal}
    \adjustbox{max width=\linewidth}{
    \begin{tabular}{@{}lcccccccccc@{}}
    \toprule
    \multirow{2}{*}{\textbf{Dataset}} & \multirow{2}{*}{\textbf{Subset}} & \multicolumn{3}{c}{\textbf{PESQ}} & \multicolumn{3}{c}{\textbf{SI-SDR}} & \multicolumn{3}{c}{\textbf{STOI}} \\
    \cmidrule(r{4pt}){3-5} \cmidrule(l{2pt}r{2pt}){6-8} \cmidrule(l{4pt}){9-11}
     & & \texttt{in} & \texttt{out} & $\Delta$ & \texttt{in} & \texttt{out} & $\Delta$ & \texttt{in} & \texttt{out} & $\Delta$ \\
    \midrule
    \multirow{6}{*}{LibriCSS} & 0L & $1.58$ & $2.12$ & $0.54$ & $-29.13$ & $-10.94$ & $18.19$ & $0.50$ & $0.85$ & $0.35$  \\
     & 0S & $1.54$ & $2.07$ & $0.53$ & $-30.40$ & $-12.46$ & $17.94$ & $0.47$ & $0.83$ & $0.36$ \\
     & OV10 & $1.40$ & $1.84$ & $0.44$ & $-30.24$ & $-10.44$ & $19.80$ & $0.47$ & $0.86$ & $0.39$ \\
     & OV20 & $1.35$ & $1.76$ & $0.41$ & $-31.60$ & $-10.50$ & $21.10$ & $0.45$ & $0.85$ & $0.40$ \\
     & OV30 & $1.27$ & $1.66$ & $0.39$ & $-30.76$ & $-11.21$ & $19.55$ & $0.45$ & $0.83$ & $0.38$ \\
     & OV40 & $1.24$ & $1.63$ & $0.39$ & $-32.37$ & $-10.59$ & $21.78$ & $0.43$ & $0.84$ & $0.41$ \\
    \hline \hline
    AMI & \texttt{test} & $1.20$ & $1.45$ & $0.25$ & $-17.87$ & $-6.34$ & $11.53$ & $0.61$ & $0.74$ & $0.13$ \\
    \hline \hline
    AliMeeting & \texttt{test} & $1.23$ & $1.52$ & $0.29$ & $-34.14$ & $-16.22$ & $17.92$ & $0.27$ & $0.57$ & $0.30$ \\
    \bottomrule
    \end{tabular}}
\end{table}

\subsection{Improvement in target-speaker information}

\begin{table}[t]
    \centering
    \caption{Improvement in target-speaker information for enhanced signal compared to mixed signal, measured in terms of $\Delta_{\text{spk}}$. The columns denote the speaker embedding extractor used.}
    \label{tab:gss_spk}
    \adjustbox{max width=\linewidth}{
    \begin{tabular}{@{}lccc@{}}
    \toprule
    \textbf{Dataset} & \textbf{Subset} & \texttt{xvect} & \texttt{ecapa} \\
    \midrule
    \multirow{6}{*}{LibriCSS} & 0L & $0.013$ & $0.055$ \\
     & 0S & $0.012$ & $0.057$ \\
     & OV10 & $0.012$ & $0.068$ \\
     & OV20 & $0.013$ & $0.089$ \\
     & OV30 & $0.015$ & $0.124$ \\
     & OV40 & $0.018$ & $0.167$ \\
    \hline \hline
    AMI & \texttt{test} & $0.025$ & $0.152$ \\
    \hline \hline
    AliMeeting & \texttt{test} & $0.034$ & $0.185$ \\
    \bottomrule
    \end{tabular}}
\end{table}

For our first extrinsic evaluation, we compute $\Delta_{\text{spk}}$ as defined in \eqref{eq:gss_spk} for all three datasets using \texttt{xvect} and \texttt{ecapa}, as shown in Table~\ref{tab:gss_spk}.
As expected, the $\Delta_{\text{spk}}$ values are positive across the board, indicating that the target-speaker information improves through enhancement.
On the LibriCSS data, as the overlap ratio increases, $\Delta_{\text{spk}}$ also increases, primarily because it is very hard to extract good speaker embeddings from mixed speech with high overlaps.
The improvement is also, in general, larger for AMI and AliMeeting compared to LibriCSS, since these datasets often contain very short phrases which may not contain enough speaker information, especially in the presence of background noise or interfering speakers.
Among the embedding extractors, we obtained larger improvements using $\texttt{ecapa}$ compared to $\texttt{xvect}$, perhaps since the former is trained with an additive margin softmax loss.
Nevertheless, the two embedding spaces are likely very different so the absolute values are not comparable between the extractors.

\subsubsection{Far-field ASR performance}

\begin{table}[t]
\centering
\caption{Comparison of close-talk and far-field ASR performance for meeting datasets. The GSS setting uses 7 channels for LibriCSS and 8 channels for AMI and AliMeeting. $^\dagger$LibriCSS IHM refers to the corresponding LibriSpeech utterances. $^\#$For AliMeeting, the numbers are CER.}
\label{tab:gss_asr}
\begin{tabular}{@{}llrrrr@{}}
\toprule
\textbf{Dataset} & \textbf{Setting} & \textbf{Ins.} & \textbf{Del.} & \textbf{Sub.} & \textbf{WER/CER} \\ \midrule
\multirow{3}{*}{LibriCSS} & IHM$^\dagger$ & 0.25 & 0.22 & 1.74 & 2.21 \\
 & SDM & 1.06 & 3.12 & 6.59 & 10.77 \\
 & GSS & 0.31 & 0.89 & 2.14 & 3.34 \\
\hline \hline
\multirow{3}{*}{AMI} & IHM & 2.22 & 4.51 & 11.31 & 18.04 \\
 & SDM & 4.01 & 9.59 & 18.50 & 32.10 \\
 & GSS & 2.43 & 6.07 & 14.33 & 22.83 \\
\hline \hline
\multirow{3}{*}{AliMeeting$^\#$} & IHM & 0.97 & 3.78 & 7.32 & 12.07 \\
 & SDM & 1.99 & 10.00 & 14.38 & 26.38 \\
 & GSS & 1.09 & 4.87 & 9.03 & 14.98 \\ 
\bottomrule
\end{tabular}
\end{table}

Finally, we measure the improvement in far-field ASR performance when using GSS with oracle segmentation, as shown in Table~\ref{tab:gss_asr}. 
The IHM and SDM settings may be considered as the lower and upper bounds on WER (or CER), respectively. 
We found that across all the datasets, GSS improved ASR performance significantly, with the \textbf{recovered error rates}, defined as $\left( W_{\mathrm{SDM}} - W_{\mathrm{GSS}} \right)/ \left( W_{\mathrm{SDM}} - W_{\mathrm{IHM}} \right)$,  being 86.8\%, 65.9\%, and 80.4\% for LibriCSS, AMI, and AliMeeting, respectively. 
As expected, most of the improvement was obtained from recovered deletion and substitution errors, possibly from better recognition of overlapped speech segments.
To verify this conjecture, we grouped all segments in the AMI \texttt{test} set by their overlap ratios to form bins of 20\% intervals, and computed the average WER for each such bin.
The resulting bar plot is shown in Fig.~\ref{fig:gss_ovl_wer}.
We can see that as the overlap ratio increases, the average WER increases more strongly for the SDM setting as compared to the GSS-enhanced utterances.

\begin{figure}[t]
    \centering
    \includegraphics[width=0.6\linewidth]{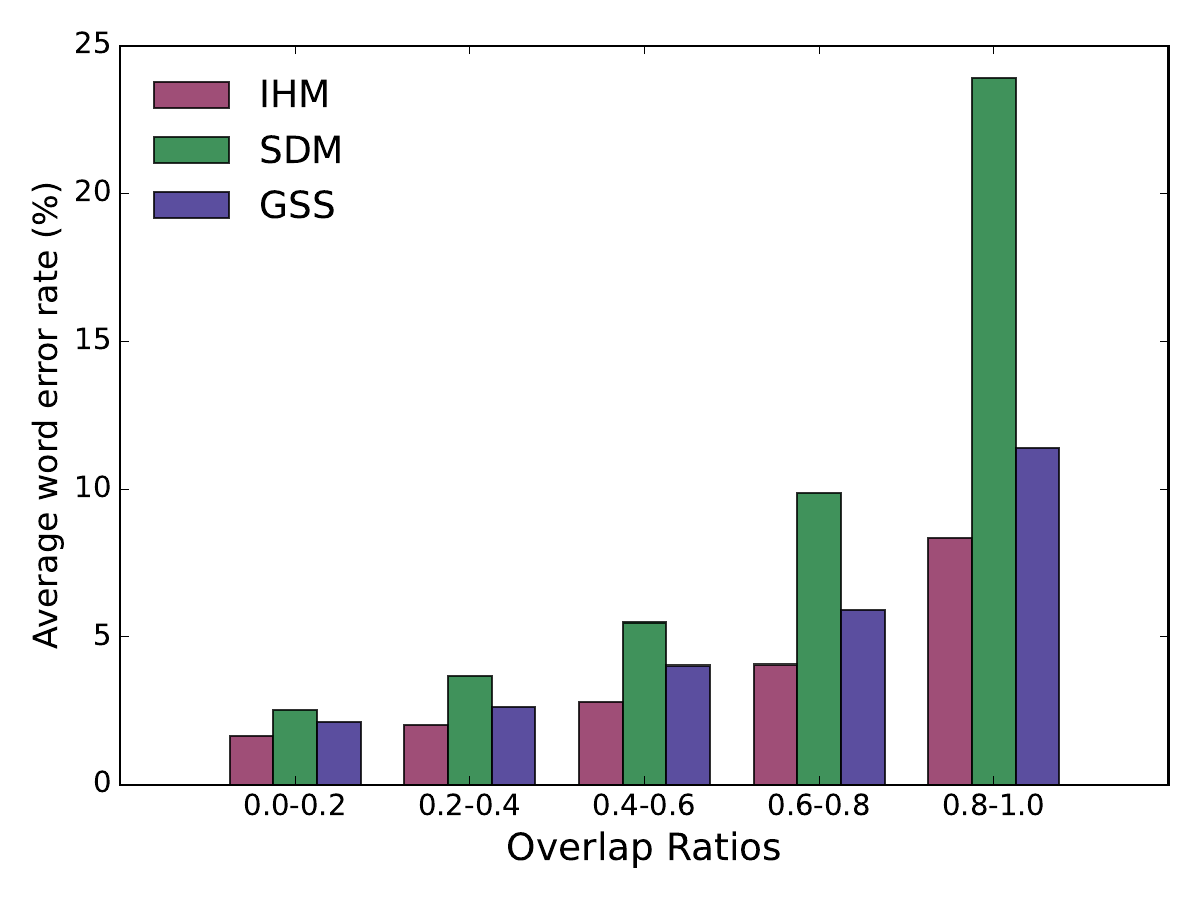}
    \caption{Average WER for different microphone settings and overlap ratios.}
    \label{fig:gss_ovl_wer}
\end{figure}

\subsubsection{Which factors are most important for GSS?}

\begin{figure}[t]
\begin{subfigure}{0.49\linewidth}
\centering
\includegraphics[width=\linewidth]{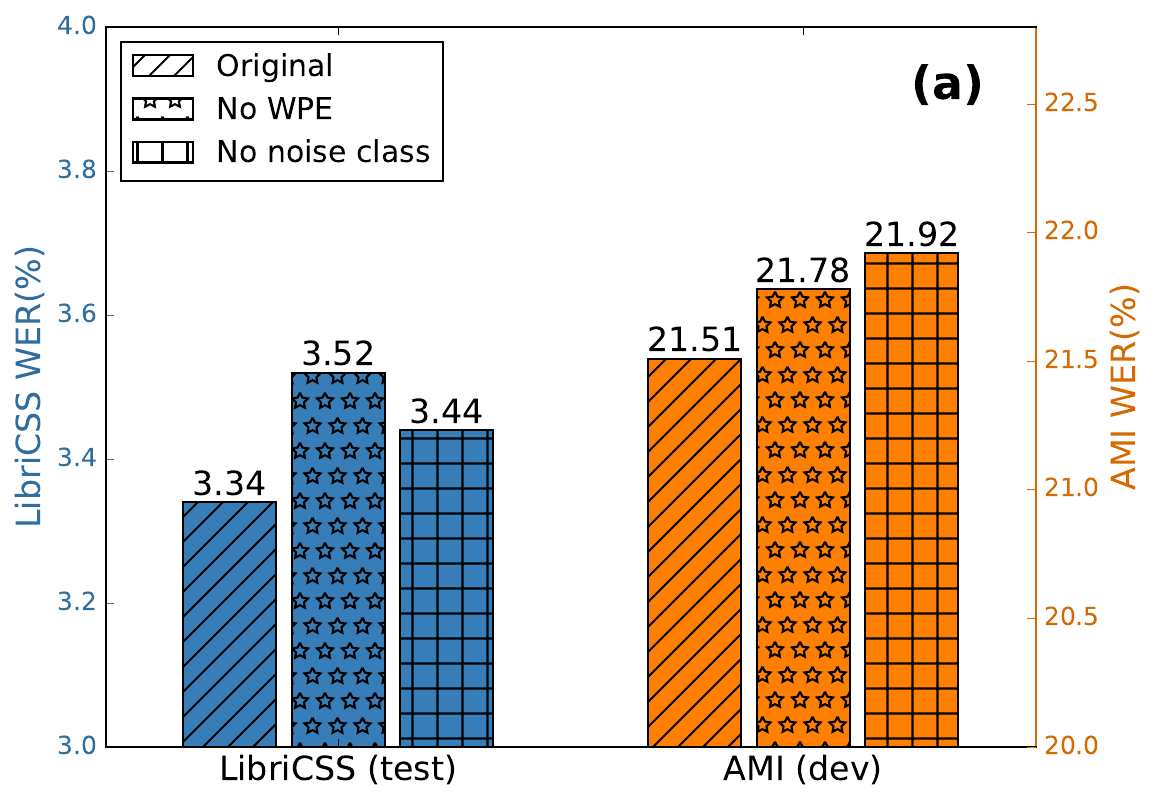}
\captionlistentry{}
\label{fig:wpe}
\end{subfigure}
\begin{subfigure}{0.49\linewidth}
\centering
\includegraphics[width=\linewidth]{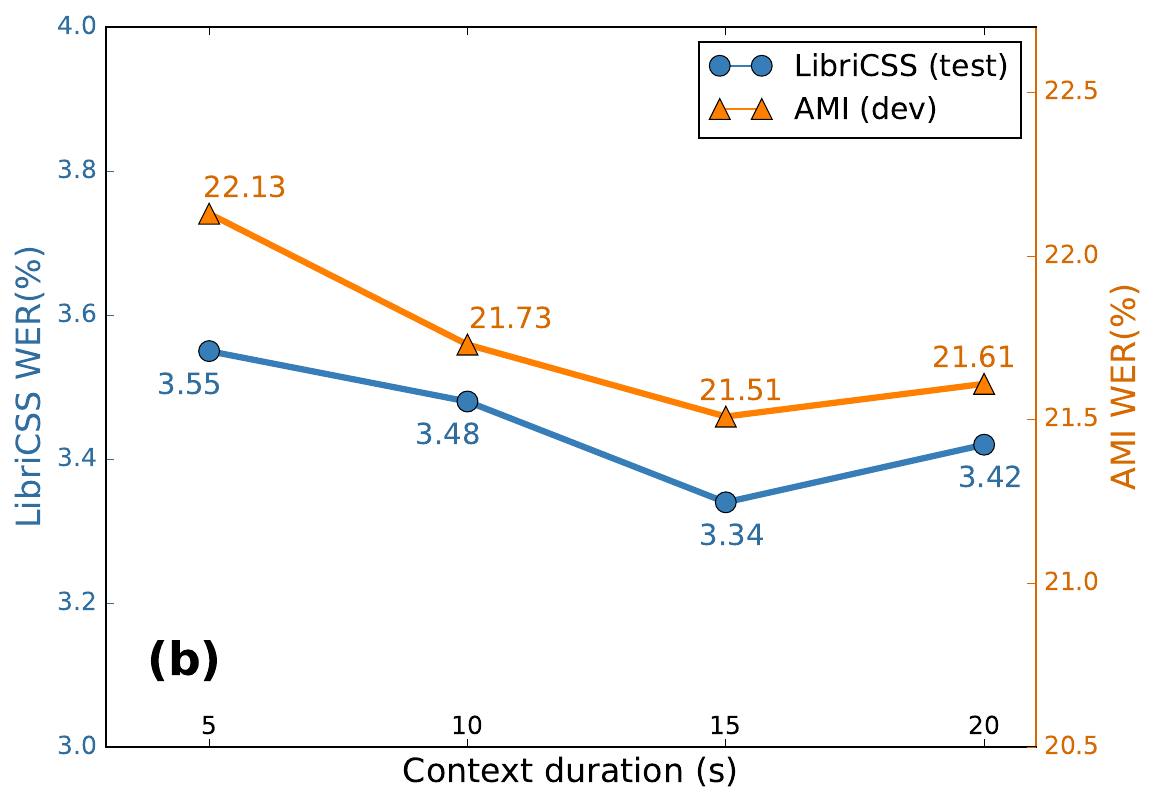}
\captionlistentry{}
\label{fig:context}
\end{subfigure}\hfill
\begin{subfigure}{0.49\linewidth}
\centering
\includegraphics[width=\linewidth]{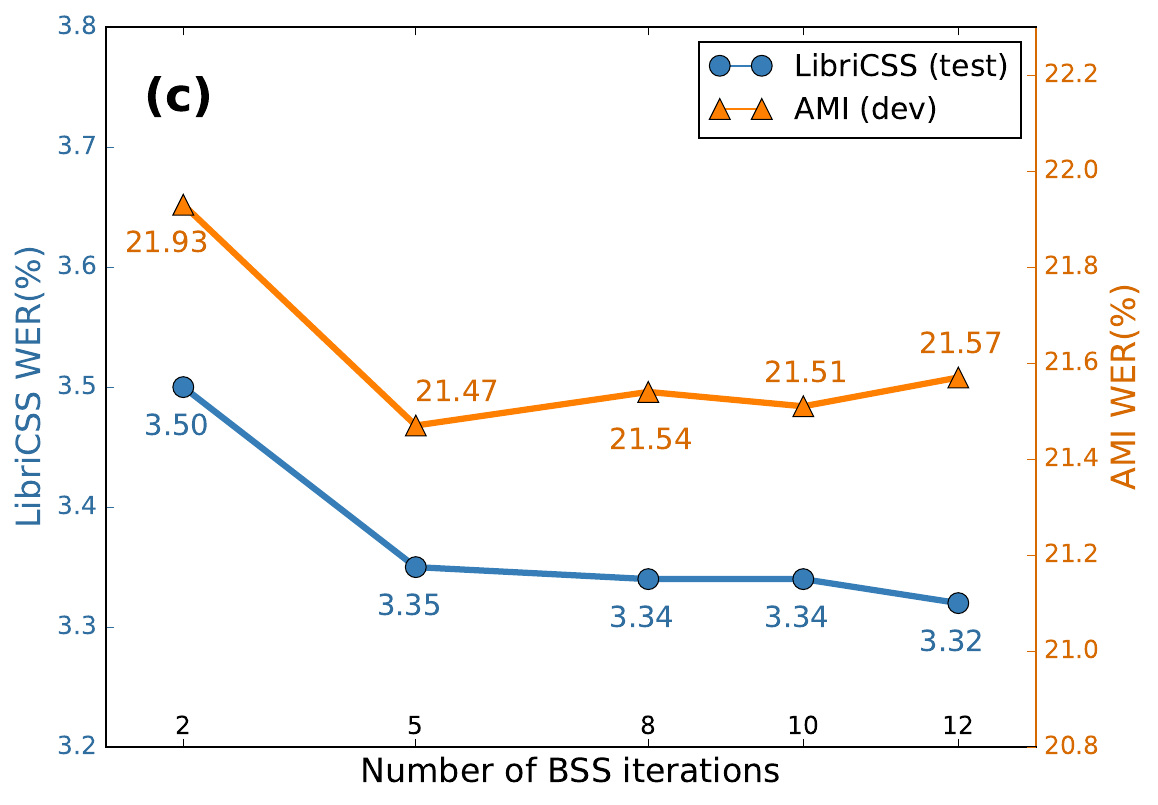}
\captionlistentry{}
\label{fig:iter}
\end{subfigure}
\begin{subfigure}{0.49\linewidth}
\centering
\includegraphics[width=\linewidth]{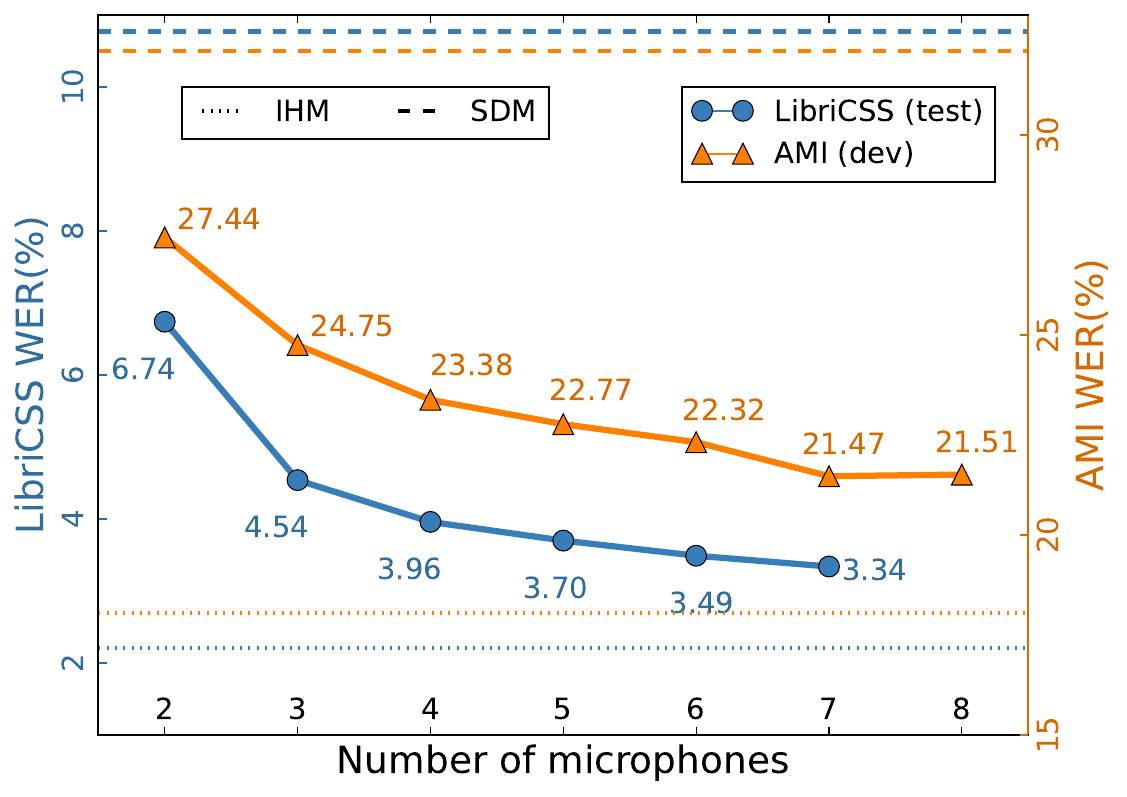}
\captionlistentry{}
\label{fig:channels}
\end{subfigure}\hfill
\caption{Impact of several factors on ASR performance. In each figure, the left and right y-axes denoten WERs for LibriCSS (test) and AMI (dev), respectively, with the axes scaled according to the range of corresponding WER values.}
\label{fig:gss_analysis}
\end{figure}

% Effect of WPE-based dereverberation
% Effect of noise class
% Effect of context duration
% Effect of number of channels
% Effect of number of BSS iterations
We performed ablation studies to investigate the effect of several GSS parameters --- WPE, noise class, context duration, number of iterations for CACGMM inference, and number of input channels --- on the downstream ASR performance, as shown in Fig.~\ref{fig:gss_analysis}. 
WPE was found to be more important for LibriCSS, while using an additional noise class was more important for AMI (Fig.~\ref{fig:wpe}). 
This may be because AMI contains occasional background noise, which is absent in LibriCSS. 
Increasing the context duration from 5s to 15s resulted in consistent WER gains (Fig.~\ref{fig:context}). 
\citet{Boeddeker2018FrontendPF} have previously made a similar observation for CHiME-5~\citep{Barker2015TheT}, where a 15s context resulted in better WER compared to a 2s context.
However, we found that adding further context beyond 15s degraded WER performance.
If the context is expanded too far, the context window may also contain target-speaker segments, which would contaminate the statistics of the noise class.

For both datasets, increasing the number of BSS iterations (for CACGMM inference) beyond 5 did not result in any WER improvements (Fig.~\ref{fig:iter}).
Finally, using more input channels was found to be the single most important factor for better WER performance. 
For example, using seven input channels resulted in relative WER reduction of 50.4\% and 21.8\% on LibriCSS and AMI, respectively, compared to using two channels. 
Nevertheless, it follows the law of diminishing returns, as evident by the rate of decay in Fig.~\ref{fig:channels}.

\subsubsection{Analysis of speed-up}
\label{sec:gss_speed}
% Speed-up compared to original implementation (on CHiME-6?)

We compared our GSS implementation with the original GSS on the CHiME-6 development set in terms of wall clock time and ASR performance, as shown in Table~\ref{tab:gss_speed}. 
For ASR inference, we used the publicly available Kaldi recipe and pretrained models from JHU-CLSP's submission to the CHiME-6 challenge\footnote{\url{https://github.com/kaldi-asr/kaldi/blob/master/egs/chime6/s5b_track1}}$^,$\footnote{\url{https://kaldi-asr.org/models/m12}}~\cite{Arora2020TheJM}. 
We found that our implementation obtained an effective speed-up of 290 without any degradation in WER. 
Since CHiME-6 has non-stationary speakers, we disabled segment batching for this experiment.
We can obtain even further speed-ups by enabling this for meeting-like data where speakers are stationary.

\begin{table}[t]
\centering
\caption{Compute time for our GSS implementation compared with original on CHiME-6 dev set, using all available channels, 15s context, and 20 BSS iterations. ``Time'' is the actual wall clock time (in hours), while ``cum. time'' is the effective total time for all jobs. Speedup is the ratio of the cumulative times.}
\label{tab:gss_speed}
\adjustbox{max width=\linewidth}{
\begin{tabular}{@{}lcrrrrr@{}}
\toprule
\textbf{GSS} & \textbf{Compute} & \textbf{Time} & \textbf{Cum. time} & \textbf{Speedup} & \textbf{WER} & \textbf{+RNNLM} \\
\midrule
\textbf{Original} & 80 x Xeon & 19.3 & 1542.6 & 1.0 & 44.7 & 43.5 \\
\textbf{Ours} & 4 x V100 & 1.3 & 5.3 & 292.2 & 44.2 & 43.1  \\
% \textbf{Ours} & 10s & 5 & RTX (24G) x 4 & 2.33 & 9.31 &  & \\
\bottomrule
\end{tabular}
}
\end{table}

\section{Conclusion}

In this chapter, we described the problem of target speaker extraction, and showed how to do it using guided source separation (GSS).
Although GSS was first proposed in \citet{Boeddeker2018FrontendPF} for the CHiME-5 challenge, it has seen limited use for meeting transcription, partially because of a slow and iterative implementation.
This has also made it challenging to perform detailed ablation studies for this method.
To solve this issue, we proposed a novel GPU-accelerated implementation of this technique, drawing inspiration from modern deep learning pipelines. 
On the CHiME-6 benchmark, it was found to be 300x faster than the original implementation, thus removing the computational bottleneck associated with this technique. 

Through experiments conducted on LibriCSS, AMI, and AliMeeting, we showed that GSS provides between 10 and 20 dB SI-SDR improvements compared to single-channel far-field recordings, when measured against close-talk recordings, and significant perceptual and intelligibility improvements.
Extrinsic evaluation on ASR and speaker similarity showed that GSS-based enhancement can recover up to 80\% of the WER difference, in going from close-talk to far-field conditions, and also drastically improve target-speaker information in the signal.
We performed several ablation studies to study the effect of GSS parameters, and showed that using more input channels is the single most important factor for better ASR performance.

With this implementation and analysis, we have developed another important component in our modular system for speaker-attributed multi-talker ASR.
A fast and accurate target-speaker extraction method is essential to be able to obtain single-speaker segments from multi-talker mixtures, which can then be transcribed by an ASR module.
Furthermore, using GSS for this task allows us to leverage the overlap-aware diarization system that we have developed earlier.
Our ablation studies show that we can obtain good TSE performance with GSS with few iterations, provided we have enough input channels.
This is useful because most devices, such as smart speakers, already contain array microphones which can provide such multi-channel inputs.
More importantly, the GSS method is invariant to the number of input channels --- this means that the same TSE module can be used across a wide variety of scenarios, and the output can be fed into a universal ASR component.
In the next chapter, we will tie together all of these components to build a complete pipeline for speaker-attributed transcription.

% \cleardoublepage

\chapter{A Modular Framework for Multi-talker Speech Recognition}
\label{chap:modular}

In the last several chapters, we formulated and analyzed individual components that make up the modular framework for multi-talker ASR.
With these components in place, we can build a system that takes as input a multi-channel recording and generates a speaker-attributed transcription of the recording.
In this chapter, we describe this modular system and use it to perform transcription for several meeting-style benchmark data.
We will then analyze the performance to understand the role of the different components in causing and propagating errors in transcription.

\section{Probabilistic formulation of multi-talker ASR}

Let $R$ be a long, multi-channel recording, recorded using $M$ microphones.
The objective of the multi-talker ASR (or ``who spoke what'') problem is to obtain a speaker-attributed transcription, $W(R)$, defined as
\begin{equation}
    W(R) = \{\mathbf{w}_1,\ldots,\mathbf{w}_K\},
\end{equation}
where $K$ is the number of speakers in the recording $R$, and $\mathbf{w}_k$ denotes the transcription for speaker $k$.
For evaluation, we will consider the concatenated minimum-permutation word error rate (cpWER) metric, described in Section~\ref{sec:intro_metrics}, and so all permutations of $\mathbf{w}_k$'s are equivalent under cpWER.

In the modular framework, instead of directly optimizing for $W(R)$, we will instead compute a proxy solution $h(R)$, which is defined as
\begin{equation}
\label{eq:modular_task}
    h(R) = \{(\Delta_j, u_j, \mathbf{y}_j): 1\leq j \leq N\},
\end{equation}
where $\Delta_j = (t_j^{\mathrm{st}},t_j^{\mathrm{en}})$ denotes the start and end times for segment $j$, $u_j \in [K]$ is the speaker label assigned to segment $j$, $N$ denotes the number of estimated \textit{segments}, and $\mathbf{y}_j$ denotes the predicted transcript for segment $j$.
It is easy to show that a deterministic function $\mathcal{Y}$ exists that maps $h(R)$ to some $W(R)$ by concatenating the transcripts of same-speaker segments in order of start time.
Let $\hat{h}$ maximize $P(h(R) \mid R)$ and $\hat{W}$ maximize $P(W \mid R)$.
Since $\mathcal{Y}$ is many-to-one, we cannot guarantee that $\hat{W} = \mathcal{Y}(\hat{h})$ in general.
Nevertheless, we will use this proxy task since it provides convenient factorization.

From a probabilistic perspective, finding the optimal $h(R)$ can be formulated as a maximum \textit{a posteriori} problem such that
\begin{equation}
    h(R) = \text{arg}\max_{h^{\prime}} P(h^{\prime} \mid R), \quad \text{where} \quad h^{\prime} = \{(\Delta_j,u_j,\mathbf{y}_j):1\leq j \leq N\}.
\end{equation}
We can simplify the $P(h^{\prime} \mid R)$ term through basic probability rules and conditional independence assumptions.
Suppose $\Delta_1^N = \{\Delta_j: 1\leq j \leq N\}$, and likewise for $u_1^N$ and $\mathbf{y}_1^N$.
We have
\begin{align}
P(h^{\prime} \mid R) &= P\left(\Delta_1^N,u_1^N,\mathbf{y}_1^N \mid R\right) \\
&= P\left(\Delta_1^N,u_1^N \mid R\right) P\left(\mathbf{y}_1^N \mid R, \Delta_1^N,u_1^N \right). \label{eq:modular_diar}
\end{align}
Here, the first term may be considered a probabilistic formulation of the speaker diarization task.
For the second term, let us define $N$ continuous random variables $\mathbf{X}_1,\ldots,\mathbf{X}_N$ representing the ``target speaker signal'' for the $N$ segments, respectively.
We can then marginalize over the $\mathbf{X}_1^N$ to get
\begin{align}
P\left(\mathbf{y}_1^N \mid R, \Delta_1^N, u_1^N \right) &= \int_{\mathbf{X}_1^N} P\left(\mathbf{X}_1^N,\mathbf{y}_1^N \mid R, \Delta_1^N, u_1^N \right) \\
&= \int_{\mathbf{X}_1^N} P\left(\mathbf{X}_1^N \mid R, \Delta_1^N, u_1^N \right) P\left(\mathbf{y}_1^N \mid R, \Delta_1^N,u_1^N,\mathbf{X}_1^N \right) \\
&= \int_{\mathbf{X}_1^N} \prod_{j=1}^N P\left(\mathbf{X}_j \mid R, \Delta_j, u_j \right) \prod_{j=1}^N P\left(\mathbf{y}_j \mid \mathbf{X}_j \right), \label{eq:modular_condasr}
\end{align}
where \eqref{eq:modular_condasr} is based on the following conditional independence assumptions.
\begin{enumerate}
    \item Given $R$, $\Delta_j$, and $u_j$, $\mathbf{X}_j$ is conditionally independent of $\mathbf{X}_i$, $\Delta_i$, and $u_i$, $\forall i \in [N]\setminus\{j\}$.
    \item Given $\mathbf{X}_j$, $\mathbf{y}_j$ is conditionally independent of all other random variables.
\end{enumerate}
We will discuss these assumptions in more detail in the following subsections.
Plugging \eqref{eq:modular_condasr} into \eqref{eq:modular_diar}, we get
\begin{equation}
    h(R) = \mathrm{arg}\max_{\Delta_1^N,u_1^N,\mathbf{y}_1^N} \left[ P\left(\Delta_1^N,u_1^N \mid R\right) \int_{\mathbf{X}_1^N} \prod_{j=1}^N P\left(\mathbf{X}_j \mid R, \Delta_j, u_j \right) \prod_{j=1}^N P\left(\mathbf{y}_j \mid \mathbf{X}_j \right) \right].
\label{eq:modular_obj}
\end{equation}
Under the ``modular'' framework, we make some further modeling choices for computational tractability.
Instead of solving for the joint optimum in \eqref{eq:modular_obj}, we first solve the diarization task to obtain
\begin{equation}
    \hat{\Delta}_1^N, \hat{u}_1^N = \underbrace{\arg \max_{\Delta_1^N, u_1^N} P\left(\Delta_1^N,u_1^N \mid R\right)}_{\text{speaker diarization}}.
\end{equation}

Next, we assume that $\mathbf{X}_j$ is approximated using a deterministic function $g(\cdot)$ which is estimated through the target speaker extraction task as
\begin{equation}
\label{eq:modular_tse}
    \hat{\mathbf{X}}_j = \underbrace{g(R,\hat{\Delta}_j,\hat{u}_j)}_{\text{target speaker extraction}}.
\end{equation}
Following this assumption, we can remove the marginalization (since we are using a point-wise estimate for $\mathbf{X}_1^N$) and the second term in \eqref{eq:modular_condasr} can be approximated as
\begin{equation}
\label{eq:modular_approx}
P\left(\mathbf{y}_1^N \mid R, \Delta_1^N, u_1^N \right) \approx \prod_{j=1}^N \underbrace{P\left(\mathbf{y}_j \mid \hat{\mathbf{X}}_j \right)}_{\text{speech recognition}}.
\end{equation}

\begin{figure}
    \centering
    \includegraphics[width=0.8\linewidth]{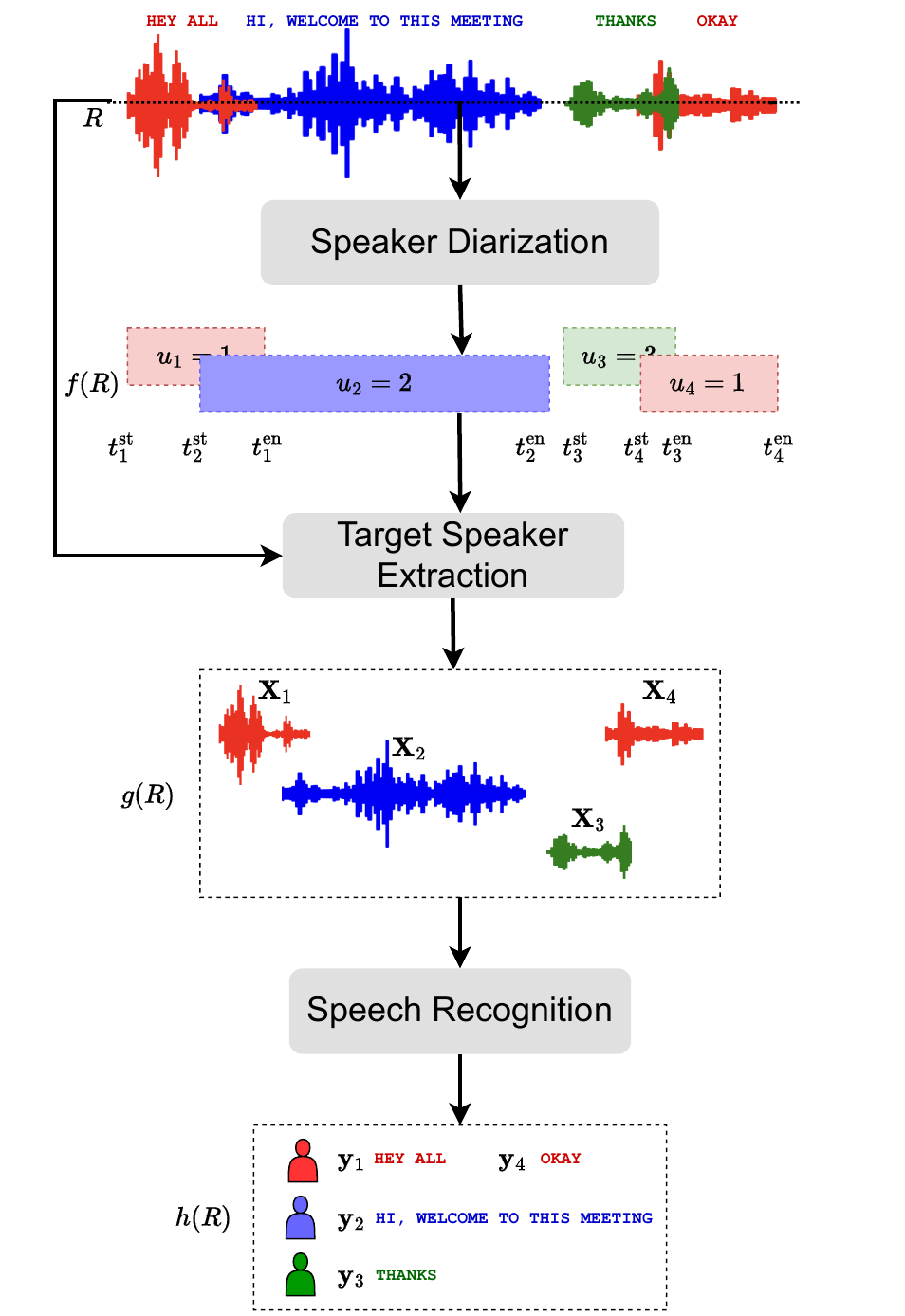}
    \caption{Flow diagram of the modular system for speaker-attributed transcription. The input is a multi-channel recording, and the output is speaker-wise transcripts.}
    \label{fig:modular_system}
\end{figure}

A block diagram representing this modular system is shown in Fig.~\ref{fig:modular_system}.
It consists of three components: (i) overlap-aware speaker diarization; (ii) target-speaker extraction; and (iii) automatic speech recognition (ASR).
In the previous chapters, we have described (i) and (ii) in some detail.
Next, we will summarize them in the context of the full system, and describe the ASR component, which is based on neural transducers.

\subsection{Overlap-aware diarization}

As described earlier in Section~\ref{sec:oasc_background}, a diarization system is first required to obtain $f(R) = \{(\Delta_j, u_j): 1\leq j \leq N\}$.
Although $R$ is a multi-channel recording, our diarization system only makes use of the first-channel to compute $f(R)$.
Specifically, we used the multi-class spectral clustering based diarization system with overlap assignment, described in Section~\ref{sec:oasc_method}, which consists of (i) a speech activity detector (SAD), (ii) a speaker embedding extractor, (iii) an overlap detector, and (iv) an overlap-aware spectral clustering module. 

As explained in Section~\ref{sec:intro_metrics}, performing well in diarization --- using the diarization error rate (DER) metric --- requires accurate estimation of $\Delta_1^N$.
However, our eventual objective is to perform speaker-attributed transcription, which is evaluated using cpWER.
While cpWER does not explicitly penalize incorrect estimates of $\Delta_1^N$, it is implicitly penalized since we need accurate time-stamps for two reasons.
First, accurate estimate of speaker activities and segmentation is required for target-speaker extraction to work well, as evident from \eqref{eq:modular_tse}.
Second, the final transcription is obtained by time-ordered concatenation of a speaker's transcripts, for which we also need reasonably accurate start times.
For these reasons, downstream cpWER performance is weakly correlated with the performance of the diarization system.

\subsection{Target-speaker extraction}

From \eqref{eq:modular_tse}, we have modeled $\mathbf{X}_j$ using target-speaker extraction as
\begin{equation}
\nonumber
\hat{\mathbf{X}}_j = g(R,\hat{\Delta}_j,\hat{u}_j).
\end{equation}

$\hat{\mathbf{X}}_j$ contains only the signal relevant to speaker $\hat{u}_j$, and suppresses background noise or interfering speakers.
In our pipeline, this target-speaker extraction is performed using the guided source separation (GSS) method described in Chapter~\ref{chap:gss}.
By using GSS, our conditional independence assumptions in \eqref{eq:modular_condasr} are valid, since the extraction is performed for each segment independent of other segments.

\subsection{Automatic speech recognition}

Finally, given the output of the targer-speaker extractor, an ASR model is applied on each segment $\hat{\mathbf{X}}_j$ to obtain the corresponding transcript $\mathbf{y}_j$.
These are then combined with $f(R) = \{(\Delta_j,u_j):1 \leq N\}$ to obtain the desired output, i.e., $h(R)=\{(\Delta_j,u_j,\mathbf{y}_j):1 \leq N\}$.

In general, any offline ASR model can be used for this component, such as hybrid HMM-DNN~\cite{Bourlard1996HybridCM}, connectionist temporal classification (CTC)~\cite{Graves2006ConnectionistTC}, or attention-based encoder-decoder (AED)~\cite{Chorowski2015AttentionBasedMF}. 
We used neural transducers in our implementation.
Neural transducers (using RNNs or transformers)~\cite{Graves2012SequenceTW} have become the dominant modeling technique in end-to-end on-device speech recognition~\cite{He2019StreamingES, Wu2020StreamingTA, Li2020TowardsFA, Shangguan2019OptimizingSR}, since they allow streaming transcription similar to CTC models~\cite{Zhang2021BenchmarkingLC, Li2019ImprovingRT}, while still retaining conditional dependence, like AEDs~\cite{Chan2016ListenAA,Kim2017JointCB}.
We will briefly describe the formulation of neural transducers for ASR.

In conventional single-talker ASR, audio features for a segmented utterance $\mathbf{X} \in \mathbb{R}^{T\times F}$ (where $T$ and $F$ denote the number of time frames and the input feature dimension, respectively) are provided as input to the system, and we are required to predict the transcript $\mathbf{y} = (y_1,\ldots,y_U)$, where $y_u \in \mathcal{V}$ denotes output units such as graphemes or word-pieces, and $U$ is the length of the label sequence. 
For the case of discriminative training, this requires computing the conditional likelihood $P(\mathbf{y} \mid \mathbf{X})$ (or its log for numerical stability). For inference, we search for $\hat{\mathbf{y}} = \text{arg}\max_{\mathbf{y}}P(\mathbf{y} \mid \mathbf{X})$, often in a constrained search space using greedy or beam search.
Transducers achieve this by marginalizing over the set of all alignments $\mathbf{a} \in \bar{\mathcal{V}}^{T+U}$, where $\bar{\mathcal{V}} = \mathcal{V}\cup \{\phi\}$ and $\phi$ is called the blank label.
Formally,
\begin{equation}
P(\mathbf{y} \mid \mathbf{X}) = \sum_{\mathbf{a}\in \mathcal{B}^{-1}(\mathbf{y})} P(\mathbf{a} \mid \mathbf{X}),
\label{eq:rnnt}
\end{equation}
where $\mathcal{B}$ is a deterministic mapping from an alignment $\mathbf{a}$ to an output sequence $\mathbf{y}$. 
In the original transducer, all alignments $\mathbf{a}$ consist of $\mathbf{y}$ interspersed with $T$ blank tokens, usually represented as a $T\times U$ lattice with $\phi$ on horizontal arcs and $\mathbf{y}$ on vertical arcs. 
Since there may be ${T+U \choose U}$ such paths on the lattice, some transducer variants (such as recurrent neural aligner~\cite{Sak2017RecurrentNA} or monotonic RNN-T~\cite{Tripathi2019MonotonicRN,Moritz2022AnIO}) restrict the number of non-blank tokens emitted per time step. 
External or internal alignments may also be used to further prune the lattice for marginalization~\cite{Mahadeokar2020AlignmentRS,Kuang2022PrunedRF}.  

\begin{figure}
    \centering
    \includegraphics[width=0.6\linewidth]{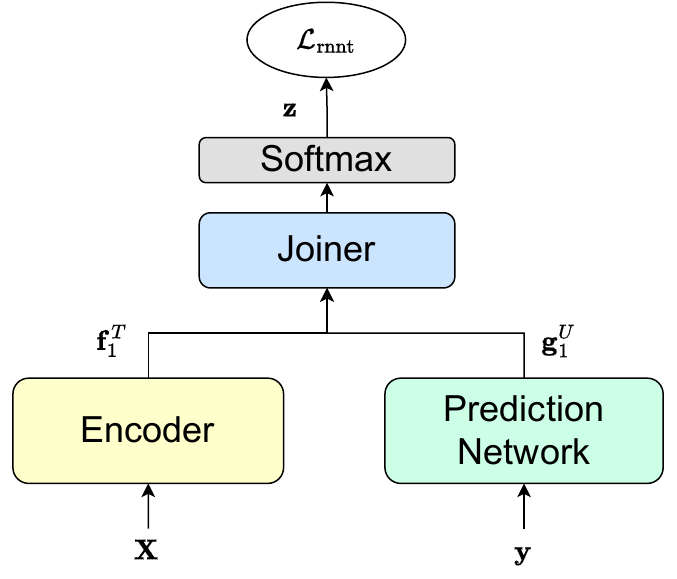}
    \caption{Block diagram of the neural transducer.}
    \label{fig:modular_transducer}
\end{figure}

Transducers parameterize $P(\mathbf{a} \mid \mathbf{X})$ with an encoder, a prediction network, and a joiner, as shown in Fig.~\ref{fig:modular_transducer}.
The encoder maps $\mathbf{X}$ into hidden representations $\mathbf{f}_1^T$, while the prediction network maps $\mathbf{y}$ into $\mathbf{g}_1^U$.
The joiner combines the outputs from the encoder and the prediction network to compute logits $\mathbf{z}_{t,u}$ which are fed to a softmax function to produce a posterior distribution over $\bar{\mathcal{V}}$ for each time step. 
Under assumptions of full context encoder, we can expand (\ref{eq:rnnt}) as
\begin{align}
    P(\mathbf{y} \mid \mathbf{X}) &= \sum_{\mathbf{a}\in \mathcal{B}^{-1}(\mathbf{y})} \prod_{t=1}^{T+U} P(\mathbf{a}_t \mid \mathbf{f}_1^T,\mathbf{g}_1^{u(t)-1}) \\
    &= \sum_{\mathbf{a}\in \mathcal{B}^{-1}(\mathbf{y})} \prod_{t=1}^{T+U} \mathrm{Softmax}(\mathbf{z}_{t,u(t)}),
\end{align}
where $u(t)\in\{1,\ldots,U\}$ denotes the index in the label sequence at time $t$.
%
% We will denote the log of this expression as $\mathcal{L}_{\text{rnnt}}(\mathbf{y},\mathbf{z})$ (or simply $\mathcal{L}_{\text{rnnt}}$) for the remainder of this paper.

\section{Experimental setup}

We show speaker-attributed transcription results on three meeting benchmarks: LibriCSS, AMI, and AliMeeting.
For these experiments, we report diarization error rates (DER) and concatenated minimum-permutation WER (cpWER) in order to analyze the impact of diarization errors on downstream ASR. 
In the case of AliMeeting, we report cpCER (instead of cpWER) since this is a Mandarin dataset.
We did not use any collars to compute DERs for LibriCSS and AMI, but a collar of 0.25 was used for AliMeeting following the original work. 

For speaker embedding extraction, we used a pre-trained ResNet101-based network~\cite{Landini2020BayesianHC}, which was trained on VoxCeleb~\cite{Nagrani2017VoxCelebAL,Nagrani2020VoxcelebLS} and CN-Celeb~\cite{Fan2020CNCelebAC}.
For SAD and overlap detection, we used an end-to-end segmentation method based on Pyannote~\cite{Bredin2021EndtoendSS}.
The segmentation model was fine-tuned on the training set annotations for AMI and AliMeeting, and used off-the-shelf for LibriCSS.
For the ASR component, we used neural transducers with the same configuration as the ones described in Section~\ref{sec:gss_setup}.
We have reproduced the key results from Table~\ref{tab:gss_asr} below for reference, as these denotes the minimum achievable cpWER if we have access to an oracle diarization system.

\begin{table}[t]
\centering
\caption{WER (or CER) results using an oracle diarization system and GSS. These represent the minimum achievable cpWER (or cpCER) under the modular framework using the same ASR.}
\label{tab:modular_asr}
\begin{tabular}{@{}lcccc@{}}
\toprule
\textbf{Dataset} & \textbf{Ins.} & \textbf{Del.} & \textbf{Sub.} & \textbf{WER} \\ \midrule
LibriCSS & 0.31 & 0.89 & 2.14 & 3.34 \\
AMI & 2.43 & 6.07 & 14.33 & 22.83 \\
AliMeeting & 1.09 & 4.87 & 9.03 & 14.98 \\ 
\bottomrule
\end{tabular}
\end{table}

\section{Results \& Discussion}
\label{sec:modular_asr}

\begin{table}[tbp]
\centering
\caption{Effect of diarization and GSS components on speaker-attributed ASR, as measured by cpWER (or cpCER).}
\label{tab:modular_diar}
\adjustbox{max width=\linewidth}{
\begin{tabular}{@{}lccccccccc@{}}
\toprule
\multirow{2}{*}{\textbf{Diarizer}} & \multicolumn{4}{c}{\textbf{DER}} & \multicolumn{1}{c}{\multirow{2}{*}{\textbf{GSS}}} & \multicolumn{4}{c}{\textbf{cpWER}} \\
\cmidrule(r{2pt}){2-5} \cmidrule(l{2pt}){7-10}
& \multicolumn{1}{c}{\textbf{FA}} & \multicolumn{1}{c}{\textbf{MS}} & \multicolumn{1}{c}{\textbf{Conf.}} & \multicolumn{1}{c}{\textbf{Total}} & \multicolumn{1}{c}{} & \multicolumn{1}{c}{\textbf{Ins.}} & \multicolumn{1}{c}{\textbf{Del.}} & \multicolumn{1}{c}{\textbf{Sub.}} & \multicolumn{1}{c}{\textbf{Total}} \\
\midrule
% LibriCSS
\multicolumn{10}{c}{\textbf{LibriCSS}} \\
\hline
\multirow{2}{*}{\textbf{Spectral}} & \multirow{2}{*}{1.19} & \multirow{2}{*}{10.37} & \multirow{2}{*}{3.37} & \multirow{2}{*}{14.93} & \xmark & 1.00 & 13.62 & 3.69 & 18.30 \\
 &  &  &  &  & \cmark & 0.73 & 12.33 & 2.80 & 15.86 \\
% \cline{1-10}
\multirow{2}{*}{\textbf{~~ + OVL}} & \multirow{2}{*}{2.22} & \multirow{2}{*}{3.79} & \multirow{2}{*}{5.33} & \multirow{2}{*}{11.34} & \xmark & 2.64 & 8.09 & 6.36 & 17.09 \\
 &  &  &  &  & \cmark & 1.62 & 7.13 & 3.38 & \textbf{12.12} \\
\hline \hline
% AMI
\multicolumn{10}{c}{\textbf{AMI}} \\
\hline
\multirow{2}{*}{\textbf{Spectral}} & \multirow{2}{*}{3.24} & \multirow{2}{*}{18.15} & \multirow{2}{*}{4.14} & \multirow{2}{*}{25.53} & \xmark & 2.64 & 20.32 & 15.49 & 38.45 \\
 &  &  &  &  & \cmark & 2.59 & 18.00 & 12.96 & 33.55 \\
% \cline{1-10}
\multirow{2}{*}{\textbf{~~ + OVL}} & \multirow{2}{*}{7.39} & \multirow{2}{*}{9.63} & \multirow{2}{*}{6.67} & \multirow{2}{*}{23.69} & \xmark & 4.37 & 14.47 & 19.70 & 38.54 \\
 &  &  &  &  & \cmark & 3.57 & 12.24 & 15.21 & \textbf{31.02} \\
\hline \hline
% AliMeeting
\multicolumn{10}{c}{\textbf{AliMeeting}} \\
\hline
\multirow{2}{*}{\textbf{Spectral}} & \multirow{2}{*}{0.17} & \multirow{2}{*}{13.60} & \multirow{2}{*}{2.60} & \multirow{2}{*}{16.37} & \xmark & 1.16 & 26.38 & 10.05 & 37.59 \\
 &  &  &  &  & \cmark & 0.86 & 24.34 & 7.23 & 32.43 \\
% \cline{1-10}
\multirow{2}{*}{\textbf{~~ + OVL}} & \multirow{2}{*}{2.83} & \multirow{2}{*}{5.96} & \multirow{2}{*}{5.64} & \multirow{2}{*}{14.43} & \xmark & 2.32 & 18.82 & 14.30 & 35.44 \\
 &  &  &  &  & \cmark & 1.69 & 17.00 & 9.76 & \textbf{28.45} \\
\bottomrule
\end{tabular}}
\end{table}

Table~\ref{tab:modular_diar} presents the main results for this chapter. 
For each of the three meeting benchmarks, we show the DER and corresponding cpWER results when the diarization system is (or is not) overlap-aware.
Furthermore, for each of these systems, we compare the performance when using GSS-based target-speaker extraction, as opposed to simply extracting the corresponding segment from the input mixture.
These results demonstrate the effect of (i) the overlap-aware spectral clustering and (ii) GSS-based TSE modules, on downstream speaker-attributed ASR performance.

\subsection{Results for different benchmarks}

Let us first consider the results for the different benchmarks.
As expected, LibriCSS is the easiest dataset since it is simulated from LibriSpeech utterances and contains long utterances.
On the other hand, AMI and AliMeeting are significantly harder due to the presence of natural conversational speech (such as backchannels).
Overall, the breakdown of errors in cpWER follows that of the diarization errors, i.e., insertion, deletion, and substitutions are roughly proportional to false alarms, missed speech, and speaker confusion, respectively.
Most errors come from deletions, possibly due to missed detection of overlapped speech.
For the real meeting benchmarks, substitutions are also significantly higher than LibriCSS.
We conjecture that very short utterances, such as those found in these datasets, are hard for speaker attribution, since the speaker embedding extractors are usually trained on longer utterances.

\subsection{Effect of overlap-aware diarization}

Next, let us consider the impact of overlap-aware diarization, i.e., ``Spectral'' versus ``+OVL``.
Overlap assignment improves diarization performance significantly, as measured by DER.
However, when no GSS is performed, the impact on downstream cpWER is negligible.
For example, the cpWER on LibriCSS improves by 1.21\% absolute, while it actually degrades for AMI (from 38.45\% to 38.54\%).
This is due to the fact that although our system detects more overlapping segments, these cannot be transcribed correctly since the ASR cannot handle overlapping speech.
This finding corroborates the results of the winning CHiME-6 system~\cite{Medennikov2020TheSS}, which was able to substantially improve ASR performance on unsegmented recordings using TS-VAD based diarization~\cite{Medennikov2020TargetSpeakerVA} and GSS-based enhancement.

\subsection{Effect of target-speaker extraction}

Finally, using GSS results in significant improvements, with relative cpWER (or cpCER) reductions of 29.1\%, 19.5\%, and 19.7\% on LibriCSS, AMI, and AliMeeting, respectively. 
Nevertheless, the cpWER performance is still far from the WERs obtained using an oracle diarizer.
While part of the difference is due to incorrect speaker attribution of the ASR output, other reasons are undetected overlapping segments and poor mask estimation due to inaccurate speaker activities.

\subsection{Qualitative analysis}

\begin{table}[tp]
\centering
\caption{Examples of speaker-attributed transcription for a short segment of the session ES2011a from AMI \texttt{dev} set, between 817s and 833s. The first block shows the reference transcript for this segment. Each line in the transcript denotes a different speaker. Different segments of the same speaker are separated with a pipe (``|'') character.}
\label{tab:modular_example}
\adjustbox{max width=\linewidth}{
\begin{tabular}{@{}lcp{9cm}r@{}}
\toprule
\textbf{Diarizer} & \textbf{GSS} & \textbf{Transcript} & \textbf{cpWER} \\
\midrule
-- & -- & \scriptsize \begin{tabular}{@{}p{9cm}}
\texttt{I ALSO THINK THOUGH THAT IT SHOULDN'T HAVE TOO MANY BUTTONS 'CAUSE I HATE THAT WHEN THEY HAVE TOO MANY BUTTONS AND I MEAN I KNOW IT HAS TO HAVE ENOUGH FUNCTIONS BUT LIKE I DON'T KNOW YOU JUST HAVE LIKE EIGHT THOUSAND BUTTONS AND YOU'RE LIKE NO YOU NEVER USE HALF OF THEM | SO} \\
\texttt{YEAH I AGREE | B BUTTON AND THE F BUTTON THEY DON'T DO ANYTHING} \\
\texttt{UM OH WE JUST | YEAH} \\
\texttt{YEAH YEAH YEAH}
\end{tabular} & 0\% \\
\hline

Oracle & \cmark & \scriptsize \begin{tabular}{@{}p{9cm}}
\texttt{I ALSO THINK THOUGH THAT IT SHOULDN'T HAVE TOO MANY BUTTONS 'CAUSE I HATE THAT ONLY HAVE TOO MANY BUTTONS AND I MEAN I KNOW IT HAS TO HAVE ENOUGH FUNCTIONS BUT LIKE JUST HAVE LIKE EIGHT THOUSAND BUTTONS AND YOU'RE LIKE NO YOU NEVER USE HALF OF THEM} \\
\texttt{YEAH I AGREE M THE BUTTON ON THE OFF BUTTON THEY DON'T DO ANYTHING} \\
\texttt{UM OH WE'RE JUST USING THAT | YEAH} \\
\texttt{YEAH YEAH}
\end{tabular} &  20.3\% \\
\hline

Spectral & \cmark & \scriptsize \begin{tabular}{@{}p{9cm}}
\texttt{I ALSO THINK THOUGH THAT IT SHOULDN'T HAVE TOO MANY BUTTONS 'CAUSE I HATE THAT ONLY HAVE TOO MANY BUTTONS AND I MEAN I KNOW IT HAS TO HAVE MANY FUNCTIONS BUT LIKE | I DUNNO JUST HAVE LIKE EIGHT THOUSAND BUTTONS AND YOU'RE LIKE NO YOU NEVER USE HALF OF THEM} \\
\end{tabular} & 40.5\% \\
\hline

\multirow{2}{*}{~~ + OVL} & \cmark & \scriptsize \begin{tabular}{@{}p{9cm}}
\texttt{I ALSO THINK THOUGH THAT IT SHOULDN'T HAVE TOO MANY BUTTONS 'CAUSE I HAD THAT ONLY HAVE TOO MANY BUTTONS AND I MEAN I KNOW IT HAS TO HAVE ENOUGH FUNCTIONS BUT LIKE | I DUNNO JUST HAVE LIKE EIGHT THOUSAND BUTTONS AND YOU'RE LIKE NO YOU NEVER USE HALF OF THEM} \\
\texttt{S YEAH I AGREE M THE BUTTON ON F BUTTON THEY DON'T DO ANYTHING} \\
\end{tabular} &  29.1\% \\
\cline{2-4}

 & \xmark & \scriptsize \begin{tabular}{@{}p{9cm}}
\texttt{I ALSO THINK THAT IT SHOULDN'T HAVE TOO MANY BUTTONS 'CAUSE I HATED NOT ONLY HAVE TOO MANY BUTTONS AND THINGS BUT I MEAN I KNOW IT HAS TO HAVE NO MANY FUNCTIONS BUT LIKE | I DUNNO JUST HAVE LIKE EIGHT THOUSAND BUTTONS AND YOU'RE LIKE YOU KNOW YOU NEVER USE HALF THE TIME} \\
\texttt{IT SHOULDN'T HAVE TOO MANY BUTTONS 'CAUSE I HATE THAT ONLY HAVE TOO MANY BUTTONS AND THINGS BUT I MEAN I KNOW IT HAS TO HAVE NO MANY FUNCTIONS BUT LIKE} \\
\end{tabular} &  72.2\% \\ 
\bottomrule
\end{tabular}}
\end{table}

In the previous sections, we have quantitatively described the effect of overlap-aware diarization and GSS on transcription in terms of cpWER.
Let us now consider an illustrative example of the transcriptions produced by the various systems, in order to qualitatively understand the differences.
In Table~\ref{tab:modular_example}, we show the reference and hypotheses transcriptions for a 16-second segment (from 817s to 833s) of session ES2011a, which is part of the AMI \texttt{dev} set.
The reference contains several segments of speech from four different speakers, as shown in the first block of the table.
Two of the four speakers only have very short utterances or back-channels.

The best possible system, obtained using an oracle diarizer and GSS-based TSE, obtains a cpWER of 20.3\% on this segment.
We see that most of the segments are transcribed relatively well with small errors due to conversational speech and noise.
If we replace the oracle diarizer with spectral clustering, we only obtain a single speaker, which is the dominant speaker in the segment.
This is the most common pattern of error in such systems, where interfering speakers are completely missed.
The resulting cpWER for this segment is 40.5\%, of which 35.4\% is caused by deletion errors.

On using overlap-aware spectral clustering with GSS, the second speaker is recovered and transcribed correctly, resulting in a reduction of cpWER to 29.1\%.
As expected, the deletion error is reduced to 19.0\%.
However, it is still difficult to recover very short utterances and backchannels from the other speakers.
We believe that a more powerful diarizer, such as TS-VAD~\cite{Medennikov2020TargetSpeakerVA}, may be better for these difficult segments.
Finally, if we do not use GSS for target-speaker extraction (and simply cut out the segments from the mixture), the cpWER jumps to 72.2\%.
Most of this is caused by repeated transcription, where the first speaker is transcribed again even in the second speaker's segment, possibly because the first speaker is louder.
This example demonstrates the importance of a TSE component even when we have overlap-aware diarization.

\section{Limitations}

In this chapter, we described the complete pipeline for modular transcription of multi-talker conversations.
By formalizing the problem through a probabilistic lens and making appropriate conditional independence assumptions, we were able to identify key components of the system: speaker diarization, target-speaker extraction, and ASR.
We then leveraged the methods proposed in the previous chapters for these components, and stringed them together for the overall transcription process.
Empirical evaluations on synthetic and real meeting benchmarks demonstrated that this modular perspective can be a viable solution to the problem of multi-talker ASR.

Nevertheless, our experiments also revealed important limitations about the approach, most of which result from the approximations made in the probabilistic formulation.
Instead of optimizing the total probability distribution, we approximated the solution by optimizing the diarization component and using the result for TSE and ASR.
Such an approximation may result in ``error propagation'', where errors made due to inaccurate modeling of diarization result in degraded performance for the downstream components.

Even if the diarization results are perfect, the system can still produce suboptimal solutions due to the various conditional independence assumptions made in the TSE and ASR components.
For instance, the transcript $\mathbf{y}_{j+1}$ is very likely to depend on the dialog history, i.e., $\mathbf{y}_{1:j}$, but this dependence is not taken into account in the ASR modeling.
While such assumptions make the modeling more convenient (i.e., the components can be trained independently on easily available data), they may also make it impossible to obtain the optimal solution.

Another major limitation of this approach is its inability to perform streaming transcription.
The diarization task, i.e., computing $f(R)$, is formulated as a clustering problem, which makes it inevitable to have to wait until the end of the recording when all the $N$ segments are available, in order to assign the relative speaker labels $u_j$.
As a result, since the TSE and ASR components are dependent on the diarization output, they cannot produce streaming results either.

For these reasons, it is attractive to consider the alternative, ``end-to-end'' perspective of multi-talker ASR, which aims to directly optimize to address the original problem, while avoiding any approximations or modeling assumptions.
Such a framework is driven by the availability of data and computational resources at scale, and leverages neural networks to model the distribution.
In the second part of this dissertation, we will focus our attention on one such end-to-end framework ---- the Streaming Unmixing and Recognition Transducer (SURT) ---  which is based on neural transducers that we have so far used for ASR.
In particular, we will consider various challenges in designing and training such a model for the task of multi-talker ASR, and propose ways to jointly perform transcription and speaker attribution in the same model.

% \cleardoublepage

\chapter{Streaming Unmixing and Recognition Transducers}
\label{chap:surt}

In the last chapter, we described a modular pipeline for speaker-attributed transcription.
Despite showing promising results on real meeting benchmarks, the pipeline approach had several limitations, most notably error propagation and an inability to perform streaming transcription.
In the second part of this dissertation, we will approach the problem from an alternative perspective, inspired by end-to-end models for speech recognition.
In particular, building on the success of neural transducers for ASR, we will propose an extension known as the Streaming Unmixing and Recognition Transducer (SURT).
In this chapter, we will describe SURT in some detail, focusing specifically on aspects such as model design, training data simulation, loss functions, and training techniques.
We will constrain the problem to multi-talker transcription, delegating the speaker attribution aspect to the following chapter.
Instead, we will focus on problems in transcription arising due to sparsely overlapped speech, long sequences, and quick turn-taking. 

\section{Introduction}
\label{sec:intro}

The conventional modeling approach for multi-talker ASR is through a cascade of separation and recognition systems. 
This approach leverages advancements in speech separation research to obtain single-speaker audio~\cite{Wang2017SupervisedSS}, which can then be used with a regular ASR component~\cite{Raj2020IntegrationOS}. 
However, such a model may be sub-optimal since the components are independently optimized, and may also require greater engineering efforts for maintenance~\cite{Wu2021InvestigationOP}. 

% Jointly optimized separation + recognition
Due to these limitations with cascaded systems, researchers have proposed jointly optimized models that combine separation and ASR and directly solve for the task of multi-talker transcription, often using a permutation-invariant training (PIT) objective~\cite{Yu2017RecognizingMS}. 
Such a paradigm has been explored in the context of hybrid HMM-DNN systems~\cite{Qian2017SingleChannelMS}, and more recently for end-to-end ASR~\cite{seki-etal-2018-purely}, with most research focusing on attention-based encoder-decoders (AEDs). 
For meeting transcription, a well-studied framework is \textit{serialized output training} (SOT), wherein multiple references including those that correspond to overlapped utterances are serialized into a single prediction sequence, using special tokens to demarcate speaker changes~\cite{Kanda2020SerializedOT}. 
A detailed review of related work is presented in Section~\ref{sec:surt_related}.

% SURT and MT-RNNT
In this chapter, we describe a framework called Streaming Unmixing and Recognition Transducer (SURT), which is based on neural transducers.
This is beneficial since neural transducers (using RNNs or transformers) have become the standard modeling technique for on-device speech recognition~\cite{he2019streaming, Wu2020StreamingTA, Li2019RNNT} in single-speaker settings.
The key step in SURT is to separate overlapping speech into multiple simultaneous \textit{branches} (or channels), each of which is transcribed by a shared transducer.
SURT was originally proposed by \citet{Lu2020StreamingEM} for the case of two-speaker single-turn conversations, and later extended to handle long-form multi-turn recordings in \citet{Raj2021ContinuousSM}. 
The original SURT has also been used to jointly perform speaker identification~\cite{Lu2021StreamingEM}, endpointing~\cite{Lu2022EndpointDF}, and segmentation~\cite{Sklyar2022SeparatorTransducerSegmenterSR}, although these studies have all been restricted to the single-turn setting.
Concurrent to these explorations, a similar modeling strategy called Multi-turn RNN-T (MT-RNNT) has also been proposed in \citet{Sklyar2021MultiTurnRF}.
Here, we refer to these class of models as SURT, but the same ideas should also be applicable to MT-RNNT.

We will begin by describing the formulation of SURT in detail, focusing on its application to continuous, streaming, multi-talker ASR.
A naive implementation of the model may suffer from several limitations. 
For example, performance often degrades on multi-turn sessions due to \textit{omission} and \textit{leakage} related errors. 
Here, omission refers to the case when an utterance is missed by all output branches, whereas leakage happens when a non-overlapping segment is transcribed on multiple branches.
Additionally, SURT requires training on long sessions with the transducer loss, which may be computationally prohibitive, or even infeasible using typical academic computing resources. 
Furthermore, it is not clear whether the models trained using synthetic mixtures, as proposed in \citet{Lu2020StreamingEM} and \citet{Sklyar2021MultiTurnRF}, would transfer well to real-world settings.

We will then discuss several aspects of SURT relating to the model design, network architecture, training mixture simulation, loss functions, and training schemes, which are designed to solve one or more of the above challenges. 
We will conduct ablation studies to demonstrate the impact of each of these design choices, and show that SURT is a viable framework even in the case of real meeting transcription.
This will lay the foundation towards streaming speaker-attributed transcription in the next chapter.

\section{Related Work}
\label{sec:surt_related}

The problem of multi-talker speech recognition has traditionally been addressed using a cascade of separation and transcription systems~\cite{Raj2021IntegrationOS}.
Since these methods pose the same challenges as the modular system described in the first part of this dissertation, there has been increasing interest in joint modeling.
Early work on joint separation and ASR involved hybrid HMM-DNN models as the ASR backbone~\cite{Qian2017SingleChannelMS,Chang2018MonauralMS}. 
These models were often trained with auxiliary speaker information~\cite{Chang2018AdaptivePI} or using transfer learning from single-speaker acoustic models~\cite{Tan2018KnowledgeTI}. 
With the success and flexibility of end-to-end ASR systems~\cite{Graves2006ConnectionistTC, Graves2012SequenceTW,Lu2016OnTT, Chorowski2015AttentionBasedMF, Chiu2018StateoftheArtSR, he2019streaming, Li2019RNNT, li2021recent}, researchers quickly adapted these into jointly optimized multi-talker ASR pipelines~\cite{seki-etal-2018-purely, Chang2020EndToEndMS}. 
Similar training schemes --- speaker embeddings, curriculum learning, or knowledge distillation --- were used to improve these pipelines~\cite{Denisov2019EndtoEndMS,Zhang2020ImprovingES,Lin2022SeparatetoRecognizeJM}.
%
% Multi-channel extensions of these models have also been proposed that seek to improve separation capabilities through neural beamforming techniques~\cite{Chang2019EndtoendMM,Shi2022TrainFS}. 
% %
% For the single-channel case, time-domain modeling has been used to improve separation and consequently benefit downstream ASR performance~\cite{vonNeumann2019EndtoEndTO}. 
%
Nevertheless, these early models transcribed each speaker on a different output channel, and were thus limited by the knowledge of number of speakers in the mixture, as we will see in Section~\ref{sec:surt_main}.
Furthermore, they were usually evaluated on the synthetic setting of fully-overlapping two-speaker mixtures, which does not resemble the sparsely overlapped multi-talker speech in real settings.

A solution for transcribing arbitrary numbers of overlapping speakers was proposed by \citet{Kanda2020SerializedOT}, who used existing AED architectures with ``serialized output training''.
In this strategy, all the utterances in the mixture are transcribed on the same output channel by serializing them in order of their start times.
SOT was extended to perform joint speaker counting and speaker identification using an auxiliary speaker inventory~\cite{Kanda2020JointSC,Kanda2020InvestigationOE}.
A token-level variant of SOT (t-SOT), in conjunction with neural transducers, has shown good performance on streaming multi-talker ASR~\cite{Kanda2022StreamingMA,Kanda2022StreamingSA}, and has also been combined with multi-channel front-ends~\cite{Kanda2022VarArrayMT} and large-scale pre-training~\cite{Kanda2021LargeScalePO}. 
An advantage of t-SOT is that it allows the same model and training scheme to be used for both single and multi-talker settings. 
The recently concluded M2MeT challenge used SOT as the baseline system~\cite{Yu2021M2MetTI,Yu2022SummaryOT}.
Despite its promise, SOT usually requires large scale training on synthetic mixtures, and its streaming variant requires complex interleaving of tokens across overlapping utterances.

With the observation that real multi-talker conversations rarely contain overlaps of 3 or more speakers, researchers have recently proposed the task of continuous speech separation (CSS)~\cite{Chen2020ContinuousSS}.
CSS refers to the task of generating overlap-free speech signals from a continuous audio stream consisting of multiple potentially overlapped utterances spoken by different people. 
This task has its origins in early work on \textit{unmixing transducers}~\cite{Wu2020AnEA,Yoshioka2018RecognizingOS}. 
The original model (which used PIT-based supervised training of BLSTM encoders) has been improved by leveraging better architectures such as Conformer~\cite{Chen2021ContinuousSS}, by two-stage training~\cite{Wu2021InvestigationOP}, and by large-scale semi-supervised and self-supervised learning~\cite{Wang2022LeveragingRC,Chen2022SpeechSW}.
While the original CSS used PIT-based training, other training methods such as recurrent selective attention network (RSAN)~\cite{Kinoshita2013TheRC,Zhang2021ContinuousSS} and Graph-PIT~\cite{vonNeumann2021GraphPITGP,vonNeumann2023SegmentLessCS} have also been investigated. 
%
% Multi-channel extensions of CSS have been proposed using complex spatial features~\cite{Wang2020MultimicrophoneCS}, low-latency beamforming~\cite{Yoshioka2019LowlatencySC}, and direction-of-arrival (DOA) based source localization~\cite{Wang2021LocalizationBS}. 
%
The CSS strategy mitigates the problem of transcribing arbitrary number of speakers by fixing the number of output channels.
Furthermore, unlike t-SOT, it does not require any serialization/deserialization of tokens.
However, performing multi-talker ASR using CSS as a front-end still presents the same issues as cascaded systems, mainly arising from error propagation.
In this chapter, we will describe the SURT model, which may be regarded as a jointly optimized version of CSS with transducer-based ASR.

\section{Multi-talker ASR with transducers}
\label{sec:surt_main}

In multi-talker ASR, the input $\mathbf{X} \in \mathbb{R}^{T\times F}$ is an unsegmented mixture containing $N$ utterances from $K$ speakers, i.e., $\mathbf{X} = \sum_{n=1}^N \mathbf{x}_n$, where $\mathbf{x}_n$ is the $n$-th utterance ordered by start time, shifted and zero-padded to the length of $\mathbf{X}$. 
The desired output is $\mathbf{Y} = \{\mathbf{y}_n: 1\leq n \leq N\}$, where $\mathbf{y}_n$ is the reference transcript corresponding to $\mathbf{x}_n$.
Suppose $\mathbf{y}_n = (y_1^n,\ldots,y_{U_n}^n)$, where $U_n$ is the length of the label sequence and $y$'s are output units, e.g. words or sub-word units, from some vocabulary $\mathcal{Y}$.

Recall, from Chapter~\ref{chap:modular}, that transducers compute $P(\mathbf{y}|\mathbf{X})$ by marginalizing over all possible alignments $\mathbf{a} \in \mathcal{B}^{-1}(\mathbf{y})$, where $\mathbf{a} = (a_1,\ldots,a_T)$, $y_t \in \mathcal{Y} \cup \{\phi\}$, and $\mathcal{B}$ is a mapping that removes the blank label $\phi$, i.e.,
\begin{equation}
    P(\mathbf{y}\mid \mathbf{X}) = \sum_{\mathbf{a}\in \mathcal{B}^{-1}(\mathbf{y})} P(\mathbf{a}\mid \mathbf{X}) = \sum_{\mathbf{a}\in \mathcal{B}^{-1}(\mathbf{y})} \prod_{t=1}^T P(a_t \mid \mathbf{X},\mathbf{a}_{1:t-1}).
\end{equation}

The problem we are interested in solving is known as \textit{continuous}, \textit{streaming}, multi-talker ASR.
As the name suggests, there are two constraints that our solution must satisfy.
First, the output should be continuous, i.e., no external segmentation should be required.
Second, the transcription must be streaming: (i) we are not allowed to use right context from the input, and (ii) overlapping utterances must be transcribed simultaneously.
The continuous requirement is satisfied through end-to-end modeling, i.e., we do not rely on any diarization/segmentation systems to generate the segments $\mathbf{x}_n$.
For the streaming requirement, we first constrain the encoders to only model the left context, which results in an approximation of $P(\mathbf{y}\mid \mathbf{X})$ as
\begin{equation}
    P(\mathbf{y}\mid \mathbf{X}) \approx \sum_{\mathbf{a}\in \mathcal{B}^{-1}(\mathbf{y})} \prod_{t=1}^T P(a_t \mid \mathbf{X}_{1:t-1},\mathbf{a}_{1:t-1}).
\end{equation}

More importantly, transcribing overlapping utterances at the same time necessitates a multi-branch architecture. 
Assuming conditional independence of the utterances given the input, we can approximate $P(\mathbf{Y}\mid \mathbf{X})$ as
\begin{equation}
    P(\mathbf{Y}\mid \mathbf{X}) = P(\mathbf{y}_1,\ldots,\mathbf{y}_N\mid \mathbf{X}) \approx \prod_{n=1}^N P(\mathbf{y}_n \mid \mathbf{X}).
\end{equation}

Alternatively, if we know the assignment of utterances to speakers, we can relax the above assumption slightly, and only assume that the transcript of one speaker is conditionally independent of other speakers.
Formally, this is given as
\begin{equation}
P(\mathbf{Y}\mid \mathbf{X}) = P(\Tilde{\mathbf{y}}_1,\ldots,\Tilde{\mathbf{y}}_K\mid \mathbf{X}) \approx \prod_{k=1}^K P(\Tilde{\mathbf{y}}_k \mid \mathbf{X}),
\end{equation}
where $\Tilde{\mathbf{y}}_k = \mathbf{y}_{1}^k \lozenge \ldots \lozenge \mathbf{y}_{N_k}^k$, $\lozenge$ denotes concatenation, and $\mathbf{y}_{n}^k$ are consecutive utterances of speaker $k$. 
This results in multi-output models such as MIMO-Speech~\cite{Chang2019MIMOSpeechEM}, where each output branch transcribes one speaker.
Such a formulation is limited by the requirement of knowing the number of speakers $K$ in advance.
The number of output branches also increases with $K$, resulting in an increase in model parameters.
Furthermore, if permutation-invariant training is used to resolve the problem of output permutation, this requires computing the loss for $K^2$ pairs, which is computationally expensive.
Instead, we want to constrain the number of output channels to a small number by making certain assumptions about the nature of overlapping speech in the mixture.
For this, let us formalize the problem in terms of graph coloring, as proposed by \citet{vonNeumann2021GraphPITGP}.

\subsection{Channel assignment as graph coloring}

Consider a graph $\mathcal{G} = (V,E)$ which contains $N$ nodes, each representing one of the utterances in the mixture.
For every pair of utterances, we connect the corresponding nodes with an edge if the utterances are overlapping.
An example of such a graph is shown in Fig.~\ref{fig:surt_assignment} (top).

\begin{figure}[t]
    \centering
    \includegraphics[width=0.8\linewidth]{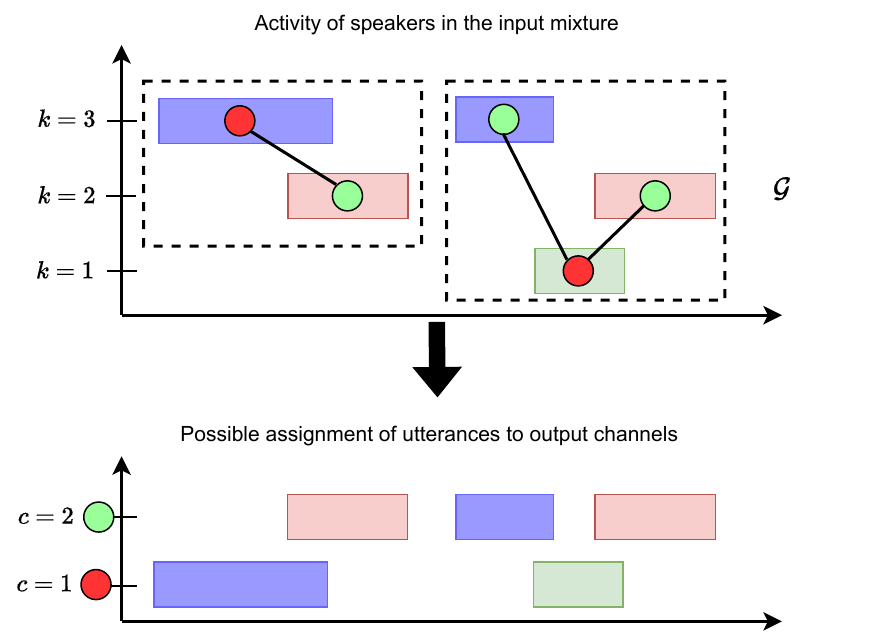}
    \caption{Assigning 5 utterances to 2 output channels. The black dashed boxes represent different utterance groups. Figure based on \citet{vonNeumann2021GraphPITGP}.}
    \label{fig:surt_assignment}
\end{figure}

Each connected component in $\mathcal{G}$ is called an utterance group.
In Fig.~\ref{fig:surt_assignment}, the black dashed boxes represent two utterance groups in the mixture.
We can also define an utterance group in non-graphical terms as follows.

\begin{definition}
An \textit{utterance group} is a set of utterances connected by speaker overlaps. If $\mathbf{U} = \{\mathbf{y}_1,\ldots,\mathbf{y}_N\}$ is an utterance group containing $N$ utterances ordered by start time, then we must have, $\forall n \in [1, N]$, $t_n^{\mathrm{st}} \leq \max_{n^{'} \in [1,n-1]} t_{n^{'}}^{\mathrm{en}}$.
\end{definition}

As mentioned previously, we want to assign utterances to a fixed number of \textit{channels}, $C$, such that no utterances within the same channel are overlapping.
Formally, we can define the channel assignment problem as follows.

\begin{definition}[Channel assignment]
Given a set of $N$ utterances represented by the graph $\mathcal{G}$, and a fixed number of channels (or colors) $C$, the objective of channel assignment is to find a $C$-vertex coloring of $\mathcal{G}$, i.e.,
\begin{align}
    &\zeta: N \rightarrow \{1,\ldots,C\},~~\text{such that} \\
    &\zeta(u) \neq \zeta(v),~~\forall u,v \in E.
\end{align}
\end{definition}

The minimum possible $C$ such that such an assignment exists is called the \textit{chromatic number} of the graph.
For our case, this is equal to the maximum number of simultaneously active speakers in the mixture.

Recall from Section~\ref{sec:intro_data} that most of the overlaps in meetings are between at most 2 speakers, and 3 or more speaker overlaps are very rare.
Therefore, for the rest of our discussion, we will set $C$ as 2, i.e., we will try to map the $N$ utterances to 2 channels.
Suppose each such $\zeta$ creates channel-wise references $\mathbf{Y}_1$ and $\mathbf{Y}_2$.
Then, for permutation-invariant training (PIT), we have
\begin{align}
P(\mathbf{y}_1,\ldots,\mathbf{y}_N \mid \mathbf{X}) &= 
    \max_{\zeta} P(\mathbf{Y}_1, \mathbf{Y}_2 \mid \mathbf{X}) \\
    &\approx \max_{\zeta} P(\mathbf{Y}_1\mid \mathbf{X}) P(\mathbf{Y}_2\mid \mathbf{X}),
\end{align}
where we have again assumed conditional independence of the outputs.
Such a channel assignment of references solves the problem of increase in model size with number of speakers or utterances.
Additionally, the conditional independence assumptions are less strong, since we only assume the channel outputs to be independent.
The corresponding PIT loss is given as
\begin{equation}
\label{eq:surt_pit}
\mathcal{L}_{\mathrm{pit}}(\mathbf{y}_{1:N},\mathbf{X};\Theta) = \min_{\zeta} \left[ -\log P_{\Theta}(\mathbf{Y}_1\mid \mathbf{X}) -\log P_{\Theta}(\mathbf{Y}_2\mid \mathbf{X}) \right].
\end{equation}

\subsection{Heuristic error assignment training (HEAT)}

\begin{figure}
    \centering
    \includegraphics[width=\linewidth]{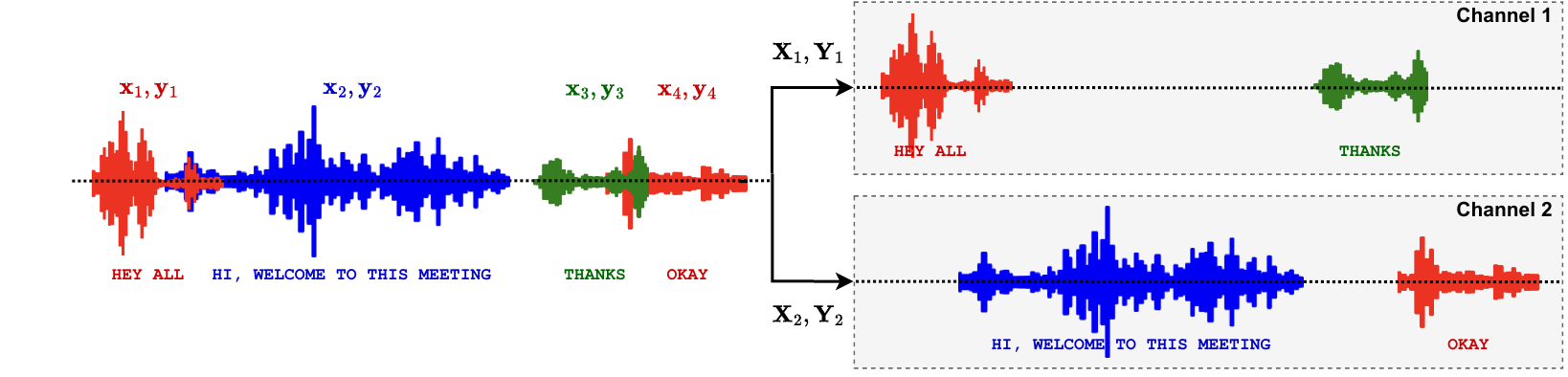}
    \caption{Example of heuristic error assignment training (HEAT). Each color represent a different speaker. The utterances ($\mathbf{x}_u$,$\mathbf{y}_u$) are assigned, in order of start times, to the next available channel. Such an assignment avoids the exponential complexity associated with permutation invariant training (PIT).}
    \label{fig:surt_heat}
\end{figure}

If there are $\Tilde{N}$ utterance groups (or connected components) in the graph, it is easy to see that we can have $2^{\Tilde{N}}$ possible mappings $\zeta$.
This can make it computationally prohibitive to evaluate \eqref{eq:surt_pit} when $\Tilde{N}$ is large.
To remedy this problem, we fix the order of the 2 channels and create channel-wise references $\mathbf{Y}_1$ and $\mathbf{Y}_2$ by assigning $\mathbf{y}_n$'s to the first available channel, in order of start time, as shown in Fig.~\ref{fig:surt_heat}. 
This technique is known as \textit{heuristic error assignment training} or HEAT, and was first introduced in \citet{Lu2020StreamingEM}.
Formally, if $t_n^{\mathrm{st}}$ is monotonically increasing in $n$, we have
\begin{equation}
\label{eq:heat_labels}
    \zeta_{\mathrm{heat}}(n) = 
    \begin{cases}
    1, &\text{if}~~t_n^{\mathrm{st}} \geq \max_{i\in \zeta^{-1}(1)}\theta_i^{\mathrm{en}} \\
    2, &\text{otherwise},
    \end{cases}
\end{equation}
and $\zeta_{\mathrm{heat}}(n)$'s are assigned sequentially.
Such an assignment allows us to avoid permutation-invariant training, reducing the computation significantly.
The corresponding HEAT loss is given as
\begin{equation}
\label{eq:surt_heat1}
\mathcal{L}_{\mathrm{heat}}(\mathbf{y}_{1:N},\mathbf{X};\Theta) =  -\log P_{\Theta}(\mathbf{Y}_1\mid \mathbf{X}) -\log P_{\Theta}(\mathbf{Y}_2\mid \mathbf{X}).
\end{equation}

\subsection{HEAT vs. PIT}
\label{sec:heat_vs_pit}

Since HEAT computes training loss for a fixed assignment of references to channels instead of minimizing over all permutations, one could argue that it may result in a suboptimal solution.
To investigate this phenomenon further, we set up simple experiments on 2-utterance mixtures generated from LibriSpeech \texttt{train-clean} set.
We prepared two kinds of mixtures with utterance delays of 2.0 and 0.0 seconds, respectively, and trained a vanilla SURT model (which will be described in detail in the next section) using both HEAT and PIT losses. 
In Fig.~\ref{fig:surt_pit}, we show the training dynamics for both the experiments and also plot the \% correct output assignment. 
This quantity represents how often the model assigns $\mathbf{Y}_1$ to output channel 1.

As expected, for the case of mixtures with 2.0s delay, HEAT quickly learned the output assignment order. 
In fact, even when training with PIT, the same heuristic was learned (albeit slower), and both models started to converge only after this point was reached (denoted by the vertical line in Fig.~\ref{fig:surt_delay2}. 
Thereafter, using PIT is wasteful, especially in our case of the expensive RNN-T loss computation. 
In the absence of utterance delay (Fig.~\ref{fig:surt_delay0}), PIT produced a random output assignment. 
Surprisingly, HEAT still learned the correct assignment, but on decoding with the trained model, we found that it learned a degenerate solution where both output channels produce the exact same hypothesis.

\begin{figure}[t]
\begin{subfigure}{0.49\linewidth}
\centering
\includegraphics[width=\linewidth]{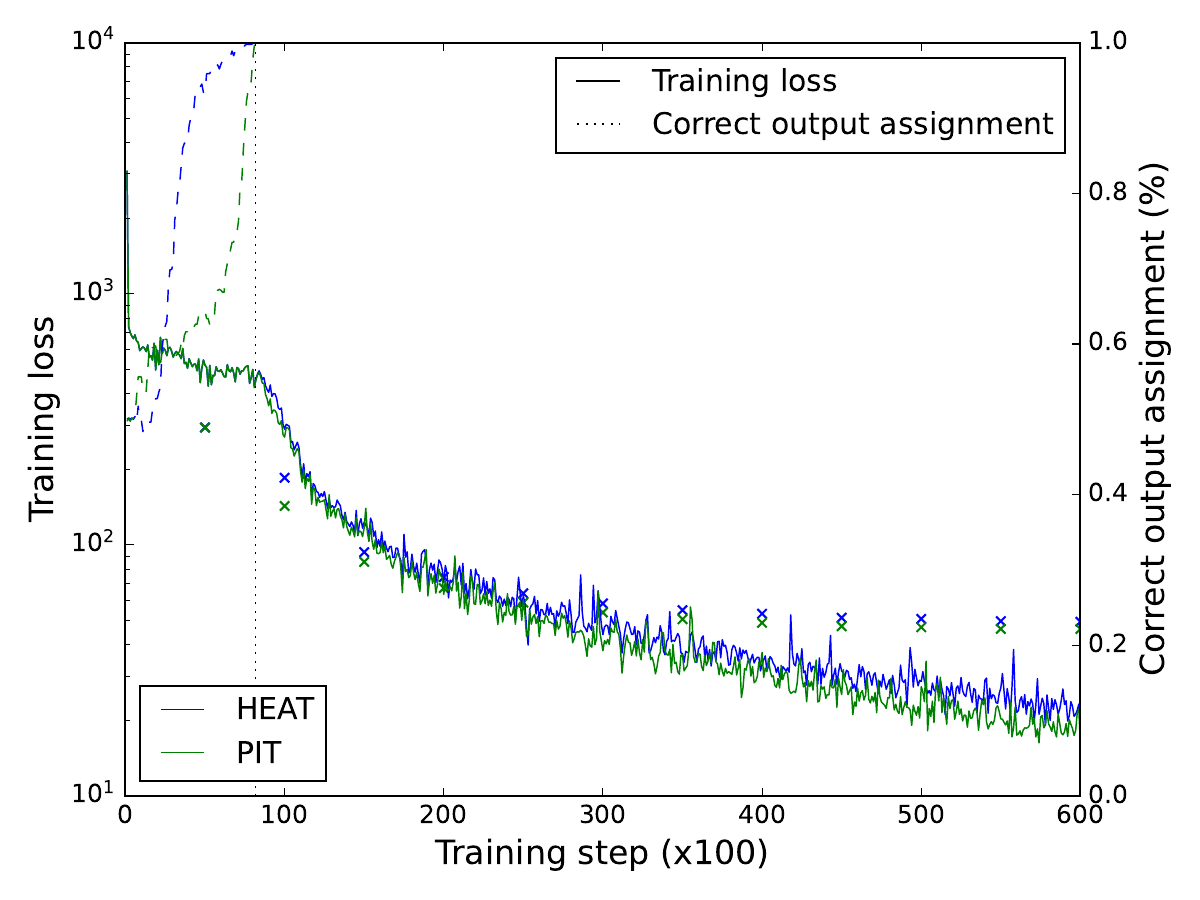}
\caption{Delay = 2.0 s}
\label{fig:surt_delay2}
\end{subfigure}
\begin{subfigure}{0.49\linewidth}
\centering
\includegraphics[width=\linewidth]{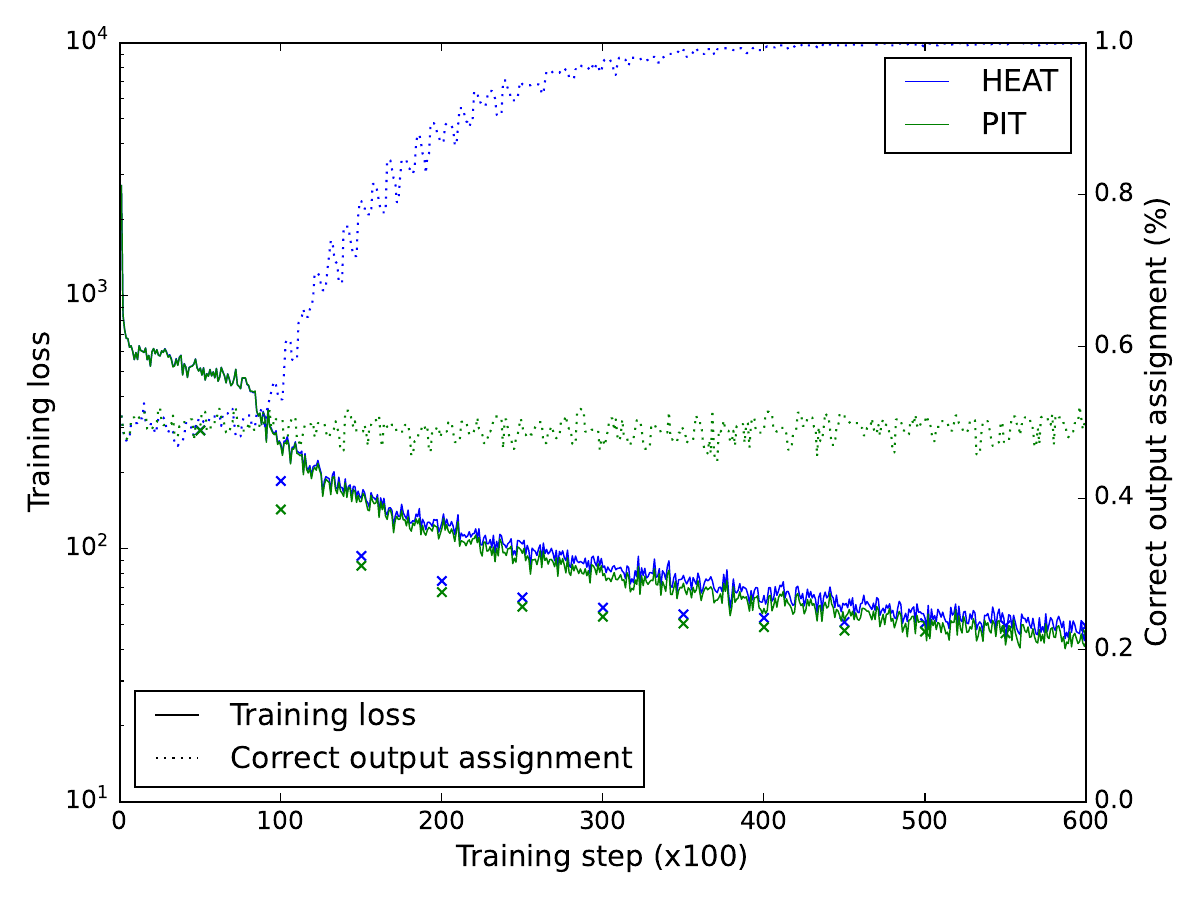}
\caption{Delay = 0.0 s}
\label{fig:surt_delay0}
\end{subfigure}\hfill
\caption{Training dynamics for HEAT versus PIT based loss for different utterance delays: (a) 2.0 s, and (b) 0.0 s.}
\label{fig:surt_pit}
\end{figure}

\section{The SURT model}
\label{sec:surt_model}

\begin{figure}[t]
\centering
\includegraphics[width=\linewidth]{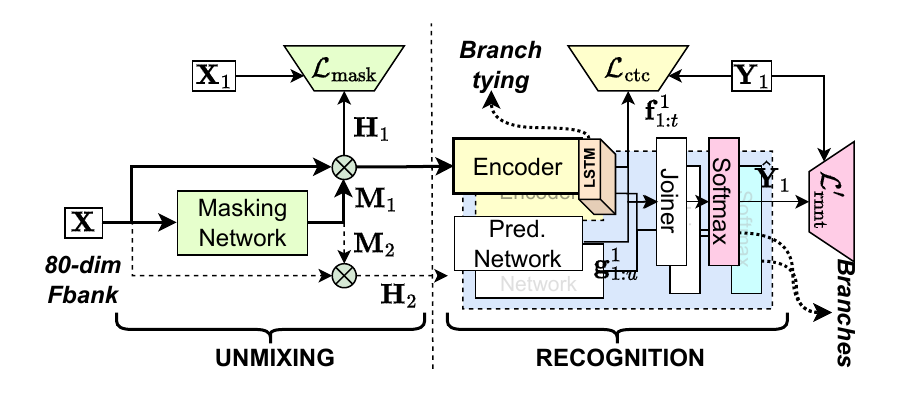}
\caption{Block diagram of SURT. $\mathbf{X}$ denotes the input mixture. $\mathbf{X}_1$ and $\mathbf{Y}_1$ are concatenated sources and references for the first branch (only required during training). The ``recognition'' component is pre-trained on single-speaker data.}
\label{fig:surt_model}
\end{figure}

SURT estimates $\hat{\mathbf{Y}} = [\hat{\mathbf{Y}}_1,\hat{\mathbf{Y}}_2] = f_{\text{surt}}(\mathbf{X})$ by combining an ``unmixing'' component and a ``recognition'' component, hence the name.
The overall block diagram of the model is shown in Fig.~\ref{fig:surt_model}.
The unmixing part first produces non-overlapping streams $\mathbf{H}_1$ and $\mathbf{H}_2$ from the input mixture $\mathbf{X}$.
The original formulation of SURT by \citet{Lu2020StreamingEM} used dual mix/mask encoders that projected input 257-dim STFTs into high dimensional representations as
\begin{align}
\label{eq:surt_unmixing}
& \mathbf{H}_1 = \mathbf{M} \ast \bar{\mathbf{X}}, \quad \mathbf{H}_2 = (\mathbbm{1} - \mathbf{M}) \ast \bar{\mathbf{X}}, \\
& \mathbf{M} = \sigma(\mathrm{MaskEnc}(\mathbf{X})) ~ \text{and} ~ \bar{\mathbf{X}} = \mathrm{MixEnc}(\mathbf{X}), \nonumber
\end{align}
where $\bar{\mathbf{X}},\mathbf{M},\mathbbm{1}\in \mathbb{R}^{T\times D}$ (for latent dim. $D$) is a matrix of ones, $\sigma$ is the sigmoid function, and $\ast$ is Hadamard product.
Clearly, this design constrains SURT to have exactly two output branches, and the separated features are not interpretable. 
Instead, we use 80-dimensional log Mel filter-banks as inputs, and replace the mix/mask encoders with a simple masking network that can generate arbitrary number of masks. 
Formally, given output channel count $C$, our unmixing module generates masks
\begin{equation}
[\mathbf{M}_1,\ldots,\mathbf{M}_C]^T = \mathrm{MaskNet}(\mathbf{X}),
\end{equation}
where $\mathbf{M}_c \in \mathbb{R}^{T\times F}$. 
These masks are applied to the input $\mathbf{X}$ to obtain channel-specific features: $\mathbf{H}_c = \mathbf{M}_c \ast \mathbf{X}$. 
Such a design has three advantages: (i) it allows the use of arbitrary number of output branches $C$, (ii) the masked representations $\mathbf{H}_c$ are interpretable as clean features, and (iii) it allows pre-training of the recognition module on single-speaker speech. 

The recognition module is similar to a conventional single-speaker ASR system based on transducers. 
Channel-wise features $\mathbf{H}_c$ are fed into an encoder which generates hidden representations $\mathbf{f}_{1:T}^c$. 
The corresponding label sequence $\mathbf{Y}_c$ (created according to the HEAT strategy) is fed into a prediction network, generating hidden representations $\mathbf{g}_{1:U}^c$. 
A joiner combines $\mathbf{f}_{1:T}^c$ and $\mathbf{g}_{1:U}^c$ to generate logits $\mathbf{Z}_c$ for each branch. 
The parameters of the encoder, prediction network, and joiner are shared among all the output branches, as shown in the figure.
Finally,
\begin{equation}
\label{eq:surt_heat2}
\mathcal{L}_{\text{heat}} = \mathcal{L}_{\text{rnnt}}(\mathbf{Y}_1, \mathbf{Z}_1) + \mathcal{L}_{\text{rnnt}}(\mathbf{Y}_2, \mathbf{Z}_2),
\end{equation}
where $\mathcal{L}_{\text{rnnt}}$ is the standard RNN-T loss~\cite{Graves2012SequenceTW}.
This is the same formulation as \eqref{eq:surt_heat1}, since $\mathcal{L}_{\text{rnnt}}(\mathbf{Y}, \mathbf{Z}) = -\log P_{\Theta}(\mathbf{Y}\mid \mathbf{X})$.

Multi-talker ASR with SURT requires the model to perform well on three challenging sub-tasks:
\begin{enumerate}[itemindent=0em]
\item continuous separation of sparsely overlapped speech; \item long-form speech recognition; and
\item modeling quick turn-taking among multiple speakers. 
\end{enumerate}

For most multi-talker ASR models, failure cases may be linked to model degeneration on one or more of these sub-tasks. 
For instance, \textit{omission} and \textit{leakage} errors may be attributed to (1), while high error rates on quick turn-taking scenarios with short silences (see Section~\ref{sec:surt_results}) may be caused by (3). 
Deletion errors on long sequences may be caused by both (1) and (2). 
In the following subsections, we will describe several SURT details related to the network architecture, training objective, mixture simulation, and training scheme, which are carefully designed to overcome the three challenges.

\subsection{Network architecture}
\label{sec:surt_arch}

% 1. Masking network uses DPRNN (instead of conv layers)
% 2. Encoder uses streaming zipformer
% 3. Prediction network is stateless
% 4. Branch tying of encoder branches

\subsubsection{Unmixing module}

\begin{figure}[t]
\centering
\includegraphics[width=0.8\linewidth]{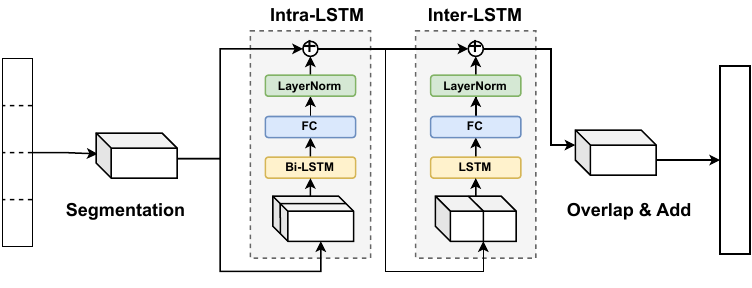}
\caption{Overview of the streaming dual path LSTM model. The intra-chunk LSTM is bidirectional, whereas the inter-chunk LSTM is unidirectional.}
\label{fig:surt_dplstm}
\end{figure}

For the masking network, we use dual-path LSTMs (DP-LSTMs) since they provide strong long range modeling capability for unmixing.
First proposed in \citet{Luo2020DualPathRE}, DP-LSTM consists of an \textit{intra} and an \textit{inter} LSTM per layer, as shown in Fig.~\ref{fig:surt_dplstm}. 
Input sequences are segmented into (overlapping) chunks and first fed into the bidirectional intra-LSTM, which processes each chunk independently. 
The output is then passed into the inter-LSTM which is unidirectional and performs strided processing over chunks. 
By choosing the chunk width to be approximately square root of the sequence length $l$, we can ensure that both the LSTMs get similar length inputs. 
Since the intra-LSTM is bidirectional, a latency equal to the chunk width is introduced in this model.

\subsubsection{Chunk width randomization}

Dual-path models trained with a fixed chunk width may not be suitable for evaluation on diverse sequence lengths due to mismatch in train-test input size for the inter block. 
We propose training with chunk width randomization (CWR), wherein we vary the CW between a minimum and maximum value for each mini-batch. 
CWR increases the train time diversity in sequence length for both the intra and inter blocks and makes the model robust to such variations at test time. 

\subsubsection{Recognition module}

As mentioned earlier, the recognition module is a conventional neural transducer, consisting of an encoder, a prediction network, and a joiner.
For the encoder, we use a recently proposed variant of the Conformer~\cite{Gulati2020ConformerCT}, known as the Zipformer~\cite{Yao2023ZipformerAF}. 
As shown in Fig.~\ref{fig:surt_zipformer}, the zipformer encoder contains multiple encoder blocks running at different rates, with the middle ones more strongly down-sampled (by up to a factor of 8). 
This makes training more efficient as there are fewer frames to evaluate. 
Each block may contain one or more encoder ``layers'' operating at the same frame rate.
Each layer performs self-attention twice with shared attention weights, and a trainable bypass is introduced for each layer dimension.

\begin{figure}[t]
    \centering
    \includegraphics[width=0.8\linewidth]{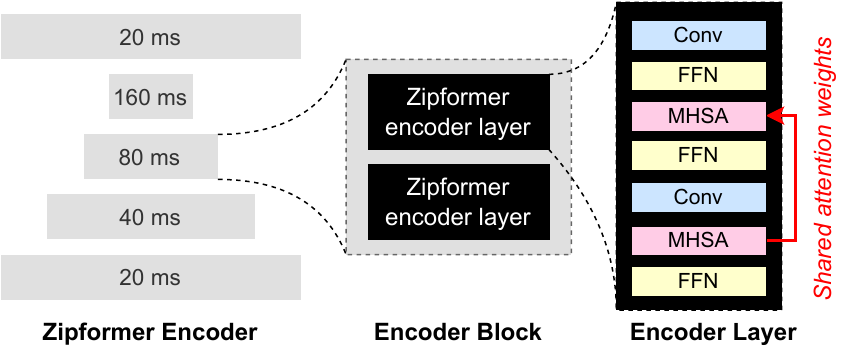}
    \caption{Illustration of the zipformer encoder architecture. The encoder contains multiple ``blocks'' running at different frame rates (left). Each block contains several ``layers'' (middle), and each layer performs self-attention twice with shared attention weights (right).}
    \label{fig:surt_zipformer}
\end{figure}

We also propose ``branch tying'' of the encoder outputs to \textit{jointly} learn representation across all branches, i.e.,
\begin{equation}
\label{eq:surt_bt}
    [\hat{\mathbf{h}}_1^{\text{enc}},\ldots,\hat{\mathbf{h}}_C^{\text{enc}}]^T = \mathrm{LSTM}\left( [\mathbf{h}_1^{\text{enc}},\ldots,\mathbf{h}_C^{\text{enc}}]^T \right).
\end{equation}
The motivation for branch-tied encoders is to reduce errors from \textit{omission} and \textit{leakage}, which usually happen when output branches do not communicate. 

Finally, we use a ``stateless'' prediction network~\cite{Ghodsi2020RnnTransducerWS} which uses a 1-D convolutional layer with a kernel of size 2, thus restricting the context to the last 2 units. 
In addition to improving computational efficiency, we conjecture that a stateless network should also be better suited to handle quick turn-taking (sub-task~(3) described earlier).
We set the segment length for the bi-directional intra-LSTM (of the DP-LSTM network) equal to the chunk size of the causal Zipformer encoder; this is the overall latency of the SURT model.
Further details about model hyperparameters are given in Section~\ref{sec:surt_implementation}.

\subsection{Training objective}
\label{sec:loss}

% 1. Pruned-RNNT instead of regular RNNT
% 2. CTC objective at encoder output
% 3. Graph-PIT at masking network output

The full-sum transducer loss, which marginalizes $P(\mathbf{y}\mid \mathbf{X})$ over all possible alignments $\mathbf{a} \in \mathcal{B}^{-1}(\mathbf{y})$, suffers from high memory usage, since the marginalization is done over a logit tensor of size ($B,T,U,D$). 
Even with an efficient implementation such as the one described in \citet{Li2019RNNT}, this severely limits the training sequence length --- for example, \citet{Raj2021ContinuousSM} trained SURT with mixtures containing at most 4 speaker turns. 
To remedy this issue, we replace the full-sum transducer loss with the recently proposed \textit{pruned} transducer loss, which prunes the alignment lattice using a simple linear joiner before computing the full sum on the pruned lattice~\cite{Kuang2022PrunedRF}. 
This idea of pruning the lattice is similar to \citet{Mahadeokar2020AlignmentRS}, but it does not require external forced alignments.

Recall from Section~\ref{sec:surt_model} that the encoder and the prediction network generate hidden representations $\mathbf{f}_1^T$ and $\mathbf{g}_1^U$, respectively.
In a pruned transducer, these are first projected to $\mathbb{R}^{|\bar{\mathcal{Y}}|}$, where $\bar{\mathcal{Y}} = \mathcal{Y}\cup \{\phi\}$. 
Let us denote these projected representations as $\hat{\mathbf{f}}_1^T$ and $\hat{\mathbf{g}}_1^U$, respectively. 
A simple additive joiner is then used to compute \textit{trivial} logits $\hat{\mathbf{z}}_{t,u}$ as
\begin{align}
   & \hat{\mathbf{z}}_{t,u} = \hat{\mathbf{f}}_{t} + \hat{\mathbf{g}}_{u} - \hat{\mathbf{z}}^{\mathrm{norm}}_{t,u}, \\
    \text{where}~~ &\hat{\mathbf{z}}^{\mathrm{norm}}_{t,u} = \log \sum_{v} \exp \left( \hat{\mathbf{f}}_{t} + \hat{\mathbf{g}}_{u} \right). \label{eq:surt_log_matrix}
\end{align}
Here, \eqref{eq:surt_log_matrix} may be interpreted as log-space matrix multiplication, and is easy to implement through simple matmul operations. 
Gradients from this simple joiner are used to compute locally optimal pruning bounds for the lattice, and the full joiner output, $\mathbf{z}_{t,u}$, is only computed on the pruned lattice.
The pruning is ``locally'' optimal in the sense that the optimization is performed per-frame, as opposed to a ``globally'' optimal treatment which would consider the whole path through the lattice.
This local strategy may result in path discontinuity that is later resolved through adjustments.
We used the open-source implementation available in \texttt{k2}: \texttt{\url{https://github.com/k2-fsa/k2}}.

We use auxiliary loss functions to regularize SURT training and improve separation of sparsely overlapped speech. 
For the former, we add a connectionist temporal classification (CTC) loss~\cite{Graves2006ConnectionistTC}, which has been shown to provide regularization capabilities due to monotonicity in alignments~\cite{Kim2016JointCB,Sudo20224DAJ}. This is given as
\begin{equation}
\label{eq:ctc}
\mathcal{L}_{\text{ctc}} = \log \sum_{\mathbf{a}\in \mathcal{B}_{\text{ctc}}^{-1}(\mathbf{y})} \prod_t P(\mathbf{a}_t \mid \mathbf{f}_1^T),
\end{equation}
where $\mathbf{a}$ is a $T$-length sequence that deterministically maps to $\mathbf{y}$ through transformation $\mathcal{B}_{\text{ctc}}$, which removes repeated tokens and $\phi$.
To improve mask estimation, we apply a masking loss\footnote{Note that including masking loss in model training assumes the availability of the clean signals $\mathbf{X}_c$ and hence applicable only to training on synthetic mixtures.} directly on the outputs $\mathbf{H}_c$ generated by the mask encoder as
\begin{equation}
\label{eq:mask}
    \mathcal{L}_{\text{mask}} = \sum_{c \in C}\mathrm{MSE}(\mathbf{H}_c, \mathbf{X}_c),
\end{equation}
where $\mathrm{MSE}$ denotes mean-squared error, and $\mathbf{X}_c$ is obtained by summing clean inputs $\mathbf{x}_u$ assigned to branch $c$ (Fig.~\ref{fig:surt_heat}). 

The overall training objective is given as
\begin{equation}
    \label{eq:loss}
    \mathcal{L} = \mathcal{L}^{\prime}_{\text{rnnt}} + \lambda_{\text{ctc}}\mathcal{L}_{\text{ctc}} + \lambda_{\text{mask}}\mathcal{L}_{\text{mask}},
\end{equation}
where $\mathcal{L}^{\prime}_{\text{rnnt}}$ denotes the pruned transducer loss and $\lambda$'s are hyperparameters.

\subsection{Mixture simulation}
\label{sec:surt_sim}

% 1. Warmup with high overlap mixtures
% 2. Pause/overlap distribution learned from target meetings
% 3. Use short segments as source utterances (cut LibriSpeech utterances at long pauses)
% 4. Comment on on-the-fly simulation

\begin{algorithm}[tp]
\DontPrintSemicolon
  
  \KwInput{$\mathcal{X}$, $\mathcal{M}$, $K$, $T$}
  \KwOutput{$\mathcal{S}$}
  
  $D_{=\text{spk}}, D_{\neq\text{spk}}, D_{\text{ovl}}, \mathcal{S}$ = $\phi$
  
  \tcp*[l]{Fit distributions to $\mathcal{M}$}
  \For{$M$ in $\mathcal{M}$}{
    \For{$i$ in range($|M|$)}{
      $t$ = $M_i.start - M_{i-1}.end$
      
      \eIf{$M_i$.spk == $M_{i-1}$.spk}{
        $D_{=\text{spk}}$ = $D_{=\text{spk}} \cup \{t\}$ 
      }{\eIf{$t > 0$}{
          $D_{\neq\text{spk}}$ = $D_{\neq\text{spk}} \cup \{t\}$ 
        }{
          $D_{\text{ovl}}$ = $D_{\text{ovl}} \cup \{-t\}$ 
        }
      }
    }
  }

  $P_{\text{ovl}}$ = $\frac{|D_{\text{ovl}}|}{|D_{\neq\text{spk}}|+|D_{\text{ovl}}|}$; $D_{\ast}$ = histogram($D_{\ast}$)

  \tcp*[l]{Generate mixtures using $\mathcal{X}$}
  $\mathcal{X}$ = $\{\mathcal{X}_1,\ldots,\mathcal{X}_S\}$ \tcp*[l]{speaker wise bucketing of $\mathcal{X}$}

  \While{any($|\mathcal{X}_s| > 0$)}{
    \tcp*[l]{Select speakers}
    $k \leftarrow \text{sample}(K)$; $\mathcal{X}_{s_1},\ldots,\mathcal{X}_{s_k} \leftarrow \text{sample}(\mathcal{X})$

    \tcp*[l]{Select utterances for each speaker}
    \For{$i$ in range($k$)}{
      $U_{s_k} \leftarrow \text{sample}(\mathcal{X}_{s_k}), \text{s.t.} \left(\sum_{u\in U_{s_k}} u.dur\right) < T$ 

      $\mathcal{X}_{s_k} \leftarrow \mathcal{X}_{s_k}\setminus U_{s_k}$
    }

    $U = \text{shuffle}(U_{s_1},\ldots,U_{s_k})$; offset = 0

    \tcp*[l]{Get offsets for each utterance}
    $\mathcal{S}_{\mathrm{cur}}\leftarrow \phi$ \tcp*[l]{initialize empty mixture}
    
    \For{$i$ in range($|U|$)}{
      \eIf{$U_i$.spk == $U_{i-1}$.spk}{
        ot = sample($D_{=\text{spk}}$)
      }{\eIf{Bernoulli($P_{\text{ovl}}$>0.5)}{
        ot = --sample($D_{\text{ovl}}$)
      }{
        ot = sample($D_{\neq\text{spk}}$)
      }
      }
      offset = offset + ot
      
      $\mathcal{S}_{\mathrm{cur}} = \mathcal{S}_{\mathrm{cur}} \cup \{U_i,\text{offset}\}$
    }
    $\mathcal{S} = \mathcal{S} \cup \mathcal{S}_{\mathrm{cur}}$
  }

\caption{Training mixture simulation}
\label{alg:mixture}
\end{algorithm}

Multi-talker ASR models are often trained on synthetic mixtures of \textit{full} utterances, which may result in prohibitively long sequences~\cite{Raj2021ContinuousSM}; for e.g., LibriSpeech \texttt{train} has an average duration of 12.4s. 
Instead, we use \textit{sub-segments} instead of full utterances as the source for mixture simulation, so that resulting mixtures are shorter while retaining multiple turns of conversation.
These sub-segments are obtained using word-level alignment information, by breaking up the utterances at pauses longer than a threshold $\tau$.
As an example, using $\tau=0.2$ for LibriSpeech resulted in sub-segments that were 2.8s on average. 
This allowed each training session to contain up to 9 turns of conversation while still being 36.5\% shorter than those generated using full utterances.

For the simulation process itself, we learn histograms of pause/overlap distribution statistics from the target sessions, and sample from these distributions for mixing the segments.
Such a strategy has been successfully applied to improve end-to-end neural diarization~\cite{Landini2022FromSM}. 
Our mixture simulation algorithm is described in Algorithm~\ref{alg:mixture}, and is similar to the conversation simulation algorithm from \citet{Landini2022FromSM}. 
We assume that the input to the algorithm is the source segments $\mathcal{X}$, target sessions $\mathcal{M}$ (to learn statistics), maximum number of speakers $K$ in each mixture, and maximum duration $T$ of a speaker in a mixture. 
The algorithm returns the training mixtures $\mathcal{S}$. 
For training, the input speech mixture $\mathbf{X}$ is obtained by digitally adding the utterances in $\mathcal{S}$ with the specified offsets.
The utterance-wise labels $\mathbf{y}_n$, along with $t_n^{\mathrm{st}}$ and $t_n^{\mathrm{en}}$ (uniquely determined by the offset and duration of utterance $n$) are used to obtain the reference transcripts for both channels using \eqref{eq:heat_labels}.
We optionally convolve $\mathcal{S}$ with real room impulse responses (RIRs) to train models for far-field reverberant conditions.
The overall simulation workflow is shown in Fig.~\ref{fig:simulation}.

\begin{figure}[t]
    \centering
    \includegraphics[width=\linewidth]{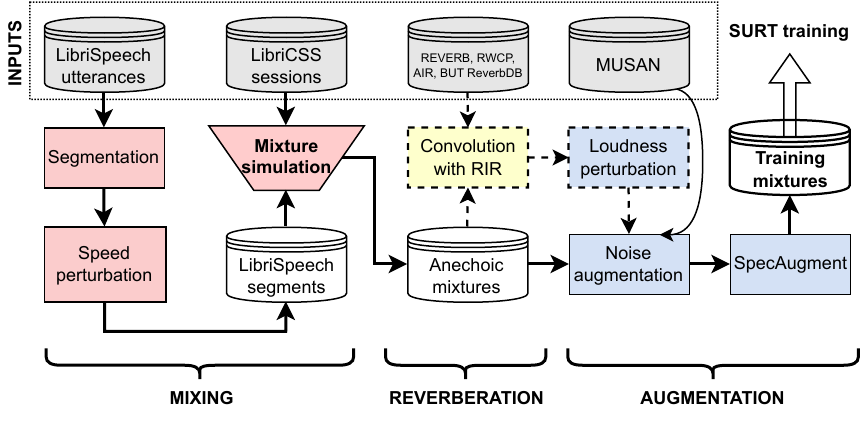}
    \caption{Mixture simulation workflow for LibriSpeech-based training. Gray cylinders denote external data used during simulation: LibriSpeech utterances, statistics from LibriCSS, real RIRs, and MUSAN noises. The simulation workflow can be conceptually divided into three phases: (i) \textit{mixing}, which creates anechoic meetings, (ii) \textit{reverberation}, which convolves with RIRs, and (iii) \textit{augmentation}, which perturbs the reverberant mixtures with various schemes. The dotted path is optional for anechoic training.}
    \label{fig:simulation}
\end{figure}

\subsection{Pre-training \& adaptation}
\label{sec:pretrain}

% 1. Initialize zipformer transducer with model trained on LibriSpeech

Using sub-segments for mixture simulation improves training efficiency while allowing multiple speaker turns; however, it creates a train-test mismatch for duration of individual segments, which could degrade the model's performance on sub-task (2), i.e., to recognize long utterances from the same speaker. 
We solve this problem by pre-training the transducer module on single-speaker utterances (e.g., on LibriSpeech \texttt{train} set).
Such a pre-training strategy also decouples the tasks of learning to separate from learning to transcribe, and helps the SURT model converge faster. 
Recall from Section~\ref{sec:surt_model} that this pre-training is possible in SURT (but not in earlier work such as \citet{Raj2021ContinuousSM}) because our masking network generates masked filter-banks instead of high-dimensional latent representations.

Despite convolving with real RIRs, the acoustic characteristics of the mixtures used to train SURT may still be mismatched from real meeting recordings. 
As a final step of SURT training, we perform model adaptation by training on in-domain data for a small number of iterations.
In this step, we adjust the training objective in \eqref{eq:loss} to omit the $\mathcal{L}_{\text{mask}}$ term.

\section{Experimental Setup}
\label{sec:setup}

\subsection{Evaluation}
\label{sec:surt_evaluation}

We performed evaluations on three publicly-available meeting datasets in English: LibriCSS, AMI, and ICSI. 
The summary statistics for these datasets are given in Section~\ref{sec:intro_data}.

From Fig.~\ref{fig:surt_heat}, it is clear that SURT generates speaker-agnostic transcription with no strict correspondence between speaker and channel. 
This makes it impossible to evaluate SURT with the conventional word error rate (WER) metric used for single-speaker ASR systems.
Instead, we use the optimal reference combination WER (ORC-WER) metric proposed independently in \citet{Sklyar2021MultiTurnRF} and \citet{Raj2021ContinuousSM}, and described in Section~\ref{sec:intro_metrics}. 
ORC-WER computes the minimum word error rate (WER) based on an optimal assignment of references to the channels, and can be considered a lower bound on cpWER (which we used in the previous chapter).
We used the polynomial-time implementation of ORC-WER using multi-dimensional Levenshtein distance, available in the \texttt{meeteval}\footnote{\url{https://github.com/fgnt/meeteval}} toolkit~\cite{vonNeumann2022OnWE}.
In the remainder of this paper, we will use the abbreviation WER to actually mean ORC-WER, unless explicitly mentioned.
For LibriCSS, we report WERs on the officially provided (approximately 1 min. long) ``segments,'' following prior work on continuous input evaluation~\cite{Chen2020ContinuousSS,Chen2021ContinuousSS,Raj2021ContinuousSM}. 
For AMI and ICSI, we use utterance-group based evaluation similar to \cite{Kanda2021LargeScalePO,Huang2022AdaptingSM}.
This results in \texttt{dev}/\texttt{test} segments of duration 7.0s/8.4s and 2.6s/2.9s for AMI and ICSI, respectively.

\subsection{Implementation details}
\label{sec:surt_implementation}

\noindent
\textbf{Network architecture.}
We experimented with two variants of SURT --- \textit{base} and \textit{large}. The \textit{base} model contains four 256-dim DP-LSTM layers as the masking network trained with chunk width randomization~\cite{Raj2021ContinuousSM}. 
The encoder consists of 5 zipformer blocks with 2 self-attention layers per block.
% \footnote{Recall from Section~\ref{sec:surt_arch} that the attention weights are shared across all the layers of a zipformer block.}
%
Each block consists of a 192-dim attention distributed across 8 heads, and a 768-dim feed-forward layer. 
Downsampling factors of (1,2,4,8,2) were used in the zipformer blocks. 
The prediction network contains a single 512-dim Conv1D layer. 
The \textit{large} model contains 6 layers in the masking network, and (2,4,3,2,4) self-attention layers in the 5 zipformer blocks.
The chunk size for the intra-LSTM and the Zipformer is set to 32 frames, resulting in a modeling latency of 320 ms.

\noindent
\textbf{Training data.}
For LibriCSS experiments, we first created anechoic mixtures, \textbf{LSMix-clean}, using each speed-perturbed LibriSpeech train sub-segment ($\tau=0.2$; cf. \S~\ref{sec:surt_sim}) once, resulting in approx. 2200h of training data.
For this data set, we set $\mathcal{M}$ as the LibriCSS \texttt{dev} set (excluding the 0L and OV10 sessions), $K=3$, and $T=15$ in Algorithm~\ref{alg:mixture}.
A reverberated copy of LSMix, named \textbf{LSMix-reverb}, was generated by convolving LSMix with real RIRs collected from the REVERB~\cite{Kinoshita2013TheRC} dataset.
We also added isotropic noises from the REVERB data, and perturbed the loudness between -20 dB and -25 dB using the \texttt{pyloudnorm} tool~\cite{steinmetz2021pyloudnorm}.
We will hereafter refer to the combination of LSMix-clean and LSMix-reverb as \textbf{LSMix-full}.
Ablation experiments were conducted on the \textit{anechoic} LibriCSS by training SURT on LSMix-clean. 
For these experiments, we used on-the-fly noise augmentation using noises from the MUSAN corpus~\cite{Snyder2015MUSANAM}.
For final evaluation, we trained SURT on LSMix-full ($\sim$4400h), so that the same model can be used on both anechoic and replayed LibriCSS.
All models were trained using on-the-fly SpecAugment~\cite{Park2019SpecAugmentAS}.
We found it beneficial to use high overlap mixtures in the warm-up stage of SURT training to encourage better mask estimation~\cite{Boeddeker2023TSSEPJD}. 
%
% For this, we synthesized mixtures based on statistics from the OV40 sessions of the LibriCSS dataset.
%
For pre-training on single-speaker data, we used LibriSpeech \texttt{train} set, optionally convolved with synthetic RIRs (for the final evaluation).
This pre-training was done for 10 epochs.

Since AMI and ICSI have similar characteristics, we trained a combined SURT model for them using synthetic mixtures created from close-talk utterances, again using sub-segments obtained from forced alignments ($\tau=0.5$).
These mixtures were obtained by setting $D_{=\text{spk}}$, $D_{\neq\text{spk}}$, $D_{\text{ovl}}$, and $P_{\text{ovl}}$ as 0.5, 0.5, 1.0, and 0.8, respectively, in Algorithm~\ref{alg:mixture}.
$K$ and $T$ were set to 3 and 15s, respectively, similar to LSMix-clean simulation.
We refer to these mixtures as \textbf{AIMix-clean}, their reverberant copy as \textbf{AIMix-reverb}, and the combination as \textbf{AIMix-full}.
The models were subsequently adapted by combining the real train sessions from all microphone settings.
%
% For adaptation, we segmented the sessions into utterance groups by cutting at maximum pause durations of 0.1s and 0.5s (i.e., two copies with different segmentation).

\noindent
\textbf{Hyper-parameters.}
The auxiliary loss scales, $\lambda_{\text{ctc}}$ and $\lambda_{\text{mask}}$, were set to 0.2 each.
$\mathcal{L}_{\text{mask}}$ was not used during adaptation since ground-truth separated audio is not available for real data. 
We trained the models with the ScaledAdam optimizer following the standard zipformer-transducer recipes in \texttt{icefall}~\cite{zipformer}.
This is a variant of Adam where each parameter's update is scaled proportional to the norm of that parameter.
The learning rate was warmed up to 0.004 for 5000 iterations, and decayed exponentially thereafter.
All models were trained for 30 epochs using 4 GPUs.
We used either Titan RTX (with batch size 500s) or V100 (with batch size 650s) depending on availability on our compute cluster.

\noindent
\textbf{Decoding.}
Checkpoints from the last 9 epochs were averaged for decoding.
We conducted experiments with both greedy decoding and beam search. 
For the ablation experiments, we used greedy decoding for faster turn-around. 
For the final evaluation, we used a ``modified'' version of beam search with beam size of 4.
This variant constraints the emission of at most 1 non-blank token at each time step, which allows batched decoding. 
We normalized the replayed LibriCSS recordings to -23 dB using \texttt{pyloudnorm} for inference.

\section{Results \& Discussion}
\label{sec:surt_results}

\subsection{Results on LibriCSS}
\label{sec:surt_res_baseline}

First, we evaluated the base and large SURT models on the LibriCSS data for both anechoic and replayed conditions, and compared it against the MT-RNNT baseline~\cite{Sklyar2021MultiTurnRF}.
These comparisons are shown in Table~\ref{tab:surt_results}.
We cannot compare SURT with published results on t-SOT~\cite{Kanda2022StreamingMA} since the latter is evaluated using the \texttt{asclite}-based speaker-agnostic WER (SAg-WER) metric.
SAg-WER requires token-level time-stamp estimation.
More importantly, it does not penalize an utterance being split into multiple ``channels,'' which is penalized by ORC-WER, making them difficult to compare.
The models were trained on LSMix-full, and adapted to LibriCSS using the \texttt{dev} set. 
For adaptation, we segmented the \texttt{dev} sessions in 2 ways, by cutting at maximum pause durations of 0.1s and 0.5s, respectively.
This resulted in 381 sub-sessions of average duration 17.9s, totalling approximately 1.88h each from anechoic and replayed conditions and containing 18.6\% overlapped speech.
The model was then trained on the combined adaptation data for 8 epochs (for \textit{base}) or 15 epochs (for \textit{large}).
The learning rate was warmed up to 0.0004 for 2 (or 4) epochs and decayed thereafter.
A single V100 GPU was used for adaptation.

\begin{table}[tp]
\centering
\caption{Performance of SURT on (a) ``anechoic'' and (b) ``replayed'' versions of LibriCSS \texttt{test} set. The SURT 2.0 models were decoded using beam search with a beam size of 4.}
\label{tab:surt_results}

\begin{subtable}[h]{\linewidth}
\centering
\caption{Anechoic}
\adjustbox{max width=\linewidth}{
\begin{tabular}{@{}lcccccccr@{}}
\toprule
\textbf{Model} & \textbf{Size (M)} & \textbf{0L} & \textbf{0S} & \textbf{OV10} & \textbf{OV20} & \textbf{OV30} & \textbf{OV40} & \textbf{Avg.} \\ \midrule

% Original SURT~\cite{Raj2021ContinuousSM} & 42.9 & 6.9 & 18.9 & 19.6 & 21.9 & 23.9 & 28.7 & 20.0 \\
% Multi-turn RNN-T~\cite{Sklyar2021MultiTurnRF} & 81.0 & -- & -- & -- & -- & -- & -- & -- \\
% \midrule

SURT (Base) & 26.7 & 5.2 & 4.7 & 14.2 & 17.8 & 21.3 & 23.0 & 14.4\\
$\hookrightarrow$ w/ \texttt{dev} adaptation & 26.7 & 5.1 & 4.2 & 13.7 & 18.7 & 20.5 & \textbf{20.6} & 13.8 \\
% \midrule

SURT (Large) & 37.9 & \textbf{4.6} & \textbf{3.8} & 14.9 & 17.3 & 19.1 & 23.9 & 13.9 \\
$\hookrightarrow$ w/ \texttt{dev} adaptation & 37.9 & \textbf{4.6} & \textbf{3.8} & \textbf{12.7} & \textbf{14.3} & \textbf{16.7} & 21.2 & \textbf{12.2} \\
\bottomrule
\end{tabular}
}
\end{subtable}

\hfill

\begin{subtable}[h]{\linewidth}
\caption{Replayed}
\adjustbox{max width=\linewidth}{
\begin{tabular}{@{}lcccccccr@{}}
\toprule
\textbf{Model} & \textbf{Size (M)} & \textbf{0L} & \textbf{0S} & \textbf{OV10} & \textbf{OV20} & \textbf{OV30} & \textbf{OV40} & \textbf{Avg.} \\ \midrule

% Original SURT~\cite{Raj2021ContinuousSM} & 42.9 & 9.3 & 21.1 & 21.2 & 25.9 & 28.2 & 31.7 & 22.9\\
Multi-turn RNN-T~\cite{Sklyar2021MultiTurnRF} & 81.0 & 14.8 & 14.5 & 18.0 & 25.8 & 30.3 & 32.3 & 22.6 \\
% \midrule

SURT (Base) & 26.7 & 6.7 & 8.2 & 23.0 & 27.1 & 28.5 & 31.8 & 20.9 \\
$\hookrightarrow$ w/ \texttt{dev} adaptation & 26.7 & 6.8 & 7.2 & 21.4 & 24.5 & 28.6 & 31.2 & 20.0 \\
% \midrule

SURT (Large) & 37.9 & \textbf{5.9} & 7.8 & 21.2 & 25.7 & 27.8 & 29.9 & 19.7 \\
$\hookrightarrow$ w/ \texttt{dev} adaptation & 37.9 & 6.4 & \textbf{6.9} & \textbf{17.9} & \textbf{19.7} & \textbf{25.2} & \textbf{25.5} & \textbf{16.9} \\
\bottomrule
\end{tabular}
}
\end{subtable}

\end{table}

As the overlap ratio in Table~\ref{tab:surt_results} increases from 0\% to 40\%, the WER also increases, which is expected. 
On the anechoic setting, SURT-base obtains a WER of 14.4\% without any adaptation.
This may be because the anechoic training mixtures are well matched to the evaluation condition in the absence of far-field artifacts.
Unlike other multi-talker ASR models which degrade performance on single-speaker input, SURT obtained very low WERs on the 0L and 0S settings. 
We attribute this primarily to pre-training on single-speaker data, which allows the model to handle non-overlapping speech well.
On using model adaptation, the anechoic WER further improved by 0.6\% absolute, with consistent improvements across most overlap conditions.
The largest improvement was obtained for the OV40 sessions, where WER reduced from 23.0\% to 20.6\%.
When evaluated on replayed LibriCSS, SURT-base was better than MT-RNNT on average, but slightly worse on overlapped conditions like OV10 and OV20.
This may be because we used a limited set of real RIRs for simulating reverberant training mixtures, whereas SURT and MT-RNNT used on-the-fly simulated RIRs.
We experimented with using simulated RIRs, but we found that it consistently degraded WERs on the 0S condition, similar to the observation in \citet{Raj2021ContinuousSM}.
Adaptation on the \texttt{dev} set improved performance across all settings, with the resulting average WER reducing to 20.0\%.

The \textit{large} model followed similar trends as the \textit{base} model, but provided consistent improvements in WER across most conditions.
For unadapted models, larger improvement was observed on the replayed setting compared to the anechoic setting (5.7\% vs. 3.5\% relative).
We conjecture that the larger masking network (6 DP-LSTM layers) may be better suited for unmixing reverberant features, leading to improved WERs.
We also found that the \textit{large} model benefited more from adaptation on in-domain data, perhaps due to higher representation capacity.
The relative WER improvement from adaptation was 12.2\% and 14.2\% for the anechoic and replayed conditions, respectively, whereas for the \textit{base} model, the improvements were 4.2\% and 4.3\%.
Overall, our SURT-large model provided relative WER improvements of 25.2\% (22.6\%~$\rightarrow$~16.9\%) over an MT-RNNT baseline.

\subsection{Effect of network architecture}
\label{sec:surt_results_arch}

Recall from Section~\ref{sec:surt_arch} that the SURT network architecture contains several components which were carefully selected for various reasons. 
These include: (i) DP-LSTMs in the unmixing module, (ii) branch-tied encoders, and (iii) a stateless prediction network. 
We performed ablation experiments to evaluate the effect of each choice, as shown in Table~\ref{tab:surt_arch}. 
Each of the first three rows denote the performance when one of the components is changed (shown in red), while the last row shows the configuration of the final SURT model.
We selected model configurations such that all models have roughly the same number of parameters.
All models were trained on LSMix-clean until convergence,
% \footnote{We had to reduce the learning rate to 0.001 when using Conv2D for the masking network.}
and evaluated on the anechoic LibriCSS setting with greedy decoding. 
The recognition modules were pre-trained for all setups, but no auxiliary losses or adaptation were used.
We also tried replacing the zipformer encoder with a DP-LSTM, but this model did not converge.
This may be because the input sequences (utterance groups) for the encoder in SURT are relatively long, which may affect convergence in architectures that do not use self-attention.

\begin{table}[t]
\centering
\caption{Effect of architectural choices for various components, shown on ``anechoic'' LibriCSS \texttt{test} set. The last row denotes the final network architectures for SURT.}
\label{tab:surt_arch}
\adjustbox{max width=\linewidth}{
\begin{tabular}{@{}l@{\hspace{0.1\tabcolsep}}c@{\hspace{0.1\tabcolsep}}ccrrrrrrr@{}}
\toprule
\begin{tabular}{@{}l}\textbf{Masking}\\\textbf{network}\end{tabular} & \begin{tabular}{@{}c}\textbf{Branch}\\\textbf{tying}\end{tabular} &  \begin{tabular}{@{}c}\textbf{Pred.}\\\textbf{network}\end{tabular} & \begin{tabular}{@{}c}\textbf{Size}\\\textbf{(M)}\end{tabular} & \textbf{0L} & \textbf{0S} & \textbf{OV10} & \textbf{OV20} & \textbf{OV30} & \textbf{OV40} & \textbf{Avg.} \\
 \midrule
\textcolor{red}{Conv2D} & \cmark & Conv1D & 25.3 & 13.9 & 6.9 & 27.0 & 35.4 & 41.0 & 45.4 & 28.3 \\
DP-LSTM & \textcolor{red}{\xmark} & Conv1D & 24.6 & 6.6 & 5.4 & 21.3 & 26.6 & 33.1 & 41.2 & 22.4 \\
DP-LSTM & \cmark & \textcolor{red}{LSTM} & 28.1 & 7.6 & 6.3 & \textbf{17.2} & 26.7 & 26.8 & 34.7 & 19.9 \\
\midrule
DP-LSTM & \cmark & Conv1D & 26.7 & \textbf{6.4} & \textbf{5.1} & 17.5 & \textbf{23.5} & \textbf{25.4} & \textbf{33.3} & \textbf{18.5} \\
\bottomrule
\end{tabular}
}
\end{table}

The largest performance degradation was caused by replacing the DP-LSTM based masking network with Conv2D.
This is consistent with speech separation research where dual-path encoders usually provide large improvements due to their ability to model long sequences~\cite{Luo2020DualPathRE}.
Without the DP-LSTM masking, the unmixing module was effectively futile, as evident by the high WERs on the overlapping sessions.

For the encoder architecture, applying branch tying using equation~(\ref{eq:surt_bt}) provided significant improvements on high overlap conditions.
For example, the relative WER reduction with branch tying was 23.3\% and 19.2\% on the OV30 and OV40 settings, respectively.
The improvement was largely due to a reduction in deletion errors (at the cost of a small increase in insertions).
For example, on the OV40 setting, the deletion error rate reduced from 20.6\% to 14.8\%, with the insertion increasing from 3.5\% to 5.5\%.
This supports our conjecture that branch tying helps in alleviating errors caused by \textit{omissions} where parts of the transcripts ``fall through the cracks'', i.e., are skipped by both branches.
We will analyze these errors further in Section~\ref{sec:surt_results_other}.

Finally, we observed small but consistent improvements by replacing the LSTM-based prediction network (used in the original SURT) with a stateless network (i.e., using a Conv1D layer).
The largest relative improvement for this change was seen in the 0S setting, where WER improved by 19.0\% (compared to 7.0\% relative WER improvement overall).
As mentioned in Section~\ref{sec:surt_arch}, we conjecture that this may be because a stateless decoder is more suited to frequent context switching that is required for modeling quick turn-taking.

\subsection{Effect of auxiliary objectives}
\label{sec:surt_results_aux}

\begin{table}[t]
\centering
\caption{Effect of auxiliary objectives on ``anechoic'' LibriCSS \texttt{test} set. All models used the SURT-base architecture.}
\label{tab:surt_auxiliary}
\adjustbox{max width=\textwidth}{
\begin{tabular}{@{}ccccccccr@{}}
\toprule
$\mathcal{L}_{\mathrm{ctc}}$ & $\mathcal{L}_{\mathrm{mask}}$ & \textbf{0L} & \textbf{0S} & \textbf{OV10} & \textbf{OV20} & \textbf{OV30} & \textbf{OV40} & \textbf{Avg.} \\ \midrule
\xmark & \xmark & 6.4 & 5.1 & 17.5 & 23.5 & 25.4 & 33.3 & 18.5 \\
\cmark & \xmark & 6.0 & 5.2 & 17.9 & 23.7 & 22.8 & 29.6 & 17.5 \\
\xmark & \cmark & \textbf{5.6} & \textbf{4.9} & 16.3 & 21.3 & 24.6 & 29.5 & 17.1 \\
\cmark & \cmark & 6.1 & 5.0 & \textbf{13.6} & \textbf{19.0} & \textbf{21.1} & \textbf{26.5} & \textbf{15.2} \\
\bottomrule
\end{tabular}
}
\end{table}

We also performed ablation experiments to study the effect of $\mathcal{L}_{\text{ctc}}$ and $\mathcal{L}_{\text{mask}}$. 
Similar to Section~\ref{sec:surt_results_arch}, we trained SURT-base models on LSMix-clean for this investigation, and evaluated the models on anechoic LibriCSS using greedy search. 
The results are shown in Table~\ref{tab:surt_auxiliary}.

First, we see that adding CTC loss on the encoder improved WER mainly for highly overlapped sessions. 
We obtained relative improvements of 10.2\% and 11.1\% on the OV30 and OV40 sessions, respectively. 
For sessions with more overlaps and turn-taking, both output branches may contain several speech segments. 
We conjecture that an auxiliary CTC objective may be useful in aligning the segments to the corresponding audio during training, resulting in better modeling for high overlap sessions. 

\begin{figure}[tbp]
\centering
  \subfloat[Without $\mathcal{L}_{\text{ctc}}$\label{fig:no_ctc_grad}]{%
       \includegraphics[width=0.8\linewidth]{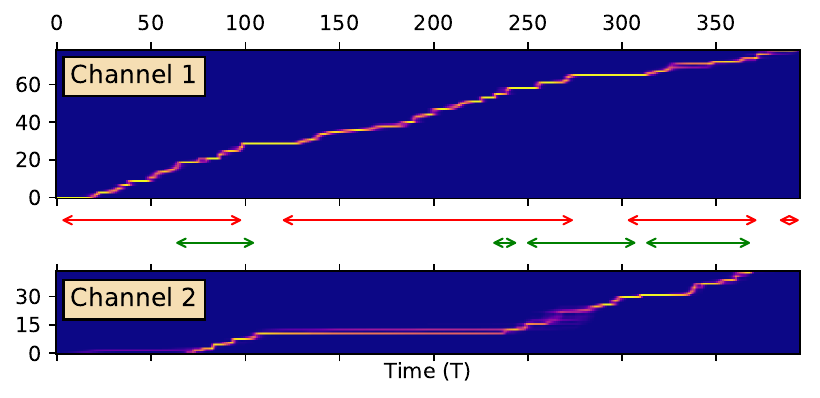}}
    \hfill
  \subfloat[With $\mathcal{L}_{\text{ctc}}$\label{fig:ctc_grad}]{%
        \includegraphics[width=0.8\linewidth]{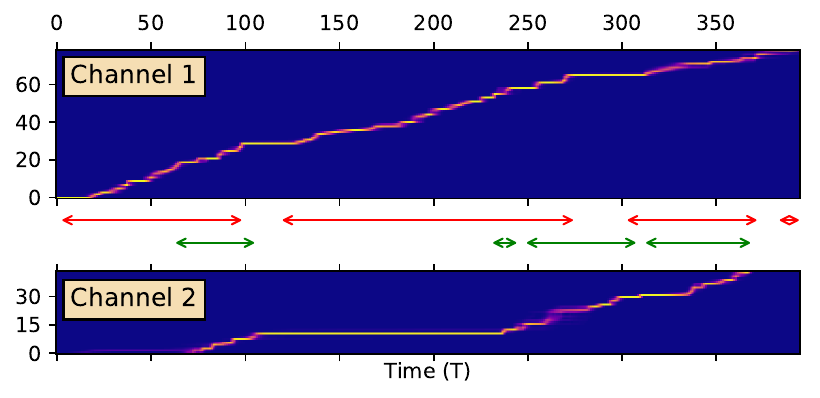}}
\caption{Occupation probabilities of nodes in the $T\times U$ lattice for both output channels, for models trained with and without auxiliary CTC loss ($\mathcal{L}_{\text{ctc}}$. $T$ and $U$ are on $x$ and $y$ axes, respectively. For this mixture, $T=395$, $U_1 = 79$, and $U_2 = 44$. The double arrows between the plots denote the reference segments: red for channel 1 and green for channel 2. Brighter colors denote higher occupation probabilities.}
\label{fig:grad}
\end{figure}

To validate our conjecture, we plotted the the occupation probabilities for the nodes in the RNN-T lattice (of shape $T \times U$), as obtained from the gradients of the simple additive joiner used in the pruned transducer loss~\cite{Kuang2022PrunedRF}. 
These values should correspond to a soft alignment between the input and the label sequence\footnote{It makes more sense to use this value instead of CTC alignments because (i) the model trained without $\mathcal{L}_{\text{ctc}}$ cannot provide corresponding alignments, and (ii) we use the transducer head for the ASR task.}.
In Fig.~\ref{fig:grad}, we show example plots for a randomly selected mixture from the training set, using models trained with and without $\mathcal{L}_{\text{ctc}}$.
When the auxiliary CTC loss was used, the model was able to better align the silence region (time frames 100 to 250) in channel 2, as discernible through the bright horizontal line.

Using auxiliary masking loss $\mathcal{L}_{\text{mask}}$, as defined in equation~(\ref{eq:mask}), again improved WER performance over the SURT-base model, as shown in the third row of Table~\ref{tab:surt_auxiliary}.
Surprisingly, we observed most improvements in the low-overlap conditions --- for instance, 12.5\% relative WER reduction for 0L.
Most of this improvement again resulted from reduction in deletion errors (2.0\%~$\rightarrow$~1.1\% for 0L).
We conjecture that this may be a result of fewer leakage-related errors in low overlap regions.
We also experimented with using a graph-PIT based masking loss~\cite{vonNeumann2021GraphPITGP} instead of the HEAT-based loss, but it did not provide similar improvements.
This may be because our final transducer objective uses the HEAT formulation.
Finally, the best WER results were obtained on combining both the auxiliary objectives.
The resulting model demonstrated a relative WER reduction of 17.8\% over the SURT-base model trained without these objectives.

\subsection{Effect of pre-training}
\label{sec:surt_results_pretrain}

\begin{figure}[tp]
\centering
\includegraphics[width=0.6\linewidth]{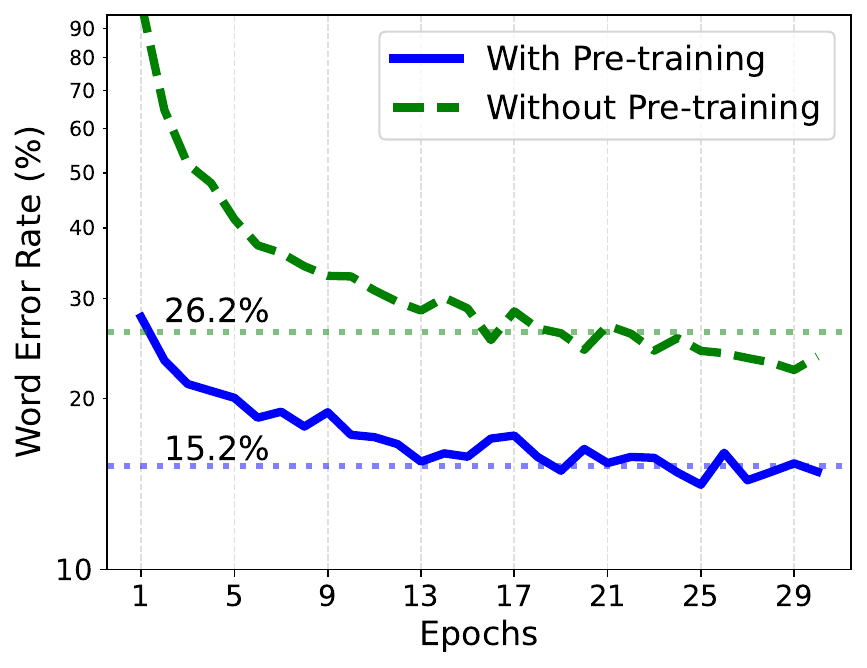}
\caption{Effect of transducer pre-training with LibriSpeech. WERs are shown for the LibriCSS \texttt{dev} set after each training epoch. The dotted horizontal lines show the final WERs on \texttt{test} set with model averaging.}
\label{fig:pretrain}
\end{figure}

Pre-training on single-speaker utterances was found to be one of the most effective strategies for faster and better convergence of SURT models.
Such a pre-training strategy has also been used in recent work on multi-channel serialized output training~\cite{Kanda2022VarArrayMT}.
To quantify this improvement, we computed the WER on the anechoic LibriCSS \texttt{dev} set (averaged across all overlap conditions) after each epoch of training a SURT-base model on the LSMix-clean data, as shown in Fig.~\ref{fig:pretrain}.
We observe that when pre-training was used, the SURT model converged much faster and to a better WER.
In fact, it surpassed the model without pre-training after just 5 epochs.
On the anechoic \texttt{test} set, the models obtained WER of 26.2\% and 15.2\%, respectively, as shown by the dotted horizontal lines.
The increase in computational time is marginal, since we only pre-train the transducer for 10 epochs (instead of training to convergence).
Alternatively, an off-the-shelf streaming transducer may also be plugged in for this purpose, since the \textit{branch tying} of the encoder is only added at the time of SURT training.
Such a pre-training scheme is analogous to a curriculum learning strategy using single-speaker utterances, with the masks fixed as $\mathbf{M}_1 = J_{T,F}$ and $\mathbf{M}_c = 0 \cdot J_{T,F}$, $\forall c \neq 1$, where $J_{t,f}=1$, $\forall t,f$.

\subsection{Measuring leakage and omission}
\label{sec:surt_results_other}

\begin{figure}[tp]
\centering
  \subfloat[\textbf{Anechoic}\label{fig:anechoic_beam}]{%
       \includegraphics[width=0.8\linewidth]{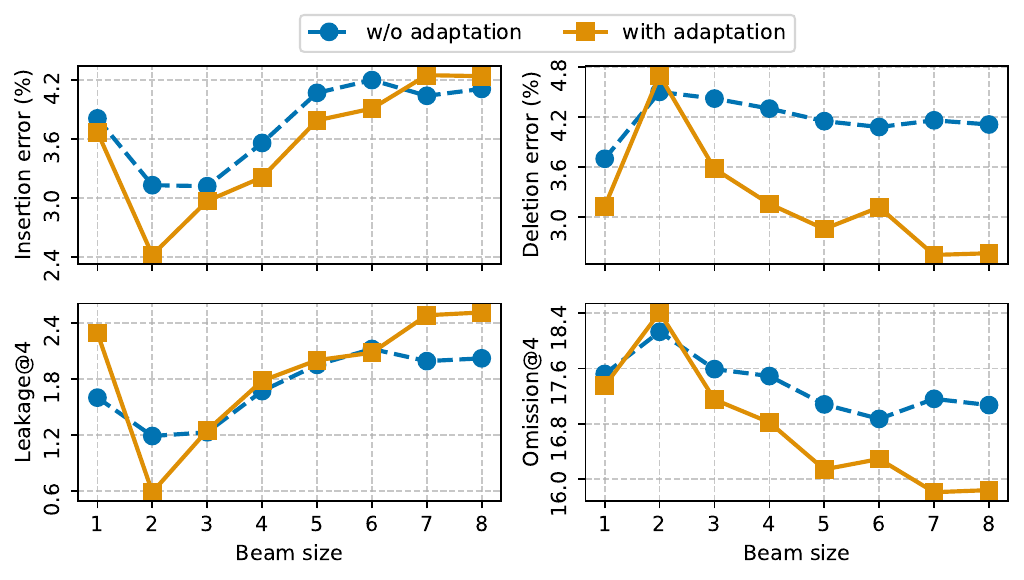}}
    \hfill
  \subfloat[\textbf{Replayed}\label{fig:replayed_beam}]{%
        \includegraphics[width=0.8\linewidth]{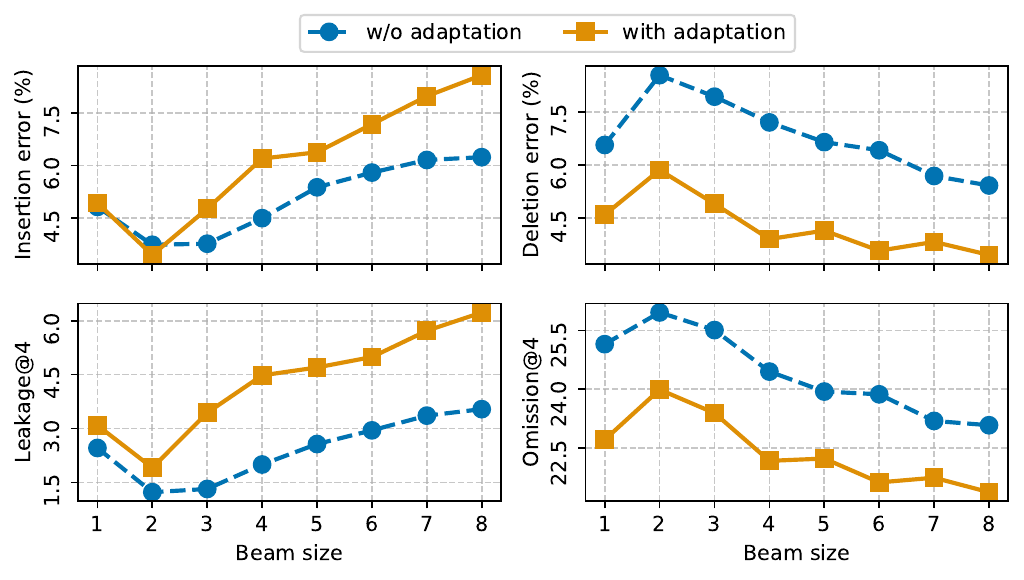}}
\caption{Effect of decoding beam size on insertion and deletion errors, versus effect on leakage and omission (for $n=4$), for SURT-base models with and without adaptation, for (a) ``anechoic'' and (b) ``replayed'' conditions, averaged across all overlap settings. The top row shows insertion and deletion errors, while the bottom row contains leakage@4 and omission@4.}
\label{fig:beams}
\end{figure}

% - quantify omission and leakage in terms of n-gram counts
% - above for beam search with increasing beam size

Throughout this paper, we have mentioned \textit{leakage} and \textit{omission}, first identified in \citet{Raj2021ContinuousSM}, as major contributors of errors in SURT.
In this section, we provide a metric for quantifying these error sources, in terms of n-gram counts on the reference $\mathbf{Y}$ and hypotheses $\hat{\mathbf{Y}}$, as defined in Section~\ref{sec:surt_model}.

\begin{description}
    \item[omission@n] For some $n$, the fraction of all unique n-grams in $\mathbf{Y}$ not present in any $\hat{\mathbf{Y}}_c$.
    \item[leakage@n] For some $n$, the fraction of all unique n-grams in $\mathbf{Y}$ present in multiple $\hat{\mathbf{Y}}_c$.
\end{description}

If emission time-stamps are known, these definitions may be modified to include time windows for n-gram search, e.g., an n-gram is considered to be present in more than one $\mathbf{Y}_c$ only if the corresponding time-spans overlap.
In the absence of time marks, the leakage@n and omission@n rates defined above are upper and lower bounds on the actual leakage and omission, respectively.
Nevertheless, by quantifying these error types explicitly, we can gain some insights into model behavior.
Our SURT model in this chapter does not predict token time-stamps, so we estimate omission and leakage over the entire hypotheses.

In Table~\ref{tab:surt_results}, we showed WERs achieved by the SURT models with and without adaptation, when decoded using beam search with a beam of size 4.
Typically, ASR model performance improves by increasing the beam size, but we found that for SURT models, the WER first improved (up to a beam size of 4) and then degraded.
This trend can be explained by looking at the leakage and omission errors, as shown in Fig.~\ref{fig:beams}.
We used the SURT-base model (with and without adaptation) for decoding the anechoic and replayed sets, using beam sizes varying from 1 to 8.
In the figure, we show insertion and deletion error rates, as well as leakage@4 and omission@4.
We observe that leakage first decreased (from beam size 1 to 2) but then increased gradually, while the opposite trend was observed for omission.
Furthermore, adapted models reduce omissions significantly, at the cost of increase in leakage (particularly in the replayed setting).

In the analysis above, the overall trend for leakage and omissions followed those of insertion and deletion errors, which is expected.
However, this may not always be the case.
In Table~\ref{tab:surt_auxiliary}, we showed results for ablation experiments done to investigate the effect of auxiliary objectives.
We calculated the leakage@4 and omission@4 rates for those models on the anechoic \texttt{dev} set, and compared them with the corresponding insertion and deletion rates.
The comparison is shown in Table~\ref{tab:surt_aux2}.
We see that although the models trained with $\mathcal{L}_{\text{ctc}}$ and $\mathcal{L}_{\text{mask}}$ have similar L@4 and O@4, the difference in insertion and deletion errors is comparatively large.
This suggests that the additional insertions or deletions are not caused due to unmixing errors, and are most likely due to errors in the recognition module.
This hypothesis seems reasonable because $\mathcal{L}_{\text{ctc}}$ should be more effective than $\mathcal{L}_{\text{mask}}$ in reducing purely ASR-related errors.
When we further include $\mathcal{L}_{\text{mask}}$ in training (last row), both O@4 and deletion reduce by roughly 5\% relative, indicating that the improvement almost entirely results from recovered sub-segments.

\begin{table}[t]
\setlength{\tabcolsep}{5pt}
\centering
\caption{Comparison of leakage@4 and omission@4 with insertion and deletion errors on the ``anechoic'' LibriCSS \texttt{dev} set.}
\label{tab:surt_aux2}
\adjustbox{max width=0.7\textwidth}{
\begin{tabular}{@{}ccrrrr@{}}
\toprule
$\mathcal{L}_{\mathrm{ctc}}$ & $\mathcal{L}_{\mathrm{mask}}$ & \textbf{Ins.} & \textbf{Del.} & \textbf{L@4} & \textbf{O@4} \\ \midrule
\xmark & \xmark & 3.44 & 5.96 & 2.44 & 20.63 \\
\cmark & \xmark & 2.56 & 4.87 & 1.51 & 20.10 \\
\xmark & \cmark & 2.96 & 5.88 & 1.49 & 20.34 \\
\cmark & \cmark & 3.03 & 4.65 & 1.68 & 19.27 \\
\bottomrule
\end{tabular}
}
\end{table}

\subsection{Results on AMI and ICSI}
\label{sec:surt_results_ami}

Finally, we evaluated our SURT models on two meeting benchmarks, AMI and ICSI, to verify their efficacy on real multi-talker speech. 
For these experiments, we initialized SURT \textit{base} and \textit{large} models with the final checkpoint from the LSMix-full training, and continued training on the AIMix-full simulated mixtures for 30 epochs. 
The models were further adapted on the AMI and ICSI \texttt{train} sessions, by combining the IHM-Mix, SDM, and beamformed MDM recordings. 
For the adaptation, we cut the sessions at pauses of 0.0 and 0.5 seconds (thus creating two copies with different segmentations).
We trimmed the long sessions to approximately 30s each based on utterance end time marks. 
This process resulted in 330k train samples of average duration 6.4s, totaling 590h and 18.4\% overlapped speech.
We adapted the SURT models on these sub-sessions by training with a LR of 0.0001 for 20 epochs, and used the final checkpoint for inference. 
The results with and without in-domain adaptation are shown in Table~\ref{tab:surt_ami}.

We found that the \textit{large} model obtained consistently better WERs than the \textit{base} model, which is expected. 
All models obtained lower WERs on ICSI, which may be due to its lower overlapped speech ratio (13.6\%) compared to AMI (21.0\%). 
A similar increase in WERs with higher overlaps was also observed for LibriCSS. 
Without model adaptation, the SURT models obtained reasonable WERs for the IHM-Mix and beamformed MDM settings, but not for SDM. 
For instance, the SURT-base model obtained 64.3\% relatively worse WERs on SDM compared to IHM-Mix. 
This suggests that purely simulated mixtures cannot compensate for real, far-field training data. 
SDM performance improved significantly with in-domain adaptation, with a relative WER reduction of 28.3\% and 43.5\% for SURT-base on AMI and ICSI, respectively.

\begin{table}[t]
\centering
\caption{Results on the AMI and ICSI \texttt{test} sets, under different microphone settings, using SURT \textit{base} and \textit{large} models.}
\label{tab:surt_ami}
\adjustbox{max width=0.9\linewidth}{
\begin{tabular}{@{}lc@{\hspace{2\tabcolsep}}ccc@{\hspace{2\tabcolsep}}cc@{}}
\toprule
\multirow{2}{*}{\textbf{Model}} & \multirow{2}{*}{\textbf{Adapt.}} & \multicolumn{3}{c}{\textbf{AMI}} & \multicolumn{2}{c}{\textbf{ICSI}} \\
\cmidrule(l{1pt}r){3-5} \cmidrule(l{1pt}r{1pt}){6-7}
 & & IHM-Mix & SDM & MDM & IHM-Mix & SDM \\
\midrule
\multirow{2}{*}{\textbf{Base}} & \xmark & 39.8 & 65.4 & 46.6 & 28.3 & 60.0 \\
 & \cmark & 37.4 & 46.9 & 43.7 & 26.3 & 33.9 \\
\midrule
\multirow{2}{*}{\textbf{Large}} & \xmark & 36.8 & 62.5 & 44.4 & 27.8 & 59.7 \\
 & \cmark & 35.1 & 44.6 & 41.4 & 24.4 & 32.2 \\
\bottomrule
\end{tabular}
}
\end{table}

Error analysis revealed that most of the errors were caused by deletion, as shown in Table~\ref{tab:surt_wer}. 
For the unadapted models, deletion comprised 89.4\% and 83.2\% of the overall errors for the base and large SURT variants, respectively. 
We attribute this primarily to the model missing very short utterances and back-channels such as ``Okay,'' ``Hmm,'', and so on, which form a significant fraction of such meetings, and which may be useful for downstream dialog understanding tasks.
When we adapted the models using real, in-domain data, the deletion errors reduced significantly, with minor increase in insertion and substitution.
For the adapted models, deletion comprised 68.9\% and 66.8\% of the total errors for the base and large SURT, respectively.
In future work, it may be interesting to investigate models and training objectives which explicitly try to avoid suppression of short overlapping segments.

\begin{table}[t]
\centering
\caption{WER breakdown on the AMI \texttt{test} set for the SDM microphone condition, using SURT \textit{base} and \textit{large} models.}
\label{tab:surt_wer}
\adjustbox{max width=0.54\linewidth}{
\begin{tabular}{@{}lccccc@{}}
\toprule
\textbf{Model} & \textbf{Adapt.} & \textbf{Ins.} & \textbf{Del.} & \textbf{Sub.} & \textbf{WER} \\
\midrule
\multirow{2}{*}{\textbf{Base}} & \xmark & 0.8 & 55.8 & 8.8 & 65.4 \\
 & \cmark & 1.8 & 32.3 & 12.8 & 46.9 \\
\midrule
\multirow{2}{*}{\textbf{Large}} & \xmark & 0.9 & 52.0 & 9.7 & 62.5 \\
 & \cmark & 2.1 & 29.8 & 12.8 & 44.6 \\
\bottomrule
\end{tabular}
}
\end{table}

\section{Conclusion}

We performed a detailed investigation of the SURT model for multi-talker speech recognition.
By decomposing the challenges faces by this model in continuous, streaming, multi-talker ASR into the three components of sparsely overlapped speech separation, long-form ASR, and quick turn-taking modeling, we were able to identify model design strategies that improve performance on one or more of these sub-problems.

We modeled SURT's unmixing component as mask estimation on the original filter-bank inputs used in single-speaker ASR, which allowed the use of transducer pre-training on single-speaker utterances.

To improve training efficiency, we applied zipformer blocks in the encoder which aggressively subsample the input sequence and use shared attention masks within the blocks.

We also used sub-segments instead of full utterances to simulate training mixtures, which resulted in more frequent turn-taking without increasing the training sequence length.

For training, we used the recently proposed pruned transducer instead of the full-sum transducer loss to reduce memory requirement.

We found that using auxiliary objectives for the encoder and the masking network also improves the model's performance, for instance by producing better soft alignments of the input and output sequences during training, or by reducing leakage in single-speaker regions.

To further reduce errors caused by leakage and omission, we used dual-path LSTMs instead of convolutional layers in unmixing, and added branch tying of encoder outputs in the recognition component.

We trained SURT in multiple stages: (i) single-speaker pre-training, (ii) training on simulated mixtures, and (iii) adaptation on in-domain real data, to outperform the larger and computationally expensive MT-RNNT proposed previously on LibriCSS.

Finally, we also demonstrated the viability of these models for real meeting benchmarks, namely AMI and ICSI.

In this chapter, our focus with the SURT model design was to perform multi-talker ASR without concerning ourselves with speaker attribution of each utterance, as measured by ORC-WER.
Now that we have a framework for this in place, we can build upon this model to jointly perform transcription and speaker attribution.
In the next chapter, we will describe new methods for speaker attribution in the SURT framework, which will allow us to perform streaming, speaker-attributed ASR in an end-to-end fashion.
% \cleardoublepage

\chapter{Speaker Attribution in the SURT Framework}
\label{chap:surt2}

In the previous chapter, we introduced the SURT framework, and evaluated its efficacy in the task of speaker-agnostic multi-talker ASR.
The SURT model thus far is only capable of transcribing all the utterances in the input mixture and keeping consecutive words of an utterance together, but it cannot attribute non-overlapping utterances across a session to speakers.
In this chapter, we will fill this gap and describe how to modify SURT such that it can predict speaker labels jointly with ASR tokens.
This will enable us to perform streaming, speaker-attributed transcription using a single end-to-end model.

\section{Introduction}
\label{sec:surt2_intro}

% Jointly optimized diarization + ASR
In previous chapters, we have seen that modular systems for speaker-attributed transcription have several limitations, such as error propagation, expensive maintenance, etc.
Due to these limitations with modular systems, researchers have proposed jointly optimized models that combine diarization and ASR to directly solve for the task of speaker-attributed transcription. 
The most popular of these is the speaker-attributed ASR (SA-ASR) model based on attention-based encoder-decoders (AEDs)~\cite{Chorowski2015AttentionBasedMF}.
It uses serialized output training (SOT) to handle overlapped speech and registered speaker profiles (called a speaker inventory) to handle speaker attribution~\cite{Kanda2020SerializedOT,Kanda2020JointSC}.
Several modifications to this model have leveraged transformer-based encoders~\cite{Kanda2021EndtoEndSA} and large-scale pre-training~\cite{Kanda2021ACS}, and have proposed methods for inference on long recordings~\cite{Chang2021HypothesisSF} without the dependence on a speaker inventory~\cite{Kanda2020InvestigationOE}.
There have been further investigations of methods for speaker attribution within SA-ASR, and its extension to multi-channel and contextualized ASR~\cite{Yu2022ACS,Shi2022ACS,Shi2023CASAASRCS}.
By modifying SOT to be performed at the token-level (known as t-SOT), Kanda et al.~\cite{Kanda2022StreamingSA} performed streaming transcription of overlapping speech, which was not feasible in the original SA-ASR due to the use of AEDs and utterance-level serialization.
Enforcing monotonicity in this manner also allows these models to be built upon neural transducers~\cite{Graves2012SequenceTW} instead of AEDs.
Nevertheless, t-SOT requires complicated interleaving/deserialization of tokens based on timestamps to accommodate overlapping speech on a single output channel, and the use of several ``channel change'' tokens may impact ASR training adversely.
Other methods for speaker-attributed transcription have been proposed that jointly model ASR and speaker labels in the output unit~\cite{Shafey2019JointSR}, perform speaker-conditioned acoustic modeling with EEND~\cite{Chetupalli2021SpeakerCA}, or attach a sidecar separator for speech activity prediction~\cite{Meng2023UnifiedMO}.

% \begin{figure}
%     \centering
%     \includegraphics[trim={0.5cm 0 0 0},clip,width=\linewidth]{surt2_chapter/figs/surt.pdf}
%     \caption{An overview of the SURT for speaker-agnostic transcription.}
%     \label{fig:surt2_surt}
% \end{figure}

% SURT and MT-RNNT
Arguably, a more natural approach for continuous, streaming, multi-talker ASR, as described in the last chapter, involves transcribing overlapping utterances on parallel output channels by unmixing them inside the model.
We refer to this two-branch strategy as Streaming Unmixing and Recognition Transducer, or SURT~\cite{Lu2020StreamingEM}, and the model is shown in Fig.~\ref{fig:surt_model}.
In the literature, SURT has been extended to handle long-form multi-turn recordings~\cite{Raj2021ContinuousSM,Sklyar2021MultiTurnRF}, and to jointly perform endpointing and segmentation~\cite{Lu2022EndpointDF,Sklyar2022SeparatorTransducerSegmenterSR}.
\citet{Lu2021StreamingMS} also proposed joint speaker identification with SURT, but their model relied on a speaker inventory and was only used for single-turn synthetic mixtures.
As shown in Fig.~\ref{fig:surt_model}, the SURT model consists of an ``unmixing'' component that separates the mixed audio into non-overlapping streams, and a ``recognition'' component that transcribes each of these streams.
Since there is no explicit emission of speaker labels in this modeling scheme, SURT has thus far been limited to \textit{speaker-agnostic} transcription.
Specifically, words uttered by two speakers in a pair of overlapping utterances are indeed transcribed on separate channels, but there is no way to attribute two non-overlapping utterances of a speaker to that speaker, even within the same utterance group.
In this chapter, our objective is to extend the SURT model for \textit{speaker-attributed} transcription of an arbitrary number of speakers without any speaker inventory.

We achieve this by adding an auxiliary speaker transducer to the recognition module of SURT.
We constrain this branch to emit a speaker label corresponding to each ASR token predicted by using HAT-style~\cite{Variani2020HybridAT} blank factorization of the output logits.
This is intended to ensure that two non-overlapping utterances of a speaker within an utterance group are assigned the same speaker label.
We additionally propose a novel ``speaker prefixing'' method to ensure that the speaker labels are consistent across different utterance groups in the recording.
We validate our methods through ablation experiments on LibriSpeech mixtures, and finally demonstrate streaming speaker-attributed transcription on real meetings from the AMI corpus.

\section{Preliminary}
\label{sec:surt2_mt-asr}

\subsection{Speech recognition with neural transducers}

In single-talker ASR, audio features for a segmented utterance $\mathbf{X} \in \mathbb{R}^{T\times F}$, where $T$ and $F$ denote the number of time frames and the input feature dimension, respectively, are provided as input, and the system predicts the transcript $\mathbf{y} = (y_1,\ldots,y_U)$, where $y_u \in \mathcal{V}$ are output units such as graphemes or word-pieces, and $U$ is the length of the label sequence. 
For discriminative training, we achieve this by minimizing the negative conditional log-likelihood, $\mathcal{L} = -\log P(\mathbf{y}|\mathbf{X})$. 
Since the alignment between $\mathbf{X}$ and $\mathbf{y}$ is not known, transducers compute $\mathcal{L}$ by marginalizing over the set of all alignments $\mathbf{a} \in \bar{\mathcal{V}}^{T+U}$, where $\bar{\mathcal{V}} = \mathcal{V}\cup \{\phi\}$, $\phi$ is called the blank label, and the non-blank labels in $\mathbf{a}$ equal $\mathbf{y}$.
Formally,
\begin{equation}
P(\mathbf{y}|\mathbf{X}) = \sum_{\mathbf{a}\in \mathcal{B}^{-1}(\mathbf{y})} P(\mathbf{a}|\mathbf{X}),
\label{eq:surt2_rnnt}
\end{equation}
where $\mathcal{B}$ is the deterministic mapping from an alignment $\mathbf{a}$ to the sub-sequence of its non-blank symbols. 
Transducers parameterize $P(\mathbf{a}|\mathbf{X})$ with an encoder, a prediction network, and a joiner (see Fig.~\ref{fig:modular_transducer}).
The encoder maps $\mathbf{X}$ into hidden representations $\mathbf{f}_1^T$, while the prediction network maps $\mathbf{y}$ into $\mathbf{g}_1^U$.
The joiner combines the outputs from the encoder and the prediction network to compute logits $\mathbf{z}_{t,u}$ which are fed to a softmax function to produce a posterior distribution over $\bar{\mathcal{V}}$. 
Under the assumption of a streaming encoder, we can expand \eqref{eq:surt2_rnnt} as
\begin{align}
    P(\mathbf{y}|\mathbf{X}) &= \sum_{\mathbf{a}\in \mathcal{B}^{-1}(\mathbf{y})} \prod_{t=1}^{T+U} P(\mathbf{a}_t|\mathbf{f}_1^t,\mathbf{g}_1^{u(t)-1}) \\
    &= \sum_{\mathbf{a}\in \mathcal{B}^{-1}(\mathbf{y})} \prod_{t=1}^{T+U} \mathrm{Softmax}(\mathbf{z}_{t,u(t)}), \label{eq:surt2_rnnt_softmax}
\end{align}
where $u(t)\in\{1,\ldots,U\}$ denotes the index in $\mathbf{y}$ of the last non-blank symbol in $\mathbf{a}$ up to time $t$.
The negative log of this expression is known as the RNN-T or transducer loss.
In practice, to make training more memory-efficient, we often approximate the full sum, for example using the pruned transducer loss~\cite{Kuang2022PrunedRF}.
This loss function was denoted as $\mathcal{L}_{\text{rnnt}}$ in Chapter~\ref{chap:surt}.

\subsection{Multi-talker ASR with SURT}

In multi-talker ASR, the input $\mathbf{X}\in\mathbb{R}^{T\times F}$ is an unsegmented mixture containing $N$ utterances from $K$ speakers, i.e., $\mathbf{X} = \sum_{n=1}^N \mathbf{x}_n$, where $\mathbf{x}_n$ is the $n$-th utterance ordered by start time, shifted left and zero-padded to the length of $\mathbf{X}$. 
The desired output is $\mathbf{Y} = \{\mathbf{y}_n: 1\leq n \leq N\}$, where $\mathbf{y}_n$ is the reference transcript corresponding to $\mathbf{x}_n$. 
Assuming at most a two-speaker overlap, the \textit{heuristic error assignment training} (HEAT) paradigm~\cite{Lu2020StreamingEM} is used to create channel-wise references $\mathbf{Y}_1$ and $\mathbf{Y}_2$ by assigning $\mathbf{y}_n$'s to the first available channel, in order of start time. 
SURT estimates $\hat{\mathbf{Y}} = [\hat{\mathbf{Y}}_1,\hat{\mathbf{Y}}_2] = f_{\text{surt}}(\mathbf{X})$ as follows. 
First, an unmixing module computes $\mathbf{H}_1$ and $\mathbf{H}_2$ as
\begin{align}
\begin{split}
\label{eq:surt2_surt}
& \mathbf{H}_1 = \mathbf{M}_1 \ast \mathbf{X}, \quad \mathbf{H}_2 = \mathbf{M}_2 \ast \mathbf{X},~~\text{where} \\
& [\mathbf{M}_1,\mathbf{M}_2]^T = \mathrm{MaskNet}(\mathbf{X}),
\end{split}
\end{align}
$\mathbf{M}_c\in \mathbb{R}^{T\times F}$ is a soft mask per channel and $\ast$ is Hadamard product. 
$\mathbf{H}_1$ and $\mathbf{H}_2$ are fed into a transducer-based ASR, producing logits $\mathbf{Z}_1$ and $\mathbf{Z}_2$. 
Finally,
\begin{align}
\begin{split}
\label{eq:surt2_heat}
&\mathcal{L}_{\text{heat}} = \mathcal{L}(\mathbf{X}, \mathbf{Y}_1, \mathbf{Z}_1) + \mathcal{L}(\mathbf{X}, \mathbf{Y}_2, \mathbf{Z}_2),~~\text{where}\\
&\mathcal{L} = \mathcal{L}_{\text{rnnt}} + \lambda_{\text{ctc}}\mathcal{L}_{\text{ctc}} + \lambda_{\text{mask}}\mathcal{L}_{\text{mask}},
\end{split}
\end{align}
where $\mathcal{L}_{\text{ctc}}$ and $\mathcal{L}_{\text{mask}}$ denote auxiliary CTC loss on the encoder~\cite{Graves2006ConnectionistTC} and mean-squared error loss on the masking network, respectively, and $\lambda$'s are hyperparameters. 
In this formulation, SURT only performs \textit{speaker-agnostic} transcription, and is evaluated using ORC-WER, as described in Chapter~\ref{chap:surt}.

\section{SURT for speaker-attributed transcription}

% Research questions:
% \begin{enumerate}
%     \item How can we deal with overlapping speech for speaker attribution?
%     \item How to synchronize speaker label prediction with ASR token prediction?
%     \item How to reconcile relative speaker labels across utterance groups in a long recording?
% \end{enumerate}

For speaker-attributed ASR, the desired output is $\mathbf{Y}=\{(\mathbf{y}_n,s_n):1\leq n \leq N, s_n \in [1,K]\}$, where $K$ is the number of speakers in the mixture.
SURT estimates $\mathbf{y}_n$ by mapping the utterances to two channels $\hat{\mathbf{Y}}_1$ and $\hat{\mathbf{Y}}_2$, as described in Section~\ref{sec:surt2_mt-asr}.
A popular method for speaker attribution in multi-talker settings (other than SURT) is to predict speaker change tokens that segment the output into speaker-specific regions, followed by speaker label assignment to each segment.
However, this kind of training is prone to over-estimate the speaker change tokens, and may also adversely affect the ASR performance.
Instead, we want to perform speaker attribution without affecting the output of the ASR branch, for example by predicting a speaker label for each ASR token emitted.

In order to perform such a streaming speaker attribution jointly with the transcription, the following questions arise: 
\begin{enumerate}
    \item How do we deal with overlapping speech? 
    \item How do we synchronize speaker label prediction with ASR token prediction?
    \item How to reconcile relative speaker labels across utterance groups in a long recording?
\end{enumerate}
We will answer each of these questions in the following subsections.

\subsection{Auxiliary speaker transducer}

\begin{figure}[t]
    \centering
    \includegraphics[width=\linewidth]{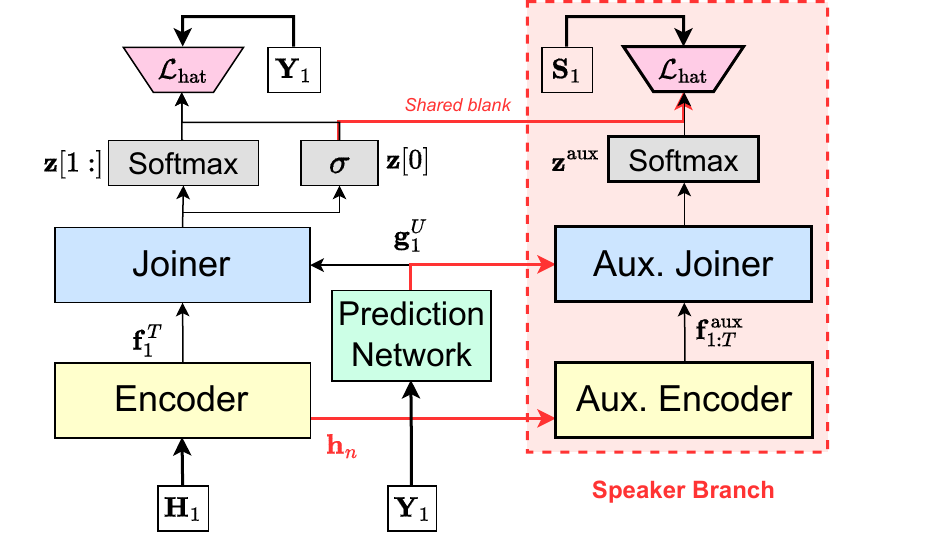}
    \caption{Auxiliary speaker transducer (red box) with shared blank label. The auxiliary encoder takes as input a hidden layer representation $\mathbf{h}_n$ from the main encoder, and generates $\mathbf{f}_{1:T}^{\mathrm{aux}}$. The blank logit $\mathbf{z}[0]$ from the main joiner is shared with the speaker branch to compute the HAT loss.}
    \label{fig:surt2_aux_branch}
\end{figure}

We map speaker labels $s_n$ to two channels according to the HEAT strategy, obtaining $\hat{\mathbf{S}}_1$ and $\hat{\mathbf{S}}_2$.
During training, we repeat $s_n$ as many times as there are tokens in $\mathbf{y}_n$, i.e., we want to predict the speaker label for each lexical token.
Thereafter, we use the non-overlapping streams $\mathbf{H}_c$ to estimate $\hat{\mathbf{S}}_c$ in the same two-branch approach as the ASR transducer.
For this, we add an auxiliary speaker transducer to each of the two branches in the recognition module, as shown in Fig.~\ref{fig:surt2_aux_branch}.
Intermediate representations $\mathbf{h}_n$ from the $n^{\text{th}}$ layer of the main encoder are fed into an auxiliary encoder, producing $\mathbf{f}_{1:T}^{\mathrm{aux}}$.
An auxiliary joiner combines $\mathbf{f}_{1:T}^{\mathrm{aux}}$ with $\mathbf{g}_{1}^U$ to produce auxiliary logits $\mathbf{z}_{t,u}^{\mathrm{aux}}$, which are used to obtain a distribution over the speaker labels and the blank label.
Combing the auxiliary encoder representation with representations from the ASR prediction network allows the speaker branch to leverage lexical content for predicting speaker labels.
Such a use of lexical information has been shown to be beneficial for speaker diarization using clustering-based~\cite{Flemotomos2019LanguageAS, Park2019SpeakerDW} or end-to-end neural approaches~\cite{Khare2022ASRAwareEN}.

\subsection{Synchronizing speaker labels with ASR tokens}

Since transducers perform frame-synchronous decoding with the blank label, the above formulation has several issues.
First, we cannot ensure that the number of ASR tokens $\lvert \hat{\mathbf{Y}}_c\rvert$ predicted on branch $c$ is equal to the the number of speaker labels $\lvert \hat{\mathbf{S}}_c \rvert$.
Even if we can ensure this, synchronizing the speaker labels with the ASR tokens can be hard, as illustrated in the following example containing two speakers saying the words ``hello'' and ``hi'', and the correct number and sequence of speaker labels.
\begin{figure}[ht]
    \centering
    \includegraphics[width=0.9\linewidth]{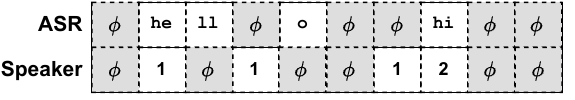}
\end{figure}

% Even though we predicted the correct speaker labels, it is hard to assign them to the corresponding ASR tokens since they are not synchronized by frame.
%
To solve these problems, we need to ensure that SURT emits blank labels on the same frames for both the ASR and speaker branches.
We achieve this by factoring out the blank label separately in the style of the hybrid auto-regressive transducer (HAT) model~\cite{Variani2020HybridAT}, i.e., we replace the alignment posterior $P(\mathbf{a}_t\mid \mathbf{f}_1^t, \mathbf{g}_1^{u(t)-1})$ in \eqref{eq:surt2_rnnt_softmax} with
\begin{equation}
\label{eq:surt2_hat}
P(\mathbf{a}_t\mid \mathbf{f}_1^t, \mathbf{g}_1^{u(t)-1}) = 
\begin{cases} 
b_{t,u}, ~~\text{if}~~ \mathbf{a}_t = \phi, \\
(1-b_{t,u})~\mathrm{Softmax(\mathbf{z}_{t,u}[1:])}, ~~\text{otherwise},
\end{cases}
\end{equation}
where $b_{t,u} = \sigma(\mathbf{z}_{t,u}[0])$, and $\sigma$ denotes the sigmoid function.
By setting $\mathbf{z}_{t,u}^{\mathrm{aux}}[0] = \mathbf{z}_{t,u}[0]$, i.e., by sharing the blank logit for the ASR and speaker outputs, we ensure that blank emission is synchronized between the two branches.
The speaker branch is trained with a similar HAT loss, i.e.,
\begin{equation}
\label{eq:surt2_spk_hat}
\mathcal{L}_{\mathrm{aux}} = \mathcal{L}_{\mathrm{hat}}(\mathbf{H}_1,\mathbf{Z}_1^{\mathrm{aux}}) + \mathcal{L}_{\mathrm{hat}}(\mathbf{H}_2,\mathbf{Z}_2^{\mathrm{aux}}).
\end{equation}
Such a synchronization strategy has also recently been proposed for performing word-level diarization using transducers~\cite{Huang2023TowardsWE}.
For both ASR and speaker branches, we use a pruned version of the HAT loss similar to pruned RNNT~\cite{Kuang2022PrunedRF}.

\subsection{Maintaining state across utterance groups}

\begin{figure}[t]
    \centering
    \includegraphics[width=0.9\linewidth]{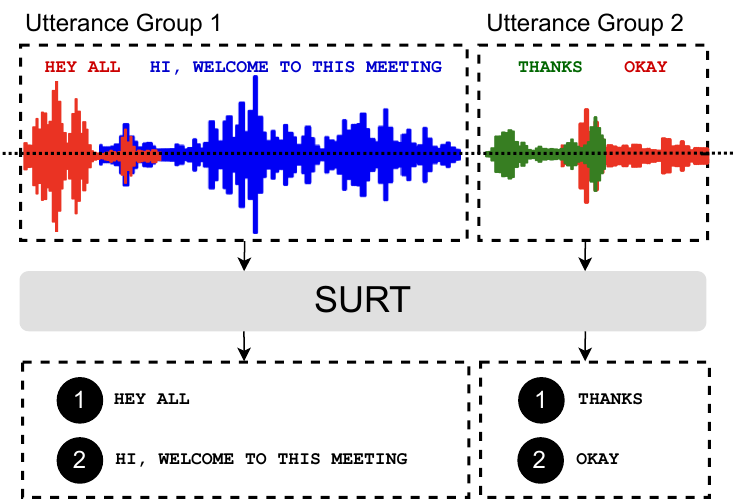}
    \caption{Illustration of problem caused by relative speaker label within utterance groups. For each group, speakers are assigned labels in FIFO order, which may result in speaker permutation errors for the recording.}
    \label{fig:surt2_utterance_groups}
\end{figure}

A common approach for inference of long-form audio is by chunking in some way (e.g., at silences or fixed-length chunks), processing each chunk separately, and then combining the outputs.
For SURT, we assume that the recording has been chunked at silences to create utterance groups, which are sets of utterances connected by speaker overlaps.
For multi-talker ASR methods such as SA-ASR (c.f. Section~\ref{sec:surt2_intro}) which predict \textit{absolute} speaker identities using external speaker profiles, combining chunk-wise outputs is relatively straightforward since there is no issue of speaker label permutation.
However, the auxiliary speaker branch in SURT is trained to predict \textit{relative} speaker labels in FIFO order within a chunk or utterance group, and these labels must be reconciled across all chunks within a recording in order to obtain the final speaker-attributed transcript.
Fig.~\ref{fig:surt2_utterance_groups} illustrates this problem with a simple example consisting of two chunks, where each color denotes a different speaker.

\subsubsection{What does the auxiliary encoder encode?}

Speaker label reconciliation across different chunks for long-form diarization or speaker-attributed ASR is often done through clustering of speaker embeddings estimated from the chunks.
For example, the EEND-VC model for speaker diarization extends EEND for diarization of long-recordings by applying clustering over chunk-wise speaker vectors~\cite{Kinoshita2020IntegratingEN}.
This method delays the output prediction at least until the end of the chunk so that the re-clustering may be done.
To remedy this issue, SA-ASR based on t-SOT estimates speaker change based on cosine similarity between consecutive speaker vectors, and applies re-clustering of all vectors every time a speaker change is detected~\cite{Kanda2022StreamingSA}.
Nevertheless, solving label permutation through such clustering requires that the chunk-wise speaker vectors should represent absolute speaker identities.
This requirement is not satisfied in the SURT model since the auxiliary speaker branch is trained to predict relative speaker labels in their order of appearance in the mixture.

\begin{figure*}[tbp]
\centering
    \captionsetup[subfigure]{labelformat=empty}
    \begin{subfigure}{0.49\linewidth}
    \centering
    \includegraphics[width=\linewidth]{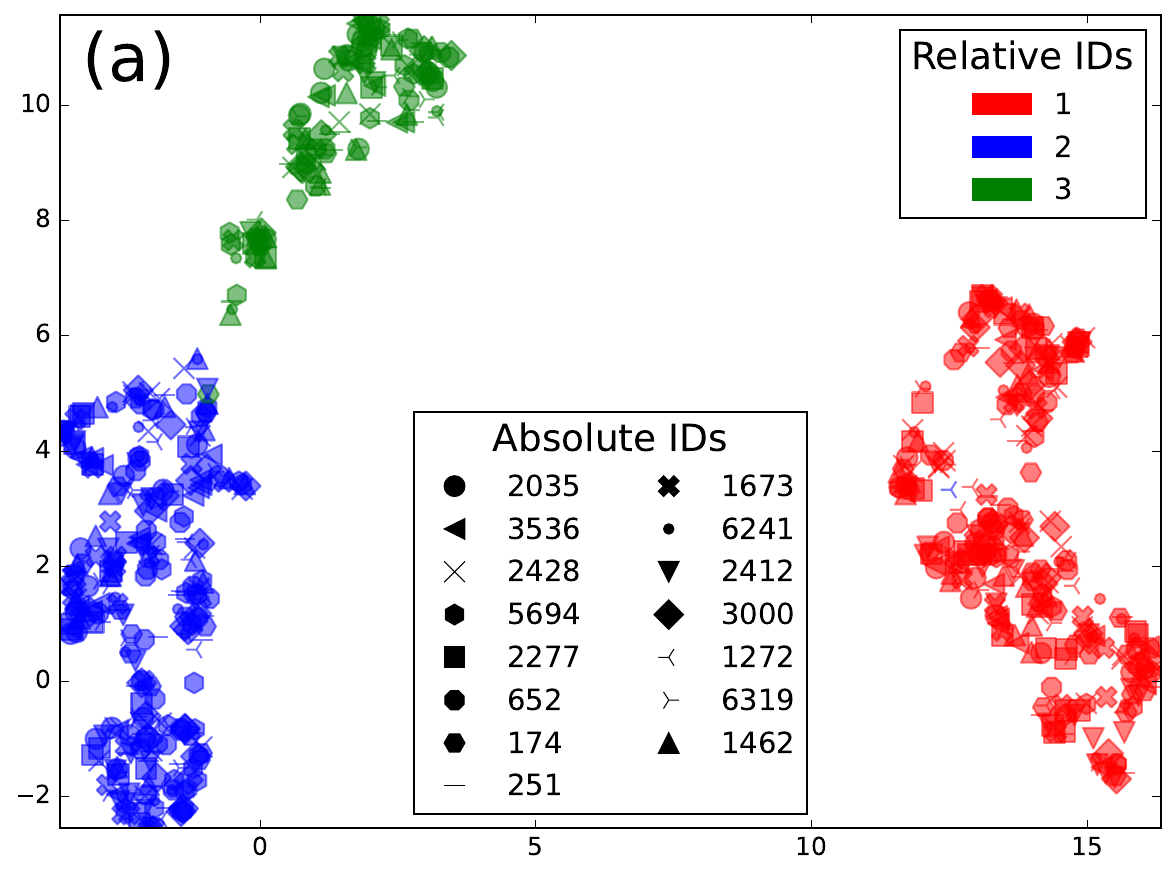}
    \caption{\label{fig:umap}}
    \end{subfigure}
    \begin{subfigure}{0.49\linewidth}
    \centering
    \includegraphics[width=\linewidth]{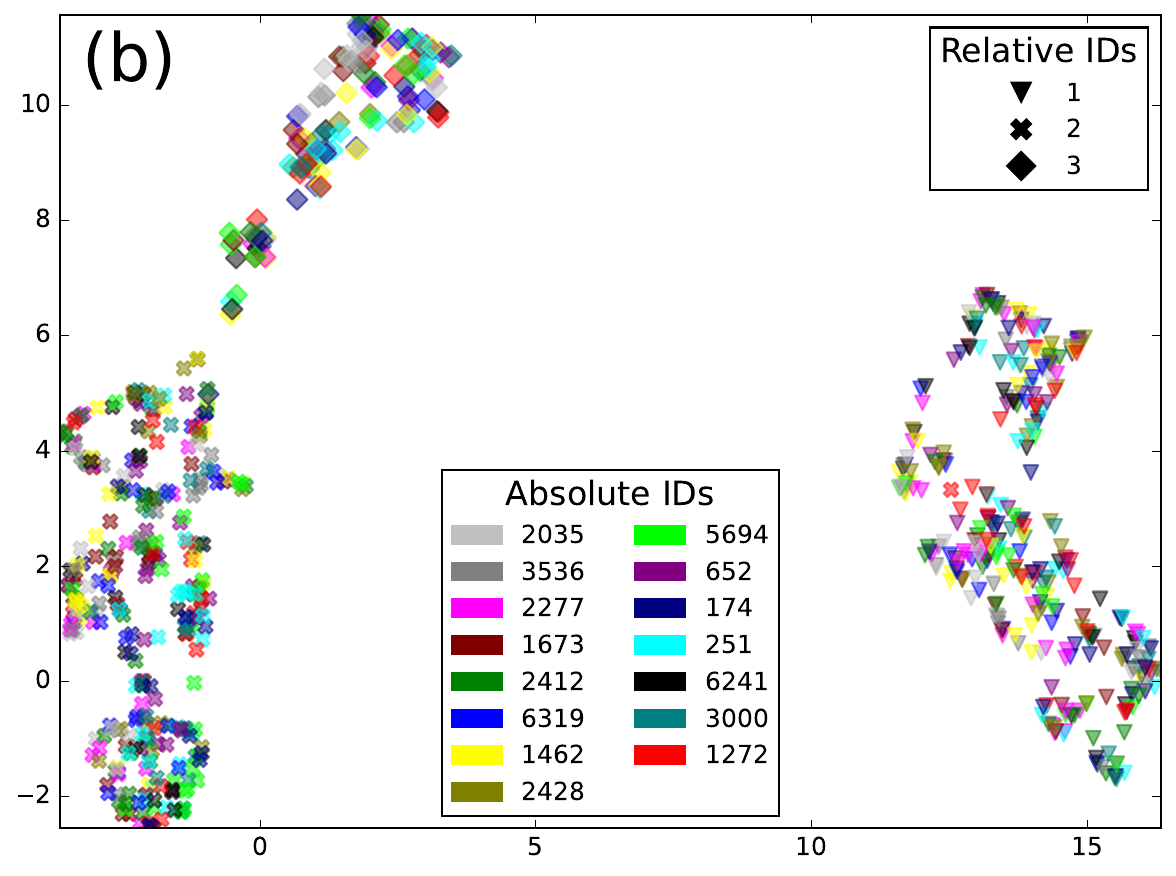}
    \caption{\label{fig:umap_inv}}
    \end{subfigure}
    \begin{subfigure}{0.49\linewidth}
    \centering
    \includegraphics[width=\linewidth]{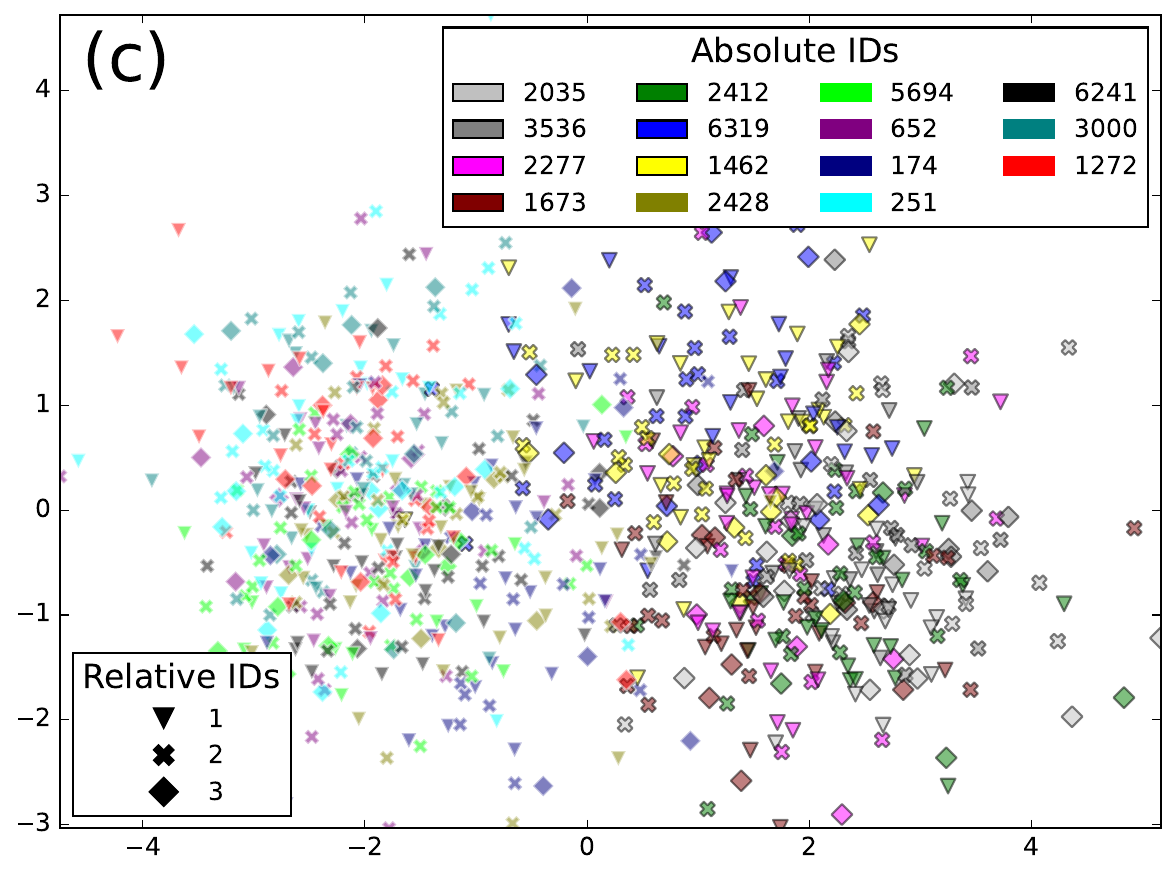}
    \caption{\label{fig:lda_inv}}
    \end{subfigure}
\caption{Projections of auxiliary encoder representations for a subset of LSMix \texttt{dev}. Each point denotes the representation of one speaker in a mixture, averaged over the frames on which the model emits a non-blank label. (a) and (b) denote UMAP projection, and (c) shows LDA projection using absolute speaker classes.}
\label{fig:surt2_spk_clusters}
\end{figure*}

To verify this, we applied SURT with the auxiliary speaker branch on synthetic mixtures of LibriSpeech utterances (described in Section~\ref{sec:surt2_data}) consisting of 2 or 3 speakers per mixture.
We collected the 256-dimensional encoder representations for the frames where the auxiliary branch predicts a speaker label, and averaged each speaker's embeddings over the mixture.
In Fig.~\ref{fig:surt2_spk_clusters}, we show UMAP and LDA projections of these embeddings for 15 different speakers in the LSMix \texttt{dev} set.
In Fig.~\ref{fig:umap}, the three colors denote the relative speaker label assigned to the speaker during SURT inference and the markers denote absolute speaker identities.
Fig.~\ref{fig:umap_inv} shows the same plot, but in this case colors denote absolute speaker identities and markers denote relative order within the chunk.
It is easy to see that the embeddings learned by the model of Fig.~\ref{fig:surt2_aux_branch} cluster by relative speaker labels instead of absolute speaker identities, validating our conjecture that the auxiliary encoder extracts relative speaker position in the chunk.
Even when LDA using absolute speaker labels is used for the low-dimensional projection (as shown in Fig.~\ref{fig:lda_inv}), we did not find clusters of absolute speaker labels.
Interestingly, the embeddings did retain information about the speaker's gender. 
In the figure, the points with and without a black border respectively denote female and male speakers, and they appear well-separated into gender-based clusters.

\subsubsection{The speaker prefixing method}

Inspired by the use of a speaker tracing buffer in the EEND model for online diarization~\cite{Xue2020OnlineEN}, we propose a novel \textit{speaker prefixing} strategy to solve the problem of speaker label permutation across utterance groups.
The key idea of speaker prefixing is to append to the beginning of each chunk or utterance group high-confidence frames of speakers we have seen so far in the recording, in the order of their predicted label.
Formally, let $\{\mathbf{X}_1,\ldots,\mathbf{X}_M\}$ be the input features corresponding to $M$ utterance groups in a recording, such that $\mathbf{X}_m \in \mathbb{R}^{T_m \times F}$.
For some chunk $m$, let $K_m \in [0,K]$ be the number of speakers seen so far in the recording.
We define some function $\mathscr{S}$ which selects frames of a given speaker in the previous chunks, i.e.,
\begin{equation}
    \mathscr{S}(\mathbf{X}_1,\ldots,\mathbf{X}_{m-1},k) = \mathbf{B}_k,
\end{equation}
where $k$ is one of the $K_m$ speakers and $\mathbf{B}_k \in \mathbb{R}^{\tau \times F}$, for some $\tau$ (which is a hyperparameter), is analogous to a speaker ``buffer.''
Then, the speaker-prefixed input for chunk $m$ is given as
\begin{equation}
    \Tilde{\mathbf{X}}_m = \left[\mathbf{B}_1^T;\ldots;\mathbf{B}_{K_m}^T;\mathbf{X}^T\right]^T.
\end{equation}

We use $\Tilde{\mathbf{X}}_m$ instead of $\mathbf{X}_m$ as input for this chunk, with the expectation that the speaker buffers would enforce a relative ordering of speakers in past chunks among speakers in the current chunk.
At the output of the main and auxiliary encoders, we remove the representation corresponding to the prefix, which is of length $\frac{K_m \times \tau}{s}$, where $s$ is the subsampling factor of the encoder.
During inference, we set $\mathscr{S}$ to select a sequence of $\tau$ frames (from the previous chunks) with the largest sum of confidence value, as predicted by its logit $\mathbf{z}^{\mathrm{aux}}[k]$.
During training, we randomly select $\tau$ from a specified list, and $\kappa$ speakers to prefix from all speakers in the batch.
Such a strategy mimics the expected inference time scenario, where not all prefixed speakers will be seen in every chunk.
For each selected speaker, we randomly sample a range of $\tau$ frames from all the segments of that speaker.

% https://www.geeksforgeeks.org/find-maximum-minimum-sum-subarray-size-k/

\section{Experimental Setup}

\subsection{Network architecture}

The main SURT model is similar to the one described in Chapter~\ref{chap:surt}.
The masking network comprises four 256-dim DP-LSTM layers~\cite{Luo2019DualPathRE}.
Masked features are reduced to half the original length through a convolutional layer, and the subsampled features are fed into a zipformer encoder~\cite{Yao2023ZipformerAF}.
The ASR encoder consists of 6 zipformer blocks subsampled at different frame rates (up to 8x in the middle).
Each such block consists of 2 self-attention layers with shared attention weights and separate feed-forward layers.
The encoder output is further down-sampled such that the overall subsampling factor is 4x.
The representations from an intermediate layer of the ASR encoder are passed to the auxiliary encoder.
This is another zipformer comprising 3 blocks with smaller attention and feed-forward dimensions.
Branch tying is used at the output of both encoders using unidirectional LSTM layers~\cite{Raj2023Surt20}.
The ASR prediction network, which is shared with the speaker branch, contains a single 512-dim Conv1D layer.
The complete SURT model contains 38.0M parameters, divided up into 6.0M, 23.6M, and 8.4M for the masking network, the ASR branch, and the speaker branch, respectively. 
The chunk size for the intra-LSTM and the Zipformer is set to 32 frames, resulting in a modeling latency of 320 ms.

\subsection{Data}
\label{sec:surt2_data}

\begin{figure}[tb]
    \begin{subfigure}{0.49\linewidth}
    \centering
    \includegraphics[width=\linewidth]{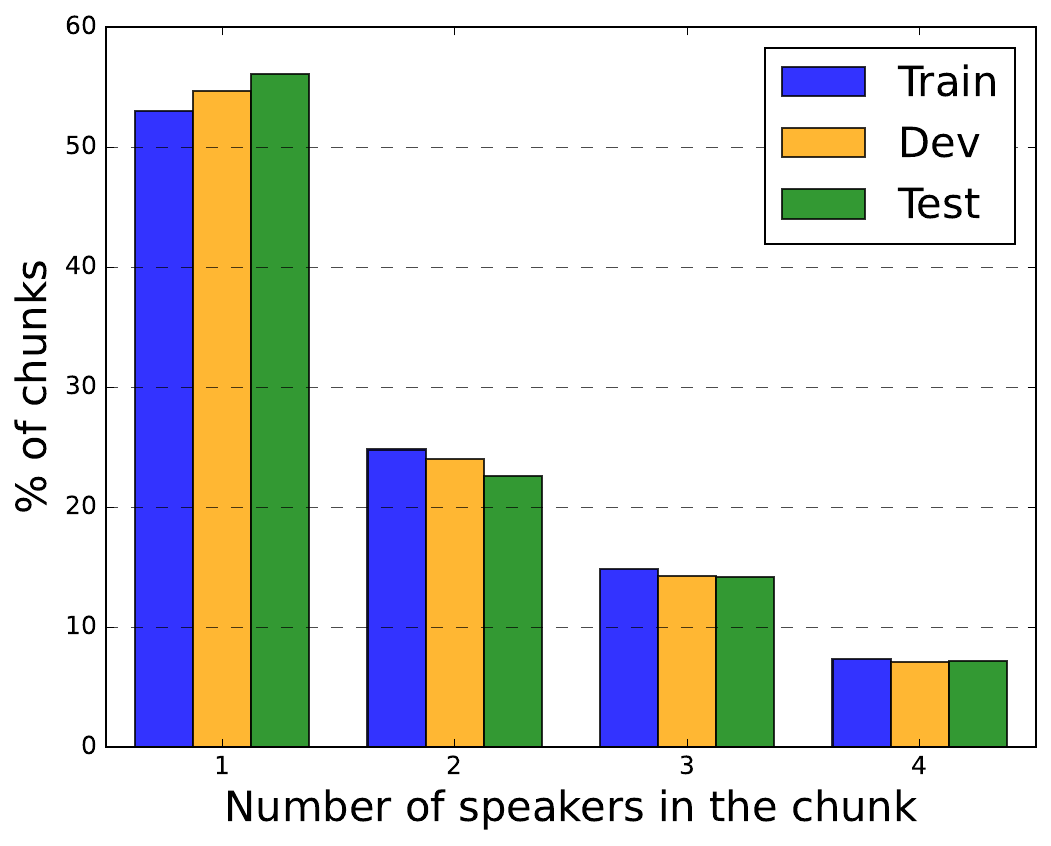}
    \label{fig:spk_chunk}
    \end{subfigure}
    \begin{subfigure}{0.49\linewidth}
    \centering
    \includegraphics[width=\linewidth]{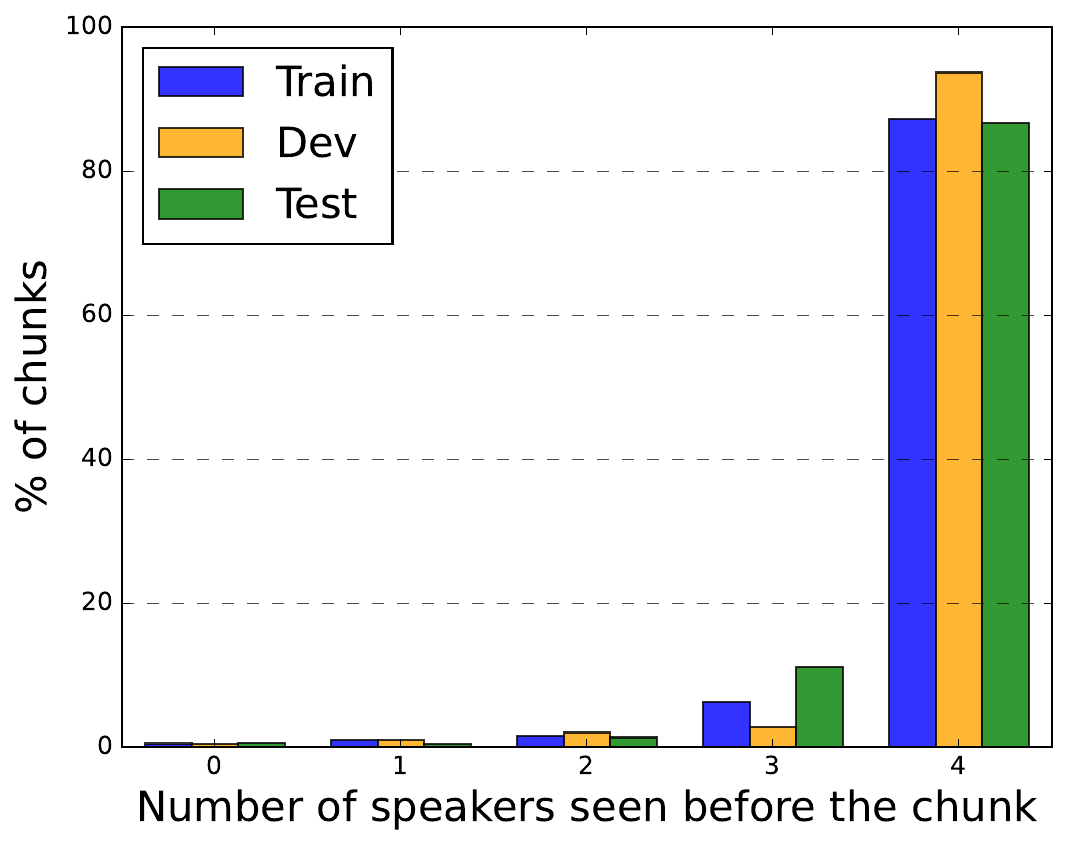}
    \label{fig:spk_seen}
    \end{subfigure}
\caption{Utterance group statistics of the AMI meeting corpus: (a) number of speakers in the group, and (b) number of speakers seen before the group.}
\label{fig:surt2_ami_stats}
\end{figure}

We conducted our experiments on synthetically mixed LibriSpeech utterances (called LSMix) and the AMI meeting corpus, and their statistics are shown in Table~\ref{tab:surt2_stats}\footnote{Detailed statistics for the AMI corpus are given in Table~\ref{tab:intro_stats}.}.
To create LSMix, we first cut LibriSpeech utterances at 0.2 second pauses, and then mixed speed-perturbed versions of these segments using Algorithm~\ref{alg:mixture}.
The resulting mixtures were 17s long on average, and contain 2--3 speakers and up to 9 turns of conversation.
We created \texttt{train} and \texttt{dev} splits of LSMix using the corresponding LibriSpeech partitions.
We used this evaluation data to perform ablations for developing the auxiliary speaker branch in SURT.

To train SURT models for AMI, we used synthetic mixtures of AMI and ICSI~\cite{Janin2003TheIM} utterances (known as AIMix) as described in Chapter~\ref{chap:surt}.
We first trained the models on 1841h of the AIMix data, and then adapted them on real AMI sessions.
For training with the speaker buffer, we randomly sampled $\tau$ from [64,96,128] frames for each batch, and chose $K_m$ to be less than or equal to 4 speakers.
From Fig.~\ref{fig:surt2_ami_stats}, it is clear that the during inference, most chunks will be processed with 4 speaker buffers.

\begin{table}[t]
\centering
\caption{Statistics of datasets used for evaluations. The overlap durations are in terms of fraction of total speaking time.}
\label{tab:surt2_stats}
\adjustbox{max width=\linewidth}{
\begin{tabular}{@{}lrrrrr@{}}
\toprule
\multirow{2}{*}{} & \multicolumn{2}{c}{\textbf{LSMix}} & \multicolumn{3}{c}{\textbf{AMI}} \\
\cmidrule(r{2pt}){2-3} \cmidrule(l{2pt}){4-6}
 & \textbf{Train} & \textbf{Dev} & \textbf{Train} & \textbf{Dev} & \textbf{Test} \\ 
\midrule
\textbf{Duration (h:m)} & 2193:57 & 4:19 & 79:23 & 9:40 & 9:03 \\
\textbf{Num. sessions} & 486440 & 897 & 133 & 18 & 16 \\
\textbf{Silence (\%)} & 3.4 & 3.2 & 18.1 & 21.5 & 19.6 \\
% \textbf{1-speaker (\%)} & 71.6 & 74.0 & 75.5 & 74.3 & 73.0 \\
\textbf{Overlap (\%)} & 28.4 & 26.0 & 24.5 & 25.7 & 27.0 \\
\bottomrule
\end{tabular}}
\end{table}

\subsection{Training details}
\label{sec:surt2_details}

The hyperparameters and training setup for the ASR branch follows Chapter~\ref{chap:surt}.
For the auxiliary speaker branch, we tried sequential and joint training strategies, as described in Section~\ref{sec:surt2_seq_vs_joint}.
For the former, the SURT model was trained for 40 epochs.
For the latter, the ASR branch was first trained for 30 epochs; it was then frozen and the speaker branch was trained for 20 epochs.
In all cases, the ASR transducer was initialized from a pre-trained transducer model, trained for 10 epochs on LibriSpeech.
We averaged model checkpoints from the last 5 epochs for inference, and used greedy decoding for reporting all results.
For evaluation on AMI, we initialized the masking network and the ASR branch using the parameters from the SURT model trained on LSMix.
We then trained this model in a sequential process, i.e., ASR branch followed by speaker branch, on AIMix followed by adaptation on real AMI training sessions.

\subsection{Evaluation}
\label{sec:surt2_evaluation}

For speaker-agnostic transcription, SURT was evaluated using the optimal reference combination word error rate (ORC-WER) metric, proposed independently in \cite{Sklyar2021MultiTurnRF} and \cite{Raj2021ContinuousSM}.
ORC-WER computes the minimum total WER obtained using the optimal assignment of reference utterances to the output channels.
In this chapter, since we have extended SURT to perform speaker-attributed transcription, we measure its performance using the concatenated minimum-permutation WER (cpWER)~\cite{Watanabe2020CHiME6CT}.
This metric finds the best permutation of reference and hypothesis speakers which minimizes the total WER across all speakers.

We also want to measure speaker attribution errors independently of transcription errors. 
The conventional metric for this is known as diarization error rate (DER), and measures the duration ratio of speaking time for which the predicted speakers do not match the reference speakers.
However, since SURT is a streaming model, the ASR tokens and the respective speaker labels may be emitted with some latency compared to their actual reference time-stamp.
This can artificially escalate the DER even when there are few speaker attribution errors.
To circumvent this issue, we report a word-level diarization error rate (WDER) inspired by \citet{Shafey2019JointSR}.
Originally, WDER was defined as the fraction of correctly recognized words which have incorrect speaker tags.
We modify the metric for SURT by using the ORC-WER reference assignment to identify the correct words and the speaker mapping from the cpWER computation to check for speaker equivalence.

% \begin{table}[t]
% \centering
% \caption{Speaker-agnostic ASR performance for RNN-T vs. HAT training, on LS-Mix \texttt{dev} measured by ORC-WER.}
% \label{tab:rnnt_vs_hat}
% \adjustbox{max width=\linewidth}{
% \begin{tabular}{@{}lrrrr@{}}
% \toprule
% \textbf{Loss} & \textbf{Ins.} & \textbf{Del.} & \textbf{Sub.} & \textbf{WER} \\ \midrule
% $\mathcal{L}_{\mathrm{rnnt}}$ & 0.72 & 2.76 & 5.11 & 8.59 \\
% %  & Beam search & 0.85 & 2.34 & 4.89 & 8.09 \\
% $\mathcal{L}_{\mathrm{hat}}$ & 0.99 & 2.22 & 5.32 & 8.53 \\
% %  & Beam search & 0.84 & 1.89 & 4.89 & 7.61 \\
% \bottomrule
% \end{tabular}
% }
% \end{table}

\section{Results \& Discussion}

Speaker attribution in SURT requires: (i) token-level speaker assignment within an utterance group, and (ii) speaker label reconciliation across utterance groups.
We will first demonstrate several ablation experiments on the LSMix \texttt{dev} set to identify optimal settings (for training strategy, auxiliary encoder position, and left context frames) for speaker attribution within utterance groups.
Once we have identified these settings, we will show the results of speaker-attributed transcription for AMI, where speaker attribution across utterance groups is performed by speaker prefixing.

\subsection{RNN-T vs. HAT for speaker-agnostic ASR}

Since our formulation requires replacing the conventional RNN-T loss, i.e. \eqref{eq:surt2_rnnt_softmax}, with the HAT loss given by \eqref{eq:surt2_hat}, we want to ensure that the speaker-agnostic ASR performance of the model does not degrade.
To verify this, we trained SURT (without an auxiliary branch) using $\mathcal{L}_{\mathrm{rnnt}}$ and $\mathcal{L}_{\mathrm{hat}}$ on the LSMix \texttt{train} set, and evaluated the resulting models on the \texttt{dev} set.
We found that the \textbf{HAT model obtained 8.53\% ORC-WER}, compared to 8.59\% using regular RNN-T.
The error breakdown showed marginally higher insertions but fewer deletions, which may be due to explicit modeling of the blank token.

\subsection{Sequential vs. joint training}
\label{sec:surt2_seq_vs_joint}

The auxiliary speaker branch of the SURT model can be trained in several ways, as shown in Table~\ref{tab:surt2_training_strategy}.
In ``sequential'' training, the main SURT model is first trained using \eqref{eq:surt2_heat} and then frozen while the auxiliary branch is trained using \eqref{eq:surt2_spk_hat}.
In ``joint'' training, the full model is trained from scratch with the multi-task objective.
Finally, we can combine the above approaches by first training the branches sequentially and then fine-tuning them jointly.
We found that \textbf{both sequential and joint training resulted in similar cpWER performance}, but joint training degrades WDER.
Furthermore, joint fine-tuning after sequential training degraded performance on both ASR metrics.
Since sequential training allows decoupling of ASR and speaker attribution performance, we used this strategy for the experiments in the remainder of this chapter.

\begin{table}[t]
\centering
\caption{Comparison of different training strategies for SURT with auxiliary speaker branch.}
\label{tab:surt2_training_strategy}
\adjustbox{max width=\linewidth}{
\begin{tabular}{@{}lrrr@{}}
\toprule
\textbf{Strategy} & \textbf{ORC-WER} & \textbf{WDER} & \textbf{cpWER} \\ \midrule
Sequential & 8.53 & \textbf{3.99} & 14.96 \\
Joint & \textbf{8.43} & 4.46 & \textbf{14.95} \\
Seq. + Joint & 9.17 & 4.25 & 15.33 \\
\bottomrule
\end{tabular}
}
\end{table}

\subsection{Auxiliary encoder position}

The input $\mathbf{h}_{0}^{\mathrm{aux}}$ to the auxiliary encoder is obtained from an intermediate representation of the main encoder.
We trained several SURT models with different positions for the auxiliary input, in order to find the optimal representation for the speaker branch, and the results are shown in Table~\ref{tab:surt2_aux_position}.
The models were trained sequentially and the ORC-WER was 8.53\% (same as earlier).
We found that both cpWER and WDER get progressively worse if we used representations from deeper layers, possibly because of loss in speaker information through the main encoder.
Interestingly, $\mathbf{h}_1$ (i.e. output of the first zipformer block) showed better performance than $\mathbf{h}_0$ (output from convolutional embedding layer).
We conjecture that the input to the \textbf{auxiliary encoder needs contextualized representations} since speaker labels need to be synchronized across the two branches.
These findings mirror recent studies showing that intermediate layers of the acoustic model are most suitable for extracting speaker information~\cite{Huang2023TowardsWE}.
Such analysis has also motivated ``tandem'' multi-task learning of ASR and speaker diarization using self-supervised encoders such as Wav2Vec 2.0~\cite{Zheng2022TandemMT}.

\begin{table}[t]
\centering
\caption{Speaker-attributed ASR performance on LSMix \texttt{dev} for different positions of the auxiliary encoder. $\mathbf{h}_l$ denotes the hidden representation at the $l^{\mathrm{th}}$ block of the main zipformer encoder, and $\mathbf{h}_0^{\mathrm{aux}}$ is the input to the auxiliary encoder.}
\label{tab:surt2_aux_position}
\adjustbox{max width=\linewidth}{
\begin{tabular}{@{}lcccrr@{}}
\toprule
$\mathbf{h}_{0}^{\mathrm{aux}}$ & \textbf{Ins.} & \textbf{Del.} & \textbf{Sub.} & \textbf{cpWER} & \textbf{WDER} \\ \midrule
$=\mathbf{h}_{0}$ & 3.34 & 6.04 & 7.28 & 16.66 & 5.36 \\
$=\mathbf{h}_{1}$ & \textbf{2.91} & \textbf{5.20} & \textbf{6.85} & \textbf{14.96} & \textbf{3.99} \\
$=\mathbf{h}_{2}$ & 4.58 & 6.77 & 8.24 & 19.59 & 6.73 \\
$=\mathbf{h}_{3}$ & 5.95 & 8.23 & 9.41 & 23.59 & 8.35 \\
\bottomrule
\end{tabular}
}
\end{table}

\subsection{Effect of left context}

The ASR encoder of the SURT model uses limited left context ($C_{\mathrm{left}}$=128 frames) in the self-attention computation during inference.
While ASR token prediction is usually a local decision, speaker label prediction requires looking at the full history in order to synchronize the relative FIFO labels.
We experimented with training and decoding with different histories, and the results are shown in Fig.~\ref{fig:surt2_left_context}.
For a model trained with infinite $C_{\mathrm{left}}$ (solid blue line), limiting it during inference quickly degraded WDER and cpWER performance.
When the model was trained with randomized $C_{\mathrm{left}}$ (solid green line), the degradation was less evident.
However, it was unable to make full use of infinite history at inference time, and only obtained a WDER of 6.12\%, versus 3.99\% for the model trained with infinite $C_{\mathrm{left}}$.
This indicates that using \textbf{infinite left context during training and inference is important} for the auxiliary speaker encoder.

\begin{figure}[tb]
    \begin{subfigure}{0.49\linewidth}
    \centering
    \includegraphics[width=\linewidth]{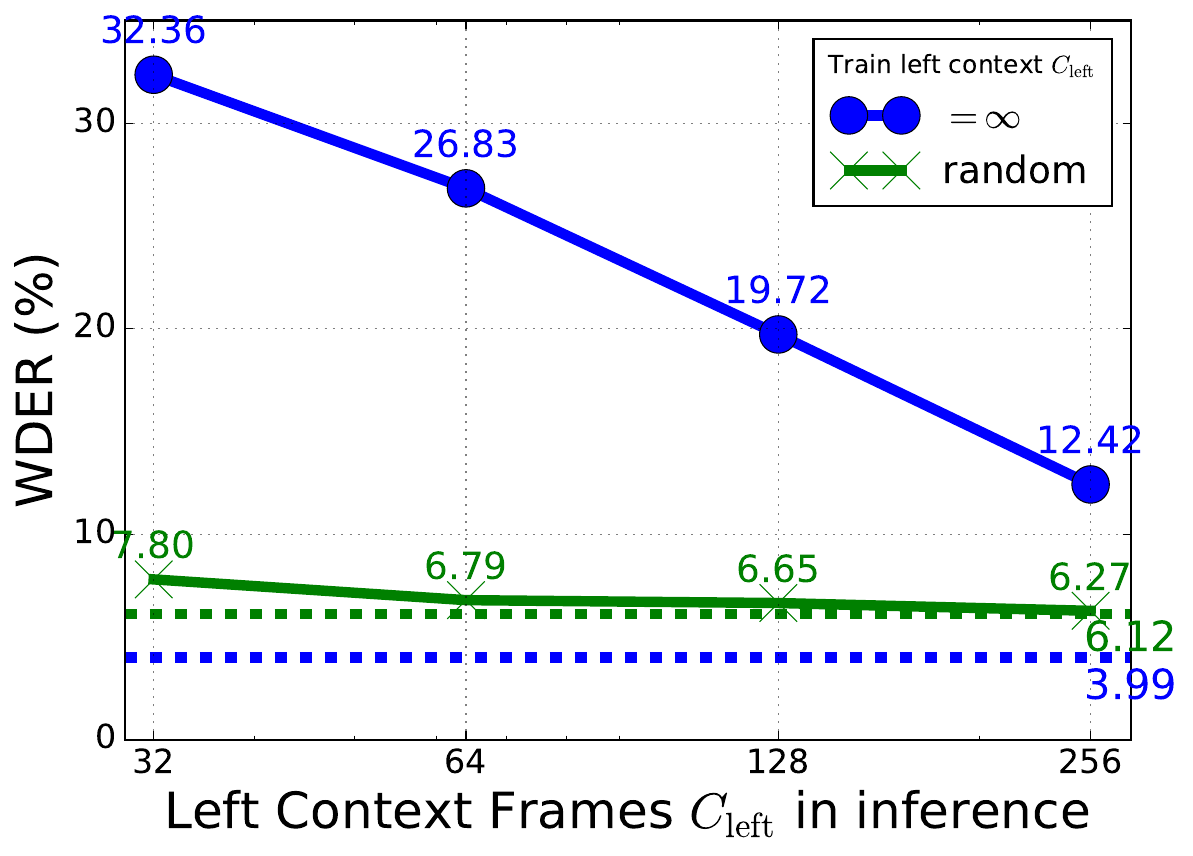}
    \label{fig:lc_der}
    \end{subfigure}
    \begin{subfigure}{0.49\linewidth}
    \centering
    \includegraphics[width=\linewidth]{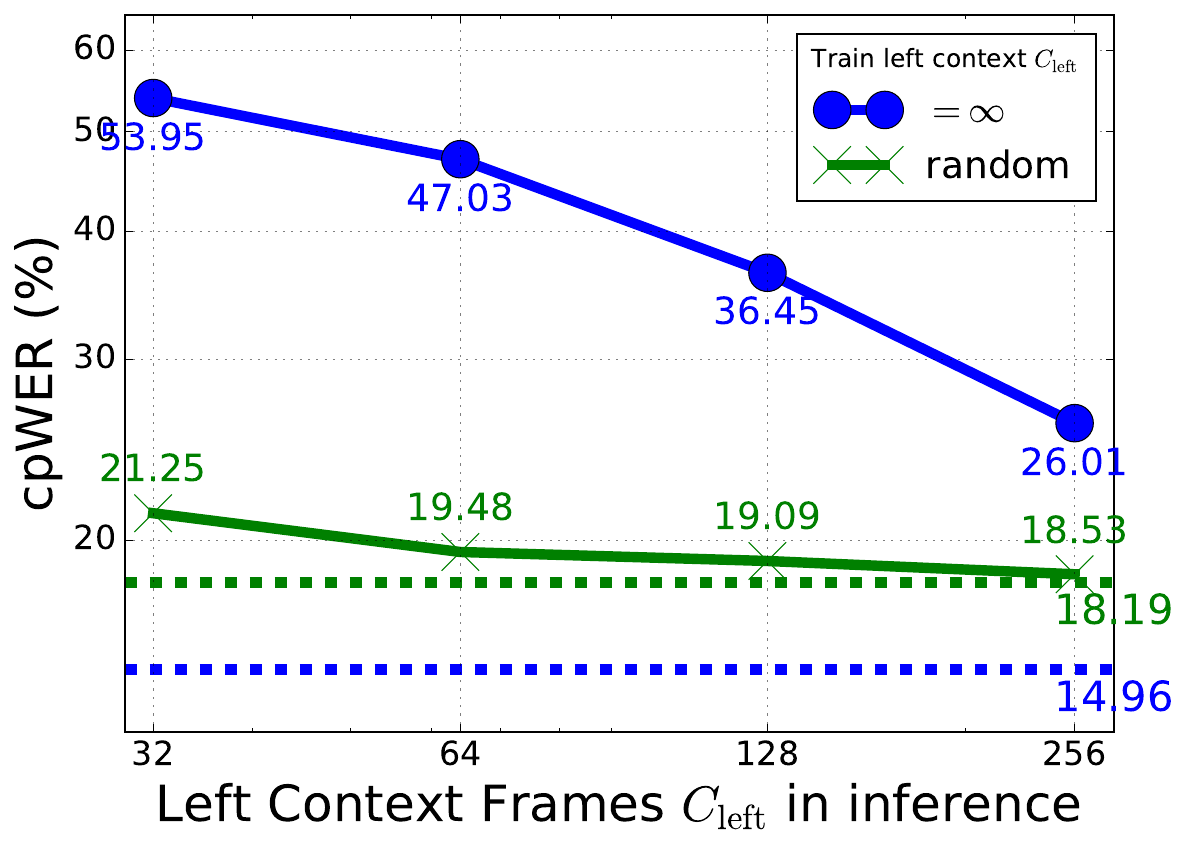}
    \label{fig:lc_wer}
    \end{subfigure}
\caption{Effect of auxiliary encoder left context on (a) WDER and (b) cpWER. Dotted lines show best performance using $\infty$ left context.}
\label{fig:surt2_left_context}
\end{figure}

\subsection{Utterance-group evaluation on AMI}

\begin{table}[t]
\centering
\caption{Performance of SURT models (with and without speaker prefixing) for different conditions on AMI \texttt{test} set, evaluated on utterance groups. ``ORC'' denotes the ORC-WER metric, and is the same for all models.}
\label{tab:surt2_ami_result}
\adjustbox{max width=\linewidth}{
\begin{tabular}{@{}lcccccccccc@{}}
\toprule
\multicolumn{2}{c}{\textbf{Prefix}} &
\multicolumn{3}{c}{\textbf{IHM-Mix}} & \multicolumn{3}{c}{\textbf{SDM}} & \multicolumn{3}{c}{\textbf{MDM}} \\
\cmidrule(r{5pt}){1-2} \cmidrule(l{2pt}r{2pt}){3-5} \cmidrule(l{2pt}r{2pt}){6-8} \cmidrule(l{4pt}){9-11}
\textbf{ID} & \textbf{Train/Decode} & ORC & WDER & cpWER & ORC & WDER & cpWER & ORC & WDER & cpWER \\
\midrule
\textbf{A} & \xmark ~/~  \xmark & 34.9 & \textbf{9.3} & \textbf{42.9} & 43.2 & \textbf{10.9} & \textbf{50.3} & 40.5 & \textbf{9.9} & \textbf{47.3} \\
\textbf{B} & \xmark ~/~ \cmark & 34.9 & 22.5 & 61.2 & 43.2 & 23.1 & 68.2 & 40.5 & 22.6 & 64.8 \\
% \cmark & \xmark & 34.9 & 11.2 & 44.2 & 43.2 & 11.5 & 50.6 & 40.5 & 11.6 & 49.0 \\
\textbf{C} & \cmark ~/~ \cmark & 34.9 & 14.0 & 49.9 & 43.2 & 16.3 & 58.9 & 40.5 & 15.5 & 56.0 \\
\bottomrule
\end{tabular}}
\end{table}

We evaluated the SURT model on different microphone settings of the AMI meeting corpus in the utterance-group scenario, and the results are shown in Table~\ref{tab:surt2_ami_result} in terms of ORC-WER, WDER, and cpWER.
Since we froze the ASR transducers for all models while training the speaker branch, all models emit identical ASR tokens, and as such, the ORC-WER is equal.
Across the board, performance \textbf{degraded from IHM-Mix to SDM} settings, which is expected since SDM contains far-field artifacts in addition to overlaps.
Beamforming with multiple microphones partially removes background noise and reverberations, thus providing a slightly easier condition than SDM.
For system A, which was trained and decoded without speaker prefixing, we obtained a cpWER of 46.8\%, on average across the three conditions.
When we used the same model for decoding with speaker prefixes (system B), the cpWER performance degraded by 38.2\% relative to the former.
Since the model has not seen short speaker buffers at train time, the auxiliary encoder is not adept at using these for generating the contextualized representations.

\begin{figure}[t]
    \centering
    \includegraphics[width=\linewidth]{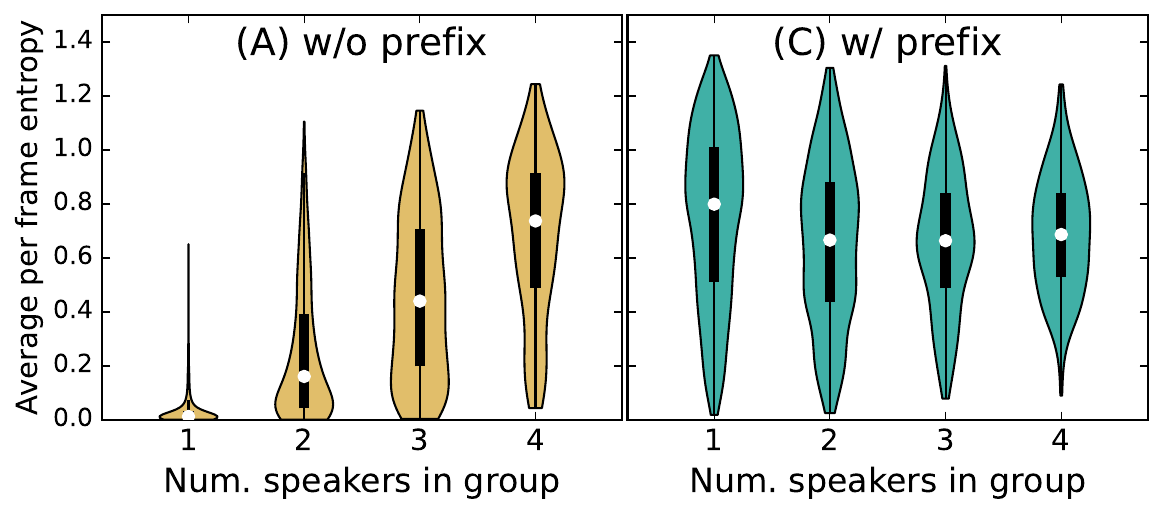}
    \caption{Average per-frame entropy for utterance groups with different number of speakers.}
    \label{fig:surt2_entropy}
\end{figure}

Next, we trained the same SURT model using speaker prefixing as described in Section~\ref{sec:surt2_details}, and found that it improved performance significantly due to matched train and test conditions.
Nevertheless, this model was \textbf{7-8\% worse than the original} model in terms of absolute cpWER performance.
To investigate this further, we computed the average framewise entropy over speaker labels for all utterance groups in the IHM-Mix \texttt{test} set, and grouped them by number of speakers in the group.
Fig.~\ref{fig:surt2_entropy} shows the distribution of these entropies for the SURT model with and without speaker prefixing.
We found that for the model without prefixing, the \textbf{entropy was very low for utterance groups with a single speaker}, and gradually increased with the number of speakers.
This indicates that the model was very confident in its prediction for few speaker cases.
The opposite trend was seen for the model with speaker prefixing, where the entropy was highest for the single-speaker case.
This is because for each frame, the model needs to decide which of the 4 prefixed speakers the frame should be assigned to, which may result in low confidence of prediction.
An implication of this is shown below, where the model with speaker prefixing is much more likely to predict speaker changes (either correct e.g., for \_\texttt{OK}), or incorrect (e.g., for \texttt{OR}).
\begin{figure}[htp]
    \centering
    \includegraphics[width=0.6\linewidth]{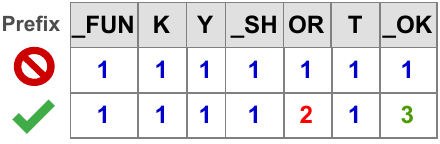}
\end{figure}

In general, we found that the performance of all models  \textbf{gets progressively worse} as the number of speakers in the group increases, as shown in Table~\ref{tab:surt2_ami_num_spk} for system A.
Interestingly, the degradation in WDER was small compared to that in the speaker counting accuracy.
This may be because of several utterance groups where some speakers participate with just a few words, which may be hard for the system to identify, but do not contribute much in overall speaker attribution error.

\begin{table}[t]
\centering
\caption{Breakdown of model (A)'s performance on IHM-Mix \texttt{test} set by number of speakers in the utterance group.}
\label{tab:surt2_ami_num_spk}
\adjustbox{max width=\linewidth}{
\begin{tabular}{@{}lrrrrr@{}}
\toprule
\textbf{\#spk} & \textbf{1} & \textbf{2} & \textbf{3} & \textbf{4} & \textbf{Avg.} \\ \midrule
WDER ($\downarrow$) & 0.1 & 3.4 & 13.0 & 23.9 & 9.3 \\
Count. ($\uparrow$) & 98.6 & 61.9 & 26.8 & 44.0 & 75.9 \\
cpWER ($\downarrow$) & 17.2 & 32.4 & 51.1 & 63.6 & 42.9 \\
\bottomrule
\end{tabular}}
\end{table}

\subsection{Full-session evaluation on AMI}

Finally, we performed inference on full AMI \texttt{test} sessions and the corresponding cpWERs are reported n Table~\ref{tab:surt2_ami_result2}.
Computing the ORC-WER and WDER for full sessions was not feasible since their computational complexity depends on the number of segments in the reference, which may be very large for the whole session.
Nevertheless, we measured the ORC-WER for the same outputs as evaluated on utterance groups, and these numbers indicate that transcription performance remained consistent across all models. 
This is expected since the same ASR branch is used in all models and the only difference is in how the speaker branch is trained.

First, we see that the model without speaker prefixing obtained very high error rates, since it failed at correctly reconciling speaker labels across different utterance groups.
With speaker prefixing, we obtained \textbf{15.1\% relative cpWER improvement} on average across the mic settings.
For the speaker prefixing, we trained and evaluated the model using $\tau$ of 128 frames or 1.28s per speaker.

In a meeting transcription setup, since the participants are known before-hand, we can usually obtain an enrollment utterance for each speaker.
Instead of selecting speaker prefixes from previous chunks, if we select them from these enrollment utterances, we obtain a further \textbf{relative cpWER improvement of 29.2\%}, on average.
We conjecture that when enrollment utterances are not used, speaker attribution errors in earlier chunks can adversely impact performance on current chunk, since the buffer frames are used to guide the relative order.
Nevertheless, there still exists a significant gap of about 10--12\% absolute cpWER between full session evaluation and utterance group evaluation (shown in Table~\ref{tab:surt2_ami_result}).

\begin{table}[t]
\centering
\caption{Full-session evaluation cpWER (\%) on AMI \texttt{test} set. We also report the ORC-WER for the same outputs (evaluated on utterance groups).}
\label{tab:surt2_ami_result2}
\adjustbox{max width=\linewidth}{
\begin{tabular}{@{}lrrrrrr@{}}
\toprule
\multirow{2}{*}{\textbf{Method}} & \multicolumn{2}{c}{\textbf{IHM-Mix}} & \multicolumn{2}{c}{\textbf{SDM}} & \multicolumn{2}{c}{\textbf{MDM}} \\
\cmidrule(r{5pt}){2-3} \cmidrule(l{2pt}r{2pt}){4-5} \cmidrule(l{4pt}){6-7}
 & \textbf{ORC-WER} & \textbf{cpWER} & \textbf{ORC-WER} & \textbf{cpWER} & \textbf{ORC-WER} & \textbf{cpWER} \\
\midrule
w/o prefix & 34.84 & 100.11 & 43.15 & 97.15 & 40.34  & 96.26 \\
w/ prefix & 34.84 & 82.77 & 43.15 & 83.94 & 40.34 & 82.28 \\
~~+ enrollment & 34.84 & 53.76 & 43.15 & 62.22 & 40.34 & 60.27  \\
\bottomrule
\end{tabular}}
\end{table}

\section{Conclusion}

The SURT framework allows continuous, streaming recognition of multi-talker conversations, but it may only be used for speaker-agnostic transcription.
In this chapter, we showed how to perform streaming word-level speaker labeling with SURT, thus enabling speaker-attributed transcription using the same model.
We achieved this by adding an auxiliary speaker encoder to the recognition component of the model, and used the same two-branch strategy to handle overlapping speech.
We solved the problem of synchronization between the ASR and speaker branch outputs by factoring out the blank logit and sharing it between the branches.
Since the model predicts relative speaker labels in FIFO order, reconciling the labels across utterance groups in a recording becomes a challenge.
We showed that a simple strategy of prefixing high-confidence speaker frames for the recognized speakers can partially alleviate this problem, but it would require further investigation to bring session-level error rates closer to those for utterance groups.

% \cleardoublepage

\chapter{Conclusion and Future Work}
\label{chap:conclusion}

Our objective in this work was to solve the problem of speaker-attributed multi-talker speech recognition.
Towards this goal, we proposed two alternative perspectives --- ``modular'' and ``end-to-end'' --- which formed the first and second parts of this dissertation, respectively.
The modular perspective, described in Chapters \ref{chap:oasc}, \ref{chap:doverlap}, \ref{chap:gss}, and \ref{chap:modular}, contained a pipeline of independently-trained components, each tackling a different sub-task, such as segmentation or enhancement.
The end-to-end perspective, comprising Chapters \ref{chap:surt} and \ref{chap:surt2}, proposed a jointly trained model that performs streaming speaker-attributed transcription of multi-talker conversations.

\section{Summary of contributions}

For the modular solution, we recognized that accurate identification of homogeneous speaker segments is crucial for good transcription performance.
To achieve this, we proposed an overlap-aware speaker diarization system in Chapter~\ref{chap:oasc} by incorporating overlap assignment into a spectral clustering system.
Since such clustering-based systems are robust and ubiquitous for the task, our method provides a convenient way to extend current systems to identify overlapping speaker segments.
On the AMI data, we were able to improve DER performance by 12\% relative without requiring supervised training.
Nevertheless, our error analysis showed that since speaker embedding extractors are trained on single-speaker utterances, they may result in high confusion when used to diarize overlapping speech.
Such limitations have been instrumental in the recent rise of end-to-end neural diarization (EEND) systems~\cite{Fujita2020EndtoEndND}.

While EEND-style models show good performance on overlapping speech, they usually require matched training on synthetic mixtures, and are not as robust to large variability in number of speakers.
Recognizing that speech and natural language tasks often benefit from ensemble techniques, we proposed the DOVER-Lap algorithm in Chapter~\ref{chap:doverlap} for combining outputs from diverse diarization systems.
We described several label mapping techniques inspired by approximation algorithms, and showed that the resulting DOVER-Lap can improve DER significantly over the single-best system.
Since its publication, DOVER-Lap has become the standard tool for ensembling diarization systems in challenges such as DIHARD, CHiME, and VoxSRC, and the associated \texttt{dover-lap} package has been downloaded over 15000 times from PyPi\footnote{\url{https://www.pepy.tech/projects/dover-lap}}.
In the context of our modular system, DOVER-Lap provides a convenient method to perform late fusion by ensembling the outputs of different channels.
We showed that such a combination outperforms a conventional delay-and-sum beamformer on the AMI data.

In Chapter~\ref{chap:gss}, we turned our attention from \textit{identifying} speaker segments to \textit{extracting} the corresponding target signals, a task known as target-speaker extraction (TSE).
We revamped the guided source separation (GSS) algorithm, first proposed for the CHiME-5 challenge, through a GPU-accelerated implementation inspired by modern deep learning pipelines.
Our new implementation provided 300x faster inference, and the speed-up allowed us to perform ablations over several factors that impact GSS performance.
We showed that the number of microphones is most critical for improved TSE performance, as measured by downstream WERs.
Our open-source implementation\footnote{\url{https://github.com/desh2608/gss}} has been adopted into the community baselines for the CHiME-7 DASR~\cite{Cornell2023TheCD} and the M2MeT 2.0~\cite{Liang2023TheSM} challenges.
While the unsupervised nature of GSS makes it an attractive choice for different array configurations, it is limited by the requirement of multi-channel inputs.
It is also an offline algorithm by design due to its reliance on pre-computed speaker activities.
As a result, there has been increased focus on single-channel neural TSE methods such as SpeakerBeam~\cite{molkov2019SpeakerBeamSA,molkov2023NeuralTS}.

In Chapter~\ref{chap:modular}, we combined our proposed diarization and TSE components with a transducer-based single-speaker ASR model to obtain the meeting transcription pipeline.
Through a probabilistic formulation of the speaker-attributed ASR problem, we showed that this pipeline can be derived as an approximate solution by optimizing in parts and making several conditional independence assumptions.
Although the resulting system shows good transcription performance on several meeting benchmarks such as LibriCSS, AMI, and AliMeeting, we found that it suffers from problems such as error propagation through the modules.
Nevertheless, we believe that our publicly available implementation~\footnote{\url{https://github.com/desh2608/icefall/tree/multi_talker}} provides a strong baseline for comparing future work on multi-talker ASR.
This chapter wrapped up the first part of the dissertation.

In the second part, we began in Chapter~\ref{chap:surt} by considering end-to-end multi-talker ASR methods from first principles, and showed that overlapping speech can be transcribed on to a fixed number of output channels through a graph coloring approach.
This resulted in the Streaming Unmixing and Recognition Transducer (SURT) model, which integrates an implicit separation component with neural transducers to address multi-talker ASR.
Since training such models is challenging (particularly on academic computing resources), we proposed several improvements in network architecture, loss function, training data simulation, and training strategies.
Our final SURT model demonstrated strong ORC-WER performance on AMI and ICSI meeting benchmarks.

With the backbone in place, we extended the SURT model for joint speaker attribution in Chapter~\ref{chap:surt2}.
By adding an auxiliary speaker transducer and sharing the blank logit between the ASR and speaker branches through HAT-style factorization, we were able to synchronize the ASR token emissions with the speaker label prediction, thus enabling streaming speaker attribution of the tokens.
Finally, we used speaker prefixing to solve the problem of relative speaker permutation across different utterance groups in a single session, and showed that using enrollment utterances for this purpose significantly improves cpWER performance.

While we demonstrated that it is possible to perform streaming, speaker-attributed transcription with SURT, the results indicate that there is a large room for improvement.
There are primarily three sources of errors for SURT in its current formulation.
First, the ORC-WER for overlapping mixtures remains significantly higher than that for single-talker audio, indicating deficiencies in the masking network.
In our models, we used simple dual-path LSTMs for this component, but these could be replaced with transformer architectures, which have shown strong performance recently in speech separation~\cite{Subakan2021AttentionIA}.
Pre-training of the masking network and the use of multi-channel inputs to improve separation performance through beamforming~\cite{Zhang2021AllNeuralBF} may be other sources for improvement.
Next, using self-supervised encoders in the ASR transducer may be considered low-hanging fruit to improve transcription performance~\cite{Chen2021WavLMLS}.
In this regard, such encoders cannot be plugged directly into the SURT model, since they are usually full-context (and therefore, non-streaming) and also computationally expensive, which makes them difficult to be deployed in real settings.
A better use of these models may be in a teacher-student setting, i.e., the SURT recognition branch may be trained to generate representations similar to those obtained by the SSL encoder.
Finally, while the SURT model demonstrates good speaker attribution within the utterance group (as measured by WDER), speaker permutation across different utterance groups is still a challenge.
Training and inference with larger prefixes may be an easy solution to this issue, but comes at the cost of increased computations.
A better solution may be to have an online speaker inventory containing fixed representations, but this would require significant modeling changes~\cite{Han2020ContinuousSS}.

Nevertheless, our public implementation of SURT available in the \texttt{icefall} framework is the first such open-source streaming multi-talker ASR system.
We are positive that this will encourage the community to build upon our work in one or more of the directions described above.

\section{Other open-source work}

Besides the new methods and associated software described above, this dissertation has also resulted in contributions to the development of the next-generation Kaldi set of tools, as part of the National Science Foundation
CCRI program via Grant No. 2120435.
In summary, this set of tools comprises four separate but related packages, as follows.
\begin{enumerate}
    \item \textbf{Lhotse:} Lhotse contains recipes for data preparation into standard manifests, and tools for PyTorch-compatible datasets, samplers, and data-loaders.

    \item \textbf{k2:} It contains implementations of popular ASR loss functions such as CTC, RNN-T, etc., along with fast WFST-based beam search.

    \item \textbf{Icefall}: It is a collection of Python scripts containing training and decoding pipelines for common ASR benchmarks.

    \item \textbf{Sherpa}: It is an open-source STT inference framework focusing on deployment.
\end{enumerate}

In order to train and evaluate our methods on a large variety of meeting benchmarks, we implemented Lhotse data preparation recipes for LibriCSS, AMI, ICSI, AliMeeting, and CHiME-6, among others.
These recipes have also been used as part of our \href{https://github.com/desh2608/diarizer}{\texttt{diarizer}} and \href{https://github.com/desh2608/gss}{\texttt{gss}} packages.
We also implemented several tools in Lhotse, such as for:
\begin{itemize}
    \item incorporating and manipulating alignments in \href{https://github.com/lhotse-speech/lhotse/pull/304}{\#304};
    \item simulating far-field audio by convolving with room impulse responses in \href{https://github.com/lhotse-speech/lhotse/pull/477}{\#477};
    \item creating synthetic meeting-style mixtures with pause/overlap statistics learned from real meetings in \href{https://github.com/lhotse-speech/lhotse/pull/951}{\#951};
    \item new PyTorch-style dataset for SURT, and samplers for batched inference of long-form multi-talker ASR in \href{https://github.com/lhotse-speech/lhotse/pull/929}{\#929}.
\end{itemize}

We implemented the hybrid auto-regressive transducer (HAT)~\cite{Variani2020HybridAT} objective in k2, since it is required for synchronization of branches in the SURT model (Chapter~\ref{chap:surt2}).
Besides its utility for SURT, the blank factorization in HAT can also be used for text-only adaptation of transducer-based ASR~\cite{Meng2022ModularHA} and for joint language identification (similar to joint speaker attribution in SURT).
We also developed single-talker ASR recipes in Icefall for AMI, ICSI, and AliMeeting, since these models are used in the modular system.
These are zipformer-based transducer models~\cite{Yao2023ZipformerAF} trained using multi-condition training, and the same model works across close-talk and far-field microphone settings.
Among systems that do not use any external corpora, our models obtain state-of-the-art performance on these far-field ASR benchmarks.

% \cleardoublepage

\appendix

\chapter{Dataset Details}
\label{chap:appendix_data}

\section{LibriCSS}

Originally proposed in \citet{Chen2020ContinuousSS}, LibriCSS consists of multi-channel audio recordings of 8-speaker \textit{simulated} conversations that were created by combining utterances from the LibriSpeech \texttt{test-clean} set~\cite{Panayotov2015LibrispeechAA} and playing in real meeting rooms. 
It comprises 10 one-hour long sessions, each made up of six 10-minute ``mini sessions'' that have different overlap ratios: 0L, 0S, OV10, OV20, OV30, and OV40.
Here, 0L and 0S refer to sessions containing no overlaps with long and short silence between utterances, respectively, while the others denote overlap ratio in the sessions.
The ``meetings'' were recorded using a 7-channel circular array microphone, containing 1 center mic and 6 mics around it.
The IHM setting for this dataset will refer to the corresponding original LibriSpeech audio for the utterances, and IHM-Mix will refer to its digitally mixed version.
The original corpus did not provide any splits, but we will use session 0 as the \texttt{dev} set and the remaining for \texttt{test}~\cite{Raj2020IntegrationOS}.

LibriCSS provides a controlled test-bed due to several overlap settings and real far-field artifacts since it is replayed in real meeting rooms.
At the same time, it foregoes the lexical challenges associated with multi-party conversations due to its use of LibriSpeech utterances for the source.
For these reasons, it has been widely used for evaluating multi-talker methods, following the increasing popularity of LibriSpeech for ASR benchmarking.
In this dissertation, we will use this dataset to validate several of our methods and to perform ablation experiments.

\section{AMI}

AMI (Augmented Multiparty Interactions) is a well-known dataset in the field of multimodal signal processing and human-computer interaction. 
It was created by \citet{Carletta2005TheAM} to support research on automatic analysis of meetings, including speech and non-verbal communication.
Approximately 100 hours of data was collected in real office meeting environments and includes both audio and video recordings.

The meetings contain 4 or 5 speakers per session, and were recorded on close-talk (headset and lapel) microphones, as well as 2 linear arrays each containing 8 microphones.
AMI provides both headset and lapel microphone recordings, but we use the headset ones for the IHM condition.
The beamformed MDM setting uses officially provided beamformed recordings from the first linear array~\cite{Mir2007AcousticBF}.
The official AMI documentation\footnote{\url{http://groups.inf.ed.ac.uk/ami/corpus/datasets.shtml}} recommends three different data partitions: \texttt{scenario-only}, \texttt{full-corpus}, and \texttt{full-corpus-asr}, based on the task that the data is used for.
In this dissertation, we will always use the \texttt{full-corpus-asr} partition, since it has a larger training set and the speakers in the dev and test sets are unseen.
We use the latest version of the official annotations (1.6.2), since the version used in Kaldi recipes (1.6.1) is known to have alignment and annotation issues). 

AMI has become a popular benchmark for multi-talker tasks such as speaker diarization, overlap detection, target-speaker ASR, and so on, since it provides real conversational data along with word-level annotations.
From Table~\ref{tab:intro_stats}, we see that the sessions contain silence for approximately 20\% of the duration.
Of the speaking time duration, overlapped speech accounts for roughly 20\% of the time, with most of it being 2-speaker overlaps.
The corpus is partitioned into train, dev, and test splits with a ratio of 80:10:10.
Our measure of success in this dissertation will be based on performing speaker-attributed transcription accurately on this benchmark.

\section{ICSI}

The ICSI Meeting Corpus is a collection of 75 meetings held at the International Computer Science Institute in Berkeley during the years 2000-2002. 
The meetings range in length from 17 to 103 minutes, but generally run just under an hour each, resulting in a total of 72 hours of natural, meeting-style overlapped speech.
The meetings were simultaneously recorded using close-talk microphones for each speaker, as well as six table-top microphones: 4 high-quality omni-directional PZM microphones arrayed down the center of the conference table, and 2 inexpensive microphone
elements mounted on a mock PDA.
For the SDM setting, we use the third PZM microphone.
The data was collected at 48 kHz sample-rate, down-sampled on-the-fly to 16 kHz. 
All meetings were recorded in the same (roughly, 13 x 25 foot) instrumented meeting room.
The meeting room contains whiteboards along three walls and is equipped with projection equipment; people writing on whiteboards or projecting slides can occasionally be heard during these recordings.
Meetings involved anywhere from 3 to 10 participants, averaging 6.
There are a total of 53 unique speakers in the corpus.
The original data did not provide any partitions, so we use the speaker-disjoint partitions suggested in \citet{Renals2014NeuralNF}.

From Table~\ref{tab:intro_stats}, we can see that the ICSI data contains a larger proportion of silence duration, and the overlapped speech ratio is significantly smaller than AMI.
This characteristic makes it easier to transcribe ICSI than AMI, and we use this as an additional benchmark in our experiments.
Furthermore, while AMI mostly contains 4 speakers in each meeting, ICSI contains between 3 to 10 participants, which makes speaker counting slightly more challenging for this dataset.

\section{AliMeeting}

The AliMeeting Mandarin corpus, originally designed for ICASSP 2022 Multi-channel Multi-party Meeting Transcription Challenge (M2MeT), is recorded from real meetings, including far-field speech collected by an 8-channel microphone array as well as
near-field speech collected by each participants' headset microphone~\cite{Yu2021M2MetTI}. 
The dataset contains 118.75 hours of speech data in total, divided into 104.75 hours for training, 4 hours for validation, and 10 hours as test, according to M2MeT challenge arrangement. 
Specifically, the partitions contain 212, 8 and 20 meeting sessions respectively, and each session consists of a 15 to 30-minute discussion by 2-4 participants. 
AliMeeting covers a variety of aspects in real-world meetings, including diverse meeting rooms, various number of meeting participants and different speaker overlap ratios.

We use this dataset as an additional, non-English benchmark to evaluate some of our proposed methods.
Since it contains fewer speakers than all the previous datasets, it may be easier in terms of speaker counting.
However, from Table~\ref{tab:intro_stats}, we note that AliMeeting contains significantly higher overlapping speech than AMI or ICSI, and so it can be a challenging task for diarization and ASR.

\chapter{Proofs for DOVER-Lap Label Mapping Algorithms}

\section{Proof of Lemma~\ref{lemma:turan}}
\label{sec:appendix_turan}

\begin{proof}
Given $\mathcal{G}$, if all independent sets $U_k$ contain exactly $C$ nodes, then the statement is trivially true. 
Otherwise, for any $U_k$ containing $c_k < C$ nodes, we add $C - c_k$ dummy nodes to $U_k$, and connect each dummy node with vertices in all other independent sets. 
We assign weight 0 to all the newly added edges. Let us call this complete graph $\mathcal{G}^{\prime}$. 
It is easy to see that $\mathcal{G}^{\prime}$ is $T(CK,K)$ since it contains $CK$ vertices divided equally into $K$ subsets (each independent set is such a subset).
Now, we only need to show that any solution $\Phi$ for the graph $\mathcal{G}$ is equivalent to some solution $\Phi^{\prime}$ to $\mathcal{G}^{\prime}$. 

Suppose $\Phi$ maximizes the objective in \eqref{eq:objective} for graph $\mathcal{G}$. 
By the pigeon-hole principle, $\Phi$ must contain exactly $C$ cliques, and each clique has size at most $K$. 
We can extend $\Phi$ to $\Phi^{\prime}$ by incrementally adding dummy nodes (and associated edges) to the cliques until they become maximal. 
Since all the added edges are zero-weighted, $w(\Phi^{\prime}) = w(\Phi)$.

Similarly, if we have a solution $\Phi^{\prime}$ for $\mathcal{G}^{\prime}$, we can obtain a solution $\Phi$ for $\mathcal{G}$ with the same total weight by simply removing the dummy nodes from all the cliques in $\Phi^{\prime}$.
\end{proof}

\section{Proof of Theorem~\ref{thm:greedy}}
\label{sec:appendix_greedy}

\begin{proof}
By Lemma~\ref{lemma:turan}, W.L.O.G, suppose $\mathcal{G}$ is complete. 
First we will show, by induction on $K$, that $w(\Phi) \geq \frac{w(\mathcal{G})}{C}$. 
For $K=2$, $\mathcal{G}$ is bipartite, so the Hungarian method provides an optimal solution, i.e., $\Psi$ is a maximum-weighted matching of $\mathcal{G}$. 
Since there are $C!$ perfect matchings in $\mathcal{G}$ and each edge appears in exactly $(C-1)!$ of these matchings, we have
\begin{equation*}
\text{average weight of matching} = \frac{\sum_{e \in E}w(e)(C-1)!}{C!} 
    = \frac{\sum_{e \in E}w(e)}{C} = \frac{w(\mathcal{G})}{C}.    
\end{equation*}

Since the Hungarian method returns the maximum-weight matching, we have
$$ w(\Phi) \geq \frac{w(\mathcal{G})}{C}.$$

For the inductive case, suppose the statement holds for some $K-1$. 
Let $\psi_1$ be the matching in the first iteration (i.e., between $U_1$ and $U_2$), and $\Phi^{\prime}$ be the remaining matching. 
Let $\mathcal{G}^{\prime}$ be the graph obtained after the first merge operation.
Then, by applying the statement on $\psi_1$ and $\Phi^{\prime}$, we have
\begin{align*}
w(\Phi) &= w(\psi_1) + w(\Phi^{\prime}) 
        \geq \frac{1}{C} \sum_{e\in [U_1,U_2]}w(e) + \frac{w(\mathcal{G}^{\prime})}{C} \\
        &= \frac{1}{C} \left( \sum_{e\in [U_1,U_2]}w(e) + \sum_{e\notin [U_1,U_2]}w(e) \right) 
        = \frac{w(\mathcal{G})}{C}.
\end{align*}

Now suppose $\Phi^*$ is an optimal solution. Then, since $w(\Phi^*)\leq w(\mathcal{G})$, we have
\begin{align*}
\frac{w(\Phi)}{w(\Phi^*)} &\geq \frac{w(\Phi)}{w(\mathcal{G})} \geq \frac{1}{C}  \\
\implies w(\Phi) &\geq \frac{w(\Phi^*)}{C}.
\end{align*}
Hence, the algorithm is a $\frac{1}{C}$-approximation.
\end{proof}

\section{Proof of Lemma~\ref{lemma:neighbor_weights}}
\label{sec:appendix_neighbor}

\begin{proof}
Consider some neighbor $\Phi'$ of $\Phi$, obtained by swapping the exchangeable pair $(u_i,u_j)$ in the independent set $U_k$. Suppose that $u_i$ and $u_j$ were originally in the cliques $V_i$ and $V_j$, respectively. This means that after swapping, in the partition $\Phi'$, $u_i$ is in $V_j$ and $u_j$ is in $V_i$.

If we consider the changes from $\Phi$ to $\Phi'$ in terms of edge weights, we removed the edges from $u_i$ to all other vertices in $V_i$ (and similarly from $u_j$ to all other vertices in $V_j$), and added edges from $u_i$ to other vertices in $V_j$ (and similarly for $u_j$ to other vertices in $V_i$). By slight abuse of notation, suppose $w(V_i\setminus\{u_i\})$ denotes the sum of all edge weights in $V_i$ excluding the node $u_i$. Then, we have

\begin{align*}
    w(\Phi') - w(\Phi) &= \left( \sum_{u \in V_j\setminus\{u_j\}}w(u_i, u) + \sum_{u \in V_i\setminus\{u_i\}}w(u_j, u) \right) \\
    &- \left( \sum_{u \in V_i\setminus\{u_i\}}w(u_i, u) + \sum_{u \in V_j\setminus\{u_j\}}w(u_j, u) \right)
\end{align*}

Summing over all neighbors $\Phi'$ of $\Phi$, we get

\begin{align*}
    \sum_{\Phi' \in N(\Phi)} \left( w(\Phi') - w(\Phi) \right) &= \sum_{\Phi' \in N(\Phi)} \biggl( \sum_{u \in V_j\setminus\{u_j\}}w(u_i, u) + \sum_{u \in V_i\setminus\{u_i\}}w(u_j, u) \\ 
    &- \sum_{u \in V_i\setminus\{u_i\}}w(u_i, u) - \sum_{u \in V_j\setminus\{u_j\}}w(u_j, u) \biggr) \\
    &= 2\sum_{c_i,c_j\in [C]}w(E(V_{c_i}, V_{c_j})) - 2(C-1) \sum_{c \in [C]}w(E(V_c)),
\end{align*}
where $E(V_{c_i}, V_{c_j})$ denotes the set of edges going between cliques $V_{c_i}$ and $V_{c_j}$ in the partition $\Phi$, and $E(V_c)$ denotes the edges inside the clique $V_c$. Since $w(\Phi) = \sum_{c\in[C]}w(E(V_c))$ and $w(\mathcal{G}) - w(\Phi) = \sum_{c_i,c_j\in [C]}w(E(V_{c_i}, V_{c_j}))$, we get

\begin{align*}
&\sum_{\Phi' \in N(\Phi)} \left( w(\Phi') - w(\Phi) \right) = 2 \left( w(\mathcal{G}) - w(\Phi) \right) - 2(C-1) w(\Phi) \\ 
\implies &\sum_{\Phi' \in N(\Phi)} w(\Phi') -  \sum_{\Phi' \in N(\Phi)} w(\Phi) = 2 w(\mathcal{G}) - 2(C-1) w(\Phi) \\ 
\implies &\sum_{\Phi' \in N(\Phi)} w(\Phi') -  K \cdot \binom{C}{2} w(\Phi) = 2 w(\mathcal{G}) - 2(C-1) w(\Phi)
\end{align*}

Finally, subtracting $\frac{w(\mathcal{G})}{C}$ from both sides and dividing by $|N(\Phi)|$ (from Lemma~\ref{lemma:neighbors}, we obtain
$$ \frac{1}{|N(\Phi)|}\sum_{\Phi' \in N(\Phi)}\left( w(\Phi') - \frac{w(\mathcal{G})}{C} \right) = \left(1 - \frac{2C}{|N(\Phi)|} \right) \left( w(\Phi) - \frac{w(\mathcal{G})}{C}\right). $$
\end{proof}

\section{Label Mapping with Randomized Local Search}
\label{sec:appendix_randomized}

In this section, we will develop a randomized algorithm using local search which obtains a better approximation in expectation compared with the Hungarian algorithm. 
We use the same graphical formulation that was developed in Section~\ref{sec:dl_mapping}, so we do not repeat the construction here. 
First, we will present a deterministic local search algorithm to understand the concept of local improvements in this setting, and then extend it to a randomized version.

\subsection{Preliminary: deterministic local search}

Again, W.L.O.G., we will assume that we have a complete graph $\mathcal{G}$, i.e., all the independent sets have exactly $C$ vertices. 
The algorithms and analysis can be easily extended to the more general case without breaking any guarantees, due to Lemma~\ref{lemma:turan}. 
Let us first define some terms which will allow us to talk about local optimality of partitions.

\begin{definition}
A partition $\Phi$ of $\mathcal{G}$ is called \textit{feasible} if it divides the graph into exactly $C$ cliques, each of size $K$.
\end{definition}

\begin{definition}
A pair of vertices $u_i$ and $u_j$ are called \textit{exchangeable} if both $u_i$ and $u_j$ belong to the same $U_k$ for some $k \in [K]$.
\end{definition}

\begin{definition}
A \textit{neighbor} $\Phi'$ of a partition $\Phi$ is a partition that can be obtained by swapping an exchangeable pair in $\Phi$. We denote the \textit{neighborhood} of $\Phi$ by $N(\Phi)$.
\end{definition}

\begin{definition}
We say that a feasible partition $\Phi$ is a \textit{local optimum} for label mapping if $w(\Phi) \geq w(\Phi')$ for any $\Phi' \in N(\Phi)$.
\end{definition}

\begin{lemma}
\label{lemma:neighbors}
For any feasible partition $\Phi$,
$$ \vert N(\Phi)\vert = K \cdot \binom{C}{2}. $$
\end{lemma}

\begin{proof}
To obtain a neighbor of a partition, we first choose a set $U_k$ in $K$ ways, and then choose an exchangeable pair in $U_k$ in $\binom{C}{2}$ ways. 
Thus, there are a total of $K\cdot \binom{C}{2}$ neighbors.
\end{proof}

\begin{lemma}
\label{lemma:neighbor_weights}
For any feasible partition $\Phi$, we have
$$ \frac{1}{|N(\Phi)|} \sum_{\Phi' \in N(\Phi)} \left(w(\Phi') - \frac{w(\mathcal{G})}{C}\right) = \left(1 - \frac{2C}{|N(\Phi)|} \right) \left( w(\Phi) - \frac{w(\mathcal{G})}{C}\right). $$
\end{lemma}

For a detailed proof, please refer to Appendix~\ref{sec:appendix_neighbor}.

We can now describe a deterministic local search algorithm for label mapping. 
Given the graph $\mathcal{G}$, we initialize a feasible partition $\Phi$ and then follow the steps below.

\begin{enumerate}
    \item Search $N(\Phi)$ until we find $\Phi'$ such that 
    $$\left(w(\Phi') - \frac{w(\mathcal{G})}{C}\right) \geq \left( 1 - \frac{2C}{|N(\Phi)|}\right)\left(w(\Phi) - \frac{w(\mathcal{G})}{C}\right).$$ 
    \item Let $\Phi = \Phi'$.
    \item If $w(\Phi) < \frac{w(\mathcal{G})}{C}$, find a new feasible partition $\Phi' \in N(\Phi)$ s.t. $w(\Phi') > w(\Phi)$. Let $\Phi = \Phi'$.
\end{enumerate}

We repeat these steps until we find $\Phi$ such that $w(\Phi) \geq \frac{w(\mathcal{G})}{C}$. 
Note that Lemma~\ref{lemma:neighbor_weights} guarantees the existence of $\Phi'$ satisfying the conditions in both steps 1 and 3.

\begin{theorem}
\label{thm:ls}
Starting from an arbitrary initialization, the local search algorithm reaches a feasible partition $\Phi$ with weight $w(\Phi) \geq \frac{w(\mathcal{G})}{C}$ in polynomial number of steps.
\end{theorem}

\begin{proof}
In step 1 of the algorithm, we can find a feasible partition $\Phi'$ in at most $|N(\Phi)|$ searches, which is at most $K\cdot \binom{C}{2}$. 
So we now need to count how many such iterations are required until we obtain a solution satisfying the stopping criterion. 

Let us denote $f(\Phi) = w(\Phi) - \frac{w(\mathcal{G})}{C}$, for ease of notation.
This means that in step 1, our local improvements are of the form $f(\Phi') \geq \left( 1 - \frac{2C}{|N(\Phi)|} \right)f(\Phi)$. 
Suppose we initialize with a partition $\Phi_0$, and the consecutive partitions $\Phi^{\prime}$ obtained in step 1 are denoted as $\Phi_1$, $\Phi_2$, and so on. 
Then,
$$ f(\Phi_i) \geq f(\Phi_0)\left( 1 - \frac{2C}{|N(\Phi_0)|} \right)^i. $$

This means that for sufficiently large $i$, the RHS in the inequality becomes diminishingly small, and $f(\Phi_i) \geq 0$, which is our desired stopping criterion. 
Otherwise, step 3 of the algorithm will produce a different feasible solution $\Phi'$ with weight $w(\Phi') > w(\Phi)$. 
We can also solve for $i$ to bound the number of iterations as $\log_b w(\mathcal{G}) + 1$, where $b = \frac{|N(\Phi)|}{|N(\Phi)|-2C}$, and so $i^{\ast} \leq n\ln w(\mathcal{G})$, which implies that the algorithm is polynomial if the weights are bounded. 
In our case, $w(\mathcal{G}) \leq |E|$, so the algorithm is polytime.
\end{proof}

Clearly, the local search algorithm is also a $\frac{1}{C}$-approximation, using similar arguments as in Theorem~\ref{thm:greedy}. 
Next, we will present a randomized version of this algorithm that provides a better solution in expectation.

\subsection{Extension to randomized case}

\begin{algorithm}[t]
\setstretch{1.35}
\DontPrintSemicolon
  
  \KwInput{Graph $\mathcal{G} = (V,E,w)$}
  \KwOutput{Partition $\Phi$ = $V_1,\ldots,V_C$}
  
  $\Phi = \{\}$
  
  \tcp{Repeat for $N$ epochs}
  \For{$n$ in $[N]$}{
    
    $\Tilde{\Phi}$ = Random($V$) \tcp*{Initialize a partition at random}
    
    \tcp{Repeat for $M$ iterations}
    \For{$m$ in $[M]$}{
        
        $u_p v_q$ = Sample($E(\Tilde{\Phi}^C)$, $p = \frac{w(u_p v_q)}{w(\Tilde{\Phi}^C)}$) \tcp*{Select edge}
        
        Swap $u_p$ and $u^{\prime}_p$ with probability $\frac{1}{2}$ \tcp*{Swap incident vertex}
        
        Swap $v_q$ and $v^{\prime}_q$ with probability $\frac{1}{2}$ \tcp*{Swap incident vertex}
    }
    
    $\Phi$ = max($\Phi$, $\Tilde{\Phi}$) \tcp*{Update if weight increases}
  }

\caption{Randomized local search}
\label{alg:randomized}
\end{algorithm}

To improve the approximation ratio of local search, we make two changes. 
First, instead of arbitrarily selecting a neighbor $\Phi'$ of $\Phi$ by choosing some exchangeable pair to swap, we make this selection based on some probability distribution. 
Second, we repeat the local search process for enough iterations and choose the solution with the maximum weight.
We will show that this process gives us a $(1-\epsilon)$-approximate solution with high probability.

First, we describe the selection of an exchangeable pair. 
Let $E(\Phi^C)$ denote the complement set of the matching, i.e., the set of edges which are not in the matching $\Phi$. 
Let $e = (u_{ik},u_{j\kappa})$ denote some edge in $E(\Psi^C)$, s.t. $u_{ik} \in U_k$ and $u_{j\kappa} \in U_{\kappa}$. 
We select edge $e$ with probability $\frac{w(e)}{w(\Phi^C)}$. 
This means that a higher-weighted edge not in the matching is more likely to get selected. 
Once we have selected an edge $e$, we choose either of its incident vertices $u_{ik}$ or $u_{j\kappa}$ with probability $\frac{1}{2}$. 
If $u_{ik}$ is selected, we select a vertex $v_k \in U_k$ ($v_k \neq u_{ik}$) with uniform probability distribution (and analogously for $u_{j\kappa}$). 
Then, our exchangeable pair is ($u_{ik},v_k$).

With this selection of exchangeable pair, our entire randomized local search algorithm for label mapping is shown in Algorithm~\ref{alg:randomized}. The local search runs for $M$ \textit{iterations}, and we repeat the whole process for $N$ \textit{epochs}. 
Finally, we return the partition with the maximum weight among all the $N$ epochs.

\begin{theorem}
\label{thm:randomized}
With probability $1-\frac{1}{e}$, Algorithm~\ref{alg:randomized} returns a $(1-\epsilon)$-approximate solution for
\begin{align*}
N &= \frac{(C!)^{K-1}}{\left(1 + \epsilon^2/(4(C-1)^2 (1-\epsilon)^2)\right)^{C(K-1)}} \\
\text{and} \quad M &= C(K-1) \cdot \frac{(C !)^{K-1}}{\left(1+\epsilon^{2} /\left(4(C-1)^{2}(1-\epsilon)^{2}\right)\right)^{C(K-1)}}.
\end{align*}
\end{theorem}

\begin{proof}
Consider an optimum partition $\Phi_{OPT}$ with weight $l$. If, at some moment, $w(\Phi)=u\geq (1-\epsilon)l$, then we are done; otherwise $w(\Phi)=u < (1-\epsilon)l$. Let $\Phi_{OPT}^C$ and $\Phi^C$ denote the complement edge sets of the optimal partition and the current partition, respectively. We have,
 \begin{align*}
     & u < (1-\epsilon)l \\
     \implies & w(\mathcal{G}) - u > w(\mathcal{G}) - (1-\epsilon)l \\
     \implies & w(\Phi^C) > (w(\mathcal{G}) - l) + \epsilon l \\
     \implies & w(\Phi^C) > w(\Phi_{OPT}^C) + \epsilon l \\
     \implies & w(\Phi^C) - w(\Phi_{OPT}^C) > \frac{\epsilon u}{1-\epsilon}.
 \end{align*}

Here, we used the fact that $w(\mathcal{G}) - w(\Phi) = w(\Phi^C)$, i.e., the weights of a partition and its complement add up to the total weight of the graph.
We note that any edge in the set $\Phi^C \setminus \Phi_{OPT}^C$ is \textit{wrongly} assigned, since it is present within cliques in the optimal partition. If such an edge is selected in line 5 of the algorithm, then the algorithm would correct it with probability $\frac{1}{2}\times \frac{1}{C-1}$. Since each edge is chosen with probability $\frac{w(e)}{w(\Phi^C)}$, the probability that the algorithm corrects at least one wrong entry is

\begin{align*}
    \sum_{e \in \Phi^C\setminus \Phi_{OPT}^C} \frac{1}{2(C-1)}\frac{w(e)}{w(\Phi^C)}
    &= \frac{1}{2(C-1)}\frac{w(\Phi^C \setminus \Phi_{OPT}^C)}{w(\Phi^C)} \\
    &> \frac{1}{2(C-1)}\frac{\epsilon u}{(1-\epsilon)u} \\
    &= \frac{\epsilon}{2(C-1)(1-\epsilon)}.
\end{align*}

Suppose in some epoch, we choose a initial partition $\Phi$ with $t_k$ wrongly assigned vertices in set $U_k$, for $2 \leq k \leq K$ (i.e., $t_k$ vertices are assigned to wrong cliques). This happens with probability $\prod_{k=2}^K p_k$, where
\begin{equation*}
p_k = \begin{cases}
\binom{C}{t_k}\frac{(t_k-1)!}{C!}, ~~t_k \geq 2 \\
\frac{1}{C!} ~~ t_k = 0
\end{cases}
\end{equation*}

For such a partition, the algorithm finds an optimum partition with probability at least
$$ \prod_{k=2}^K \left( \frac{\epsilon}{2(C-1)(1-\epsilon)}  \right)^{t_k}.$$

Summing over all possible choices of $(t_2,\ldots,t_K)$, we conclude that the probability of using local search to find one initial partition is at least 
\begin{equation*}
\begin{array}{l}
\frac{1}{(C!)^{K-1}} \prod_{k=2}^{K}\left(1+\sum_{t_{k}=2}^{K}\left(\binom{K}{t_k}\right)\left(t_{k}-1\right) !\left(\frac{\epsilon}{2(C-1)(1-\epsilon)}\right)^{t_{k}}\right) \\
\geq \frac{1}{(C !)^{K-1}}\left(1+\frac{\epsilon^{2}}{4(k-1)^{2}(1-\epsilon)^{2}}\right)^{C(K-1)}=\frac{1}{N}
\end{array}
\end{equation*}

Since we choose $N = \frac{(C!)^{K-1}}{\left(1 + \epsilon^2/(4(C-1)^2 (1-\epsilon)^2)\right)^{C(K-1)}}$ initial independent partitions, the probability of failing is at most
$$ (1-\frac{1}{N})^N < \frac{1}{e}.$$

Hence, we obtain a $(1-\epsilon)$-approximate solution with probability at least $1-\frac{1}{e}$.

\end{proof}
% \cleardoublepage

\printbibliography[title={References}]

\end{document}